\documentclass[12pt,a4paper,reqno]{amsart}
\usepackage{amsmath}
\usepackage{amsthm}
\usepackage{amsfonts}
\usepackage{amssymb}
\usepackage{color}
\usepackage{mathrsfs}
\usepackage{graphicx}

\numberwithin{equation}{subsection}

\theoremstyle{plain}
\newtheorem{thm}{Theorem}[subsection]

\theoremstyle{plain}
\newtheorem{lem}{Lemma}[subsection]

\theoremstyle{plain}
\newtheorem{cor}{Corollary}[subsection]

\theoremstyle{plain}
\newtheorem*{propa}{Property A}

\theoremstyle{plain}
\newtheorem*{propb}{Property B}

\theoremstyle{plain}
\newtheorem*{theorema}{Theorem A}

\theoremstyle{plain}
\newtheorem*{theoremb}{Theorem B}

\theoremstyle{plain}
\newtheorem*{theoremc}{Theorem C}

\theoremstyle{plain}
\newtheorem*{theoremd}{Theorem D}

\theoremstyle{plain}
\newtheorem*{theoreme}{Theorem E}

\theoremstyle{plain}
\newtheorem*{theorem811c}{Theorem \ref{thm8.1.1}}

\theoremstyle{remark}
\newtheorem*{remark}{Remark}

\theoremstyle{remark}
\newtheorem*{iremark}{Important remark}

\theoremstyle{remark}

\theoremstyle{remark}
\newtheorem{example}{Example}
\numberwithin{example}{subsection}

\theoremstyle{remark}
\newtheorem*{example721c}{Continuation of Example \emph{\ref{ex7.2.1}}}

\theoremstyle{definition}
\newtheorem*{part0}{The shortcut-ancestor process}

\theoremstyle{definition}
\newtheorem*{part1}{Finding a relevant $\bfitM$-invariant subspace and the edge cutting lemma}

\theoremstyle{definition}
\newtheorem*{part2}{An important subspace}

\theoremstyle{definition}
\newtheorem*{part3}{Algorithm to find the street-spreading matrix}

\theoremstyle{definition}
\newtheorem*{2maze}{Definition of the 2-power square-maze translation surface}

\def\bfu{\mathbf{u}}
\def\bfv{\mathbf{v}}
\def\bfw{\mathbf{w}}

\def\bfz{\mathbf{z}}

\def\bfA{\mathbf{A}}

\def\bfD{\mathbf{D}}

\def\bfS{\mathbf{S}}

\def\bfU{\mathbf{U}}
\def\bfV{\mathbf{V}}

\def\bfitM{\boldsymbol{M}}

\def\ii{\mathrm{i}}

\def\eps{\varepsilon}

\def\Cc{\mathbb{C}}

\def\Rr{\mathbb{R}}

\def\Zz{\mathbb{Z}}

\def\PPP{\mathcal{P}}

\def\VVV{\mathcal{V}}

\def\SSSS{\mathscr{S}}

\def\WWWW{\mathscr{W}}

\DeclareMathOperator{\length}{length}

\DeclareMathOperator{\Bil}{Bil}
\DeclareMathOperator{\LCM}{LCM}
\DeclareMathOperator{\ANC}{ANC}
\DeclareMathOperator{\image}{Im}
\DeclareMathOperator{\cp}{CP}
\DeclareMathOperator{\HS}{HS}
\DeclareMathOperator{\VS}{VS}

\renewcommand{\le}{\leqslant}
\renewcommand{\ge}{\geqslant}

\def\nrightarrow{{+\hspace{-11pt}\rightarrow}}
\def\nuparrow{{\hspace{-0.5pt}{\scriptstyle-}\hspace{-8.5pt}\uparrow}}

\title[Non-integrable systems (IV)]
{Quantitative behavior\\
of non-integrable systems (IV)}

\author[Beck]{J. Beck}

\address{Department of Mathematics, Rutgers University, Hill Center for the Mathematical Sciences, Piscataway NJ 08854, USA}

\email{jbeck@math.rutgers.edu}

\author[Chen]{W.W.L. Chen}

\address{Department of Mathematics and Statistics, Macquarie University, Sydney NSW 2109, Australia}

\email{william.chen@mq.edu.au}

\author[Yang]{Y. Yang}

\address{Department of Mathematics, Rutgers University, Hill Center for the Mathematical Sciences, Piscataway NJ 08854, USA}

\email{yy458@math.rutgers.edu}



\begin{document}

\begin{abstract}
In this paper, there are two sections.
In Section~7, we simplify the eigenvalue-based surplus shortline method for arbitrary finite polysquare translation surfaces.
This makes it substantially simpler to determine the irregularity exponents of some infinite orbits, and quicker to find the escape rate to infinity of some orbits in some infinite models.
In Section~8, our primary goal is to extend the surplus shortline method, both this eigenvalue-based version as well as the eigenvalue-free version, for application to a large class of $2$-dimensional flat dynamical systems beyond polysquares, including all Veech surfaces, and establish time-quantitative equidistribution and time-quantitative superdensity of some infinite orbits in these new systems.
\end{abstract}

\maketitle

\thispagestyle{empty}

%%%%%%%%%%
%
% SECTION 7
%
%%%%%%%%%%

\section{More on the eigenvalue-based shortline method}\label{sec7}

%%%%%%%%%%
%
% SECTION 7.1
%
%%%%%%%%%%

\subsection{The shortline method and the edge cutting lemma}\label{sec7.1}

Here in part (IV) we assume that the reader is more or less familiar with the earlier parts \cite{BDY1,BDY2,BCY}.
The eigenvalue-based surplus shortline method has been developed in the special case of the L-surface in \cite[Section~3]{BDY1} and \cite[Section~4]{BDY2}.

For a $4$-direction billiard flow in a general finite polysquare region, we can apply \textit{unfolding} introduced in \cite{BDY1} to convert the problem to one concerning a $1$-direction geodesic flow on some related finite polysquare translation surface.
We concentrate therefore on $1$-direction geodesic flow on a finite polysquare translation surface.

We show that the surplus shortcut-ancestor process, introduced earlier in \cite{BDY1,BDY2} in connection with the L-surface, can be adapted to \textit{every} finite polysquare translation surface with $1$-direction geodesic flow. 
For a quick introduction, we first illustrate the method by applying it to a $1$-direction geodesic flow with a particular quadratic irrational slope on the surface $S_2$ given in the picture on the left in Figure~7.1.1.

\begin{displaymath}
\begin{array}{c}
\includegraphics[scale=0.8]{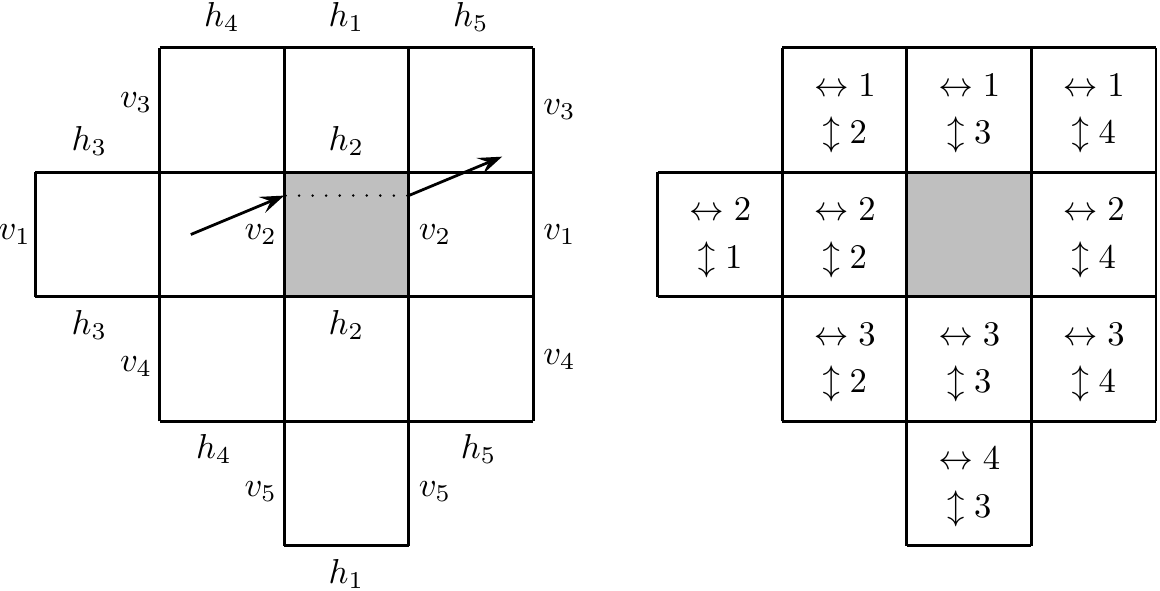}
\vspace{3pt}\\
\mbox{Figure 7.1.1: the surface $S_2$ with a gap}
\end{array}
\end{displaymath}

We compute the $2$-step transition matrix, and also determine the irregularity exponent for this particular quadratic irrational slope.
We illustrate the fact that this $2$-step transition matrix is usually \textit{very redundant}.
This in turn motivates the so-called edge cutting lemma which we formulate later.

We consider $1$-direction geodesic flow on~$S_2$.
Here the boundary pairings come from simple perpendicular translations.
However, the pair of vertical edges $v_2$ and the pair of horizontal edges $h_2$ around the missing square represent a slightly less straightforward form of perpendicular translation.

The surface $S_2$ has $1$ vertical street of length~$1$, $3$ vertical streets of length~$3$, $1$ horizontal street of length~$1$, and $3$ horizontal streets of length~$3$, so the street-LCM is equal to~$3$.
The picture on the right in Figure~7.1.1 shows the streets where, for instance, the entries $\leftrightarrow2$ and $\updownarrow4$ in a square indicates that the square face is on the $2$-nd horizontal and $4$-th vertical street.

To apply the surplus shortline method, we have to restrict our study to slopes $\alpha$ or $\alpha^{-1}$ for which the continued fraction of $\alpha$ has the special form
\begin{equation}\label{eq7.1.1}
\alpha=[3c_0;3c_1,3c_2,3c_3,\ldots]
=3c_0+\frac{1}{3c_1+\frac{1}{3c_2+\frac{1}{3c_3+\cdots}}},
\end{equation}
where the digits $c_0,c_1,c_2,c_3,\ldots$ are positive integers.

For illustration we consider here the simplest slope satisfying \eqref{eq7.1.1}, namely
\begin{equation}\label{eq7.1.2}
\alpha=[3;3,3,3,\ldots]=3+\frac{1}{3+\frac{1}{3+\frac{1}{3+\cdots}}}=\frac{3+\sqrt{13}}{2},
\end{equation}
and its reciprocal~$\alpha^{-1}$.
We study the long-term behavior of two particular geodesics $V(t)$ and $H(t)$ on $S_2$ that start from the origin, which is some chosen vertex of one of the square faces of~$S_2$.
Crucially, this guarantees that $V(t)$ and $H(t)$ are \textit{surplus shortlines} of each other.
The almost vertical geodesic $V(t)$ has slope~$\alpha$, while the almost horizontal geodesic $H(t)$ has slope $\alpha^{-1}$.
Following \cite[Section~3]{BDY1}, we briefly elaborate on the details of the surplus shortline method in this particular case.

We distinguish the $20$ types of almost vertical units $a_i$, $1\le i\le20$, in the picture on the left in Figure~7.1.2.

\begin{displaymath}
\begin{array}{c}
\includegraphics[scale=0.8]{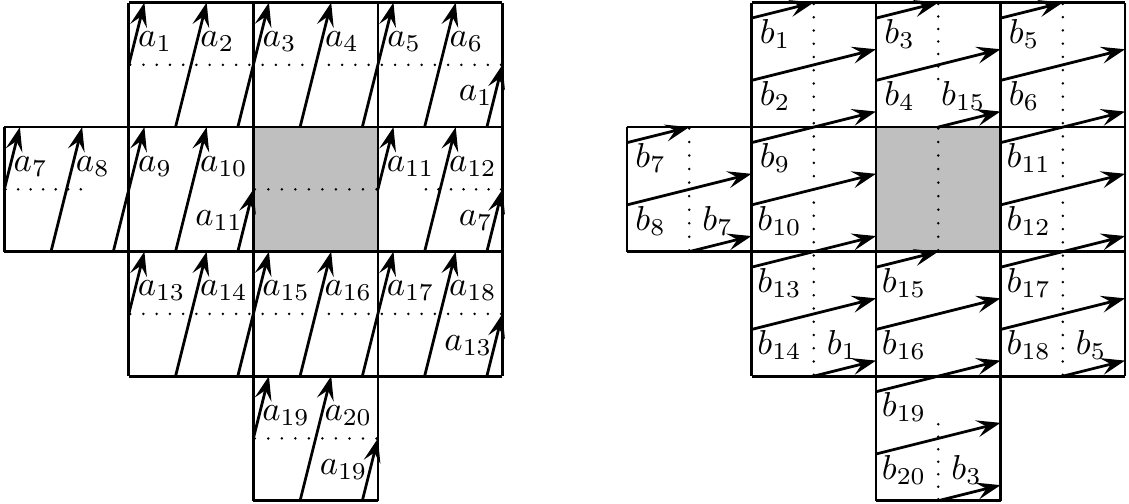}
\vspace{3pt}\\
\mbox{Figure 7.1.2: the $20$ types of almost vertical units and the $20$ types}
\\
\mbox{of almost horizontal units in the surface $S_2$}
\end{array}
\end{displaymath}

We say that each of the almost vertical units
\begin{displaymath}
a_2,a_4,a_6,a_8,a_{10},a_{12},a_{14},a_{16},a_{18},a_{20}
\end{displaymath}
is of type~$\uparrow$, and starts at the bottom edge of a square face and ends at the top edge of the same square face, while each of the almost vertical units
\begin{displaymath}
a_1,a_3,a_5,a_7,a_9,a_{11},a_{13},a_{15},a_{17},a_{19}
\end{displaymath}
is of type $\nuparrow$, and starts at the bottom edge of a square face and ends at the top edge of an adjoining square face.
Of particular interest is the unit $a_{11}$ which hits the left edge of the gap at some point and continues from the corresponding point on the identified right edge of the gap.

Similarly, we distinguish the $20$ types of almost horizontal units $b_j$, $1\le j\le20$, in the picture on the right in Figure~7.1.2.

Note that the surface $S_2$ exhibits a $45$-degree reflection symmetry. 
But symmetry is not necessary, and indeed totally irrelevant, for the success of the shortline method.
To emphasize this point, we have deliberately used labelling in the picture on the right in Figure~7.1.2 which is not a 45-degree reflection of the labelling in the picture on the left in Figure~7.1.2.

Applying a straightforward adaptation of the surplus shortline method discussed in \cite[Section~3]{BDY1}, we can determine the surplus ancestor units of~$a_i$,
$1\le i\le20$.
Similarly, we can determine the surplus ancestor units of $b_j$, $1\le j\le20$.

Consider first the almost vertical unit~$a_1$.
It is not difficult to see from Figure~7.1.3, which shows only the top horizontal street of~$S_2$, that $a_1$ is the shortcut of an almost horizontal detour crossing of the this horizontal street, made up of a fractional almost horizontal unit $b_{11}$, full almost horizontal units $b_2$, $b_4$ and~$b_6$, and then a fractional almost horizontal unit~$b_1$.
We apply the delete end rule, meaning that we keep the initial fractional almost horizontal unit $b_{11}$ as a full unit and discard the final fractional almost horizontal unit~$b_1$, and call $b_{11},b_2,b_4,b_6$ the ancestor units of~$a_1$.

\begin{displaymath}
\begin{array}{c}
\includegraphics[scale=0.8]{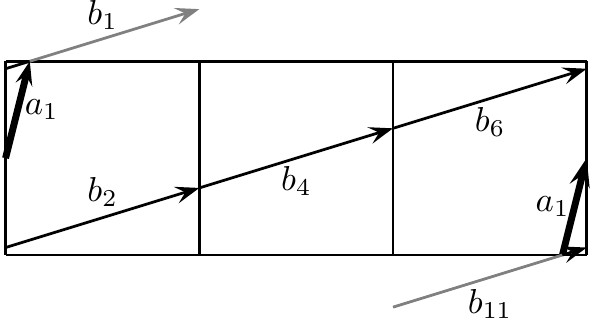}
\vspace{3pt}\\
\mbox{Figure 7.1.3: almost horizontal detour crossing for which $a_1$ is the shortcut}
\end{array}
\end{displaymath}

We do likewise for $a_i$, $2\le i\le20$, and can summarize in the form
\begin{align}\label{eq7.1.3}
a_1&\rightharpoonup b_{11},b_2,b_4,b_6,
\nonumber
\\
a_2&\rightharpoonup b_9,b_4,b_6,
\nonumber
\\
a_3&\rightharpoonup b_9,b_4,b_6,b_2,
\nonumber
\\
&\ \ \vdots
\nonumber
\\
a_{19}&\rightharpoonup b_3,b_{20},b_{20},b_{20},
\nonumber
\\
a_{20}&\rightharpoonup b_3,b_{20},b_{20}.
\end{align}
This leads to a transition matrix $M_1(S_2)$, where the $i$-th row captures the information in the $i$-th ancestor relation in \eqref{eq7.1.3}, with the entry on the $j$-th column displaying the multiplicity of~$b_j$.

A similar exercise with the roles of the horizontal and vertical interchanged, again using the delete end rule, results in
\begin{align}\label{eq7.1.4}
b_1&\rightharpoonup a_1,a_{14},a_{10},a_2,
\nonumber
\\
b_2&\rightharpoonup a_1,a_{14},a_{10},
\nonumber
\\
b_3&\rightharpoonup a_3,a_{20},a_{16},a_4,
\nonumber
\\
&\ \ \vdots
\nonumber
\\
b_{19}&\rightharpoonup a_{19},a_{16},a_4,a_{20},
\nonumber
\\
b_{20}&\rightharpoonup a_{19},a_{16},a_4.
\end{align}
This leads to an analogous transition matrix $M_2(S_2)$.

\begin{remark}
Instead of using the delete end rule, we may also use the keep end rule, meaning that we keep the final fractional unit as a full unit and discard the initial fractional unit.
Both rules are for bookkeeping purposes only, as we do not want to count any unit twice.
Depending on which rule we use, \eqref{eq7.1.3} and \eqref{eq7.1.4}, and hence also the matrices $M_1(S_1)$ and $M_2(S_2)$, may be a little different, but this will not affect the subsequent argument.
\end{remark}

Composition of the two transition matrices arising from \eqref{eq7.1.3} and \eqref{eq7.1.4} leads to the product matrix $M_1(S_2)M_2(S_2)$.
The transpose of this product matrix happens to be the $20\times20$ matrix
\begin{equation}\label{eq7.1.5}
\bfA(S_2)=\begin{pmatrix}
1\ \ 0\ \ 1\ \ 1\ \ 1\ \ 1\ \ 0\ \ 0\ \ 0\ \ 0\ \ 0\ \ 0\ \ 0\ \ 1\ \ 1\ \ 0\ \ 0\ \ 0\ \ 0\ \ 0\\
0\ \ 1\ \ 1\ \ 0\ \ 0\ \ 0\ \ 1\ \ 1\ \ 1\ \ 1\ \ 2\ \ 1\ \ 1\ \ 1\ \ 1\ \ 1\ \ 1\ \ 1\ \ 0\ \ 0\\        
1\ \ 1\ \ 1\ \ 0\ \ 1\ \ 1\ \ 0\ \ 0\ \ 0\ \ 0\ \ 0\ \ 0\ \ 0\ \ 0\ \ 0\ \ 0\ \ 0\ \ 0\ \ 1\ \ 1\\
0\ \ 0\ \ 0\ \ 1\ \ 1\ \ 0\ \ 0\ \ 0\ \ 0\ \ 0\ \ 0\ \ 0\ \ 1\ \ 1\ \ 1\ \ 1\ \ 2\ \ 1\ \ 4\ \ 3\\
1\ \ 1\ \ 1\ \ 1\ \ 1\ \ 0\ \ 0\ \ 0\ \ 0\ \ 0\ \ 0\ \ 0\ \ 1\ \ 0\ \ 0\ \ 0\ \ 0\ \ 1\ \ 0\ \ 0\\
1\ \ 0\ \ 0\ \ 0\ \ 0\ \ 1\ \ 2\ \ 1\ \ 1\ \ 1\ \ 1\ \ 1\ \ 2\ \ 1\ \ 1\ \ 1\ \ 1\ \ 1\ \ 0\ \ 0\\
0\ \ 0\ \ 0\ \ 0\ \ 0\ \ 0\ \ 1\ \ 1\ \ 2\ \ 1\ \ 1\ \ 1\ \ 0\ \ 0\ \ 0\ \ 0\ \ 0\ \ 0\ \ 0\ \ 0\\
0\ \ 0\ \ 0\ \ 0\ \ 0\ \ 0\ \ 2\ \ 3\ \ 5\ \ 2\ \ 2\ \ 2\ \ 0\ \ 0\ \ 0\ \ 0\ \ 0\ \ 0\ \ 0\ \ 0\\
0\ \ 1\ \ 1\ \ 0\ \ 0\ \ 0\ \ 1\ \ 1\ \ 1\ \ 0\ \ 1\ \ 1\ \ 0\ \ 0\ \ 0\ \ 0\ \ 0\ \ 0\ \ 0\ \ 0\\
1\ \ 1\ \ 2\ \ 1\ \ 1\ \ 1\ \ 0\ \ 0\ \ 0\ \ 1\ \ 1\ \ 0\ \ 1\ \ 1\ \ 2\ \ 1\ \ 1\ \ 1\ \ 0\ \ 0\\
1\ \ 0\ \ 0\ \ 0\ \ 0\ \ 1\ \ 1\ \ 1\ \ 1\ \ 1\ \ 1\ \ 0\ \ 0\ \ 0\ \ 0\ \ 0\ \ 0\ \ 0\ \ 0\ \ 0\\
2\ \ 1\ \ 1\ \ 1\ \ 1\ \ 1\ \ 1\ \ 0\ \ 0\ \ 0\ \ 0\ \ 1\ \ 2\ \ 1\ \ 1\ \ 1\ \ 1\ \ 1\ \ 0\ \ 0\\
0\ \ 0\ \ 0\ \ 0\ \ 0\ \ 0\ \ 0\ \ 0\ \ 0\ \ 1\ \ 1\ \ 0\ \ 1\ \ 0\ \ 1\ \ 1\ \ 1\ \ 1\ \ 0\ \ 0\\
1\ \ 1\ \ 2\ \ 1\ \ 1\ \ 1\ \ 1\ \ 1\ \ 1\ \ 1\ \ 2\ \ 1\ \ 0\ \ 1\ \ 1\ \ 0\ \ 0\ \ 0\ \ 0\ \ 0\\
0\ \ 0\ \ 0\ \ 1\ \ 1\ \ 0\ \ 0\ \ 0\ \ 0\ \ 0\ \ 0\ \ 0\ \ 1\ \ 1\ \ 1\ \ 0\ \ 1\ \ 1\ \ 0\ \ 0\\
1\ \ 1\ \ 1\ \ 1\ \ 2\ \ 1\ \ 0\ \ 0\ \ 0\ \ 0\ \ 0\ \ 0\ \ 0\ \ 0\ \ 0\ \ 1\ \ 1\ \ 0\ \ 4\ \ 3\\
0\ \ 0\ \ 0\ \ 0\ \ 0\ \ 0\ \ 1\ \ 0\ \ 0\ \ 0\ \ 0\ \ 1\ \ 1\ \ 1\ \ 1\ \ 1\ \ 1\ \ 0\ \ 0\ \ 0\\
2\ \ 1\ \ 1\ \ 1\ \ 1\ \ 1\ \ 2\ \ 1\ \ 1\ \ 1\ \ 1\ \ 1\ \ 1\ \ 0\ \ 0\ \ 0\ \ 0\ \ 1\ \ 0\ \ 0\\
0\ \ 0\ \ 0\ \ 0\ \ 0\ \ 0\ \ 0\ \ 0\ \ 0\ \ 0\ \ 0\ \ 0\ \ 0\ \ 0\ \ 0\ \ 1\ \ 1\ \ 0\ \ 3\ \ 2\\
1\ \ 1\ \ 1\ \ 1\ \ 2\ \ 1\ \ 0\ \ 0\ \ 0\ \ 0\ \ 0\ \ 0\ \ 1\ \ 1\ \ 1\ \ 1\ \ 2\ \ 1\ \ 1\ \ 1
\end{pmatrix}.
\end{equation}
We use the transpose, because in the edge cutting lemma, to be formulated later, we are interested in eigenvectors as column vectors of~$\bfA(S_2)$.
And it helps that MATLAB, like any other linear algebra computer program, automatically computes right or column eigenvectors.

We start with the $3$ eigenvalues of \eqref{eq7.1.5} with the largest absolute values.
By MATLAB, the largest eigenvalue is
\begin{equation}\label{eq7.1.6}
\lambda_1= \lambda_1(\bfA(S_2))=\frac{11+3\sqrt{13}}{2}=\left(\frac{3+\sqrt{13}}{2}\right)^2=\alpha^2,
\end{equation}
in view of \eqref{eq7.1.2}, with corresponding eigenvector
\begin{align}\label{eq7.1.7}
\bfv_1
&
=\bfv_1(\bfA(S_2))
=(v_1(1),v_1(2),v_1(3),\ldots,v_1(20))^T
\nonumber
\\
&
=(c,1,c,1,c,1,c,1,c,1,c,1,c,1,c,1,c,1,c,1)^T,
\end{align}
where
\begin{displaymath}
c=\frac{\sqrt{13}-1}{6}.
\end{displaymath}
The second largest eigenvalue is
\begin{equation}\label{eq7.1.8}
\lambda_2= \lambda_2(\bfA(S_2))=3+2\sqrt{2}=(1+\sqrt{2})^2,
\end{equation}
with corresponding eigenvector
\begin{align}\label{eq7.1.9}
\bfv_2
&
=\bfv_2(\bfA(S_2))
=(v_2(1),v_2(2),v_2(3),\ldots,v_2(20))^T
\nonumber
\\
&
=(c_1,c_2,c_1,c_3,c_1,c_2,c_4,c_5,c_4,c_6,c_4,c_6,c_1,c_2,c_1,c_3,c_1,c_2,c_7,1)^T,
\end{align}
where
\begin{align}
&
c_1=\frac{2\sqrt{2}+1}{14},
\quad
c_2=-\frac{5\sqrt{2}+6}{14},
\quad
c_3=\frac{10\sqrt{2}+19}{14},
\quad
c_4=-\frac{4\sqrt{2}+2}{7},
\nonumber
\\
&
c_5=-\frac{5\sqrt{2}+13}{7},
\quad
c_6=\frac{5\sqrt{2}-1}{14},
\quad
c_7=\frac{6\sqrt{2}+3}{7}.
\nonumber
\end{align}
The third largest eigenvalue is
\begin{equation}\label{eq7.1.10}
\lambda_3= \lambda_3(\bfA(S_2))=\frac{3+\sqrt{5}}{2},
\end{equation}
with corresponding eigenvector
\begin{align}\label{eq7.1.11}
\bfv_3
&
=\bfv_3(\bfA(S_2))
=(v_3(1),v_3(2),v_3(3),\ldots,v_3(20))^T
\nonumber
\\
&
=(c_1^*,c_2^*,c_1^*,c_3^*,c_1^*,c_2^*,c_4^*,c_5^*,c_4^*,c_6^*,c_4^*,c_6^*,c_1^*,c_2^*,c_1^*,c_3^*,c_1^*,c_2^*,c_7^*,1)^T,
\end{align}
where
\begin{align}
&
c_1^*=\frac{\sqrt{5}+5}{10},
\quad
c_2^*=\frac{\sqrt{5}}{5},
\quad
c_3^*=-\frac{2\sqrt{5}+5}{5},
\quad
c_4^*=-c_1^*,
\nonumber
\\
&
c_5^*=-\frac{4\sqrt{5}+5}{5},
\quad
c_6^*=-c_3^*,
\quad
c_7^*=-\frac{3\sqrt{5}+15}{10}.
\nonumber
\end{align}

There are $3$ other eigenvalues which are the algebraic conjugates of $\lambda_i$, $1\le i\le3$, and these are between $0$ and~$1$.
There are $14$ other eigenvalues with absolute value~$1$.
We shall see that $\lambda_i$, $1\le i\le3$, are the only relevant eigenvalues of $\bfA(S_2)$, and the other $17$ eigenvalues are irrelevant.
We comment that the Jordan normal form of the matrix $\bfA(S_2)$ is simple, as the matrix has $20$ different eigenvectors.

The irregularity exponent of this special slope $\alpha$ given by \eqref{eq7.1.2} is, by definition, equal to
\begin{displaymath}
\kappa_0(\alpha)
=\frac{\log\vert\lambda_2\vert}{\log\vert\lambda_1\vert}
=\frac{\log (1+\sqrt{2})}{\log\frac{3+\sqrt{13}}{2}}.
\end{displaymath}

Before discussing the substantial redundancy and repetition of the coordinates in \eqref{eq7.1.7}, \eqref{eq7.1.9} and \eqref{eq7.1.11}, we make a short detour and state two general theorems.

What we are doing here with the surface $S_2$ and have done with the L-surface in \cite{BDY1,BDY2} can be adapted for any finite polysquare translation surface, and we shall give some details in the next section.
In particular, we obtain the following result.

\begin{thm}\label{thm7.1.1}
Let $\PPP$ be a finite polysquare translation surface, with street-LCM, \textit{i.e.}, the least common multiple of the lengths of the horizontal and vertical streets of~$\PPP$,  denoted by $\LCM(\PPP)$.
Consider a geodesic flow with a quadratic irrational slope where the ordinary continued fraction digits are all divisible by $\LCM(\PPP)$.

Using the eigenvalue-based version of the surplus shortline method developed in \cite{BDY1,BDY2}, we can explicitly compute the irregularity exponent for such a slope.

Combining the irregularity exponent with the method of zigzagging introduced in \cite[Section~3.3]{BDY1}, we can also describe, for a geodesic flow on $\PPP$ with such a slope, the time-quantitative behavior of the edge cutting and face crossing numbers, as well as equidistribution relative to all convex sets.
\end{thm}

We remark that for the slopes in Theorem~\ref{thm7.1.1}, \cite[Theorem~6.4.1]{BCY} applies, and not only guarantees superdensity, but establishes superdensity for \textit{more} slopes than Theorem~\ref{thm7.1.1} can provide the explicit values of the irregularity exponents.
On top of the arithmetic condition that the continued fraction digits are divisible by the street-LCM, for superdensity in Theorem~6.4.1, we also need the boundedness of the continued fraction digits given by the badly approximable slopes.
Furthermore, we need in Theorem~\ref{thm7.1.1} the stronger condition of periodicity of the tail of the sequence of continued fraction digits, so that the slope is a quadratic irrational.

The irregularity exponent can be computed from the two eigenvalues with the largest absolute values of an appropriate $2d\times2d$ matrix, where $d$ is the number of square faces of the polysquare translation surface~$\PPP$. 

Since for every non-integrable polysquare translation surface the street-LCM is at least~$2$, the surplus shortline method does not have a chance of determining the irregularity exponent for every quadratic irrational slope.

The reason why in the special case of the L-surface we are able to determine the irregularity exponent for every quadratic irrational slope is two-fold.
First, the street-LCM is equal to~$2$.
More importantly, we combine the surplus and deficit versions of the shortline method according to the $\pm$-even type continued fraction of the slope. 

If the street-LCM of a polysquare region or translation surface is equal to~$2$, then we call it a \textit{$2$-polysquare} region or translation surface.
One such surface with the simplest boundary identification via perpendicular translation is called a \textit{$2$-polysquare snake surface}.
The special class of $2$-polysquare snake surfaces is surprisingly large, and Figure~7.1.4 gives an example.

\begin{displaymath}
\begin{array}{c}
\includegraphics[scale=0.8]{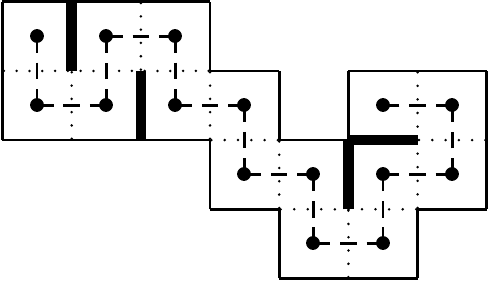}
\vspace{3pt}\\
\mbox{Figure 7.1.4: $2$-polysquare snake surface with capitals}
\\
\mbox{and its strictly alternating path}
\end{array}
\end{displaymath}

We can easily characterize every $2$-polysquare snake surface by making use of the simple graph-theoretic term of \textit{path}.

Given a $2$-polysquare snake surface, we can put a point at the center of every square face of the underlying polysquare region, and call it the \textit{capital} of the square face.
We join two capitals by a dashed line-segment if and only if the corresponding square faces share a common edge.
In this way we obtain a non self-intersecting strictly alternating h-v-path (or v-h-path) of the capitals, where ``h'' stands for a horizontal edge of unit length and ``v'' stands for a vertical edge of unit length.
The term \textit{strictly alternating} means that h is always preceded and followed by~v, and v is always preceded and followed by~h, unless the path stops.

We also have the converse, that every non self-intersecting strictly alternating h-v-path (or v-h-path) corresponds to a $2$-polysquare snake surface.
Indeed, we cover every vertex of the path with a unit size square
such that the vertex is the capital of the square, and
we may place some extra walls, denoted by boldface line-segments, if necessary.

It can be shown that the special treatment of the L-surface in \cite{BDY1,BDY2} can be adapted to \textit{every} $2$-polysquare translation surface.
The idea, as for the L-surface, is to combine the surplus and deficit versions of the eigenvalue-based shortline method according to the $\pm$-even type continued fraction expansion of the slope. 

\begin{thm}\label{thm7.1.2}
Let $\PPP$ be a finite $2$-polysquare translation surface. 
Then the LCM-divisibility condition in Theorem~\ref{thm7.1.1} may be dropped, so that the conclusion of Theorem~\ref{thm7.1.1} can be extended to every quadratic irrational slope.
\end{thm}

To explain the substantial redundancy or repetition of the coordinates in the eigenvectors \eqref{eq7.1.7}, \eqref{eq7.1.9} and \eqref{eq7.1.11}, the trick is to visualize and represent the surface $S_2$ in Figure~7.1.1 as the \textit{period} of a doubly periodic infinite polysquare translation surface $S_2(\infty)$ in Figure~7.1.5 below.

\begin{displaymath}
\begin{array}{c}
\includegraphics[scale=0.8]{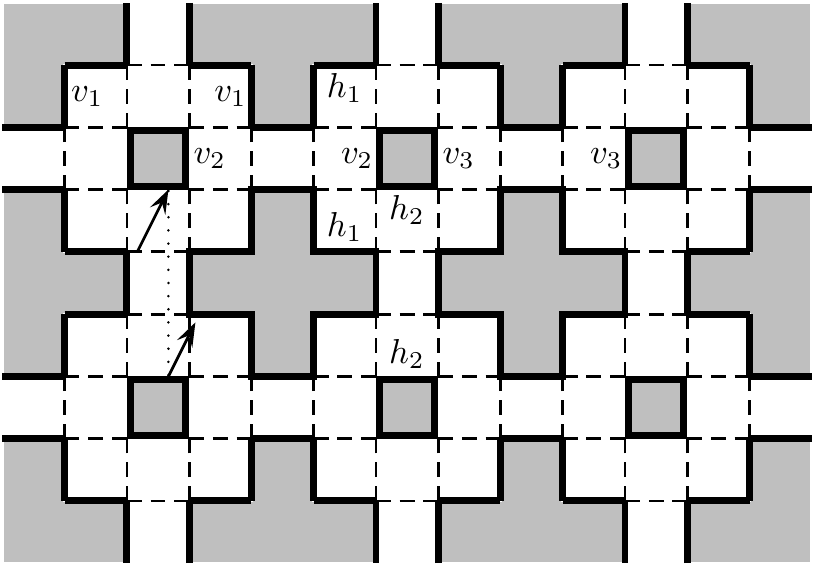}
\vspace{3pt}\\
\mbox{Figure 7.1.5: the doubly periodic polysquare translation surface $S_2(\infty)$}
\end{array}
\end{displaymath}

The infinite polysquare translation surface $S_2(\infty)$ is essentially a square lattice of copies of~$S_2$ with a noticeable difference -- observe carefully the edge pairings in~$S_2(\infty)$.
For convenience, we say that $S_2$ is the \textit{period surface} of~$S_2(\infty)$.

Following the surplus shortline method, we study two special geodesics on the infinite polysquare translation surface~$S_2(\infty)$.
The first geodesic $V_*(t)$, $t\ge0$, starts from the origin, namely a vertex of one of the square faces, and has slope $\alpha$ defined in \eqref{eq7.1.2}.
The second geodesic $H_*(t)$, $t\ge0$, also starts from the origin and has reciprocal slope~$\alpha^{-1}$.
Note that the projections of $V_*(t)$ and $H_*(t)$ on the period surface $S_2$
give back the geodesics $V(t)$ and $H(t)$ defined earlier.

A particle traveling on the almost vertical geodesic $V_*(t)$, $t\ge0$, of the doubly periodic infinite translation surface $S_2(\infty)$ moves from one copy of the period $S_2$ to another copy of $S_2$, back and forth, up and down, and generally wanders around.

Note that in~$S_2(\infty)$, each copy of $S_2$ has $4$ immediate neighbors: we shall refer to them as the North, South, East and West neighbors.
Traveling along the geodesic $V_*(t)$, $t\ge0$, is somewhat similar to a symmetric random walk on the $2$-dimensional doubly periodic
lattice~$\Zz^2$.
Escaping from a copy of $S_2$ to the North means following a unit of type $a_3$ or $a_4$ in Figure~7.1.2, while escaping to the South means a following a unit of type $a_{15}$ or~$a_{16}$.
Similarly, escaping to the East means following a unit of type~$a_7$, while escaping to the West means following a unit of type~$a_{11}$. 

To describe an arbitrary finite segment of the geodesic $V_*(t)$, $t\ge0$, of slope $\alpha$ on the infinite polysquare translation surface~$S_2(\infty)$, it suffices to consider its projection on the period surface~$S_2$.
The difference between number of North-escapes and the number of South-escapes gives the vertical change on~$S_2(\infty)$, while the difference between the number of East-escapes and the number of West-escapes gives the horizontal change on~$S_2(\infty)$.
This is how we can determine the \textit{escape rate to infinity} of the special geodesic $V_*(t)$, $t\ge0$, on $S_2(\infty)$ from the behavior of its projection $V(t)$, $t\ge0$, on the period surface~$S_2$.

Since the two special geodesics $V(t)$ and $H(t)$ are surplus shortlines of each other on the period surface~$S_2$, the time evolution of
$V(t)$, $t\ge0$, of slope $\alpha$ and starting from the origin is described by powers of the $2$-step transition matrix $\bfA(S_2)$ given by \eqref{eq7.1.5}.
Since $\bfA(S_2)$ is diagonalizable, it is particularly easy to express the number of types of units by using the relevant eigenvalues and their eigenvectors.
Since the number of types of units is expressed as a linear combination of the \textit{powers} of the eigenvalues, in case of large powers the contribution of the \textit{irrelevant} eigenvalues is clearly negligible.

Recall that escaping to the North means following a unit of type $a_3$ or $a_4$ in Figure~7.1.2, and escaping to the South means following a unit of type $a_{15}$ or~$a_{16}$.
It follows that for the vertical change, we need to know the $3$-rd, $4$-th, $15$-th and $16$-th coordinates of the relevant eigenvectors $\bfv_i$, $i=1,2,3$, in \eqref{eq7.1.7}, \eqref{eq7.1.9} and \eqref{eq7.1.11}.
Escaping to the East means following a unit of type~$a_7$, and escaping to the West means following a unit of type~$a_{11}$.
It follows that for the horizontal change, we need to know the $7$-th and $11$-th coordinates of the relevant eigenvectors. 

Recall that $v_i(j)$ denotes the $j$-th coordinate of $\bfv_i$, $i=1,2,3$ and $j=1,2,\ldots,20$.
It is easily seen from \eqref{eq7.1.7}, \eqref{eq7.1.9} and \eqref{eq7.1.11} that
\begin{equation}\label{eq7.1.12}
v_i(3)+v_i(4)=v_i(15)+v_i(16),
\quad
i=1,2,3,
\end{equation}
and
\begin{equation}\label{eq7.1.13}
v_i(7)=v_i(11),
\quad
i=1,2,3.
\end{equation}

The reader may find \eqref{eq7.1.12} and \eqref{eq7.1.13} a surprising coincidence.
However, there is a simple explanation why these equalities must hold, namely, that the geodesic $V_*(t)$, $t\ge0$, of $S_2(\infty)$ satisfies the conditions of \cite[Theorem~6.5.1]{BCY} and exhibits super-slow logarithmic escape rate to infinity; see the Remark after the proof there.
Meanwhile, a violation of \eqref{eq7.1.12} and \eqref{eq7.1.13} would imply a power-size escape rate to infinity, exponentially larger than logarithmic escape rate to infinity.

\begin{remark}
The equalities \eqref{eq7.1.12} and \eqref{eq7.1.13} remain true even if in obtaining \eqref{eq7.1.3} and \eqref{eq7.1.4} we replace the delete end rule by the keep end rule, or mix the two rules arbitrarily. 
It follows from the fact that the \textit{relevant} eigenvalues and eigenvectors remain the same under these changes.
\end{remark}

The edge cutting lemma, which we shall formulate later, is simply a far-reaching generalization of the equalities \eqref{eq7.1.12} and \eqref{eq7.1.13}.

It can be seen from \eqref{eq7.1.7}, \eqref{eq7.1.9} and \eqref{eq7.1.11} that \eqref{eq7.1.13} can be extended to include the $9$-th coordinate, so that
\begin{equation}\label{eq7.1.14}
v_i(7)=v_i(9)=v_i(11),
\quad
i=1,2,3.
\end{equation}
Note from Figure~7.1.2 that the entries in \eqref{eq7.1.14} represent almost vertical units of type $\nuparrow$ on the second horizontal street of~$S_2$, of length~$3$.
For the first horizontal street and the third horizontal street, both also of length~$3$, we have respectively the analogs
\begin{equation}\label{eq7.1.15}
v_i(1)=v_i(3)=v_i(5),
\quad
i=1,2,3,
\end{equation}
and
\begin{equation}\label{eq7.1.16}
v_i(13)=v_i(15)=v_i(17),
\quad
i=1,2,3,
\end{equation}
both of which hold, as can be seen from \eqref{eq7.1.7}, \eqref{eq7.1.9} and \eqref{eq7.1.11}.

It can also be seen from \eqref{eq7.1.7}, \eqref{eq7.1.9} and \eqref{eq7.1.11} that \eqref{eq7.1.12} can be extended to
\begin{equation}\label{eq7.1.17}
v_i(3)+v_i(4)=v_i(15)+v_i(16)=v_i(19)+v_i(20),
\quad
i=1,2,3.
\end{equation}
Note from Figure~7.1.2 that the entries in \eqref{eq7.1.17} represent almost vertical units of type $\uparrow$ or $\nuparrow$ on the third vertical street of~$S_2$, of length~$3$.
Note that we have chosen almost vertical units of type $\nuparrow$ that intersect the left edge of each square face.
In view of \eqref{eq7.1.15} and \eqref{eq7.1.16}, it would have made no difference if we had chosen instead any that intersects the right edge of the same square face.
For the second vertical street and the fourth vertical street, both also of length~$3$, we have respectively the analogs
\begin{equation}\label{eq7.1.18}
v_i(1)+v_i(2)=v_i(9)+v_i(10)=v_i(13)+v_i(14),
\quad
i=1,2,3,
\end{equation}
and
\begin{equation}\label{eq7.1.19}
v_i(5)+v_i(6)=v_i(11)+v_i(12)=v_i(17)+v_i(18),
\quad
i=1,2,3,
\end{equation}
both of which hold, as can be seen from \eqref{eq7.1.7}, \eqref{eq7.1.9} and \eqref{eq7.1.11}.

The edge cutting lemma basically says that for \textit{any} finite polysquare translation surface with $1$-direction geodesic flow, analogs of the equalities
\eqref{eq7.1.14}--\eqref{eq7.1.19} hold.
Section~\ref{sec7.2} contains a proof of this fact.
In fact, we prove even more; see Theorem~\ref{thm7.2.2}.

We conclude this section by giving a simple proof of the edge cutting lemma in the simplest special case of the L-surface. 
We include it, because this simple proof already illustrates quite well the method we use in Section~\ref{sec7.2}.

We go back to \cite[Section~3]{BDY1}, and consider the L-surface with a $1$-direction geodesic flow with slope $\alpha=1+\sqrt{2}$.
Figure~7.1.6 below shows the edges of the L-surface as well as the almost vertical units.

\begin{displaymath}
\begin{array}{c}
\includegraphics[scale=0.8]{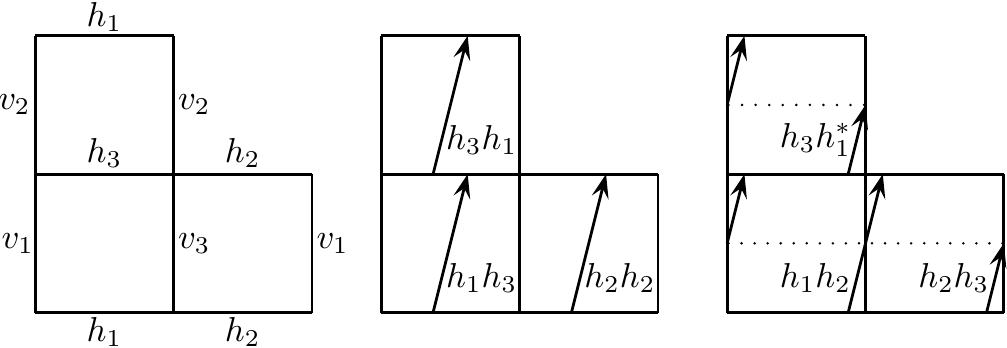}
\vspace{3pt}\\
\mbox{Figure 7.1.6: the L-surface and almost vertical units}
\end{array}
\end{displaymath}

Corresponding to the matrix $\bfA(S_2)$ for the polysquare translation surface~$S_2$, we have the $2$-step transition matrix
\begin{equation}\label{eq7.1.20}
\bfA=M_2^TM_1^T=\begin{pmatrix}
2&1&1&1&0&0\\
1&1&0&0&2&3\\
1&1&2&3&0&0\\
1&1&1&2&0&0\\
2&1&1&1&1&1\\
0&0&0&0&2&3
\end{pmatrix}.
\end{equation}
In \eqref{eq7.1.20}, we have used the lexicographic order for the almost vertical units
\begin{equation}\label{eq7.1.21}
h_1h_2,h_1h_3,h_2h_2,h_2h_3,h_3h_1,h_3h_1^*.
\end{equation}
The L-surface has precisely one vertical and one horizontal street of length greater than~$1$.
If $\bfv=(v_1,v_2,v_3,v_4,v_5,v_6)$ is a relevant eigenvector of~$\bfA$, then the coordinates correspond to the $6$ almost vertical units in \eqref{eq7.1.21}.
It follows that the L-surface analog of the equalities \eqref{eq7.1.14}--\eqref{eq7.1.16} is
\begin{equation}\label{eq7.1.22}
v_1=v_4,
\end{equation}
and the L-surface analog of the inequalities \eqref{eq7.1.17}--\eqref{eq7.1.19} is
\begin{equation}\label{eq7.1.23}
v_1+v_2=v_5+v_6.
\end{equation}
These two equations together give an invariant subspace of the matrix~$\bfA$.
To see this, suppose that
\begin{displaymath}
\bfA\bfv^T=\bfw^T=(w_1,w_2,w_3,w_4,w_5,w_6)^T.
\end{displaymath}
Then
\begin{align}
w_1&=2v_1+v_2+v_3+v_4,
\nonumber
\\
w_2&=v_1+v_2+2v_5+3v_6,
\nonumber
\\
w_3&=v_1+v_2+2v_3+3v_4,
\nonumber
\\
w_4&=v_1+v_2+v_3+2v_4,
\nonumber
\\
w_5&=2v_1+v_2+v_3+v_4+v_5+v_6,
\nonumber
\\
w_6&=2v_5+3v_6.
\nonumber
\end{align}
It is easy to see that $w_1=w_4$ follows from $v_1=v_4$.
On the other hand, we have
\begin{align}
w_1+w_2&=3v_1+2v_2+v_3+v_4+2v_5+3v_6,
\nonumber
\\
w_5+w_6&=2v_1+v_2+v_3+v_4+3v_5+4v_6,
\nonumber
\end{align}
so that
\begin{displaymath}
(w_1+w_2)-(w_5+w_6)=(v_1+v_2)-(v_5+v_6),
\end{displaymath}
and so $w_1+w_2=w_5+w_6$ follows from $v_1+v_2=v_5+v_6$.

The $2$-step transition matrix $\bfA$ has $6$ eigenvalues and $6$ independent eigenvectors in~$\Rr^6$.
The eigenvalues, in order of decreasing absolute value, are
\begin{align}\label{eq7.1.24}
\lambda_1=3+2\sqrt{2},
\quad
\lambda_2=\frac{3+\sqrt{5}}{2},
\quad
\lambda_3=1,
\nonumber
\\
\lambda_4=1,
\quad
\lambda_5=\frac{3-\sqrt{5}}{2},
\quad
\lambda_6=3-2\sqrt{2}.
\end{align}
Suppose that the corresponding eigenvectors of $\bfA$ are respectively $\bfv_i$, $1\le i\le6$.
Then
\begin{align}
&
\lambda_1\to\bfv_1
=(1,\sqrt{2},\sqrt{2},1,\sqrt{2},1)^T,
\nonumber
\\
&
\lambda_2\to\bfv_2
=\left(-\frac{1}{2},\frac{\sqrt{5}+3}{4},-\frac{\sqrt{5}}{2},-\frac{1}{2},\frac{\sqrt{5}-3}{4},1\right)^T,
\nonumber
\\
&
\lambda_3\to\bfv_3
=(0,-1,1,0,0,0)^T,
\nonumber
\\
&
\lambda_4\to\bfv_4
=(-1,1,0,0,-1,1)^T,
\nonumber
\\
&
\lambda_5\to \bfv_5
=\left(-\frac{1}{2},-\frac{\sqrt{5}-3}{4},\frac{\sqrt{5}}{2},-\frac{1}{2},-\frac{\sqrt{5}+3}{4},1\right)^T,
\nonumber
\\
&
\lambda_6\to\bfv_6
=(1,-\sqrt{2},-\sqrt{2},1,-\sqrt{2},1)^T.
\nonumber
\end{align}
Note now that the $4$-dimensional invariant subspace of the matrix $\bfA$ given by the equations \eqref{eq7.1.22} and \eqref{eq7.1.23} contains the 
eigenvectors $\bfv_1$ and $\bfv_2$ corresponding to the two relevant eigenvalues $\lambda_1$ and~$\lambda_2$.
It also contains the eigenvectors $\bfv_5$ and $\bfv_6$ of the two conjugate eigenvalues $\lambda_5$ and~$\lambda_6$.

This simple example illustrates what we call an \textit{invariant subspace} argument. 
In the next section we use a similar but more sophisticated version of this to prove the edge cutting lemma and more.

Finally observe that the edge cutting lemma has the following vague intuitive meaning.
If the surplus shortline method works for a $1$-direction geodesic on a finite polysquare translation surface with some fixed slope, then the edge cutting numbers of the edges in the \textit{same} street are \textit{nearly the same}.
We shall clarify this later.

%%%%%%%%%%
%
% SECTION 7.2
%
%%%%%%%%%%

\subsection{Reducing the 2-step transition matrix: the street-spreading matrix}\label{sec7.2}

We see from the examples in Section~\ref{sec7.1} that the $2$-step transition matrix of the shortline method is quite redundant.
Here we elaborate on this.
What we are interested in is the long-term behavior of the geodesics, and when the shortline method works, it suffices to consider the relevant eigenvalues, those with absolute value greater than~$1$, and the corresponding eigenvectors.
For example, while the $2$-step transition matrix $\bfA(S_2)$ in \eqref{eq7.1.5} has $20$ eigenvalues, only $3$ of them are relevant.
Moreover, each of the corresponding $20$-dimensional eigenvectors \eqref{eq7.1.7}, \eqref{eq7.1.9} and \eqref{eq7.1.11} has at most $8$ distinct coordinates, with some multiplicities.

Accordingly, we introduce a general reduction process that, roughly speaking, eliminates the irrelevant part of the $2$-step transition matrix~$\bfA$.
We shall give a recipe of how we can read out the relevant eigenvalues and their eigenvectors from a smaller matrix called the \textit{street-spreading matrix}.

We consider a vector space $W$ with basis made up of all the distinct types of almost vertical units of a finite polysquare translation surface~$\PPP$, of dimension equal to twice the number of distinct square faces of~$\PPP$.
We shall show that the equations in the edge cutting lemma define a subspace of~$W$, and we want to show that this subspace is $\bfA$-invariant and contains all the relevant eigenvectors.
Thus the basic idea is to find such a \textit{relevant} $\bfA$-invariant subspace, and a convenient basis of it, leading to the usually substantially smaller street-spreading matrix.
A proof of the edge cutting lemma will come as a byproduct of the reduction process.
Theorem~\ref{thm7.2.2}, the main result, will be formulated towards the end of this section.

\begin{part0}
Let $\PPP$ be a given finite polysquare translation surface, with $d$ square faces.
Let $m$ be any fixed integer multiple of the lengths of the horizontal streets of~$\PPP$, and let $n$ be any fixed integer multiple of the lengths of the vertical streets of~$\PPP$.

We adopt the following convention.

A typical square face of~$\PPP$, on the $i$-th horizontal street and $j$-th vertical street, is denoted by~$S_{i,j}$.
For any fixed~$i$, there is a finite set $J_i$ of indices $j$ such that $j\in J_i$ if and only if $S_{i,j}\subset\PPP$, so that
\begin{displaymath}
\bigcup_{j\in J_i}S_{i,j}
\end{displaymath}
is the $i$-th horizontal street of~$\PPP$.
For any fixed~$j$, there is a finite set $I_j$ of indices $i$ such that $i\in I_j$ if and only if $S_{i,j}\subset\PPP$, so that
\begin{displaymath}
\bigcup_{i\in I_j}S_{i,j}
\end{displaymath}
is the $j$-th vertical street of~$\PPP$.

\begin{remark}
Our convention here assumes that if a horizontal street and a vertical street of a polysquare translation surface intersect, then the intersection is a unique square face.
This is an over-simplification of the general situation, as it is in fact possible for two or more distinct square faces of a polysquare translation surface to lie on the same horizontal street and the same vertical street simultaneously, so that the notation $S_{i,j}$ does not allow us to distinguish between two distinct square faces that lie on the $i$-th horizontal street and the $j$-th vertical street.
However, our results remain valid in the general situation, but we have chosen to adopt the present simplification and convention here as the notation is somewhat simpler and more convenient for us to study many of the polysquare translation surfaces of interest.

Indeed, the first occasion in this paper where our notation becomes inadequate is at the end of Section~\ref{sec8.2} when we discuss the regular octagon surface, and we shall deal with the problem on a case by case basis when it arises.

Nevertheless, we make some further comments on the general case after the proof of Theorem~\ref{thm7.2.2}.
\end{remark}

We consider two types of almost vertical units in $S_{i,j}$.
Both types start from the bottom edge of~$S_{i,j}$.
However, type $\uparrow_{i,j}$ ends on the top edge of~$S_{i,j}$, whereas type $\nuparrow_{i,j}$ exits $S_{i,j}$ through the right edge.
We also consider two types of almost horizontal units in $S_{i,j}$.
Both types end on the right edge of~$S_{i,j}$.
However, type $\rightarrow_{i,j}$ starts from the left edge of~$S_{i,j}$, whereas type $\nrightarrow_{i,j}$ enters $S_{i,j}$ through the bottom edge.

\begin{displaymath}
\begin{array}{c}
\includegraphics[scale=0.8]{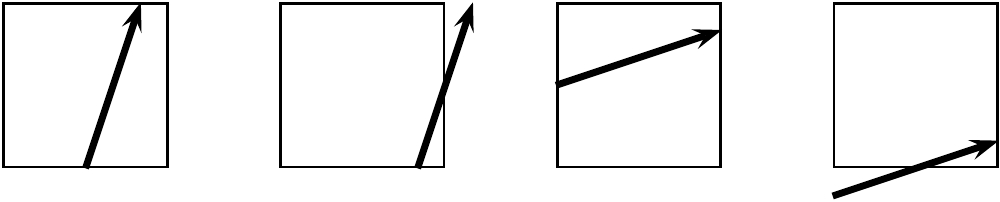}
\vspace{3pt}\\
\mbox{Figure 7.2.1: the square face $S_{i,j}$, with units of types $\uparrow_{i,j},\nuparrow_{i,j},\rightarrow_{i,j},\nrightarrow_{i,j}$}
\end{array}
\end{displaymath}

We now consider a $1$-direction almost vertical geodesic $V_0$ on~$\PPP$, starting at a vertex of $\PPP$ and with slope
\begin{equation}\label{eq7.2.1}
\alpha=[n;m,n,m,\ldots]=n+\frac{1}{m+\frac{1}{n+\frac{1}{m+\cdots}}},
\end{equation}
so that $n\le\alpha<n+1$.
The geodesic $V_0$ is made up of almost vertical units of type $\uparrow_{i,j}$ and $\nuparrow_{i,j}$, all with slope~$\alpha$.
Thus we may consider a vector space~$W$, with basis
\begin{equation}\label{eq7.2.2}
\WWWW=\{\uparrow_{i,j},\nuparrow_{i,j}:S_{i,j}\subset\PPP\}.
\end{equation}
The shortline of $V_0$ is a $1$-direction almost horizontal geodesic $H_0$ on~$\PPP$, starting at the same vertex of $\PPP$ and with
slope~$\alpha_1^{-1}$, where
\begin{displaymath}
\alpha_1=[m;n,m,\ldots]=m+\frac{1}{n+\frac{1}{m+\cdots}},
\end{displaymath}
so that $m\le\alpha_1<m+1$.
The geodesic $H_0$ is made up of almost horizontal units of type $\rightarrow_{i,j}$ and $\nrightarrow_{i,j}$, all with slope~$\alpha_1^{-1}$.
Thus we may consider a vector space~$W'$, with basis
\begin{displaymath}
\WWWW'=\{\rightarrow_{i,j},\nrightarrow_{i,j}:S_{i,j}\subset\PPP\}.
\end{displaymath}
The shortline of $H_0$ is the original $1$-direction almost vertical geodesic $V_0$ on~$\PPP$, with slope~$\alpha$.
Thus we return to the vector space~$W$, with basis~$\WWWW$.

\begin{remark}
Note that $V_0$ and $H_0$ are mutual shortlines.
Note the special property of~$\alpha$, that its continued fraction expansion has period~$2$.
\end{remark}

Our task is to start with $\WWWW$ and identify the $2$-step ancestors in $\WWWW$ of each of the almost vertical units.
Let $\ANC$ denote one step of this ancestor process.

\emph{Step~1}.\ We consider each of the almost vertical units in~$\WWWW$, and determine its ancestors in $\WWWW'$ using the delete end rule, in the same spirit as in \eqref{eq7.1.3}.
This leads to a transition matrix~$M_1$, where each row captures the information concerning the ancestors of a particular almost vertical unit, with the columns displaying the multiplicities of the individual ancestor almost horizontal units.
Indeed, $M_1^T$ is the matrix for this ancestor process $W\to W'$, where the coefficient vectors are taken as column vectors.
Note that $M_1^T$ is a $2d\times2d$ matrix.

Thus for every $i_0,j_0$ such that $S_{i_0,j_0}\subset\PPP$, using the delete end rule, we have the ancestor relationships
\begin{align}
\ANC(\{\uparrow_{i_0,j_0}\})
&
=\{\nrightarrow_{i_0,j_0}\}
\cup\{\rightarrow_{i_0,j}:j\in J_{i_0}^*\}
\setminus\{\rightarrow_{i_0,j_0}\},
\label{eq7.2.3}
\\
\ANC(\{\nuparrow_{i_0,j_0}\})
&
=\{\nrightarrow_{i_0,j_0}\}
\cup\{\rightarrow_{i_0,j}:j\in J_{i_0}^*\},
\label{eq7.2.4}
\end{align}
where the $*$ in $J_{i_0}^*$ denotes that each edge with indices $i_0,j$ with $j\in J_{i_0}$ is counted with multiplicity $m\vert J_{i_0}\vert^{-1}$, where $\vert J_{i_0}\vert$ denotes the number of distinct elements of the set~$J_{i_0}$.
We also adopt the convention that unions and complements are taken with appropriate multiplicities.

\emph{Step~2}.\ We consider each of the almost horizontal units in~$\WWWW'$, and determine its ancestors in $\WWWW$ using the keep end rule, so not quite in the same spirit as in \eqref{eq7.1.4}.
This leads to a transition matrix~$M_2$, where each row captures the information concerning the ancestors of a particular almost horizontal unit, with the columns displaying the multiplicities of the individual ancestor almost vertical units.
Indeed, $M_2^T$ is the matrix for this ancestor process $W'\to W$, where the coefficient vectors are taken as column vectors.
Note that $M_2^T$ is a $2d\times2d$ matrix.

Thus for every $i_1,j_1$ such that $S_{i_1,j_1}\subset\PPP$, using the keep end rule, we have the ancestor relationships
\begin{align}
\ANC(\{\rightarrow_{i_1,j_1}\})
&
=\{\nuparrow_{i_1,j_1}\}
\cup\{\uparrow_{i,j_1}:i\in I_{j_1}^*\}
\setminus\{\uparrow_{i_1,j_1}\},
\label{eq7.2.5}
\\
\ANC(\{\nrightarrow_{i_1,j_1}\})
&
=\{\nuparrow_{i_1,j_1}\}
\cup\{\uparrow_{i,j_1}:i\in I_{j_1}^*\},
\label{eq7.2.6}
\end{align}
where the $*$ in $I_{j_1}^*$ denotes that each edge of $I_{j_1}$ is counted with multiplicity $n\vert I_{j_1}\vert^{-1}$, where $\vert I_{j_1}\vert$ denotes the number of distinct elements of the set~$I_{j_1}$.
We also adopt the convention that unions and complements are taken with appropriate multiplicities.

\emph{Step~3}.\ Finally, we combine the two steps and end up with a $2$-step transition matrix $\bfA=(M_1M_2)^T$, of size $2d\times2d$, in the same spirit as in \eqref{eq7.1.5}.
Thus combining the ancestor relationships \eqref{eq7.2.3}--\eqref{eq7.2.6}, we deduce that
\begin{align}
&
\ANC(\ANC(\{\uparrow_{i_0,j_0}\}))
\nonumber
\\
&\quad
=\left(
\{\nuparrow_{i_0,j_0}\}
\cup\{\uparrow_{i,j_0}:i\in I_{j_0}^*\}
\right)
\nonumber
\\
&\quad\qquad
\cup\left(
\{\nuparrow_{i_0,j}:j\in J_{i_0}^*\}
\cup\{\uparrow_{i,j}:j\in J_{i_0}^*,i\in I_j^*\}
\setminus\{\uparrow_{i_0,j}:j\in J_{i_0}^*\}
\right)
\nonumber
\\
&\quad\qquad
\setminus\left(
\{\nuparrow_{i_0,j_0}\}
\cup\{\uparrow_{i,j_0}:i\in I_{j_0}^*\}
\setminus\{\uparrow_{i_0,j_0}\}
\right)
\nonumber
\\
&\quad
=\{\uparrow_{i_0,j_0}\}
\cup\{\uparrow_{i,j}:j\in J_{i_0}^*,i\in I_j^*\}
\cup\{\nuparrow_{i_0,j}:j\in J_{i_0}^*\}
\setminus\{\uparrow_{i_0,j}:j\in J_{i_0}^*\},
\label{eq7.2.7}
\\
&
\ANC(\ANC(\{\nuparrow_{i_0,j_0}\}))
\nonumber
\\
&\quad
=\left(
\{\nuparrow_{i_0,j_0}\}
\cup\{\uparrow_{i,j_0}:i\in I_{j_0}^*\}
\right)
\nonumber
\\
&\quad\qquad
\cup\left(
\{\nuparrow_{i_0,j}:j\in J_{i_0}^*\}
\cup\{\uparrow_{i,j}:j\in J_{i_0}^*,i\in I_j^*\}
\setminus\{\uparrow_{i_0,j}:j\in J_{i_0}^*\}
\right)
\nonumber
\\
&\quad
=\{\nuparrow_{i_0,j_0}\}\cup\{\uparrow_{i,j}:j\in J_{i_0}^*,i\in I_j^*\}
\nonumber
\\
&\quad\qquad
\cup\{\uparrow_{i,j_0}:i\in I_{j_0}^*\}
\cup\{\nuparrow_{i_0,j}:j\in J_{i_0}^*\}
\setminus\{\uparrow_{i_0,j}:j\in J_{i_0}^*\}.
\label{eq7.2.8}
\end{align}

\begin{remark}
We choose this particular convention regarding the delete end rule and keep end rule, as it makes the matrix reduction particularly simple.
For instance, if we use this particular convention for the surface~$S_2$, then in the $2$-step transition matrix~$\bfA$, the $14$ eigenvalues of absolute value $1$ are all the same and equal to~$1$.
Indeed, we shall give a good reason later, when we discuss a simpler approach to the process, why they are all equal to~$1$.
\end{remark}

Suppose that $\lambda_1,\ldots,\lambda_s$ are the eigenvalues of the $2$-step transition matrix~$\bfA$, with multiplicities $d_1,\ldots,d_s$ respectively, and where $\vert\lambda_1\vert\ge\ldots\ge\vert\lambda_s\vert$.
Then clearly $d_1+\ldots+d_s=2d$.
Furthermore, the space $\Cc^{2d}$ can be decomposed into a direct sum
\begin{displaymath}
\Cc^{2d}=W_1\oplus\ldots\oplus W_s,
\end{displaymath}
where each $W_i$, $i=1,\ldots,s$, is an $\bfA$-invariant subspace of $\Cc^{2d}$ and also contains an eigenvector $\Psi_i$ corresponding to the
eigenvalue~$\lambda_i$.
If $d_i=1$, then $\Psi_i$ generates $W_i$ and gives rise to a basis of~$W_i$.
If $d_i>1$, then we can find a basis $\Psi_{i,j}$, $j=1,\ldots,d_i$, of~$W_i$, with $\Psi_i=\Psi_{i,1}$.
Thus the collection
\begin{displaymath}
\Psi_{i,j},
\quad
i=1,\ldots,s,
\quad
j=1,\ldots,d_i,
\end{displaymath}
gives rise to a basis of~$\Cc^{2d}$.

The almost vertical geodesic $V_0$ starts at a vertex of~$\PPP$, and it starts with a finite succession of almost vertical units.
Let $\bfw_0$ denote the column coefficient vector of this finite collection of units with respect to the basis~$\WWWW$.
Then we can find coefficients $c_{i,j}\in\Cc$, $i=1,\ldots,s$, $j=1,\ldots,d_i$, such that
\begin{equation}\label{eq7.2.9}
\bfw_0=\sum_{i=1}^s\sum_{j=1}^{d_i}c_{i,j}\Psi_{i,j}.
\end{equation}
Suppose that $\bfw_r=\bfA^r\bfw_0$.
Then
\begin{equation}\label{eq7.2.10}
\bfw_r=\sum_{i=1}^s\sum_{j=1}^{d_i}c_{i,j}\bfA^r\Psi_{i,j}.
\end{equation}

We shall consider the special case where for every eigenvalue $\lambda_i$ of $\bfA$ with $\vert\lambda_i\vert>1$, the $\bfA$-invariant subspace $W_i$ has a basis consisting entirely of eigenvectors of $\bfA$ with eigenvalue~$\lambda_i$.
Let $\lambda_i$, $i=1,\ldots,s_0$, be these eigenvalues.
Then
\begin{displaymath}
\Psi_{i,j},
\quad
i=1,\ldots,s_0,
\quad
j=1,\ldots,d_i,
\end{displaymath}
are all eigenvectors of~$\bfA$.
Furthermore, $\vert\lambda_i\vert\le1$, $i=s_0+1,\ldots,s$.

In this case, we have
\begin{equation}\label{eq7.2.11}
\bfw_r=\sum_{i=1}^{s_0}\sum_{j=1}^{d_i}c_{i,j}\lambda_i^r\Psi_{i,j}+O(r^{D-1}),
\end{equation}
where $D=\max\{d_1,\ldots,d_s\}$, and the error term $O(r^{D-1})$ gives an upper bound on the absolute value of the coordinates of the missing vectors.
To see this, note that the matrix $\bfA$ is similar to a matrix of the form
\begin{displaymath}
\begin{pmatrix}
J(\tau_1,e_1)\\
&\ddots\\
&&J(\tau_u,e_u)
\end{pmatrix},
\end{displaymath}
where a typical Jordan block $J=J(\tau,e)$ is an $e\times e$ matrix of the form
\begin{displaymath}
J(\tau,e)=\begin{pmatrix}
\tau&1\\
&\ddots&\ddots\\
&&\ddots&1\\
&&&\tau
\end{pmatrix},
\end{displaymath}
where every entry on the diagonal is equal to an eigenvalue $\tau$ of~$\bfA$, every entry on the superdiagonal is equal to~$1$, and every other entry is equal to~$0$.

The contribution to the error term in \eqref{eq7.2.11} comes from those eigenvalues of $\bfA$ with absolute value at most~$1$.
These are related to Jordan blocks $J=J(\tau,e)$ with $\vert\tau\vert\le1$ and $e\le D$.
The matrix~$J^r$ is $e\times e$ and upper-triangular, and the entry on row $i$ and column~$j$, with $j\ge i$, is given by
\begin{displaymath}
\binom{r}{j-i}\tau^{r-(j-i)},
\end{displaymath}
where the binomial coefficient is equal to $0$ if $j-i>r$.
The bound $O(r^{D-1})$ follows on observing that
\begin{displaymath}
\binom{r}{j-i}\le r^{e-1}\le r^{D-1}
\quad\mbox{and}\quad
\vert\tau^{r-(j-i)}\vert\le1.
\end{displaymath}

\begin{iremark}
It is crucial to bring those eigenvalues with the largest absolute values into play.
The coefficient in \eqref{eq7.2.11} of the eigenvector corresponding to the largest eigenvalue is non-zero, since the left hand side of \eqref{eq7.2.11} contains terms of order of magnitude~$\lambda_1^r$.
For the coefficient of the eigenvector corresponding to the second largest eigenvalue, this is less obvious.
The trick here is to start with an almost vertical unit from the chosen vertex.
If the coefficient corresponding to the second eigenvector is non-zero, then we have a \textit{good} initial vector~$\bfw_0$.
If the coefficient is zero, then we add the next unit, and keep on doing so but stop as soon as we get a non-zero coefficient.
Take this as our starting succession of almost vertical units, and use the corresponding good~$\bfw_0$.
This process must stop after bounded time, depending only on $\PPP$ and the slope~$\alpha$, as soon every basis element in $\WWWW$ will come into play.
Of the $2d$ types of almost vertical units, it is clear that at least one of them gives rise to a non-zero coefficient.
\end{iremark}
\end{part0}

\begin{part1}
The above method in rather laborious.
Finding the eigenvalues and eigenvectors of the $2$-step transition matrix~$\bfA$, of size $2d\times2d$, is not a very pleasant task.

It is clear that the expression \eqref{eq7.2.11} is dominated by the terms arising from those eigenvalues $\lambda_i$, $i=1,\ldots,s$, with absolute values exceeding~$1$.
We call these the relevant eigenvalues, as the remaining eigenvalues make little to no contribution.
We next seek a simpler way of finding these relevant eigenvalues and their corresponding eigenvectors.

For any set $\SSSS$ of almost vertical units in~$\PPP$, counted with multiplicity, let $[\SSSS]$ denote the column coefficient vector of $\SSSS$ with respect to the
basis~$\WWWW$.
Note that $[\SSSS]\in\Cc^{2d}$, where $d$ is the number of square faces of~$\PPP$.

As the matrix $\bfA$ is the transition matrix of the $2$-step ancestor process with respect to the basis~$\WWWW$, \eqref{eq7.2.7} and
\eqref{eq7.2.8} can be rewritten in the form
\begin{align}
(\bfA-I)[\{\uparrow_{i_0,j_0}\}]
&
=[\{\uparrow_{i,j}:j\in J_{i_0}^*,i\in I_j^*\}]
\nonumber
\\
&\qquad
+[\{\nuparrow_{i_0,j}:j\in J_{i_0}^*\}]
-[\{\uparrow_{i_0,j}:j\in J_{i_0}^*\}],
\label{eq7.2.12}
\\
(\bfA-I)[\{\nuparrow_{i_0,j_0}\}]
&
=[\{\uparrow_{i,j}:j\in J_{i_0}^*,i\in I_j^*\}]
+[\{\uparrow_{i,j_0}:i\in I_{j_0}^*\}]
\nonumber
\\
&\qquad
+[\{\nuparrow_{i_0,j}:j\in J_{i_0}^*\}]
-[\{\uparrow_{i_0,j}:j\in J_{i_0}^*\}].
\label{eq7.2.13}
\end{align}
Write
\begin{align}
\bfu_{i_0}
&=[\{\uparrow_{i,j}:j\in J_{i_0}^*,i\in I_j^*\}],
\label{eq7.2.14}
\\
\bfv_{i_0}
&=[\{\nuparrow_{i_0,j}:j\in J_{i_0}^*\}]-[\{\uparrow_{i_0,j}:j\in J_{i_0}^*\}],
\label{eq7.2.15}
\\
\bfz_{j_0}
&=[\{\uparrow_{i,j_0}:i\in I_{j_0}^*\}].
\label{eq7.2.16}
\end{align}
Then $\bfu_{i_0},\bfv_{i_0},\bfz_{j_0}\in\Cc^{2d}$, and \eqref{eq7.2.12} and \eqref{eq7.2.13} can be expressed in the form
\begin{align}
(\bfA-I)[\{\uparrow_{i_0,j_0}\}]
&=\bfu_{i_0}+\bfv_{i_0},
\label{eq7.2.17}
\\
(\bfA-I)[\{\nuparrow_{i_0,j_0}\}]
&=\bfu_{i_0}+\bfv_{i_0}+\bfz_{j_0}.
\label{eq7.2.18}
\end{align}

Motivated by the equations \eqref{eq7.2.17} and \eqref{eq7.2.18}, we consider $\image(\bfA-I)$, the image of the matrix $\bfA-I$ in $\Cc^{2d}$.
Note that this is an $\bfA$-invariant subspace of~$\Cc^{2d}$.
Indeed, for every vector $\bfw\in\Cc^{2d}$, we have
\begin{displaymath}
\bfA((\bfA-I)\bfw)
=(\bfA^2-\bfA)\bfw
=(\bfA-I)(\bfA\bfw)
=(\bfA-I)\bfw'
\end{displaymath}
where $\bfw'\in\Cc^{2d}$ satisfies $\bfw'=\bfA\bfw$.
It is clear from the equations \eqref{eq7.2.17} and \eqref{eq7.2.18} that the subspace $\image(\bfA-I)$ is generated by the elements of the form
\begin{equation}\label{eq7.2.19}
\bfu_{i_0},\bfv_{i_0},\bfz_{j_0}\in\image(\bfA-I)\subset\Cc^{2d}.
\end{equation}

\begin{lem}\label{lem7.2.1}
All eigenvectors $\lambda\ne1$ of $\bfA$ belong to the subspace $\image(\bfA-I)\subset\Cc^{2d}$.
\end{lem}

\begin{proof}
Suppose that $\bfw\in\Cc^{2d}$ satisfies $\bfA\bfw=\lambda\bfw$ with $\lambda\ne1$.
Then
\begin{displaymath}
(\lambda-1)\bfw=(\bfA-I)\bfw\in\image(\bfA-I),
\end{displaymath}
so that $\bfw\in\image(\bfA-I)$.
\end{proof}

\begin{thm}[``edge cutting lemma'']\label{thm7.2.1}
Let $\PPP$ be a finite polysquare translation surface.
Let $m$ be any fixed integer multiple of the lengths of the horizontal streets of~$\PPP$, and let $n$ be any fixed integer multiple of the lengths of the vertical streets of~$\PPP$.
Consider a $1$-direction almost vertical geodesic $V_0$ on~$\PPP$, with slope $\alpha$ given by \eqref{eq7.2.1}.
Let $\bfA$ denote the transition matrix of the $2$-step ancestor process with respect to the basis~$\WWWW$, where $\WWWW$ is given by \eqref{eq7.2.2}.
Then the following hold:

\emph{(i)}
For any two distinct almost vertical units $\nuparrow_{i_0,j_1}$ in a square face $S_{i_0,j_1}$ and $\nuparrow_{i_0,j_2}$ in a square face $S_{i_0,j_2}$ on the same horizontal street of $\PPP$ and every eigenvector corresponding to a relevant eigenvalue of~$\bfA$, the two entries in the eigenvector corresponding to $\nuparrow_{i_0,j_1}$ and $\nuparrow_{i_0,j_2}$ are equal.

\emph{(ii)}
For any two distinct pairs of almost vertical units $\uparrow_{i_1,j_0},\nuparrow_{i_1,j_0}$ in a square face $S_{i_1,j_0}$ and
$\uparrow_{i_2,j_0},\nuparrow_{i_2,j_0}$ in a square face $S_{i_2,j_0}$ on the same vertical street of $\PPP$ and every eigenvector corresponding to a relevant eigenvalue of~$\bfA$, the sum of the two entries in the eigenvector corresponding to $\uparrow_{i_1,j_0},\nuparrow_{i_1,j_0}$ in the square face $S_{i_1,j_0}$ is equal to the sum of the two entries in the eigenvector corresponding to $\uparrow_{i_2,j_0},\nuparrow_{i_2,j_0}$ in the square face $S_{i_2,j_0}$.
\end{thm}

\begin{proof}
In view of Lemma~\ref{lem7.2.1}, it is enough to concentrate our attention on the subspace $\image(\bfA-I)$.
Since this subspace is generated by the vectors in \eqref{eq7.2.19}, it suffices to check that these generating vectors all satisfy (i) and~(ii).
Our argument is reduced to a simple case study.

Note that $\bfu_{i_0}$ does not involve units of type $\nuparrow_{i,j}$, so (i) holds by default.
On the other hand, $\bfu_{i_0}$ satisfies~(ii), due to cyclic vertical symmetry.
A similar argument applies to~$\bfz_{j_0}$.

On the other hand, $\bfv_{i_0}$ satisfies~(i), due to cyclic horizontal symmetry.
Clearly it also satisfies~(ii), as the count for each $\nuparrow_{i_0,j}$ cancels the count for~$\uparrow_{i_0,j}$.

This completes the proof of the edge cutting lemma.
\end{proof}
\end{part1}

\begin{part2}
Let $h=h(\PPP)$ and $v=v(\PPP)$ denote respectively the number of horizontal and vertical streets in~$\PPP$.

Let $\VVV$ be the subspace of $\Cc^{2d}$ generated by $\bfu_i,\bfv_i$, $1\le i\le h$, defined in \eqref{eq7.2.14} and \eqref{eq7.2.15}, so that the dimension of $\VVV$ is at most~$2h$.

\begin{lem}\label{lem7.2.2}
We have $\image((\bfA-I)^2)\subset\VVV$.
\end{lem}

\begin{proof}
Since $\image((\bfA-I)^2)=(\bfA-I)(\image(\bfA-I))$, and the subspace ${\rm Im}(\bfA-I)$ is generated by the vectors in \eqref{eq7.2.19}, it suffices to show that
\begin{equation}\label{eq7.2.20}
(\bfA-I)\bfu_{i_0}\in\VVV,
\quad
(\bfA-I)\bfv_{i_0}\in\VVV,
\quad
(\bfA-I)\bfz_{j_0}\in\VVV.
\end{equation}
From \eqref{eq7.2.14} and \eqref{eq7.2.16}, it is clear that both $\bfu_{i_0}$ and $\bfz_{j_0}$ involve only the almost vertical units of various types
$\uparrow_{i,j}$, so it follows from \eqref{eq7.2.17} that the first and third statements in \eqref{eq7.2.20} hold.
On the other hand, it follows from \eqref{eq7.2.14}--\eqref{eq7.2.18} that
\begin{equation}\label{eq7.2.21}
(\bfA-I)\bfv_{i_0}
=\sum_{j\in J_{i_0}^*}\bfz_j
=[\{\uparrow_{i,j}:i\in I_j^*,j\in J_{i_0}^*\}]
=\bfu_{i_0},
\end{equation}
and the second statement in \eqref{eq7.2.20} follows immediately.
\end{proof}

By definition $\VVV$ is a subspace of $\image(\bfA-I)$.
Thus the following lemma is an extension of Lemma~\ref{lem7.2.1}.

\begin{lem}\label{lem7.2.3}
The vector space $\VVV$ is an $\bfA$-invariant subspace of~$\Cc^{2d}$, and contains the eigenvectors corresponding to each of the eigenvalues
$\lambda\ne1$ of~$\bfA$.
\end{lem}

\begin{proof}
To prove that $\VVV$ is $\bfA$-invariant, it suffices to check that
\begin{displaymath}
\bfA\bfu_i\in\VVV
\quad\mbox{and}\quad
\bfA\bfv_i\in\VVV.
\end{displaymath}
This is trivial, since we know from \eqref{eq7.2.20} that $(\bfA-I)\bfu_i\in\VVV$ and $(\bfA-I)\bfv_i\in\VVV$.
Next suppose that $\bfw\in\Cc^{2d}$ satisfies $\bfA\bfw=\lambda\bfw$ for some $\lambda\ne1$.
Then
\begin{displaymath}
(\lambda-1)^2\bfw=(\bfA-I)^2\bfw\in\image((\bfA-I)^2),
\end{displaymath}
so that $\bfw\in\image((\bfA-I)^2)\subset\VVV$, 
where in the last step we have used Lemma~\ref{lem7.2.2}.
\end{proof}

Next we study how the matrix $\bfA$ acts on~$\VVV$.
First of all, the equation \eqref{eq7.2.21} implies
\begin{equation}\label{eq7.2.22}
\bfA\bfv_i=\bfu_i+\bfv_i.
\end{equation}
We already know that $\bfA\bfu_i\in\VVV$, so that $\bfA\bfu_i$ is a linear combination of the vectors $\bfu_i,\bfv_i$, $1\le i\le h$.
Next we shall find these coefficients explicitly, but this process is much more complicated than \eqref{eq7.2.22}, as the coefficients depend heavily on the combinatorial arrangement of the square faces of the given polysquare translation surface~$\PPP$.
In fact, they depend on the given triple $(\PPP;m,n)$.
This will eventually lead us to the desired $h\times h$ street-spreading matrix $\bfS=\bfS(\PPP;m,n)$ which contains all the \textit{relevant information}.

Combining \eqref{eq7.2.14} and \eqref{eq7.2.17}, one can determine the desired coefficients by using a particular algorithm.
The abstract/formal definition of this algorithm in the general case is somewhat complicated.
As with most algorithms, the best way to learn it is to study a few concrete examples with figures.
We strongly recommend the reader to study the next two examples, after which the abstract/formal definition becomes almost self-explanatory.
\end{part2}

\begin{part3}
The algorithm can easily be summarized in a nutshell: \textit{spread the horizontal streets vertically}.
We illustrate part of the recipe by two examples.

\begin{example}\label{ex7.2.1}
Consider again the polysquare translation surface $S_2$ in Figure~7.1.1, and consider a $1$-direction geodesic starting from some vertex of $S_2$ with slope $\alpha$ given by \eqref{eq7.2.1} with $m=n=3$.

We shall number the horizontal streets from top to bottom, and the vertical streets from left to right.

Consider the first horizontal street at the top, corresponding to~$\bfu_1$.
We now use \eqref{eq7.2.14}, and since $J_1^*=\{2,3,4\}$, it follows that $\bfu_1$ concerns almost vertical units of type $\uparrow$ on the vertical streets $2,3,4$.
In other words, we start with the first horizontal street, and spread along every vertical street that intersects this horizontal street.
The picture on the left in Figure~7.2.2 shows a listing of all the almost vertical units of type $\uparrow$ along these vertical streets.
Applying \eqref{eq7.2.17}, we have
\begin{align}\label{eq7.2.23}
(\bfA-I)\bfu_1
&
=(\bfA-I)[\{\uparrow_{1,2},\uparrow_{1,3},\uparrow_{1,4}\}]
+(\bfA-I)[\{\uparrow_{2,2},\uparrow_{2,4}\}]
\nonumber
\\
&\qquad
+(\bfA-I)[\{\uparrow_{3,2},\uparrow_{3,3},\uparrow_{3,4}\}]
+(\bfA-I)[\{\uparrow_{4,3}\}]
\nonumber
\\
&
=3(\bfu_1+\bfv_1)+2(\bfu_2+\bfv_2)+3(\bfu_3+\bfv_3)+(\bfu_4+\bfv_4).
\end{align}
\begin{displaymath}
\begin{array}{c}
\includegraphics[scale=0.8]{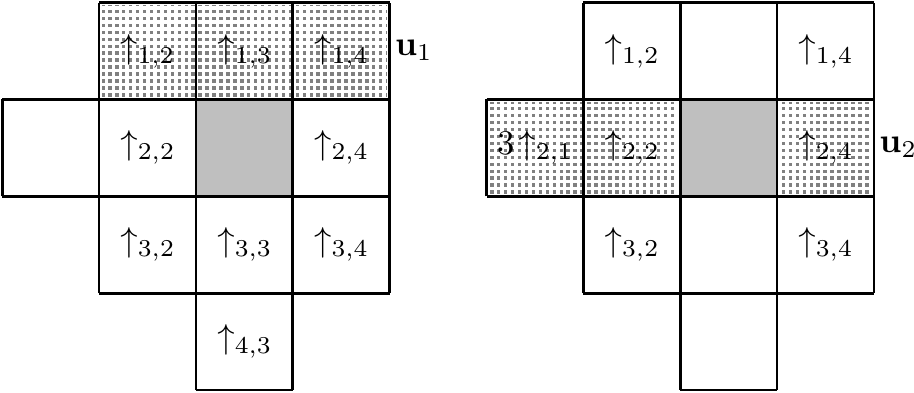}
\vspace{3pt}\\
\mbox{Figure 7.2.2: $S_2$ with $m=n=3$: almost vertical units of type $\uparrow$ in $\bfu_1,\bfu_2$}
\end{array}
\end{displaymath}

For the second horizontal street, corresponding to~$\bfu_2$, highlighted in the picture on the right in Figure~7.2.2, we have $J_2^*=\{1,2,4\}$, and a similar argument gives
\begin{align}\label{eq7.2.24}
(\bfA-I)\bfu_2
&
=(\bfA-I)[\{\uparrow_{1,2},\uparrow_{1,4}\}]
+(\bfA-I)[\{\uparrow_{2,1},\uparrow_{2,1},\uparrow_{2,1},\uparrow_{2,2},\uparrow_{2,4}\}]
\nonumber
\\
&\qquad
+(\bfA-I)[\{\uparrow_{3,2},\uparrow_{3,4}\}]
\nonumber
\\
&
=2(\bfu_1+\bfv_1)+5(\bfu_2+\bfv_2)+2(\bfu_3+\bfv_3).
\end{align}
The multiple $3$ for $\uparrow_{2,1}$ in the square face $S_{2,1}$ comes from the fact that $I_1^*=\{2,2,2\}$, since $n=3$ but the vertical street has length~$1$.

For the third horizontal street, corresponding to~$\bfu_3$, highlighted in the picture on the left in Figure~7.2.3, we have $J_3^*=\{2,3,4\}$, and a similar argument gives
\begin{align}\label{eq7.2.25}
(\bfA-I)\bfu_3
&
=(\bfA-I)[\{\uparrow_{1,2},\uparrow_{1,3},\uparrow_{1,4}\}]
+(\bfA-I)[\{\uparrow_{2,2},\uparrow_{2,4}\}]
\nonumber
\\
&\qquad
+(\bfA-I)[\{\uparrow_{3,2},\uparrow_{3,3},\uparrow_{3,4}\}]
+(\bfA-I)[\{\uparrow_{4,3}\}]
\nonumber
\\
&
=3(\bfu_1+\bfv_1)+2(\bfu_2+\bfv_2)+3(\bfu_3+\bfv_3)+(\bfu_4+\bfv_4).
\end{align}
\begin{displaymath}
\begin{array}{c}
\includegraphics[scale=0.8]{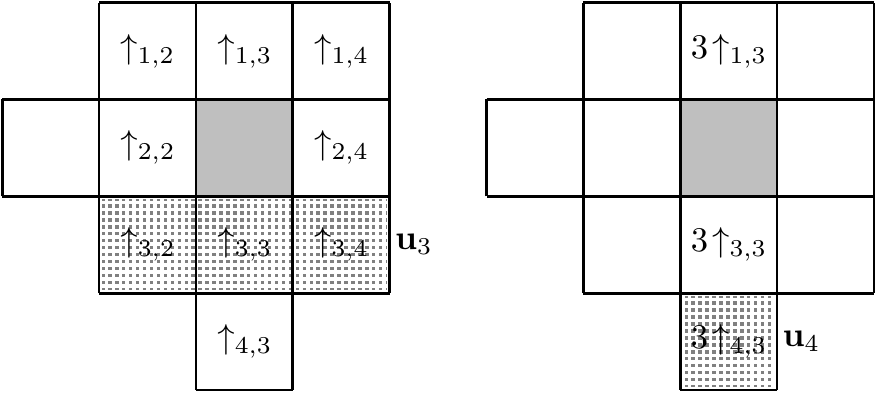}
\vspace{3pt}\\
\mbox{Figure 7.2.3: $S_2$ with $m=n=3$: almost vertical units of type $\uparrow$ in $\bfu_3,\bfu_4$}
\end{array}
\end{displaymath}

For the fourth horizontal street, corresponding to~$\bfu_4$, highlighted in the picture on the right in Figure~7.2.3, we have $J_4^*=\{3,3,3\}$, and a similar argument gives
\begin{align}\label{eq7.2.26}
(\bfA-I)\bfu_4
&
=(\bfA-I)[\{\uparrow_{1,3},\uparrow_{1,3},\uparrow_{1,3}\}]
+(\bfA-I)[\{\uparrow_{3,3},\uparrow_{3,3},\uparrow_{3,3}\}]
\nonumber
\\
&\qquad
+(\bfA-I)[\{\uparrow_{4,3},\uparrow_{4,3},\uparrow_{4,3}\}]
\nonumber
\\
&
=3(\bfu_1+\bfv_1)+3(\bfu_3+\bfv_3)+3(\bfu_4+\bfv_4).
\end{align}

We now define the street-spreading matrix by
\begin{equation}\label{eq7.2.27}
\bfS=\begin{pmatrix}
3&2&3&3\\
2&5&2&0\\
3&2&3&3\\
1&0&1&3
\end{pmatrix},
\end{equation}
where the $j$-th column comes from the coefficients of $(\bfu_i+\bfv_i)$, $i=1,2,3,4$, in the expression for $(\bfA-I)\bfu_j$, $j=1,2,3,4$, as given by
\eqref{eq7.2.23}--\eqref{eq7.2.26}.

Combining \eqref{eq7.2.22}--\eqref{eq7.2.26}, we obtain
\begin{align}
\bfA\bfu_1
&
=3(\bfu_1+\bfv_1)+2(\bfu_2+\bfv_2)+3(\bfu_3+\bfv_3)+(\bfu_4+\bfv_4)+\bfu_1,
\nonumber
\\
\bfA\bfu_2
&
=2(\bfu_1+\bfv_1)+5(\bfu_2+\bfv_2)+2(\bfu_3+\bfv_3)+\bfu_2,
\nonumber
\\
\bfA\bfu_3
&
=3(\bfu_1+\bfv_1)+2(\bfu_2+\bfv_2)+3(\bfu_3+\bfv_3)+(\bfu_4+\bfv_4)+\bfu_3,
\nonumber
\\
\bfA\bfu_4
&
=3(\bfu_1+\bfv_1)+3(\bfu_3+\bfv_3)+3(\bfu_4+\bfv_4)+\bfu_4,
\nonumber
\\
\bfA\bfv_1
&
=\bfu_1+\bfv_1,
\nonumber
\\
\bfA\bfv_2
&
=\bfu_2+\bfv_2,
\nonumber
\\
\bfA\bfv_3
&
=\bfu_3+\bfv_3,
\nonumber
\\
\bfA\bfv_4
&
=\bfu_4+\bfv_4.
\nonumber
\end{align}
For any $\bfw\in\VVV$, if we write
\begin{displaymath}
\bfw=a_1\bfu_1+a_2\bfu_2+a_3\bfu_3+a_4\bfu_4+b_1\bfv_1+b_2\bfv_2+b_3\bfv_3+b_4\bfv_4,
\end{displaymath}
and
\begin{displaymath}
\bfA\bfw=c_1\bfu_1+c_2\bfu_2+c_3\bfu_3+c_4\bfu_4+d_1\bfv_1+d_2\bfv_2+d_3\bfv_3+d_4\bfv_4,
\end{displaymath}
then
\begin{displaymath}
\begin{pmatrix}
c_1\\
c_2\\
c_3\\
c_4\\
d_1\\
d_2\\
d_3\\
d_4
\end{pmatrix}
=\begin{pmatrix}
4&2&3&3&1&0&0&0\\
2&6&2&0&0&1&0&0\\
3&2&4&3&0&0&1&0\\
1&0&1&4&0&0&0&1\\
3&2&3&3&1&0&0&0\\
2&5&2&0&0&1&0&0\\
3&2&3&3&0&0&1&0\\
1&0&1&3&0&0&0&1
\end{pmatrix}
\begin{pmatrix}
a_1\\
a_2\\
a_3\\
a_4\\
b_1\\
b_2\\
b_3\\
b_4
\end{pmatrix},
\end{displaymath}
so the $8\times8$ matrix in question is
\begin{displaymath}
\bfA\vert_\VVV=\begin{pmatrix}
\bfS+I&I\\
\bfS&I
\end{pmatrix}.
\end{displaymath}
We shall return to this example later.
\end{example}

\begin{example}\label{ex7.2.2}
Consider the polysquare translation surface~$S_3$, given in Figure~7.2.4.

\begin{displaymath}
\begin{array}{c}
\includegraphics[scale=0.8]{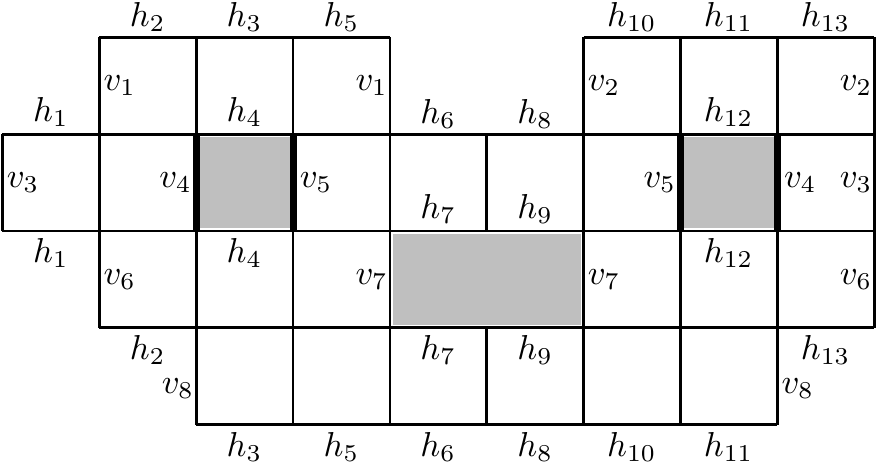}
\vspace{3pt}\\
\mbox{Figure 7.2.4: $S_3$ with edge pairings}
\end{array}
\end{displaymath}

This has $6$ horizontal streets of lengths $3,3,3,4,6,6$, and $9$ vertical streets of lengths $1,3,3,4,2,2,4,3,3$, as illustrated in Figure~7.2.5 where, for instance, $\leftrightarrow4$ and $\updownarrow6$ in a square face indicates that this square face is part of the $4$-th horizontal street and the $6$-th vertical street.

\begin{displaymath}
\begin{array}{c}
\includegraphics[scale=0.8]{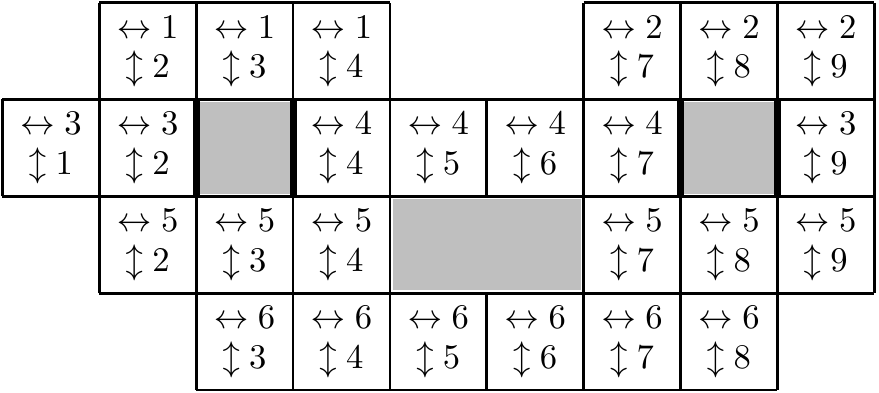}
\vspace{3pt}\\
\mbox{Figure 7.2.5: the horizontal and vertical streets of $S_3$}
\end{array}
\end{displaymath}

Consider a $1$-direction geodesic starting from some vertex of $S_3$ with slope $\alpha$ given by \eqref{eq7.2.1} with $m=n=12$.

The polysquare translation surface $S_3$ has $25$ square faces, so any $2$-step transition matrix $\bfA$ has size $50\times50$, making it an onerous task to find any eigenvalue, even with the help of MATLAB.
It is clearly much simpler to find the $6\times6$ street-spreading matrix.

First of all, we define the vectors $\bfu_i,\bfv_i$, $i=1,2,3,4,5,6$, according to \eqref{eq7.2.14} and \eqref{eq7.2.15}.

Consider the horizontal street corresponding to~$\bfu_1$, as highlighted in the left half of Figure~7.2.6.
We see that
\begin{displaymath}
J_1=\{2,3,4\}
\quad\mbox{and}\quad
J_1^*=\{2,2,2,2,3,3,3,3,4,4,4,4\},
\end{displaymath}
since $m=12$, and $\bfu_1$ concerns almost vertical units on the vertical streets $2,3,4$.
In other words, we start with the first horizontal street, and spread along every vertical street that intersects this horizontal street.
The left half of Figure~7.2.6 shows a listing of all the almost vertical units of type $\uparrow$ along these vertical streets.
Applying \eqref{eq7.2.17}, we have
\begin{align}\label{eq7.2.28}
(\bfA-I)\bfu_1
&
=(\bfA-I)[\{16\uparrow_{1,2},16\uparrow_{1,3},12\uparrow_{1,4}\}]
\nonumber
\\
&\qquad
+(\bfA-I)[\{16\uparrow_{3,2}\}]
+(\bfA-I)[\{12\uparrow_{4,4}\}]
\nonumber
\\
&\qquad
+(\bfA-I)[\{16\uparrow_{5,2},16\uparrow_{5,3},12\uparrow_{5,4}\}]
+(\bfA-I)[\{16\uparrow_{6,3},12\uparrow_{6,4}\}]
\nonumber
\\
&
=44(\bfu_1+\bfv_1)+16(\bfu_3+\bfv_3)+12(\bfu_4+\bfv_4)
\nonumber
\\
&\qquad
+44(\bfu_5+\bfv_5)+28(\bfu_6+\bfv_6).
\end{align}
\begin{displaymath}
\begin{array}{c}
\includegraphics[scale=0.8]{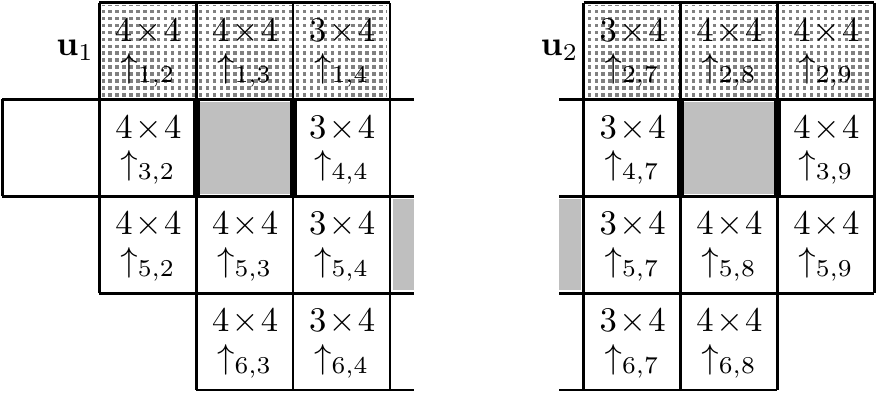}
\vspace{3pt}\\
\mbox{Figure 7.2.6: $S_3$ with $m=n=12$: almost vertical units of type $\uparrow$ in $\bfu_1,\bfu_2$}
\end{array}
\end{displaymath}

We need to explain the multiples in Figure~7.2.6.
Observe the horizontal street corresponding to $\bfu_1$ consists of $3$ square faces.
As $m=12$, the correct multiplicity is $12/3=4$.
On the other hand, the vertical street containing the square face $S_{1,4}$ consists of $4$ square faces.
As $n=12$, the correct multiplicity is $12/4=3$.
Thus the multiple for that square face is $3\times4$.

\begin{remark}
We always have $\vert I_j^*\vert=n$ for every $j=1,\ldots,v$, and $\vert J_i^*\vert=m$ for every $i=1,\ldots,h$.
The first of these corresponds to the fact that every almost vertical detour crossing of $V_0$ travels a vertical distance between $n$ and $n+1$.
The second of these corresponds to the fact that every almost horizontal detour crossing of $H_0$ travels a horizontal distance between $m$ and $m+1$.
\end{remark}

With $J_2^*=\{7,7,7,7,8,8,8,8,9,9,9,9\}$, for the horizontal street corresponding to~$\bfu_2$, as highlighted in the right half of Figure~7.2.6, a similar argument gives
\begin{align}\label{eq7.2.29}
(\bfA-I)\bfu_2
&
=44(\bfu_2+\bfv_2)+16(\bfu_3+\bfv_3)+12(\bfu_4+\bfv_4)
\nonumber
\\
&\qquad
+44(\bfu_5+\bfv_5)+28(\bfu_6+\bfv_6).
\end{align}

With $J_3^*=\{1,1,1,1,2,2,2,2,9,9,9,9\}$, for the horizontal street corresponding to~$\bfu_3$, as highlighted in Figure~7.2.7, a similar argument gives
\begin{equation}\label{eq7.2.30}
(\bfA-I)\bfu_3
=16(\bfu_1+\bfv_1)+16(\bfu_2+\bfv_2)+80(\bfu_3+\bfv_3)+32(\bfu_5+\bfv_5).
\end{equation}
\begin{displaymath}
\begin{array}{c}
\includegraphics[scale=0.8]{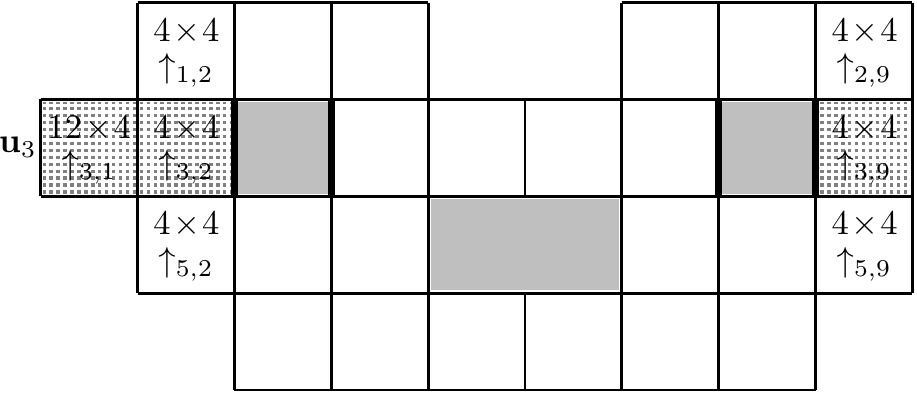}
\vspace{3pt}\\
\mbox{Figure 7.2.7: $S_3$ with $m=n=12$: almost vertical units of type $\uparrow$ in $\bfu_3$}
\end{array}
\end{displaymath}

With $J_4^*=\{4,4,4,5,5,5,6,6,6,7,7,7\}$, for the horizontal street corresponding to~$\bfu_4$, as highlighted in Figure~7.2.8, a similar argument gives
\begin{align}\label{eq7.2.31}
(\bfA-I)\bfu_4
&
=9(\bfu_1+\bfv_1)+9(\bfu_2+\bfv_2)+54(\bfu_4+\bfv_4)
\nonumber
\\
&\qquad
+18(\bfu_5+\bfv_5)+54(\bfu_6+\bfv_6).
\end{align}
\begin{displaymath}
\begin{array}{c}
\includegraphics[scale=0.8]{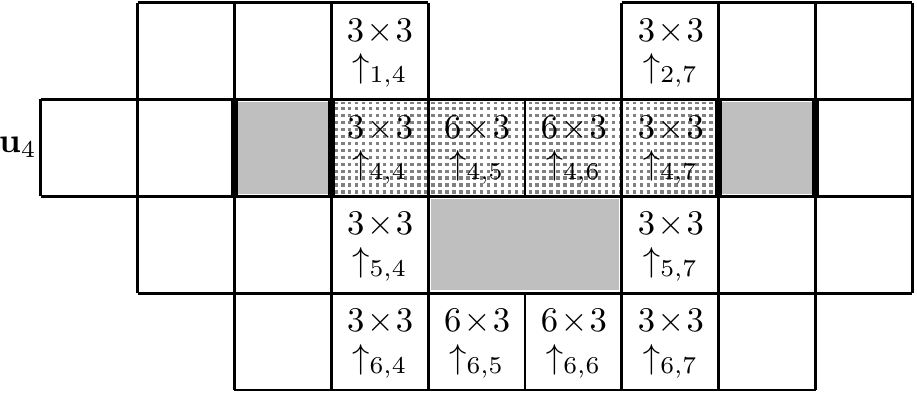}
\vspace{3pt}\\
\mbox{Figure 7.2.8: $S_3$ with $m=n=12$: almost vertical units of type $\uparrow$ in $\bfu_4$}
\end{array}
\end{displaymath}

With $J_5^*=\{2,2,3,3,4,4,7,7,8,8,9,9\}$, for the horizontal street corresponding to~$\bfu_5$, as highlighted in Figure~7.2.9, a similar argument gives
\begin{align}\label{eq7.2.32}
(\bfA-I)\bfu_5
&
=22(\bfu_1+\bfv_1)+22(\bfu_2+\bfv_2)+16(\bfu_3+\bfv_3)
\nonumber
\\
&\qquad
+12(\bfu_4+\bfv_4)+44(\bfu_5+\bfv_5)+28(\bfu_6+\bfv_6).
\end{align}
\begin{displaymath}
\begin{array}{c}
\includegraphics[scale=0.8]{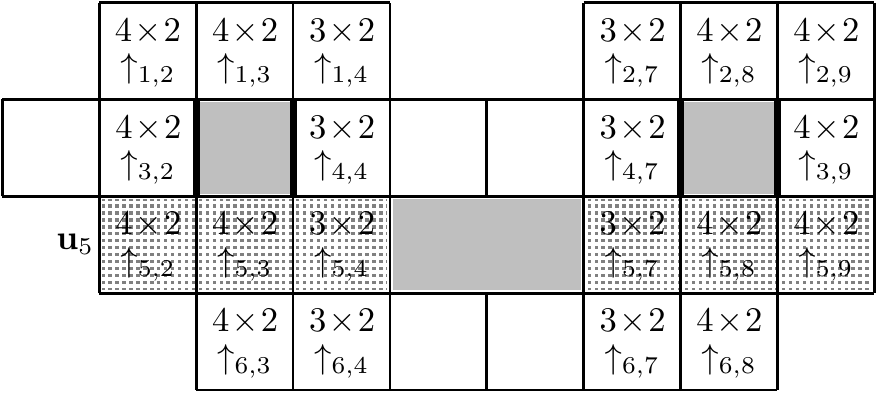}
\vspace{3pt}\\
\mbox{Figure 7.2.9: $S_3$ with $m=n=12$: almost vertical units of type $\uparrow$ in $\bfu_5$}
\end{array}
\end{displaymath}

With $J_6^*=\{3,3,4,4,5,5,6,6,7,7,8,8\}$, for the horizontal street corresponding to~$\bfu_6$, as highlighted in Figure~7.2.10, a similar argument gives
\begin{align}\label{eq7.2.33}
(\bfA-I)\bfu_6
&
=14(\bfu_1+\bfv_1)+14(\bfu_2+\bfv_2)+36(\bfu_4+\bfv_4)
\nonumber
\\
&\qquad
+28(\bfu_5+\bfv_5)+52(\bfu_6+\bfv_6).
\end{align}
\begin{displaymath}
\begin{array}{c}
\includegraphics[scale=0.8]{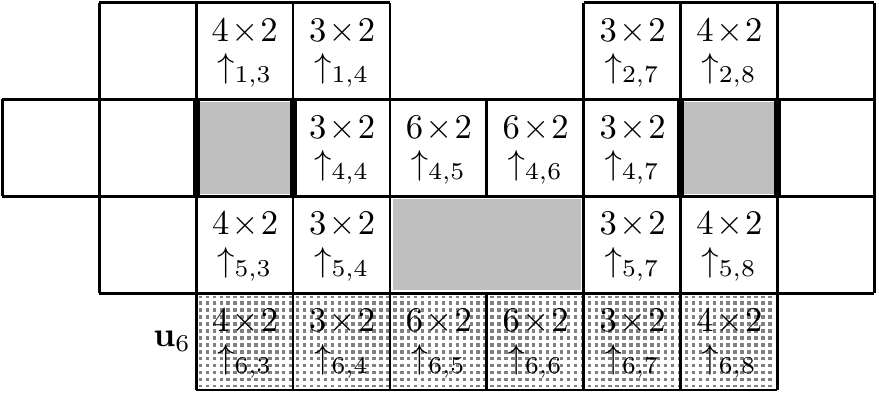}
\vspace{3pt}\\
\mbox{Figure 7.2.10: $S_3$ with $m=n=12$: almost vertical units of type $\uparrow$ in $\bfu_6$}
\end{array}
\end{displaymath}

We now define the street-spreading matrix by
\begin{displaymath}
\bfS=\begin{pmatrix}
44&0&16&9&22&14\\
0&44&16&9&22&14\\
16&16&80&0&16&0\\
12&12&0&54&12&36\\
44&44&32&18&44&28\\
28&28&0&54&28&52
\end{pmatrix},
\end{displaymath}
where the $j$-th column comes from the coefficients of $(\bfu_i+\bfv_i)$, $i=1,2,3,4,5,6$, in the expression for $(\bfA-I)\bfu_j$, $j=1,2,3,4,5,6$, as given by \eqref{eq7.2.28}--\eqref{eq7.2.33}.
\end{example}

\begin{thm}\label{thm7.2.2}
Let $\PPP$ be a finite polysquare translation surface with $d$ square faces.
Assume that $h\le v$, where $h$ is the number of horizontal streets in $\PPP$ and $v$ is the number of vertical streets in~$\PPP$.
Let $m$ be any fixed integer multiple of the lengths of the horizontal streets of~$\PPP$, and let $n$ be any fixed integer multiple of the lengths of the vertical streets of~$\PPP$.

Consider a $1$-direction almost vertical geodesic $V_0$ on~$\PPP$, with slope $\alpha$ given by \eqref{eq7.2.1}.
Let $\bfA$ denote the transition matrix of the $2$-step ancestor process with respect to the basis~$\WWWW$, where $\WWWW$ is given by
\eqref{eq7.2.2} .
Then all the relevant eigenvalues of $\bfA$ and their corresponding eigenvectors can be determined as follows:

\emph{(i)}
Let $\VVV$ denote the subspace of $\Cc^{2d}$ generated by $\bfu_i,\bfv_i$, $i=1,\ldots,h$, given by \eqref{eq7.2.14} and \eqref{eq7.2.15}.
Then $\VVV$ is an $\bfA$-invariant subspace of $\Cc^{2d}$, and contains the eigenvectors corresponding to all the relevant eigenvalues of~$\bfA$.
Furthermore, there exist non-negative integers $s_{i,j}$, $i,j=1,\ldots,h$, such that for every $j=1,\ldots,h$,
\begin{equation}\label{eq7.2.34}
(\bfA-I)\bfu_j=\sum_{i=1}^hs_{i,j}(\bfu_i+\bfv_i),
\end{equation}
and
\begin{equation}\label{eq7.2.35}
\bfA\bfv_j=\bfu_j+\bfv_j.
\end{equation}

\emph{(ii)}
Consider the street-spreading matrix
\begin{equation}\label{eq7.2.36}
\bfS=\begin{pmatrix}
s_{1,1}&\cdots&s_{1,h}\\
\vdots&&\vdots\\
s_{h,1}&\cdots&s_{h,h}
\end{pmatrix},
\end{equation}
as well as the matrix
\begin{equation}\label{eq7.2.37}
\bfA\vert_\VVV=\begin{pmatrix}
\bfS+I&I\\
\bfS&I
\end{pmatrix}.
\end{equation}
Suppose that $\tau$ is an eigenvalue of $\bfS$ with eigenvector~$\psi$, so that $\bfS\psi=\tau\psi$.
Then
\begin{equation}\label{eq7.2.38}
\lambda(\tau;\pm)=1+\frac{\tau\pm\sqrt{\tau^2+4\tau}}{2}
\end{equation}
are two eigenvalues of~$\bfA\vert_\VVV$ with product equal to~$1$, with corresponding eigenvectors
\begin{equation}\label{eq7.2.39}
\Psi(\tau;\pm)=\left(\psi,\frac{-\tau\pm\sqrt{\tau^2+4\tau}}{2}\psi\right)^T\in\Cc^{2h}.
\end{equation}
In particular, each eigenvalue of $\bfA\vert_\VVV$ comes from an eigenvalue of~$\bfS$.

\emph{(iii)}
If $\Psi(\tau;\pm)=(a_1,a_2,a_3,a_4,b_1,b_2,b_3,b_4)^T$ is an eigenvector corresponding to an eigenvalue $\lambda(\tau;\pm)$ of~$\bfA\vert_\VVV$, then $\lambda(\tau;\pm)$ is an eigenvalue of $\bfA$ with eigenvector
\begin{equation}\label{eq7.2.40}
\bfw=\sum_{i=1}^ha_i\bfu_i+\sum_{i=1}^hb_i\bfv_i.
\end{equation}
\end{thm}

\begin{remark}
The assumption that the number $h=h(\PPP)$ of horizontal streets is less than or equal to the number $v=v(\PPP)$ of vertical streets, \textit{i.e.}, $h\le v$, is one for convenience, as the street-spreading matrix is smaller.
For any polysquare translation surface $\PPP$ where $h\ge v$, we can always interchange the horizontal and vertical directions.
\end{remark}

\begin{proof}[Proof of Theorem~\ref{thm7.2.2}]
(i)
In view of Lemma~\ref{lem7.2.3}, it remains to establish \eqref{eq7.2.34} and \eqref{eq7.2.35}.
For \eqref{eq7.2.34}, note from \eqref{eq7.2.14} that each $\bfu_j$ involves only units of type~$\uparrow$, and so \eqref{eq7.2.34} follows immediately from \eqref{eq7.2.17}, on observing that $\bfu_i,\bfv_i$, $i=1,\ldots,h$, generate~$\VVV$.
On the other hand, \eqref{eq7.2.35} follows immediately from \eqref{eq7.2.21}.

(ii)
Suppose that $\Psi$ is a column eigenvector corresponding to an eigenvalue $\lambda$ of~$\bfA\vert_\VVV$.
From \eqref{eq7.2.37}, it is clear that $\bfA\vert_\VVV$ can be reduced to the matrix
\begin{displaymath}
\begin{pmatrix}
I&0\\
\bfS&I
\end{pmatrix}
\end{displaymath}
by elementary row operations.
The determinant of this matrix is clearly equal to~$1$.
It follows that $\det(\bfA\vert_\VVV)=1$, and so $\lambda\ne0$.
We can write
\begin{equation}\label{eq7.2.41}
\Psi=(\psi,\psi^*)^T,
\end{equation}
where $\psi$ and $\psi^*$ are both column vectors with $h$ entries.
Then
\begin{displaymath}
\bfA\vert_\VVV\Psi
=\begin{pmatrix}
\bfS+I&I\\
\bfS&I
\end{pmatrix}
\begin{pmatrix}
\psi\\
\psi^*
\end{pmatrix}
=\lambda\begin{pmatrix}
\psi\\
\psi^*
\end{pmatrix},
\end{displaymath}
and so
\begin{align}\label{eq7.2.42}
\bfS\psi+\psi+\psi^*&
=\lambda\psi,
\nonumber
\\
\bfS\psi+\psi^*&
=\lambda\psi^*.
\end{align}
Subtracting the second equation in \eqref{eq7.2.42} from the first, we obtain $\psi=\lambda(\psi-\psi^*)$, and so
\begin{equation}\label{eq7.2.43}
\psi^*=\frac{\lambda-1}{\lambda}\psi.
\end{equation}
Substituting this into the first equation in \eqref{eq7.2.42}, we obtain
\begin{displaymath}
\bfS\psi=\left(\lambda-\frac{2\lambda-1}{\lambda}\right)\psi,
\end{displaymath}
so that
\begin{displaymath}
\tau=\lambda-\frac{2\lambda-1}{\lambda}
\end{displaymath}
is an eigenvalue of $\bfS$ with eigenvector~$\psi$.
This is equivalent to
\begin{equation}\label{eq7.2.44}
\lambda=1+\frac{\tau\pm\sqrt{\tau^2+4\tau}}{2}.
\end{equation}
This proves that any eigenvalue $\lambda$ of $\bfA\vert_\VVV$ is obtained from an eigenvalue $\tau$ of $\bfS$ by \eqref{eq7.2.44}, and establishes \eqref{eq7.2.38}.
Finally, note that
\begin{displaymath}
\frac{\lambda-1}{\lambda}
=1-\frac{1}{\lambda}
=1-\left(1+\frac{\tau\mp\sqrt{\tau^2+4\tau}}{2}\right)
=\frac{-\tau\pm\sqrt{\tau^2+4\tau}}{2}.
\end{displaymath}
In view of \eqref{eq7.2.41} and \eqref{eq7.2.43}, this establishes \eqref{eq7.2.39}.

(iii)
Suppose that $(a_1,\ldots,a_h,b_1,\ldots,b_h)^T$ is an eigenvector corresponding to an eigenvalue $\lambda$ of~$\bfA\vert_\VVV$.
Let the vector $\bfw$ be given by \eqref{eq7.2.40}.
Then
\begin{equation}\label{eq7.2.45}
\bfA\bfw=\sum_{i=1}^hc_i\bfu_i+\sum_{i=1}^hd_i\bfv_i,
\end{equation}
where the coefficients are related by
\begin{displaymath}
(c_1,\ldots,c_h,d_1,\ldots,d_h)^T=\bfA\vert_\VVV(a_1,\ldots,a_h,b_1,\ldots,b_h)^T.
\end{displaymath}
It is clear that
\begin{displaymath}
(c_1,\ldots,c_h,d_1,\ldots,d_h)^T=\lambda(a_1,\ldots,a_h,b_1,\ldots,b_h)^T,
\end{displaymath}
and so it follows from \eqref{eq7.2.40} and \eqref{eq7.2.45} that $\bfA\bfw=\lambda\bfw$, so that $\lambda$ is an eigenvalue of $\bfA$ with
eigenvector~$\bfw$.
\end{proof}

\begin{remark}
As commented earlier, if two or more distinct square faces of a polysquare translation surface lie on the same horizontal street and the same vertical street simultaneously, then the notation $S_{i,j}$ does not allow us to distinguish between two distinct square faces that lie on the $i$-th horizontal street and the $j$-th vertical street.
In this case, we need to denote each square face by $S_\delta$, where $\delta=1,\ldots,d$.

We see from \eqref{eq7.2.34} as well as Examples \ref{ex7.2.1} and~\ref{ex7.2.2} that the crux of our argument concerns expressions of the form
\begin{displaymath}
(\bfA-I)\bfu_j,
\quad
\bfu_i
\quad\mbox{and}\quad
\bfv_i,
\end{displaymath}
so we need to find suitable generalizations of the expressions \eqref{eq7.2.12}--\eqref{eq7.2.18}.

For instance, let $\HS_{i_0}$ denote the $i_0$-th horizontal street.
For any $\delta=1,\ldots,d$, let $\HS(S_\delta)$ and $\VS(S_\delta)$ denote respectively the horizontal street and the vertical street that contains the square face~$S_\delta$.
Then the expression \eqref{eq7.2.12} can be rewritten in the form
\begin{align}
(\bfA-I)[\{\uparrow_{\delta_0}\}]
&
=[\{\uparrow_\delta:S_\delta\subseteq\VS^*(S_\eta)\mbox{ and }S_\eta\subseteq\HS^*(S_{\delta_0})\}]
\nonumber
\\
&\qquad
+[\{\nuparrow_\delta:S_\delta\subseteq\HS^*(S_{\delta_0})\}]
-[\{\uparrow_\delta:S_\delta\subseteq\HS^*(S_{\delta_0})\}],
\nonumber
\end{align}
while the expression \eqref{eq7.2.14} can then be rewritten in the form
\begin{displaymath}
\bfu_{i_0}=[\{\uparrow_\delta:S_\delta\subseteq\VS^*(S_\eta)\mbox{ and }S_\eta\subseteq\HS^*_{i_0}\}].
\end{displaymath}
In both instances, we adopt the convention that every $S_\eta$ is counted and $*$ denotes counting with the appropriate multiplicity.
We also need analogous generalizations of the other identities.

We make use of this more general version in the study of one-street polysquare translation surfaces later in this section.
We also make further use of this general version in Section~\ref{sec8} where we study geodesic flow on the regular octagon surface, the double-pentagon surface, polyrhombus surfaces and the infinite-halving staircase surface; see Figures 8.2.5, 8.5.2, 8.6.9 and 8.9.2 respectively.
\end{remark}

\begin{example721c}
Recall from \eqref{eq7.1.6}, \eqref{eq7.1.8} and \eqref{eq7.1.10} that the relevant eigenvalues of $\bfA(S_2)$ are
\begin{displaymath}
\lambda_1=\frac{11+3\sqrt{13}}{2},
\quad
\lambda_2=3+2\sqrt{2},
\quad
\lambda_3=\frac{3+\sqrt{5}}{2},
\end{displaymath}
obtained by tedious calculation using MATLAB.
Let us instead retrieve the same information using the street-spreading matrix $\bfS$ given by \eqref{eq7.2.27}.

The eigenvalues of $\bfS$ are $\tau_1=9$, $\tau_2=4$, $\tau_3=1$ and $\tau_4=0$.
Using \eqref{eq7.2.38}, the corresponding eigenvalues of $\bfA$ are
\begin{displaymath}
\lambda(9;\pm)=\frac{11\pm3\sqrt{13}}{2},
\quad
\lambda(4;\pm)=3\pm2\sqrt{2},
\quad
\lambda(1;\pm)=\frac{3\pm\sqrt{5}}{2},
\quad
\lambda(0;\pm)=1.
\end{displaymath}

Note that in view of Lemma~\ref{lem7.2.3}, $\lambda(9;\pm)$, $\lambda(4;\pm)$ and $\lambda(1;\pm)$ are the only eigenvalues of $\bfA$ that are not equal to~$1$.
Thus the multiplicity of the eigenvalue $1$ is~$14$, as we have claimed earlier in the Remark following \eqref{eq7.2.8}.
\end{example721c}

\begin{example}\label{ex7.2.3}
Consider again the L-surface, given in Figure~7.1.6, and consider a $1$-direction geodesic starting from some vertex of the L-surface with slope
$\alpha$ given by \eqref{eq7.2.1} with $m=n=2$.
It is easy to see from Figure~7.2.11 that
\begin{align}
(\bfA-I)\bfu_1
&=2(\bfu_1+\bfv_1)+2(\bfu_2+\bfv_2),
\label{eq7.2.46}
\\
(\bfA-I)\bfu_2
&=(\bfu_1+\bfv_1)+3(\bfu_2+\bfv_2).
\label{eq7.2.47}
\end{align}
\begin{displaymath}
\begin{array}{c}
\includegraphics[scale=0.8]{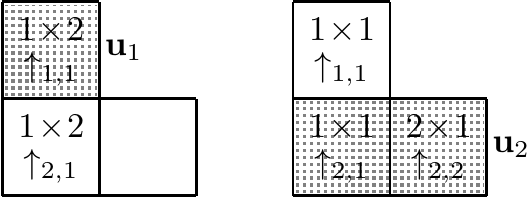}
\vspace{3pt}\\
\mbox{Figure 7.2.11: L-surface with $m=n=2$: almost vertical units}
\\
\mbox{of type $\uparrow$ in $\bfu_1,\bfu_2$}
\end{array}
\end{displaymath}

It follows from \eqref{eq7.2.46} and \eqref{eq7.2.47} that the street-spreading matrix is given by
\begin{displaymath}
\bfS=\begin{pmatrix}
2&1\\
2&3
\end{pmatrix},
\end{displaymath}
with eigenvalues $\tau_1=4$ and $\tau_2=1$.
Using \eqref{eq7.2.38} leads to the eigenvalues
\begin{displaymath}
\lambda(4;\pm)=3\pm2\sqrt{2}
\quad\mbox{and}\quad
\lambda(1;\pm)=\frac{3\pm\sqrt{5}}{2},
\end{displaymath}
including those given in \eqref{eq7.1.24} with absolute values greater than~$1$.
\end{example}

\begin{example}\label{ex7.2.4}
It is not difficult to show that geodesic flow on the surface of the unit cube is equivalent to $1$-direction geodesic flow on the polysquare translation surface $\PPP$ as shown in Figure~7.2.12 that represents the $4$-copy version of the cube surface.

\begin{displaymath}
\begin{array}{c}
\includegraphics[scale=0.8]{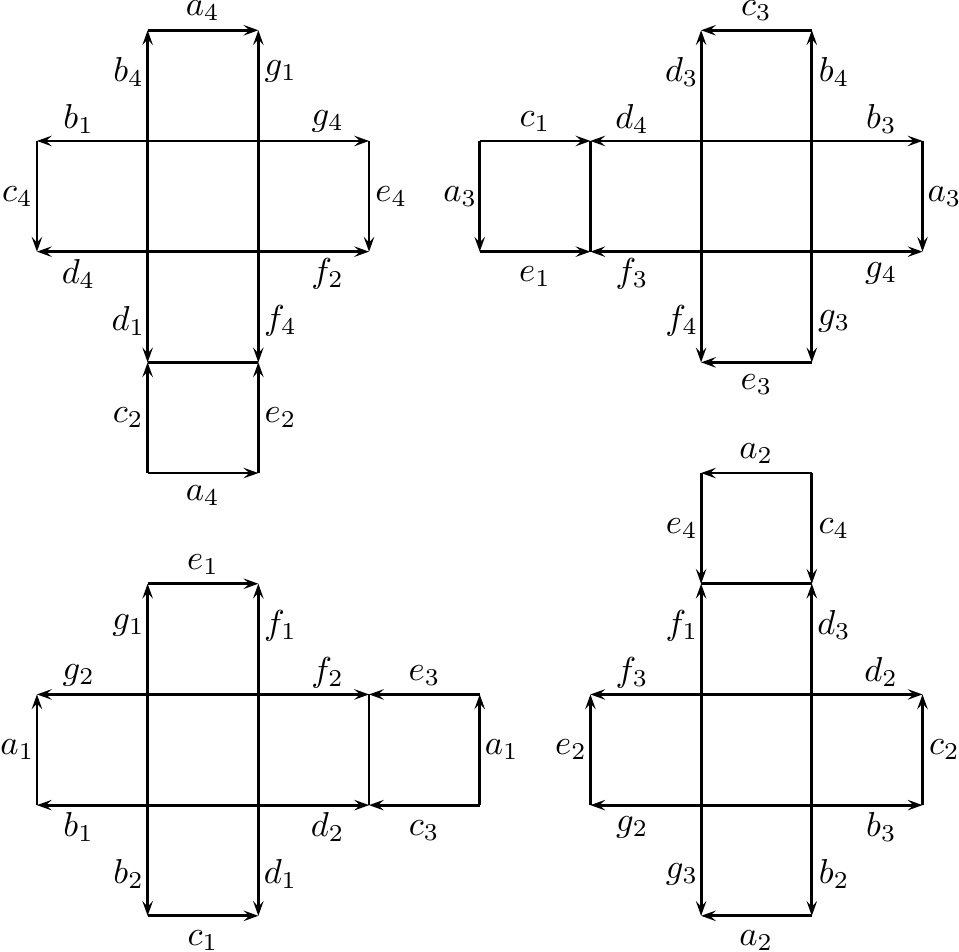}
\vspace{3pt}\\
\mbox{Figure 7.2.12: $4$-copy version of the surface of the unit cube}
\end{array}
\end{displaymath}

Note that $\PPP$ has $6$ horizontal streets and $6$ vertical streets, all of length~$4$, as shown in Figure~7.2.13, where, for instance,
$\leftrightarrow4$ and $\updownarrow6$ in a square face indicates that this square face is part of the $4$-th horizontal street and $6$-th vertical street.

\begin{displaymath}
\begin{array}{c}
\includegraphics[scale=0.8]{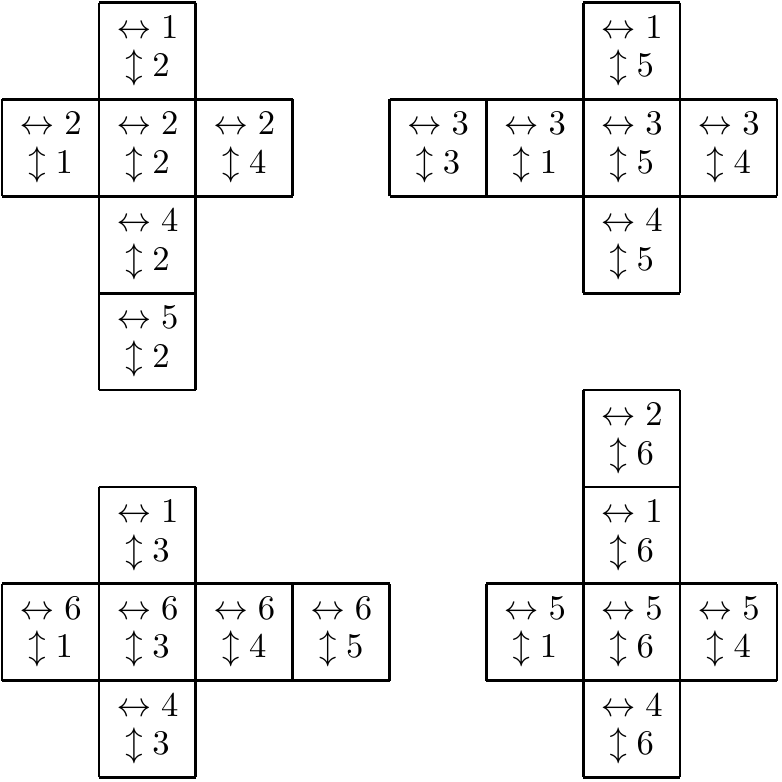}
\vspace{3pt}\\
\mbox{Figure 7.2.13: $4$-copy version of the surface of the unit cube}
\end{array}
\end{displaymath}

Figures 7.2.12 and~7.2.13 are not very convenient for us to visualize the horizontal and vertical streets clearly.
We can re-draw these two figures, and attempt to arrange the square faces in an array where each row represents a horizontal street and each column represents a vertical street.

In Figure~7.2.14, we do precisely this, where for $i,j=1,\ldots,6$, the $i$-th row represents the $i$-th horizontal street and the $j$-th column represents the $j$-th vertical street as shown in Figure~7.2.13.
Note that this exercise is somewhat cumbersome and requires very careful edge identifications as illustrated.
We include the details here for the sake of completeness, and comment that the edge identifications are not required in the process for the determination of the street-spreading matrix.

\begin{displaymath}
\begin{array}{c}
\includegraphics[scale=0.8]{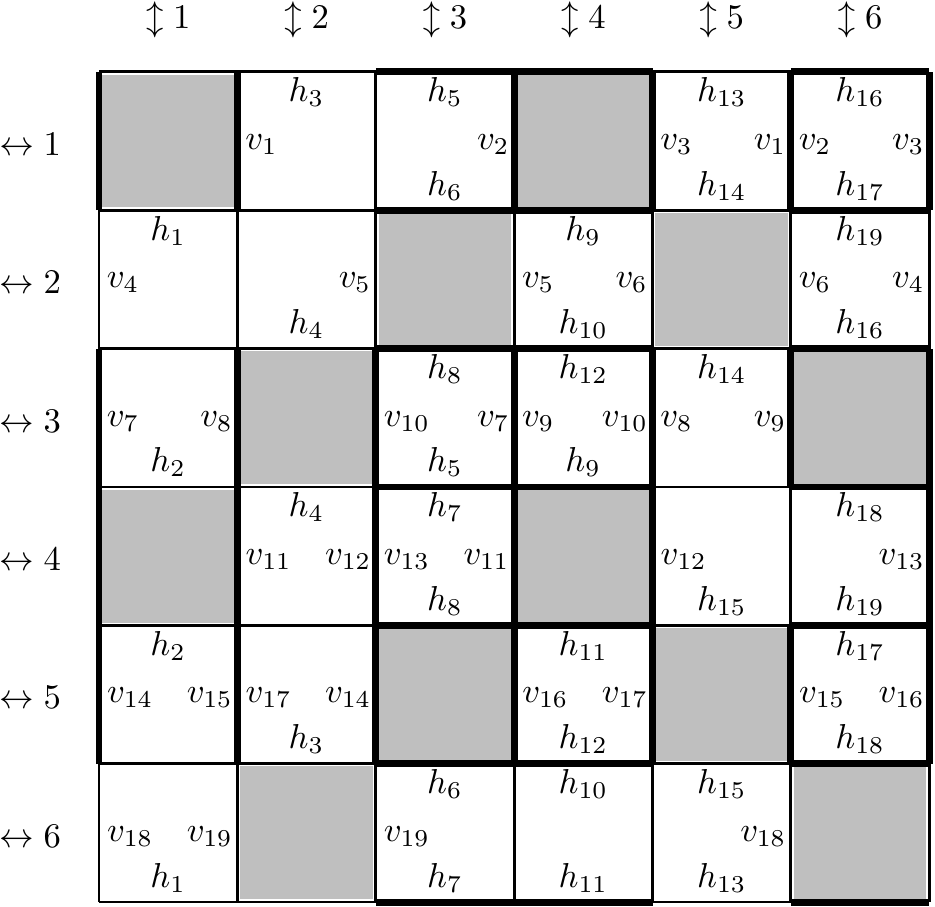}
\vspace{3pt}\\
\mbox{Figure 7.2.14: $4$-copy version of the surface of the unit cube}
\\
\mbox{with square faces arranged in an array}
\end{array}
\end{displaymath}

Consider a $1$-direction geodesic on the polysquare translation surface $\PPP$ starting from some vertex of $\PPP$ with slope $\alpha$ given by \eqref{eq7.2.1} with $m=n=4$.

The surface $\PPP$ has $24$ square faces, so any $2$-step transition matrix $\bfA$ has size $48\times48$.
It is much simpler to find the $6\times6$ street-spreading matrix.

As before, we define the matrices $\bfu_i,\bfv_i$, $i=1,\ldots,6$, according to \eqref{eq7.2.14} and \eqref{eq7.2.15}.

Consider the horizontal street corresponding to~$\bfu_1$, and also the horizontal street corresponding to~$\bfu_4$, as highlighted in the picture on the left in Figure~7.2.15.

\begin{displaymath}
\begin{array}{c}
\includegraphics[scale=0.8]{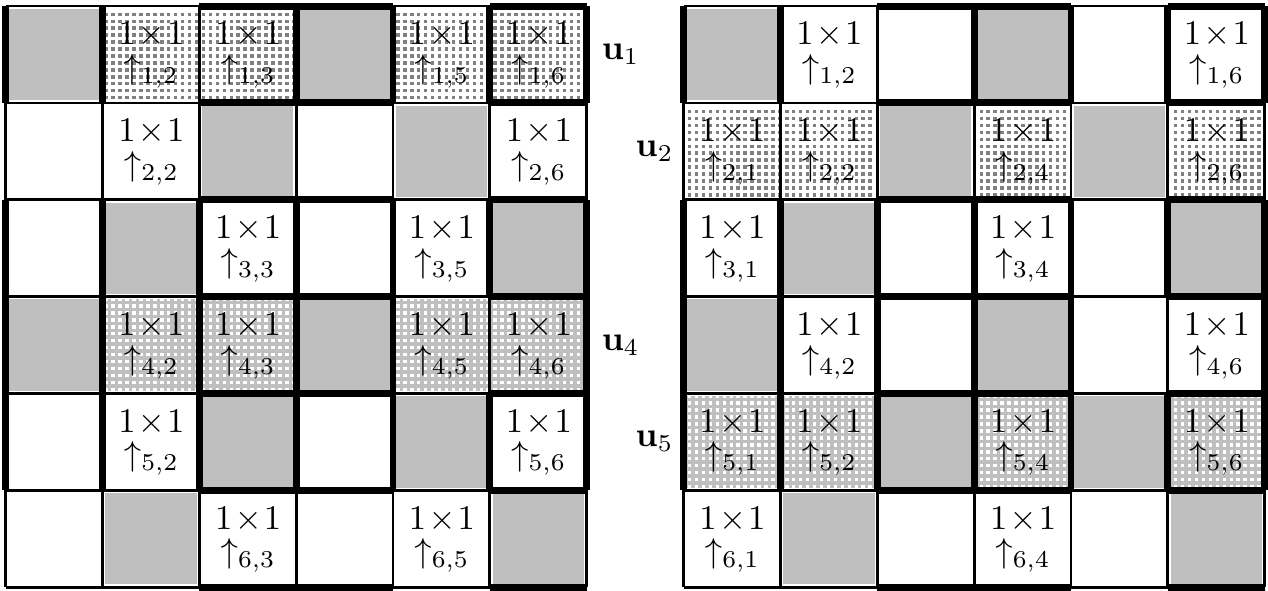}
\vspace{3pt}\\
\mbox{Figure 7.2.15: $\PPP$ with $m=n=4$: almost vertical units of type $\uparrow$ in $\bfu_1\bfu_2,\bfu_4,\bfu_5$}
\end{array}
\end{displaymath}

We see that $J_1^*=J_4^*=\{2,3,5,6\}$, and both $\bfu_1$ and $\bfu_4$ concern almost vertical units on the vertical streets $2,3,5,6$.
Applying \eqref{eq7.2.17}, we have
\begin{align}\label{eq7.2.48}
(\bfA-I)\bfu_1
&
=(\bfA-I)\bfu_4
\nonumber
\\
&
=(\bfA-I)[\{\uparrow_{1,2},\uparrow_{1,3},\uparrow_{1,5},\uparrow_{1,6}\}]+(\bfA-I)[\{\uparrow_{2,2},\uparrow_{2,6}\}]
\nonumber
\\
&\qquad
+(\bfA-I)[\{\uparrow_{3,3},\uparrow_{3,5}\}]
+(\bfA-I)[\{\uparrow_{4,2},\uparrow_{4,3},\uparrow_{4,5},\uparrow_{4,6}\}]
\nonumber
\\
&\qquad
+(\bfA-I)[\{\uparrow_{5,2},\uparrow_{5,6}\}]
+(\bfA-I)[\{\uparrow_{6,3},\uparrow_{6,5}\}]
\nonumber
\\
&
=4(\bfu_1+\bfv_1)+2(\bfu_2+\bfv_2)+2(\bfu_3+\bfv_3)
\nonumber
\\
&\qquad
+4(\bfu_4+\bfv_4)+2(\bfu_5+\bfv_5)+2(\bfu_6+\bfv_6).
\end{align}

With $J_2^*=J_5^*=\{1,2,4,6\}$, for the horizontal street corresponding to~$\bfu_2$, and also the horizontal street corresponding to~$\bfu_5$, as highlighted in the picture on the right in Figure~7.2.15, a similar argument gives
\begin{align}\label{eq7.2.49}
(\bfA-I)\bfu_2
&
=(\bfA-I)\bfu_5
\nonumber
\\
&
=2(\bfu_1+\bfv_1)+4(\bfu_2+\bfv_2)+2(\bfu_3+\bfv_3)
\nonumber
\\
&\qquad
+2(\bfu_4+\bfv_4)+4(\bfu_5+\bfv_5)+2(\bfu_6+\bfv_6).
\end{align}

With $J_3^*=J_6^*=\{1,3,4,5\}$, for the horizontal street corresponding to~$\bfu_3$, and also the horizontal street corresponding to~$\bfu_6$, as highlighted in Figure~7.2.16, a similar argument gives
\begin{align}\label{eq7.2.50}
(\bfA-I)\bfu_3
&
=(\bfA-I)\bfu_6
\nonumber
\\
&
=2(\bfu_1+\bfv_1)+2(\bfu_2+\bfv_2)+4(\bfu_3+\bfv_3)
\nonumber
\\
&\qquad
+2(\bfu_4+\bfv_4)+2(\bfu_5+\bfv_5)+4(\bfu_6+\bfv_6).
\end{align}
\begin{displaymath}
\begin{array}{c}
\includegraphics[scale=0.8]{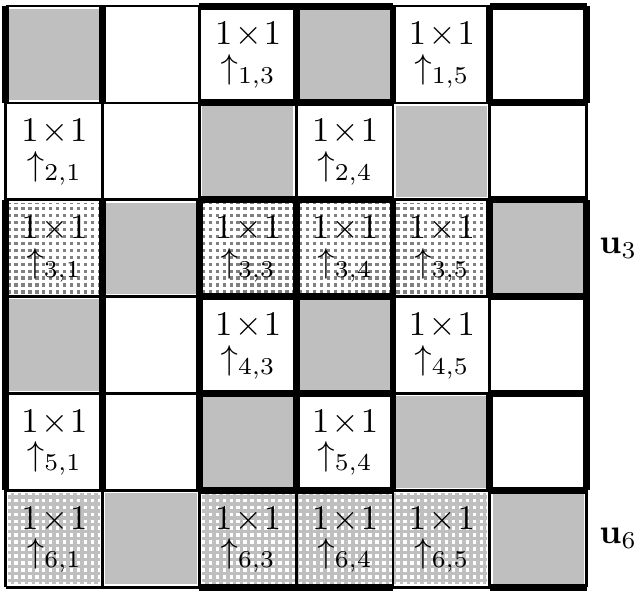}
\vspace{3pt}\\
\mbox{Figure 7.2.16: $\PPP$ with $m=n=4$: almost vertical units of type $\uparrow$ in $\bfu_3,\bfu_6$}
\end{array}
\end{displaymath}

It follows from \eqref{eq7.2.48}--\eqref{eq7.2.50} that the street-spreading matrix is given by
\begin{displaymath}
\bfS=\begin{pmatrix}
4&2&2&4&2&2\\
2&4&2&2&4&2\\
2&2&4&2&2&4\\
4&2&2&4&2&2\\
2&4&2&2&4&2\\
2&2&4&2&2&4
\end{pmatrix},
\end{displaymath}
with eigenvalues $\tau_1=16$, $\tau_2=\tau_3=4$ and $\tau_4=\tau_5=\tau_6=0$.
Using \eqref{eq7.2.38}, the corresponding eigenvalues of $\bfA$ are
\begin{displaymath}
\lambda(16;\pm)=9\pm4\sqrt{5},
\quad
\lambda(4;\pm)=3\pm2\sqrt{2},
\quad
\lambda(0;\pm)=1,
\end{displaymath}
with the appropriate multiplicities.
The two largest eigenvalues of $\bfA$ are therefore
\begin{displaymath}
\lambda_1=9+4\sqrt{5}=(2+\sqrt{5})^2
\quad\mbox{and}\quad
\lambda_2=3+2\sqrt{2}=(1+\sqrt{2})^2.
\end{displaymath}
Furthermore, the multiplicity of the eigenvalue $1$ is~$42$.

Note that
\begin{displaymath}
2+\sqrt{5}=[4;4,4,4,\ldots]
\quad\mbox{and}\quad
1+\sqrt{2}=[2;2,2,2,\ldots]
\end{displaymath}
exhibit \textit{digit-halving} in their continued fraction expansions, and are also the two largest eigenvalues of the $1$-step transition matrix of a
$1$-direction geodesic of slope $\alpha=2+\sqrt{5}$ on the L-surface; see \cite[Section~4.1]{BDY2}.

We comment that digit-halving is still exhibited if $\alpha$ is defined by \eqref{eq7.2.1} with $m=n=4k$ for any positive integers~$k$.
In this general case, the largest eigenvalue is the square of $[4k;4k,4k,4k,\ldots]$, and the second largest eigenvalue is the square of
$[2k;2k,2k,2k,\ldots]$.
\end{example}
\end{part3}

We complete this section by considering two consequences of Theorem~\ref{thm7.2.2} and its more general form.
The first one is fairly straightforward.

\begin{cor}\label{cor7.2.1}
The determinant of the $2$-step transition matrix $\bfA$ is equal to~$1$.
\end{cor}

\begin{proof}
We have shown earlier that $\det(\bfA\vert_\VVV)=1$.

Let $p(x)$ be the characteristic polynomial of~$\bfA\vert_\VVV$.
In view of Theorem~\ref{thm7.2.2}, it is clear that the characteristic polynomial of $\bfA$ is equal to $p(x)q(x)$, where $q(x)$ has degree $2d-2h$ and its roots are all the eigenvalues of $\bfA$ that are not eigenvalues of~$\bfA\vert_\VVV$.
In view of Lemma~\ref{lem7.2.3}, we have $q(x)=(x-1)^{2d-2h}$.
On the other hand, it follows from \eqref{eq7.2.38} that
\begin{displaymath}
p(x)=\prod_{i=1}^h(x-\lambda(\tau_i;+))(x-\lambda(\tau_i;-)),
\end{displaymath}
where $\tau_1,\ldots,\tau_h$ are the eigenvalues of~$\bfS$, and $\lambda(\tau_i;+)\lambda(\tau_i;-)=1$, $i=1,\ldots,h$.

Clearly the determinant of $\bfA$ is equal to $p(0)q(0)=1$.
\end{proof}

To motivate the second consequence, we start with the L-surface, and obtain a new polysquare translation surface with $6$ smaller square faces, each of area~$1/2$.
The trick is to draw two diagonals on each of the $3$ square faces of the L-surface.
We use the boundary pairing in the picture on the right in Figure~7.2.17, which gives the so-called \textit{diagonal subdivision surface} of the L-surface, or DS-L-surface.
Of course the genus remains the same; namely,~$2$.
Rotating this picture on the right anticlockwise by $45$ degrees and resizing, the DS-L-surface becomes a polysquare translation surface in the standard horizontal-vertical position.
It is easy to see that the DS-L-surface has only one horizontal street, which has length $6$ and is defined by the cycle $1,2,3,4,6,5$.
It also has only one vertical street, defined by the cycle $1,5,6,2,3,4$.
So the street size multiple of the DS-L-surface is equal to~$6$.

\begin{displaymath}
\begin{array}{c}
\includegraphics[scale=0.8]{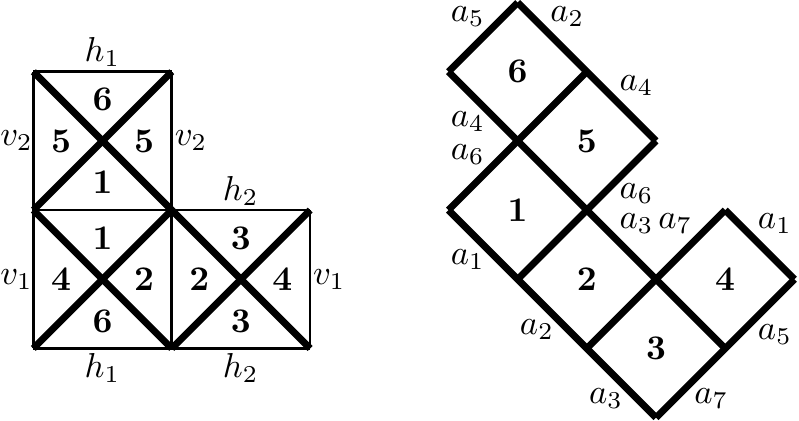}
\vspace{3pt}\\
\mbox{Figure 7.2.17: L-surface and DS-L-surface}
\end{array}
\end{displaymath}

Now \cite[Theorem~5.5.1]{BDY2} says that, applying the surplus shortline method to this particular \textit{one-street translation surface}, we obtain infinitely many explicit slopes with the property that a $1$-direction geodesic with any such slope is \textit{superuniform}.
However, this theorem is obtained by using \cite[Theorem~5.3.1]{BDY2}, the main result for the L-surface in that paper, with an exceptionally long proof.
So the mysterious connection between a concrete one-street translation surface and superuniformity there may seem at first sight a very deep, and perhaps accidental, result.
However, nothing is further from the truth.

If a polysquare translation surface has only one horizontal street, or only one vertical street, then we call it a \textit{one-street polysquare translation surface}.

A one-street polysquare translation surface does not in general satisfy the condition that if a horizontal street and a vertical street intersect, then the intersection is a unique square face.
Thus we cannot directly apply Theorem~\ref{thm7.2.2} as stated.
However, there is a general connection between such surfaces and superuniformity, and it is an extremely easy consequence of the general version of Theorem~\ref{thm7.2.2}.

Indeed, we simply need the condition that the polysquare translation surface has only one horizontal street (or one vertical street), but no restriction about the other direction. 
Then the street-spreading matrix $\bfS$ is a $1\times1$ matrix which gives rise to the largest eigenvalue of the $2$-step transition matrix~$\bfA$.
It follows that the second largest eigenvalue is irrelevant, which implies that the surplus shortline method provides explicit quadratic irrational slopes for which the irregularity exponent is equal to~$0$, so that the corresponding geodesic is superuniform.

We summarize our discussion as follows.

\begin{remark}
For every one-street polysquare translation surface, there exist infinitely many explicit quadratic irrational slopes such that every $1$-direction geodesic with any such slope is superuniform.
\end{remark}

We conclude this section by demonstrating that one-street polysquare translation surfaces are not rare.
The idea is based on a a far-reaching generalization of the DS-L-surface.

Recall the $2$-polysquare snake surface in Figure~7.1.4.
Snake surfaces are not restricted to $2$-polysquare translation surfaces, as illustrated in Figure~7.2.18.

\begin{displaymath}
\begin{array}{c}
\includegraphics[scale=0.8]{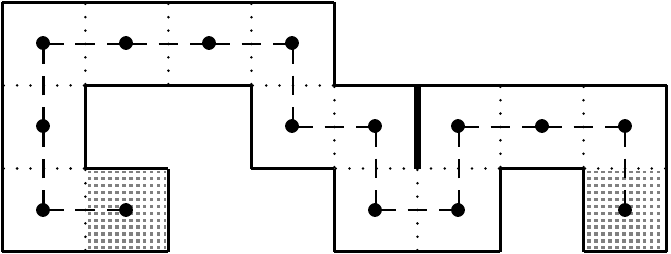}
\vspace{3pt}\\
\mbox{Figure 7.2.18: a snake surface}
\end{array}
\end{displaymath}

More precisely, let $\PPP$ be a finite polysquare translation surface.
We put a point at the center of every square face of~$\PPP$, and call it the \textit{capital} of the square face.
Consider the \textit{neighbor graph} of~$\PPP$, where the \textit{vertices} are the capitals, and where two capitals are joined by an \textit{edge} of the graph if the corresponding square faces share a common edge.
If the neighbor graph of $\PPP$ is a \textit{path}, with no branching, then $\PPP$ is called a snake surface, as shown in Figure~7.2.18.
If the neighbor graph of $\PPP$ is a \textit{tree}, then $\PPP$ is called a polysquare-tree translation surface, as shown in Figure~7.2.18 and~7.2.19.

\begin{displaymath}
\begin{array}{c}
\includegraphics[scale=0.8]{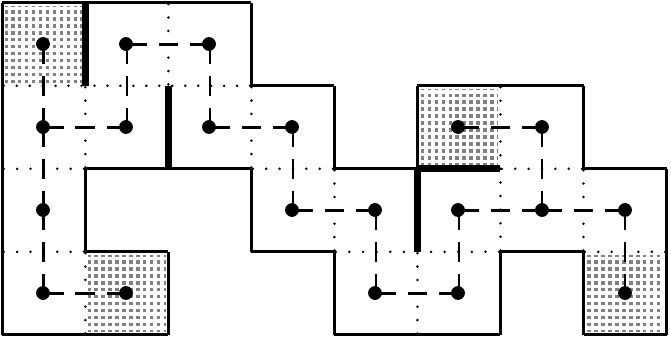}
\vspace{3pt}\\
\mbox{Figure 7.2.19: a polysquare-tree translation surface}
\end{array}
\end{displaymath}

Given any graph, the number of edges at a vertex is called the \textit{degree} of the vertex. 
In a path, the degree of every vertex is~$2$, except at the two end-vertices where the degree is~$1$.
In a tree we may have higher degrees.
For example, the neighbor tree of the polysquare-tree translation surface in Figure~7.2.19 has two vertices of degree~$3$. 

One-street polysquare translation surfaces are not rare, because the diagonal subdivision of every polysquare-tree translation surface has only one street.

The proof goes by induction on the number of square faces in the polysquare.
The inductive step is based on the elementary result in graph theory that every tree has a vertex of degree~$1$.
Indeed, it is easy to see that a tree with $n$ vertices has precisely $n-1$ edges, and that the sum of the degrees is twice the number of edges.
Since the sum of the degrees is $2(n-1)$, this implies that there must be a vertex with degree less than~$2$, and so must be equal to~$1$.

In terms of the polysquare-tree region, this result in graph theory means that there is always a square face with $3$ sides not shared with neighboring square faces.
Let us call them end square faces.
We have highlighted these end square faces in Figures 7.2.18 and~7.2.19.

Figure~7.2.20 shows the smaller polysquare-tree translation surface obtained from the polysquare-tree translation surface in Figure~7.2.19 by removing one of the end square faces.
By the induction hypothesis, the diagonal subdivision of this smaller polysquare-tree translation surface is a one-street translation surface, labelled by the numbers $1$ to~$38$.

\begin{displaymath}
\begin{array}{c}
\includegraphics[scale=0.8]{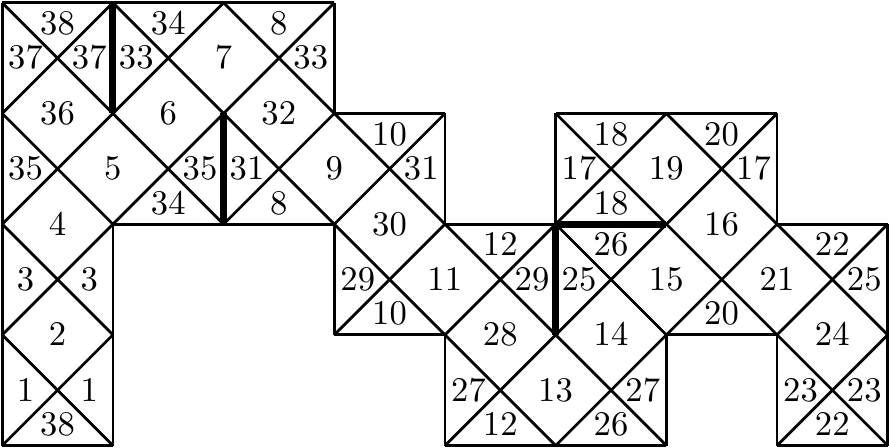}
\vspace{3pt}\\
\mbox{Figure 7.2.20: removing an end square face}
\end{array}
\end{displaymath}

Note that we carefully start the labelling from a common edge shared with the removed end square face, and keep moving in the north-east direction.
Adding back the removed end square face, Figure~7.2.21 shows an explicit one-street in the original polysquare-street translation surface, where the consecutive squares along the street are labelled with the consecutive integers from $1$ to~$40$.

\begin{displaymath}
\begin{array}{c}
\includegraphics[scale=0.8]{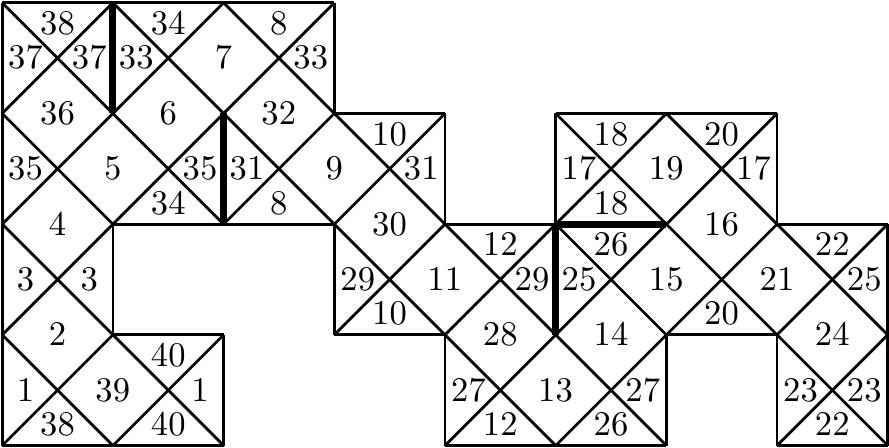}
\vspace{3pt}\\
\mbox{Figure 7.2.21: extension of the labelling}
\end{array}
\end{displaymath}

Note that the labelling in Figure~7.2.21 is obtained as a trivial extension of the labelling in Figure~7.2.20 by including $39$ and~$40$.
This completes the inductive step.

We cannot entirely drop the tree condition for the neighbor graph, as cycles may spoil the argument.
Consider the two polysquare cycles in Figure~7.2.22.

\begin{displaymath}
\begin{array}{c}
\includegraphics[scale=0.8]{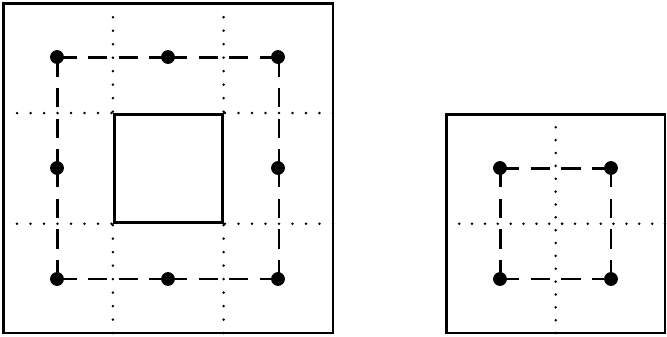}
\vspace{3pt}\\
\mbox{Figure 7.2.22: two polysquare-cycles}
\end{array}
\end{displaymath}

Neither of these polysquare-cycles leads to a one-street diagonal decomposition surface, as can be shown in Figure~7.2.23.

\begin{displaymath}
\begin{array}{c}
\includegraphics[scale=0.8]{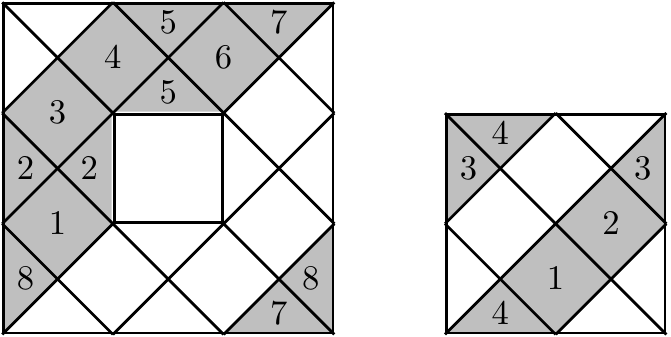}
\vspace{3pt}\\
\mbox{Figure 7.2.23: diagonal decomposition of two polysquare-cycles}
\end{array}
\end{displaymath}

In the picture on the left, we have a tilted street of length~$8$, covering only half the tilted square faces of the polysquare.
In the picture on the right, we have a tilted street of length~$4$, again covering only half the tilted square faces of the polysquare.

%%%%%%%%%%
%
% SECTION 7.3
%
%%%%%%%%%%

\subsection{Application of Theorem~\ref{thm7.2.2}: super-fast escape to infinity}\label{sec7.3}

Our goal in this section is to find explicit quadratic irrational slopes which give rise to orbits that exhibit \textit{super-fast} escape rate to infinity.
More precisely, given any $\eps>0$, we shall find explicit quadratic irrational slopes such that for any orbit with any such slope, there exists an infinite sequence $T_n$, $n\ge1$, such that $T_n\to\infty$ as $n\to\infty$, and the diameter of the initial segment of length $T_n$ of the orbit exceeds $c_0T_n^{1-\eps}$, where $c_0$ is an absolute constant depending only on the given special slope.

Billiard orbits and geodesics in integrable polysquare systems behave in the same way for all quadratic irrational slopes.
In complete contrast, billiard orbits and geodesics in infinite non-integrable polysquare systems may behave in completely different ways for different quadratic irrational slopes.
We may therefore have the full spectrum from super-slow to super-fast escape rate to infinity.

\begin{example}\label{ex7.3.1}
We return to the problem of billiard in the $\infty$-L-strip region shown in Figure~7.3.1.
In \cite[Theorem~6.7.2]{BCY}, it is shown that there are infinitely many quadratic irrational slopes such that a billiard orbit starting from a corner point with such a slope is dense on the $\infty$-L-strip region and also exhibits super-slow logarithmic escape rate to infinity.
In particular, any initial segment of such an orbit of length~$T$, $T\ge2$, has distance at most $c_0\log T$ from the starting point, where $c_0$ is an absolute constant depending only on the given special slope.

\begin{displaymath}
\begin{array}{c}
\includegraphics[scale=0.8]{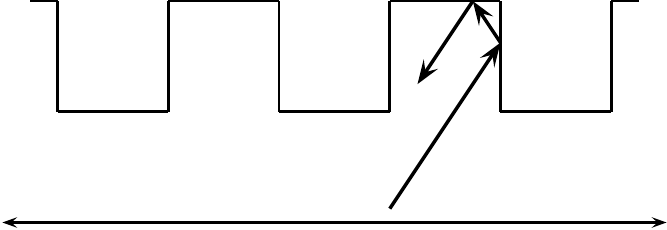}
\vspace{3pt}\\
\mbox{Figure 7.3.1: billiard in the $\infty$-L-strip region}
\end{array}
\end{displaymath}

As illustrated in Figure~7.3.2, the L-shape is the building block of the $\infty$-L-strip region.

\begin{displaymath}
\begin{array}{c}
\includegraphics[scale=0.8]{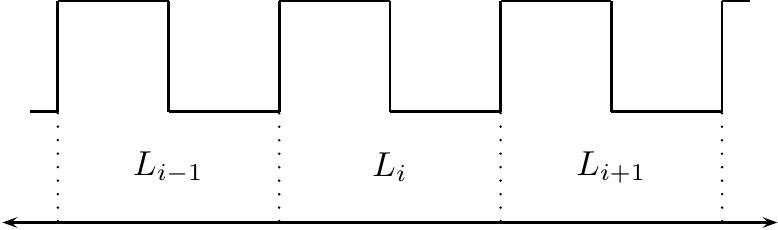}
\vspace{3pt}\\
\mbox{Figure 7.3.2: the $\infty$-L-strip region with its building blocks}
\end{array}
\end{displaymath}

The trick of unfolding reduces the problem of billiard on the $\infty$-L-strip region, a $4$-direction flow, to the problem of a $1$-direction geodesic flow on an appropriate infinite polysquare translation surface.
We call this surface the \textit{translation surface for $\infty$-L-strip billiard}, and denote it by $\Bil(\infty;1)$.

To describe $\Bil(\infty;1)$, we first unfold each of the L-shapes $L_i$ by reflecting it horizontally and vertically to obtain $4$ reflected copies and then put them together to form a $4$-copy-$L_i$.
We then obtain $\Bil(\infty;1)$ by gluing together the infinitely many copies of these $4$-copy-$L_i$.

Clearly there is billiard flow from each L-shape $L_i$ to its two immediate neighbours $L_{i-1}$ and $L_{i+1}$.
We therefore need to identify corresponding edges of these $4$-copy versions very carefully.
A simple examination will convince the reader that the edge identifications can be as illustrated in Figure~7.3.3.
Note that a $1$-direction geodesic on $\Bil(\infty;1)$ that goes from the vertical edge $v_2^{(i+1)}$ to the vertical edge $v_2^{(i)}$ corresponds to a billiard path going from $L_{i+1}$ to~$L_i$, whereas a $1$-direction geodesic on $\Bil(\infty;1)$ that goes from the vertical edge $v_3^{(i-1)}$ to the vertical edge $v_3^{(i)}$ corresponds to a billiard path going from $L_i$ to~$L_{i+1}$.

\begin{displaymath}
\begin{array}{c}
\includegraphics[scale=0.8]{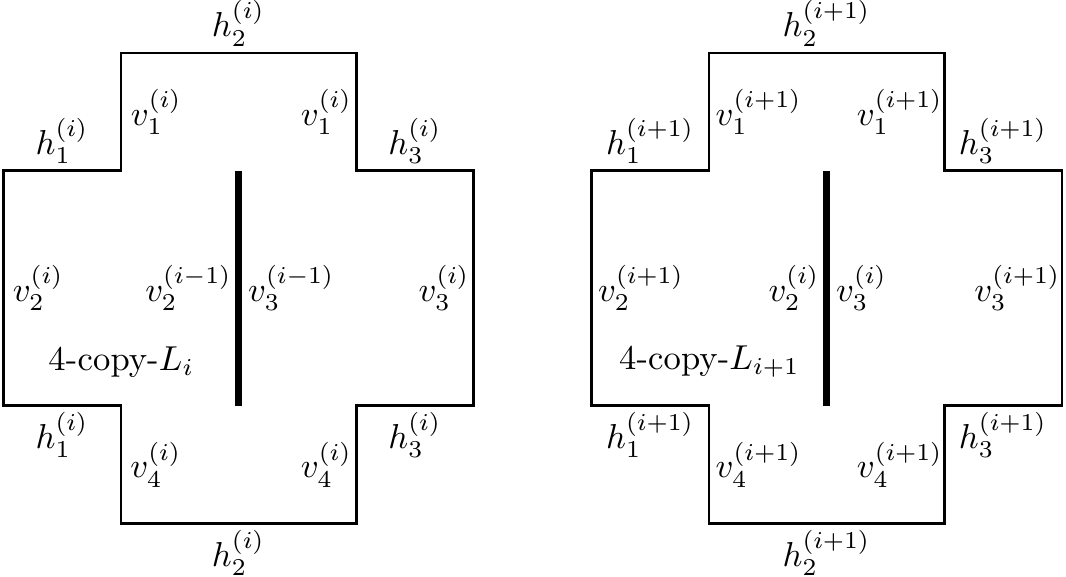}
\vspace{3pt}\vspace{3pt}\\
\mbox{Figure 7.3.3: $4$-copy-$L_i$ and $4$-copy-$L_{i+1}$ together with edge identifications}
\end{array}
\end{displaymath}

Thus the motion of the infinite L-strip billiard can be described in terms of $1$-direction geodesic flow on a finite polysquare translation surface~$\PPP$, the period-surface of $\Bil(\infty;1)$, illustrated in Figure~7.3.4, where we have rotated by $90$ degrees in the anticlockwise direction.

\begin{displaymath}
\begin{array}{c}
\includegraphics[scale=0.8]{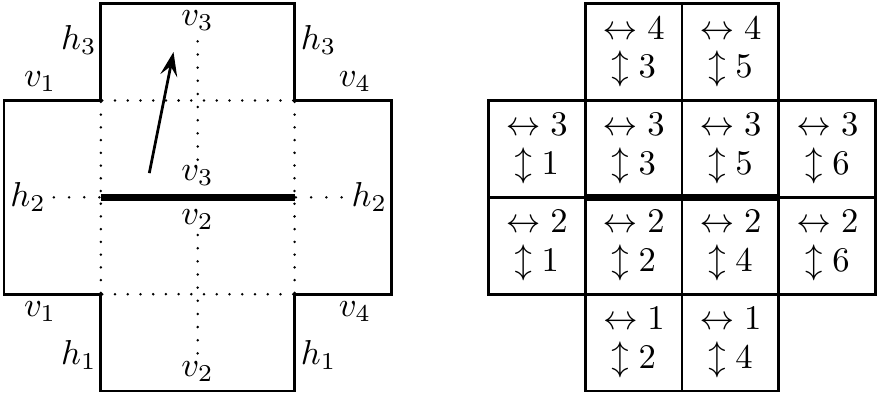}
\vspace{3pt}\vspace{3pt}\\
\mbox{Figure 7.3.4: the period surface $\PPP$ of $\Bil(\infty;1)$}
\end{array}
\end{displaymath}

Clearly $\PPP$ has $12$ square faces, $4$ horizontal streets and $6$ vertical streets, labelled according to the picture on the right where, for instance, the information
$\leftrightarrow3$ and $\updownarrow5$ in a square face indicates that this square face is on the $3$-rd horizontal street and the $5$-th vertical street.
Note that the anticlockwise rotation by $90$ degrees makes the application of Theorem~\ref{thm7.2.2} more convenient.

The boundary identification in $\PPP$ is clearly the simplest perpendicular translation.
Consider a $1$-direction geodesic on $\PPP$ with slope given by the arrow in the picture on the left in Figure~7.3.4.
If such a geodesic hits the edge $v_3$ at the top, then it jumps down vertically to the identified edge $v_3$ in the middle.
In terms of the $\infty$-L-strip billiard this means precisely that the billiard moves from $L_i$ to~$L_{i+1}$, \textit{i.e.}, \textit{moving right}.
If such a geodesic hits the edge $v_2$ in the middle, then it jumps down vertically to the  identified edge $v_2$ at the bottom.
In terms of the $\infty$-L-strip billiard this means precisely that the billiard moves from $L_i$ to~$L_{i-1}$, \textit{i.e.}, \textit{moving left}.
Thus by counting the number of edge cuttings for $v_2$ and for~$v_3$, we know the total number of steps moving to the left and the total number of steps moving to the right.
Taking the difference of these two numbers, we then know precisely which particular constituent L-shape happens to contain the billiard at a given moment. 

Now the edge cutting numbers are expressed in terms of powers of the relevant eigenvalues of the $2$-step transition matrix corresponding to the given quadratic irrational slope. 
It makes it plausible to guess that the \textit{irregularity exponent} of the period-surface $\PPP$ gives the escape rate to infinity for the infinite L-strip billiard, as long as the shortline method works.
Next we expand on this and justify this intuition.

We need to choose the parameters $m$ and $n$, which have to be integer multiples of the lengths of the horizontal and vertical streets respectively.
For this period surface, the horizontal streets have length $2$ or~$4$, while the vertical streets all have length~$2$.
Thus we let $m=4k$ and $n=2k$, where $k\ge1$ is any integer.
Consequently, we let
\begin{equation}\label{eq7.3.1}
\alpha_k=[2k;4k,2k,4k,\ldots]=k+\frac{\sqrt{4k^2+2}}{2}.
\end{equation}

Since the period-surface $\PPP$ has $12$ squares, the original shortline method leads to a $24\times24$ $2$-step transition matrix~$\bfA$.
To find the relevant eigenvalues and eigenvectors of such a large matrix directly is clearly rather inconvenient. 

Instead we apply Theorem~\ref{thm7.2.2}, which leads to a much smaller $4\times4$ street-spreading matrix $\bfS$ defined by \eqref{eq7.2.36}.
We follow the notation in Section~\ref{sec7.2}.
We consider the $\bfA$-invariant subspace $\VVV$ generated by the $8$ vectors $\bfu_i,\bfv_i$, $i=1,2,3,4$, defined by \eqref{eq7.2.14} and
\eqref{eq7.2.15}.
The relevant eigenvalues of~$\bfA$ are then eigenvalues of the matrix~$\bfA\vert_\VVV$, defined by \eqref{eq7.2.37}.

Consider the horizontal street corresponding to~$\bfu_1$, as highlighted in the picture on the left in Figure~7.3.5.
Clearly $J_1=\{2,4\}$.
Using \eqref{eq7.2.17}, we have
\begin{align}\label{eq7.3.2}
(\bfA-I)\bfu_1
&
=(\bfA-I)[\{2k^2\uparrow_{1,2},2k^2\uparrow_{1,4}\}]+(\bfA-I)[\{2k^2\uparrow_{2,2},2k^2\uparrow_{2,4}\}]
\nonumber
\\
&
=4k^2(\bfu_1+\bfv_1)+4k^2(\bfu_2+\bfv_2).
\end{align}

With $J_2=\{1,2,4,6\}$, for the horizontal street corresponding to~$\bfu_2$, as highlighted in the picture on the right in Figure~7.3.5, a similar argument gives
\begin{equation}\label{eq7.3.3}
(\bfA-I)\bfu_2
=2k^2(\bfu_1+\bfv_1)+4k^2(\bfu_2+\bfv_2)+2k^2(\bfu_3+\bfv_3).
\end{equation}
\begin{displaymath}
\begin{array}{c}
\includegraphics[scale=0.8]{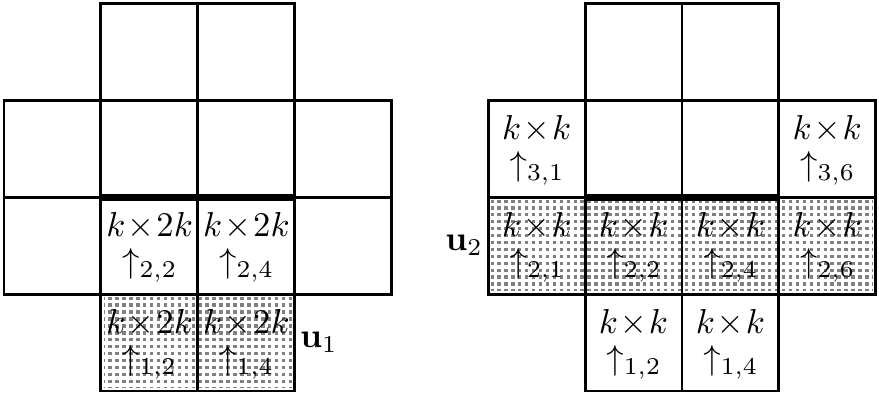}
\vspace{3pt}\vspace{3pt}\\
\mbox{Figure 7.3.5: almost vertical units of type $\uparrow$ in $\bfu_1$ and $\bfu_2$}
\end{array}
\end{displaymath}

With $J_3=\{1,3,5,6\}$, for the horizontal street corresponding to~$\bfu_3$, as highlighted in the picture on the left in Figure~7.3.6, a similar argument gives
\begin{equation}\label{eq7.3.4}
(\bfA-I)\bfu_3
=2k^2(\bfu_2+\bfv_2)+4k^2(\bfu_3+\bfv_3)+2k^2(\bfu_4+\bfv_4).
\end{equation}

With $J_4=\{3,5\}$, for the horizontal street corresponding to~$\bfu_3$, as highlighted in the picture on the right in Figure~7.3.6, a similar argument gives
\begin{equation}\label{eq7.3.5}
(\bfA-I)\bfu_4
=4k^2(\bfu_3+\bfv_3)+4k^2(\bfu_4+\bfv_4).
\end{equation}
\begin{displaymath}
\begin{array}{c}
\includegraphics[scale=0.8]{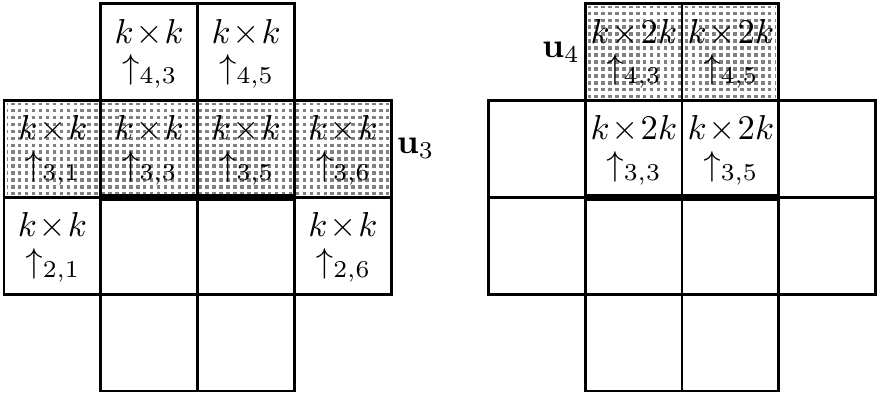}
\vspace{3pt}\vspace{3pt}\\
\mbox{Figure 7.3.6: almost vertical units of type $\uparrow$ in $\bfu_3$ and $\bfu_4$}
\end{array}
\end{displaymath}

It follows from \eqref{eq7.3.2}--\eqref{eq7.3.5} that the street-spreading matrix is given by
\begin{displaymath}
\bfS
=\begin{pmatrix}
4k^2&2k^2&0&0\\
4k^2&4k^2&2k^2&0\\
0&2k^2&4k^2&4k^2\\
0&0&2k^2&4k^2
\end{pmatrix}
=k^2\begin{pmatrix}
4&2&0&0\\
4&4&2&0\\
0&2&4&4\\
0&0&2&4
\end{pmatrix}.
\end{displaymath}
This has eigenvalues and corresponding eigenvectors
\begin{align}
\tau_1=8k^2,
&\quad\psi_1=(1,2,2,1)^T,
\nonumber
\\
\tau_2=6k^2,
&\quad\psi_2=(1,1,-1,-1)^T,
\nonumber
\\
\tau_3=2k^2,
&\quad\psi_3=(1,-1,-1,1)^T,
\nonumber
\\
\tau_4=0,
&\quad\psi_4=(1,-2,2,-1)^T.
\nonumber
\end{align}
Using the formula \eqref{eq7.2.38}, we obtain corresponding eigenvalues of $\bfA\vert_\VVV$ given by
\begin{align}
\lambda(\tau_1;\pm)
&
=1+\frac{8k^2\pm\sqrt{64k^4+32k^2}}{2}
=1+4k^2\pm2k\sqrt{4k^2+2},
\nonumber
\\
\lambda(\tau_2;\pm)
&
=1+\frac{6k^2\pm\sqrt{36k^4+24k^2}}{2}
=1+3k^2\pm k\sqrt{9k^2+6},
\nonumber
\\
\lambda(\tau_3;\pm)
&
=1+\frac{2k^2\pm\sqrt{4k^4+8k^2}}{2}
=1+k^2\pm k\sqrt{k^2+2},
\nonumber
\\
\lambda(\tau_4;\pm)
&
=1.
\nonumber
\end{align}
The three relevant eigenvalues of $\bfA\vert_\VVV$ and corresponding eigenvectors are therefore
\begin{align}
\lambda_1
=1+4k^2+2k\sqrt{4k^2+2},
&
\quad\Psi_1=(1,2,2,1,b_{11},b_{12},b_{13},b_{14})^T,
\label{eq7.3.6}
\\
\lambda_2
=1+3k^2+k\sqrt{9k^2+6},
&
\quad\Psi_2=(1,1,-1,-1,b_{21},b_{22},b_{23},b_{24})^T,
\label{eq7.3.7}
\\
\lambda_3
=1+k^2+k\sqrt{k^2+2},
&
\quad\Psi_3=(1,-1,-1,1,b_{31},b_{32},b_{33},b_{34})^T,
\nonumber
\end{align}
where the values of $b_{ij}$, $i=1,2,3$, $j=1,2,3,4$, can be calculated using \eqref{eq7.2.39}.

Recall that by counting the number of edge cuttings for $v_2$ and for~$v_3$, we know the total number of steps of the billiard orbit moving to the left and the total number of steps of the billiard orbit moving to the right.
The difference between these two numbers then gives us information on the movement of the billiard orbit.

We now proceed to find the edge cutting numbers of $v_2$ and~$v_3$.

The eigenvector of $\bfA$ corresponding to the eigenvector $\Psi_1$ of $\bfA\vert_\VVV$ is given by
\begin{displaymath}
\bfu_1+2\bfu_2+2\bfu_3+\bfu_4+b_{11}\bfv_1+b_{12}\bfv_2+b_{13}\bfv_3+b_{14}\bfv_4.
\end{displaymath}

We first find the number of almost vertical units counted here that cut the edge~$v_2$; see Figure~7.3.4.

Concerning the contribution from those units counted by~$\bfu_1$, it is clear from the picture on the left in Figure~7.3.5 that we have $2k^2$ units of type $\uparrow_{2,2}$ and $2k^2$ units of type $\uparrow_{2,4}$, making a total of~$4k^2$.

Concerning the contribution from those units counted by~$\bfu_2$, it is clear from the picture on the right in Figure~7.3.5 that we have $k^2$ units of type $\uparrow_{2,2}$ and $k^2$ units of type $\uparrow_{2,4}$, making a total of~$2k^2$.

Concerning the contribution from those units counted by $\bfu_3$ and~$\bfu_4$, it is clear from Figure~7.3.6 that there is none.

To study the contribution from those units counted by $\bfv_1,\ldots,\bfv_4$, we appeal to \eqref{eq7.2.15} and Figure~7.2.1.
There is clearly no contribution from $\bfv_1,\bfv_3,\bfv_4$, while for~$\bfv_2$, we have positive contributions from $\nuparrow_{2,1}$ and
$\nuparrow_{2,2}$, and negative contributions from $\uparrow_{2,2}$ and $\uparrow_{2,4}$, each counted with the multiplicity of $J_2^*$, and so there is perfect cancellation.

Thus for the edge~$v_2$, we have a total count of $4k^2+2(2k^2)=8k^2$.

We next find the number of almost vertical units counted here that cut the edge~$v_3$; see Figure~7.3.4.

Concerning the contribution from those units counted by $\bfu_1$ and~$\bfu_2$, it is clear from Figure~7.3.5 that there is none.

Concerning the contribution from those units counted by~$\bfu_3$, it is clear from the picture on the left in Figure~7.3.6 that we have $k^2$ units of type $\uparrow_{4,3}$ and $k^2$ units of type $\uparrow_{4,5}$, making a total of~$2k^2$.

Concerning the contribution from those units counted by~$\bfu_4$, it is clear from the picture on the right in Figure~7.3.6 that we have $2k^2$ units of type
$\uparrow_{4,3}$ and $2k^2$ units of type $\uparrow_{4,5}$, making a total of~$4k^2$.

There is clearly no contribution from $\bfv_1,\bfv_2,\bfv_3$, while for~$\bfv_4$, we have positive contributions from $\nuparrow_{4,3}$ and $\nuparrow_{4,5}$, and negative contributions from $\uparrow_{4,3}$ and $\uparrow_{4,5}$, each counted with the multiplicity of $J_4^*$, and so there is perfect cancellation.

Thus for the edge~$v_3$, we have a total count of $2(2k^2)+4k^2=8k^2$.

Since the two counts are the same, it follows that the eigenvalue $\lambda_1$ does not contribute to the difference between the edge cutting numbers of $v_2$ and~$v_3$.

\begin{remark}
The perfect cancellation is not surprising at all.
Indeed, the high powers corresponding to the largest eigenvalue represents the \textit{main term} of the problem.
Any lack of cancellation here would violate the uniformity of the geodesic.
With a quadratic irrational slope, this would specifically violate the Gutkin--Veech theorem \cite{G,V2,V3} that guarantees uniformity of any
$1$-direction geodesic in any polysquare translation surface with irrational slope.
See also \cite{Bo1,Bo2}.
\end{remark}

The eigenvector of $\bfA$ corresponding to the eigenvector $\Psi_2$ of $\bfA\vert_\VVV$ is given by
\begin{displaymath}
\bfu_1+\bfu_2-\bfu_3-\bfu_4+b_{21}\bfv_1+b_{22}\bfv_2+b_{23}\bfv_3+b_{24}\bfv_4.
\end{displaymath}

We first find the number of almost vertical units counted here that cut the edge~$v_2$.
Recall that the contribution from $\bfu_1$ and $\bfu_2$ are $4k^2$ and $2k^2$ respectively, while there is no contribution from $\bfu_3$ and
$\bfu_4$, and also no contribution from $\bfv_1,\ldots,\bfv_4$.
Thus for the edge~$v_2$, we have a total count of $4k^2+2k^2=6k^2$.

We next find the number of almost vertical units counted here that cut the edge~$v_3$.
Recall that the contribution from $\bfu_3$ and $\bfu_4$ are $2k^2$ and $4k^2$ respectively, while there is no contribution from $\bfu_1$ and
$\bfu_2$, and also no contribution from $\bfv_1,\ldots,\bfv_4$.
Thus for the edge~$v_3$, we have a total count of $-2k^2-4k^2=-6k^2$.

The difference, in absolute value, is therefore~$12k^2$.
Let $c_2=c_{2,1}$ in \eqref{eq7.2.11}.
Then it follows that the second eigenvalue $\lambda_2$ contributes $12c_2\lambda_2^rk^2$ to the difference between the edge cutting numbers of $v_2$ and~$v_3$.

So the deviation from the starting point comes from the second largest eigenvalue, and this has the order of magnitude $\lambda_2^r$ compared to the order of magnitude $\lambda_1^r$ of the main term.
Choosing $T=\lambda_1^r$, we have $\lambda_2^r\asymp T^{\kappa_0}$, where
\begin{equation}\label{eq7.3.8}
\kappa_0=\kappa_0(k)=\frac{\log\lambda_2}{\log\lambda_1}=\frac{2\log k+\log6}{2\log k+\log8}+o(1),
\end{equation}
in view of \eqref{eq7.3.6} and \eqref{eq7.3.7}, is the irregularity exponent of a $1$-direction geodesic of slope \eqref{eq7.3.1} on the period-surface in Figure~7.3.4.

Clearly $\kappa_0=\kappa_0(k)\to1$ as $k\to\infty$.
So we have just established $T^{\kappa_0}=T^{1-\eps}$ size \textit{super-fast escape rate to infinity} for the infinite L-strip billiard with the explicit class of quadratic irrational slopes in \eqref{eq7.3.1} where the parameter $k\ge1$ is any integer.

It is easy to see that the irregularity exponent $\kappa_0=\kappa_0(k)$ in \eqref{eq7.3.8} is {\it precisely} the escape rate to infinity of this infinite billiard.
Indeed, the exponent of the escape rate to infinity cannot be larger than the expression \eqref{eq7.3.8} coming from the two largest eigenvalues.
\end{example}

\begin{example}\label{ex7.3.2}
Consider the Ehrenfest wind-tree model with $a=b=1/2$, rescaled so that we have square faces of unit area.
Consider a $4$-direction billiard trajectory in the infinite region in Figure~7.3.7.
Here the building block of this infinite polysquare translation surface is an L-shape.

\begin{displaymath}
\begin{array}{c}
\includegraphics[scale=0.8]{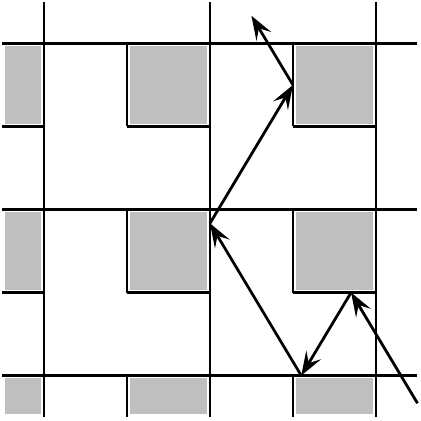}
\vspace{3pt}\\
\mbox{Figure 7.3.7: Ehrenfest wind-tree billiard model}
\end{array}
\end{displaymath}

As usual, this infinite billiard model is equivalent to a $1$-direction geodesic flow on an infinite polysquare translation surface that we denote by $\Bil(\infty;2)$, with the index $2$ here indicating that this is double periodic.
To construct $\Bil(\infty;2)$, we take one of the L-shape building blocks, and unfold the $4$-direction billiard flow on it to a $1$-direction geodesic flow on a $4$-copy version of it, obtained by reflecting horizontally and vertically.
Note that each $4$-copy L-shape has a right-neighbor, a left-neighbor, a down-neighbor and an up-neighbor in $\Bil(\infty;2)$, and we need appropriate edge identification which amount to a more complicated version of Figure~7.3.3.
The period-surface $\Bil(2)$ of $\Bil(\infty;2)$ is shown in the picture on the left in Figure~7.3.8. 

\begin{displaymath}
\begin{array}{c}
\includegraphics[scale=0.8]{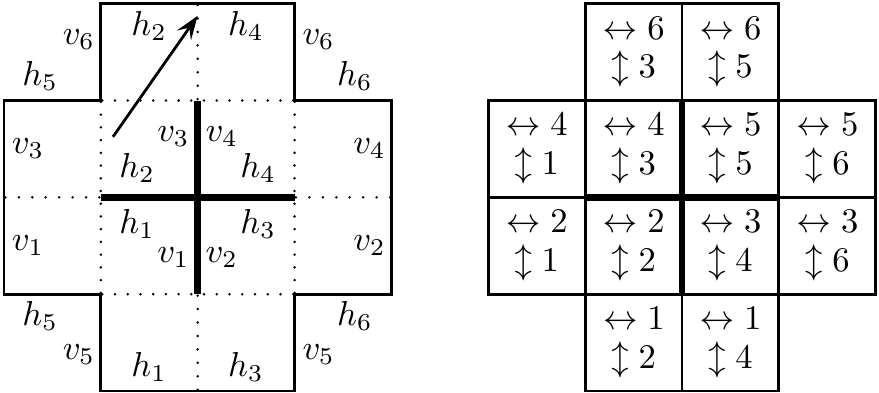}
\vspace{3pt}\\
\mbox{Figure 7.3.8: the period surface $\Bil(2)$ of $\Bil(\infty;2)$}
\end{array}
\end{displaymath}

The horizontal and vertical streets of $\Bil(2)$ are indicated in the picture on the right in Figure~7.3.8 where, for instance, the entries
$\leftrightarrow3$ and $\updownarrow4$ in a square face indicates that the square face is on the $3$-rd horizontal street and the $4$-th vertical street.
It is easy to see that $\Bil(2)$ has $6$ horizontal streets and $6$ vertical streets.

The boundary identification in $\Bil(2)$ is clearly the simplest perpendicular translation.
Consider a $1$-direction geodesic in $\Bil(2)$ with slope given by the arrow in the picture on the left in Figure~7.3.8.
If such a geodesic hits the edge $h_2$ or $h_4$ at the top, then it jumps down vertically to the identified edge $h_2$ or $h_4$ in the middle.
In terms of the Ehrenfest wind-tree billiard this means precisely that the billiard moves from one L-shape to the next above, \textit{i.e.}, \textit{moving up}.
If such a geodesic hits the edge $h_1$ or $h_3$ in the middle, then it jumps down vertically to the  identified edge $h_1$ or $h_3$ at the bottom.
In terms of the Ehrenfest wind-tree billiard this means precisely that the billiard moves from one L-shape to the next below, \textit{i.e.}, \textit{moving down}.
Thus by counting the number of edge cuttings for $h_2,h_4$ and for $h_1,h_3$, we know the total number of steps moving up and the total number of steps moving down.
Taking the difference of these two numbers, we then know precisely the vertical location of the billiard at a given moment. 

We need to choose the parameters $m$ and $n$, which have to be integer multiples of the lengths of the horizontal and vertical streets respectively.
For this period surface, the horizontal and vertical streets all have length $2$.
Thus we let $m=2k$ and $n=2k$, where $k\ge1$ is any integer.
Consequently, we let
\begin{equation}\label{eq7.3.9}
\alpha_k=[2k;2k,2k,2k,\ldots]=k+\sqrt{k^2+1}.
\end{equation}

We apply Theorem~\ref{thm7.2.2}, leading to a $6\times6$ street-spreading matrix $\bfS$ defined by \eqref{eq7.2.36}.
We follow the notation in Section~\ref{sec7.2}.
We consider the $\bfA$-invariant subspace $\VVV$ generated by the $12$ vectors $\bfu_i,\bfv_i$, $i=1,\ldots,6$, defined by \eqref{eq7.2.14} and \eqref{eq7.2.15}.
The relevant eigenvalues of~$\bfA$ are then eigenvalues of the matrix~$\bfA\vert_\VVV$, defined by \eqref{eq7.2.37}.

Consider the horizontal street corresponding to~$\bfu_1$, as highlighted in the picture on the left in Figure~7.3.9.
Clearly $J_1=\{2,4\}$.
Using \eqref{eq7.2.17}, we have
\begin{align}\label{eq7.3.10}
(\bfA-I)\bfu_1
&
=(\bfA-I)[\{k^2\uparrow_{1,2},k^2\uparrow_{1,4}\}]+(\bfA-I)[\{k^2\uparrow_{2,2},k^2\uparrow_{3,4}\}]
\nonumber
\\
&
=2k^2(\bfu_1+\bfv_1)+k^2(\bfu_2+\bfv_2)+k^2(\bfu_3+\bfv_3).
\end{align}

With $J_2=\{1,2\}$, for the horizontal street corresponding to~$\bfu_2$, as highlighted in the picture in the middle in Figure~7.3.9, a similar argument gives
\begin{equation}\label{eq7.3.11}
(\bfA-I)\bfu_2
=k^2(\bfu_1+\bfv_1)+2k^2(\bfu_2+\bfv_2)+k^2(\bfu_4+\bfv_4).
\end{equation}

With $J_3=\{4,6\}$, for the horizontal street corresponding to~$\bfu_3$, as highlighted in the picture on the right in Figure~7.3.9, a similar argument gives
\begin{equation}\label{eq7.3.12}
(\bfA-I)\bfu_3
=k^2(\bfu_1+\bfv_1)+2k^2(\bfu_3+\bfv_3)+k^2(\bfu_5+\bfv_5).
\end{equation}
\begin{displaymath}
\begin{array}{c}
\includegraphics[scale=0.8]{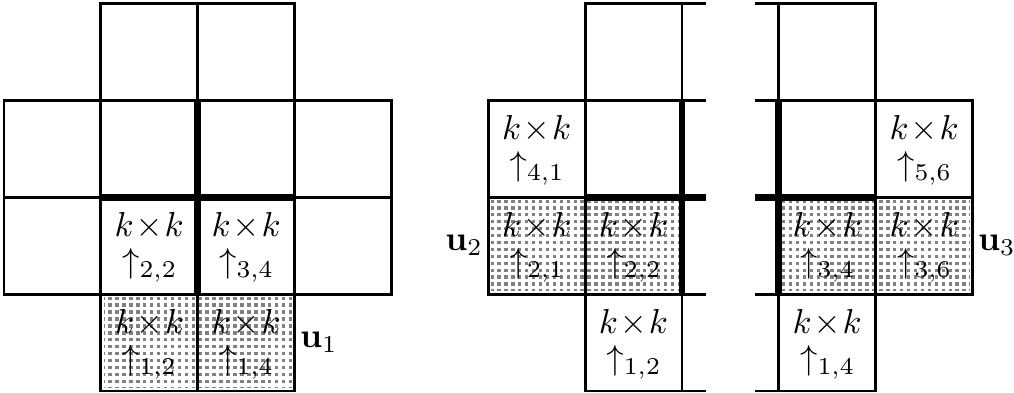}
\vspace{3pt}\vspace{3pt}\\
\mbox{Figure 7.3.9: almost vertical units of type $\uparrow$ in $\bfu_1,\bfu_2,\bfu_3$}
\end{array}
\end{displaymath}

With $J_4=\{1,3\}$, for the horizontal street corresponding to~$\bfu_4$, as highlighted in the picture on the left in Figure~7.3.10, a similar argument gives
\begin{equation}\label{eq7.3.13}
(\bfA-I)\bfu_4
=k^2(\bfu_2+\bfv_2)+2k^2(\bfu_4+\bfv_4)+k^2(\bfu_6+\bfv_6).
\end{equation}

With $J_5=\{5,6\}$, for the horizontal street corresponding to~$\bfu_5$, as highlighted in the picture in the middle in Figure~7.3.10, a similar argument gives
\begin{equation}\label{eq7.3.14}
(\bfA-I)\bfu_5
=k^2(\bfu_3+\bfv_3)+2k^2(\bfu_5+\bfv_5)+k^2(\bfu_6+\bfv_6).
\end{equation}

With $J_6=\{3,5\}$, for the horizontal street corresponding to~$\bfu_5$, as highlighted in the picture in the middle in Figure~7.3.10, a similar argument gives
\begin{equation}\label{eq7.3.15}
(\bfA-I)\bfu_6
=k^2(\bfu_4+\bfv_4)+k^2(\bfu_5+\bfv_5)+2k^2(\bfu_6+\bfv_6).
\end{equation}
\begin{displaymath}
\begin{array}{c}
\includegraphics[scale=0.8]{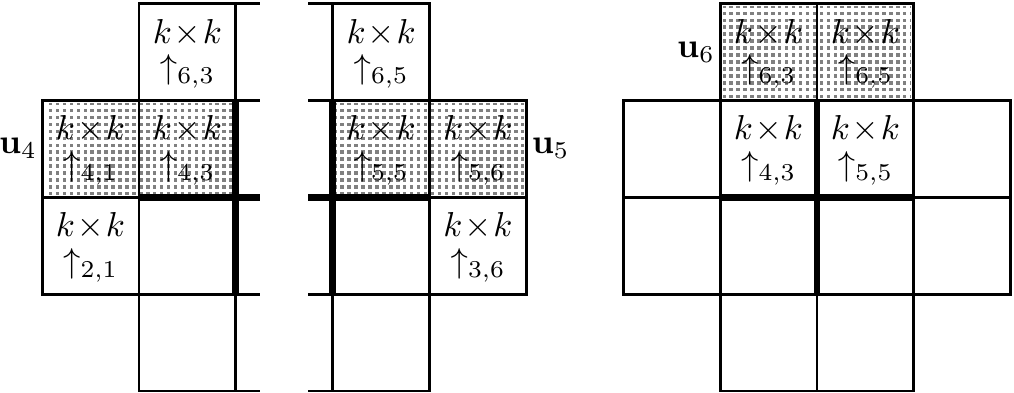}
\vspace{3pt}\vspace{3pt}\\
\mbox{Figure 7.3.10: almost vertical units of type $\uparrow$ in $\bfu_4,\bfu_5,\bfu_6$}
\end{array}
\end{displaymath}

It follows from \eqref{eq7.3.10}--\eqref{eq7.3.15} that the street-spreading matrix is given by
\begin{displaymath}
\bfS
=\begin{pmatrix}
2k^2&k^2&k^2&0&0&0\\
k^2&2k^2&0&k^2&0&0\\
k^2&0&2k^2&0&k^2&0\\
0&k^2&0&2k^2&0&k^2\\
0&0&k^2&0&2k^2&k^2\\
0&0&0&k^2&k^2&2k^2
\end{pmatrix}
=k^2\begin{pmatrix}
2&1&1&0&0&0\\
1&2&0&1&0&0\\
1&0&2&0&1&0\\
0&1&0&2&0&1\\
0&0&1&0&2&1\\
0&0&0&1&1&2
\end{pmatrix}.
\end{displaymath}
This has non-zero eigenvalues and corresponding eigenvectors
\begin{align}
\tau_1=4k^2,
&\quad\psi_1=(1,1,1,1,1,1)^T,
\nonumber
\\
\tau_2=3k^2,
&\quad\psi_{2,1}=(-1,0,-1,1,0,1)^T,
\nonumber
\\
&\quad\psi_{2,2}=(0,-1,1,-1,1,0)^T,
\nonumber
\\
\tau_3=k^2,
&\quad\psi_{3,1}=(1,0,-1,-1,0,1)^T,
\nonumber
\\
&\quad\psi_{3,2}=(0,1,-1,-1,1,0)^T,
\nonumber
\\
\tau_4=0,
&\quad\psi_4=(-1,1,1,-1,-1,1)^T.
\nonumber
\end{align}
Using the formula \eqref{eq7.2.38}, we obtain corresponding eigenvalues of $\bfA\vert_\VVV$ given by
\begin{align}
\lambda(\tau_1;\pm)
&
=1+\frac{4k^2\pm\sqrt{16k^4+16k^2}}{2}
=1+2k^2\pm2k\sqrt{k^2+1},
\nonumber
\\
\lambda(\tau_2;\pm)
&
=1+\frac{3k^2\pm\sqrt{9k^4+12k^2}}{2}
=1+\frac{3k^2}{2}\pm\frac{k\sqrt{9k^2+12}}{2},
\nonumber
\\
\lambda(\tau_3;\pm)
&
=1+\frac{k^2\pm\sqrt{k^4+4k^2}}{2}
=1+\frac{k^2}{2}\pm\frac{k\sqrt{k^2+4}}{2}.
\nonumber
\\
\lambda(\tau_4;\pm)
&
=1.
\nonumber
\end{align}
The three relevant eigenvalues of $\bfA\vert_\VVV$ and corresponding eigenvectors are therefore
\begin{align}
\lambda_1
=1+2k^2+2k\sqrt{k^2+1},
&
\quad\Psi_1=(1,1,1,1,1,1,b_{11},\ldots,b_{16})^T,
\label{eq7.3.16}
\\
\lambda_2
=1+\frac{3k^2}{2}+\frac{k\sqrt{9k^2+12}}{2},
&
\quad\Psi_{2,1}=(-1,0,-1,1,0,1,b_{21}^{(1)},\ldots,b_{26}^{(1)})^T,
\label{eq7.3.17}
\\
&
\quad\Psi_{2,2}=(0,-1,1,-1,1,0,b_{21}^{(2)},\ldots,b_{26}^{(2)})^T,
\nonumber
\\
\lambda_3
=1+\frac{k^2}{2}+\frac{k\sqrt{k^2+4}}{2},
&
\quad\Psi_{3,1}=(1,0,-1,-1,0,1,b_{31}^{(1)},\ldots,b_{36}^{(1)})^T,
\nonumber
\\
&
\quad\Psi_{3,2}=(0,1,-1,-1,1,0,b_{31}^{(2)},\ldots,b_{36}^{(2)})^T,
\nonumber
\end{align}
where the values of $b_{1j},b_{2j}^{(1)},b_{2j}^{(2)},b_{3j}^{(1)},b_{3j}^{(2)}$, $j=1,\ldots,6$, can all be calculated using \eqref{eq7.2.39}.

We now proceed to find the edge cutting numbers of $h_2,h_4$ and $h_1,h_3$.

As in Example~\ref{ex7.3.1}, the eigenvalue $\lambda_1$ does not contribute to their difference, as any lack of cancellation would violate the Gutkin--Veech theorem that guarantees uniformity.

The eigenvector of $\bfA$ corresponding to the eigenvector $\Psi_{2,1}$ of $\bfA\vert_\VVV$ is given by
\begin{equation}\label{eq7.3.18}
-\bfu_1-\bfu_3+\bfu_4+\bfu_6
-b_{21}^{(1)}\bfv_1-b_{23}^{(1)}\bfv_3+b_{24}^{(1)}\bfv_4+b_{26}^{(1)}\bfv_6.
\end{equation}
The eigenvector of $\bfA$ corresponding to the eigenvector $\Psi_{2,2}$ of $\bfA\vert_\VVV$ is given by
\begin{equation}\label{eq7.3.19}
-\bfu_2+\bfu_3-\bfu_4+\bfu_5
-b_{22}^{(2)}\bfv_2+b_{23}^{(2)}\bfv_3-b_{24}^{(2)}\bfv_4+b_{25}^{(2)}\bfv_5.
\end{equation}

We first find the number of almost vertical units counted here that cut the edges $h_2$ and~$h_4$; see Figure~7.3.8.
The count from each of $\bfu_i$, $i=1,\ldots,6$ are
\begin{displaymath}
\begin{array}{llllll}
\bfu_1\mapsto0,\ {}
&\bfu_2\mapsto0,\ {}
&\bfu_3\mapsto0,\ {}
\bfu_4\mapsto k^2,\ {}
&\bfu_5\mapsto k^2,\ {}
&\bfu_6\mapsto2k^2.
\end{array}
\end{displaymath}
On the other hand, it is easy to show that the total contribution from $\bfv_i$, $i=1,\ldots,6$ is~$0$.
Thus corresponding to the eigenvector $\Psi_{2,1}$ and \eqref{eq7.3.18}, the total count is~$3k^2$.
Corresponding to the eigenvector $\Psi_{2,2}$ and \eqref{eq7.3.19}, the total count is~$0$.

We next find the number of almost vertical units counted here that cut the edges $h_1$ and~$h_3$; see Figure~7.3.8.
The count from each of $\bfu_i$, $i=1,\ldots,6$ are
\begin{displaymath}
\begin{array}{llllll}
\bfu_1\mapsto2k^2,\ {}
&\bfu_2\mapsto k^2,\ {}
&\bfu_3\mapsto k^2,\ {}
\bfu_4\mapsto0,\ {}
&\bfu_5\mapsto0,\ {}
&\bfu_6\mapsto0.
\end{array}
\end{displaymath}
Like before, the total contribution from $\bfv_i$, $i=1,\ldots,6$ is~$0$.
Thus corresponding to the eigenvector $\Psi_{2,1}$ and \eqref{eq7.3.18}, the total count is~$-3k^2$.
Corresponding to the eigenvector $\Psi_{2,2}$ and \eqref{eq7.3.19}, the total count is~$0$.

It follows that the second eigenvalue $\lambda_2$ contributes $6c_{2,1}\lambda_2^rk^2$ to the difference between the edge cuttings numbers of $h_2,h_4$ and $h_1,h_3$.

So the deviation from the starting point comes from the second largest eigenvalue, and this has the order of magnitude $\lambda_2^r$ compared to the order of magnitude $\lambda_1^r$ of the main term.
Choosing $T=\lambda_1^r$, we have $\lambda_2^r\asymp T^{\kappa_0}$, where
\begin{equation}\label{eq7.3.20}
\kappa_0=\kappa_0(k)=\frac{\log\lambda_2}{\log\lambda_1}=\frac{2\log k+\log3}{2\log k+\log4}+o(1),
\end{equation}
in view of \eqref{eq7.3.16} and \eqref{eq7.3.17}, is the irregularity exponent of a $1$-direction geodesic of slope \eqref{eq7.3.9} on the period-surface in Figure~7.3.8.

Clearly $\kappa_0=\kappa_0(k)\to1$ as $k\to\infty$.
So we have just established $T^{\kappa_0}=T^{1-\eps}$ size super-fast escape rate to infinity for the infinite L-strip billiard with the explicit class of quadratic irrational slopes in \eqref{eq7.3.9} where the parameter $k\ge1$ is any integer.

It is easy to see that the irregularity exponent $\kappa_0=\kappa_0(k)$ in \eqref{eq7.3.20} is {\it precisely} the escape rate to infinity of this infinite billiard.
Indeed, the exponent of the escape rate to infinity cannot be larger than the expression \eqref{eq7.3.20} coming from the two largest eigenvalues.

This example also highlights that different quadratic irrational slopes can lead to $1$-direction geodesics that exhibit vastly different escape rates to infinity.
To explain this, we shall consider an infinite polysquare translation surface $\PPP$ where there is an absolute bound $\ell$ such that the length of any horizontal or vertical street of $\PPP$ is at most~$\ell$.
This is called an $\ell$-square-maze translation surface, and $1$-direction geodesics on square-maze translation surfaces are studied in \cite[Theorem~6.5.1]{BCY}.

The problem connected to the Ehrenfest wind-tree model billiard with $a=b=1/2$ that we have considered here leads to a problem of $1$-direction geodesic flow on
$\Bil(\infty;2)$.
Now $\Bil(\infty;2)$ has infinite horizontal and vertical streets, so is not a square-maze translation surface.
However, if we consider \textit{streets} in $\Bil(\infty;2)$ at $45$ degrees to the horizontal and vertical axes, then the situation becomes very different.
First of all, observe such a street in the Ehrenfest wind-tree model billiard as shown in Figure~7.3.11.

\begin{displaymath}
\begin{array}{c}
\includegraphics[scale=0.8]{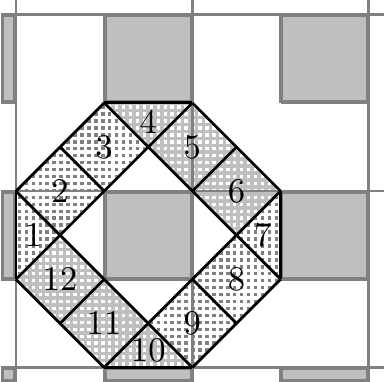}
\vspace{3pt}\\
\mbox{Figure 7.3.11: a diagonal street in the Ehrenfest wind-tree model billiard}
\end{array}
\end{displaymath}

The street indicated is contained in $4$ constituent L-shapes.
If we consider the analogous problem of $1$-direction geodesic flow in $\Bil(\infty;2)$, then the image of this street is shown in Figure~7.3.12, made up of the $4$-copy versions of these $4$ constituent L-shapes.
We have in fact shown that $\Bil(\infty;2)$ rotated by $45$ degrees becomes a square-maze translation surface where every horizontal and vertical street has length~$12$.

\begin{displaymath}
\begin{array}{c}
\includegraphics[scale=0.8]{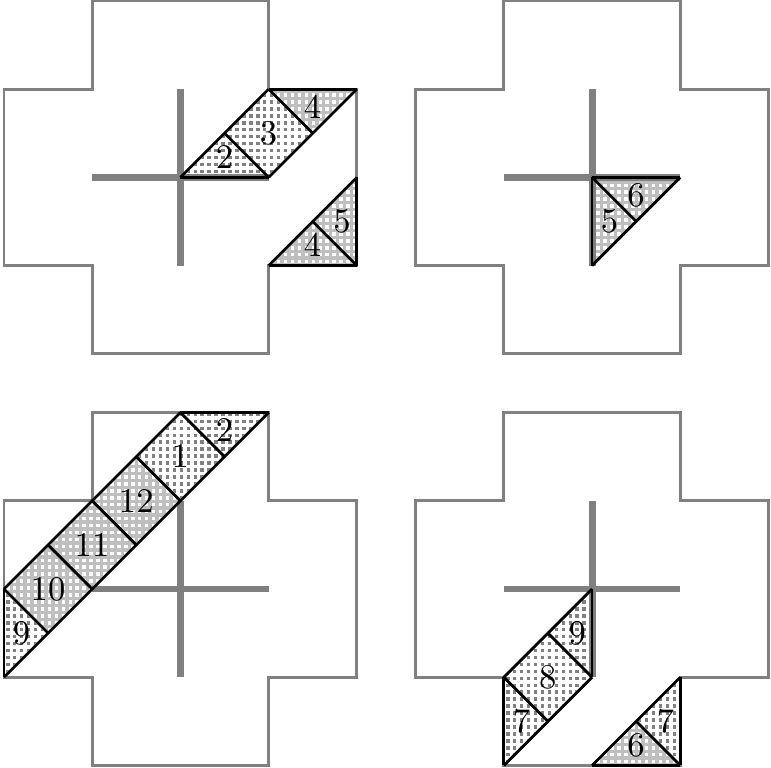}
\vspace{3pt}\\
\mbox{Figure 7.3.12: a $45$-degree street of length $12$ in $\Bil(\infty;2)$}
\end{array}
\end{displaymath}

\cite[Theorem~6.5.1]{BCY} asserts that for an $\ell$-square-maze translation surface, there are infinitely many numbers $\alpha$ of the form
\begin{displaymath}
\alpha=[a;a,a,a,\ldots]=a+\frac{1}{a+\frac{1}{a+\frac{1}{a+\cdots}}},
\end{displaymath}
where $a\ge\ell!$ is divisible by~$\ell!$, such that there exist $1$-direction geodesics with slope $\alpha$ that exhibit time-quantitative density and super-slow logarithmic escape rate to infinity.
For a square-maze translation surface where every street has length precisely~$\ell$, the condition on $a$ can be relaxed to include all integers $a\ge\ell$ that are divisible by~$\ell$.

It follows that for the rotated $\Bil(\infty;2)$, there exist infinitely many numbers $\alpha_k$ of the form
\begin{equation}\label{eq7.3.21}
\alpha_k=[12k;12k,12k,12k,\ldots]=12k+\frac{1}{12k+\frac{1}{12k+\frac{1}{12k+\cdots}}},
\end{equation}
where $k\ge1$ is an integer, such that there exist $1$-direction geodesics with slope $\alpha_k$ that exhibit super-slow logarithmic escape rate to infinity.

Now note that slopes $\alpha_k$ of the form \eqref{eq7.3.21} make up a subset of those slopes $\alpha_k$ of the form \eqref{eq7.3.9}.
It follows that slopes $\alpha_k$ of the form \eqref{eq7.3.21} give rise to $1$-direction geodesics on $\Bil(\infty;2)$ that exhibit super-fast escape rate to infinity.

It is easy to show that a line with slope~$\alpha_k$, after a clockwise rotation of $45$ degrees, now has slope equal to
\begin{equation}\label{eq7.3.22}
\alpha_k^*=\frac{\alpha_k-1}{\alpha_k+1}.
\end{equation}
Thus we have shown that there are infinitely many numbers $\alpha_k$ of the form \eqref{eq7.3.21} such that there exist $1$-direction geodesics of slope $\alpha_k$ in $\Bil(\infty;2)$ that exhibit super-fast escape rate to infinity, and also $1$-direction geodesics of slope $\alpha_k^*$ in
$\Bil(\infty;2)$ that exhibit super-slow logarithmic escape rate to infinity.
Note that in view of \eqref{eq7.3.21} and \eqref{eq7.3.22}, both $\alpha_k$ and $\alpha_k^*$ are quadratic irrationals.
\end{example}

%%%%%%%%%%
%
% SECTION 7.4
%
%%%%%%%%%%

\subsection{More on the escape rate to infinity}\label{sec7.4}

In this section, we consider some more complicated infinite billiards.

\begin{example}\label{ex7.4.1}
Consider \textit{C-wall} obstacles, as shown in Figure~7.4.1.

\begin{displaymath}
\begin{array}{c}
\includegraphics[scale=0.8]{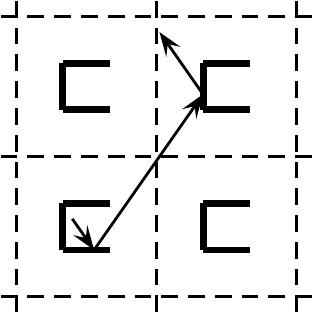}
\vspace{3pt}\vspace{3pt}\\
\mbox{Figure 7.4.1: infinite billiard with double periodic C-shaped walls}
\end{array}
\end{displaymath}

More precisely, the building block of this infinite polysquare region is a $3\times3$ square, with walls on three sides of the middle square face as shown.

As usual, this infinite billiard model is equivalent to a $1$-direction geodesic flow on an infinite polysquare translation surface that we denote by $\Bil(\infty;2;\mathrm{C})$.
Here the index $2$ indicates that this is double periodic, and the letter $\mathrm{C}$ refers to the common shape of the obstacles.
To construct $\Bil(\infty;2;\mathrm{C})$, we take one of the building blocks, and unfold the $4$-direction billiard flow on it to a $1$-direction geodesic flow on a $4$-copy version of it, obtained by reflecting horizontally and vertically.
Note that each $4$-copy version has a right-neighbor, a left-neighbor, a down-neighbor and an up-neighbor in $\Bil(\infty;2;\mathrm{C})$, and we need appropriate edge identification for gluing them together.
The period-surface $\Bil(2;\mathrm{C})$ of $\Bil(\infty;2;\mathrm{C})$ is shown in the picture on the left in Figure~7.4.2. 
The horizontal and vertical streets of $\Bil(2;\mathrm{C})$ are indicated in the picture on the right in Figure~7.4.2 where, for instance, the entries
$\leftrightarrow3$ and $\updownarrow5$ in a square face indicates that the square face is on the $3$-rd horizontal street and the $5$-th vertical street.
It is easy to see that $\Bil(2;\mathrm{C})$ has $8$ horizontal streets of length~$3$, and $2$ horizontal streets of length~$6$.
It also has $8$ vertical streets of length~$3$, $2$ vertical streets of length~$4$, and $2$ vertical streets of length~$2$.
Thus the street-LCM is equal to ~$12$. 

\begin{displaymath}
\begin{array}{c}
\includegraphics[scale=0.8]{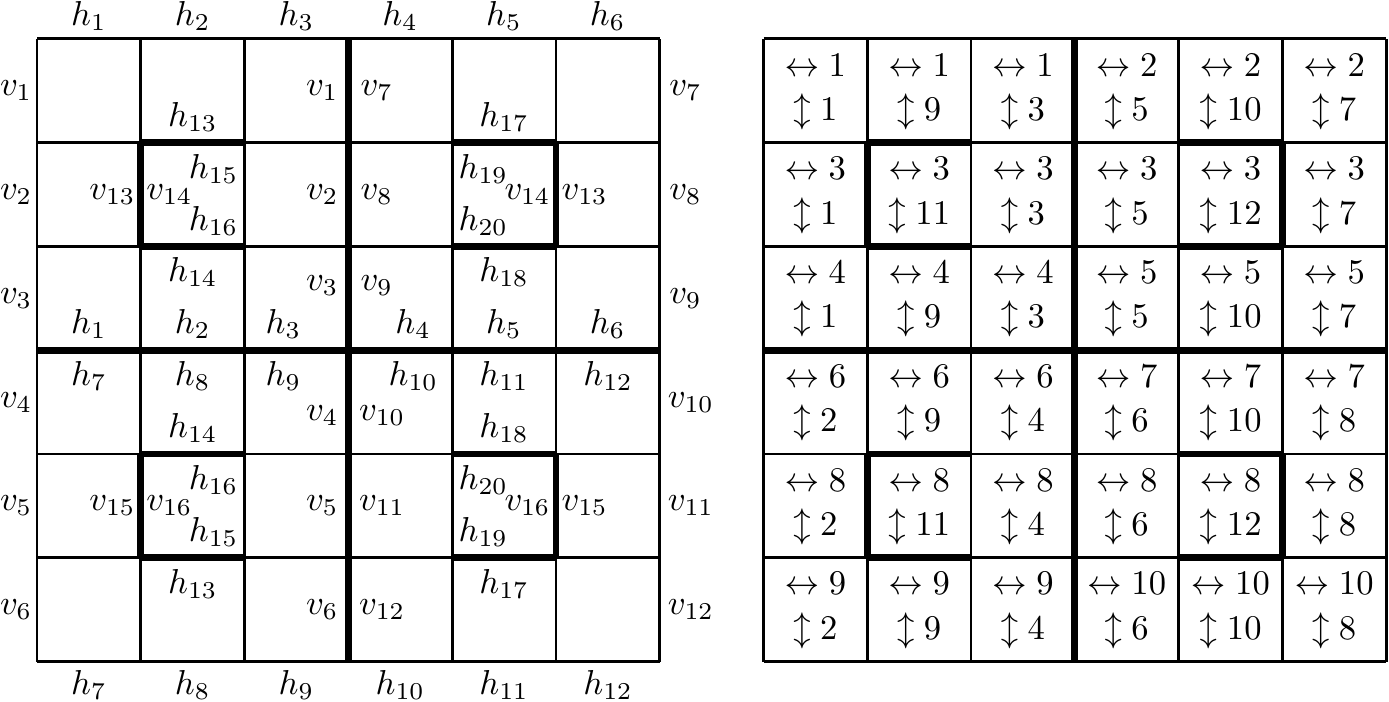}
\vspace{3pt}\vspace{3pt}\\
\mbox{Figure 7.4.2: the period surface $\Bil(2;\mathrm{C})$ of $\Bil(\infty;2;\mathrm{C})$}
\end{array}
\end{displaymath}

We shall consider an almost vertical geodesic $V_0$ in $\Bil(2;\mathrm{C})$ with slope $\alpha$ given by \eqref{eq7.2.1} with $m=n=12$.
As $\Bil(2;\mathrm{C})$ has $36$ square faces, the $2$-step transition matrix $\bfA$ is $72\times72$ which is mildly inconvenient.
Instead, we shall determine the street-spreading matrix $\bfS$ which, at size $10\times10$, is considerably smaller.

We follow the notation in Section~\ref{sec7.2}.
We consider the $\bfA$-invariant subspace $\VVV$ generated by the $20$ vectors $\bfu_i,\bfv_i$, $i=1,\ldots,10$, defined by \eqref{eq7.2.14} and \eqref{eq7.2.15}.
The relevant eigenvalues of~$\bfA$ are then eigenvalues of the matrix~$\bfA\vert_\VVV$, defined by \eqref{eq7.2.37}.

With $J_1=J_4=\{1,3,9\}$, for the horizontal streets corresponding to $\bfu_1,\bfu_4$ as highlighted in the picture on the left in Figure~7.4.3, we have, using \eqref{eq7.2.17},
\begin{align}
(\bfA-I)\bfu_1
&
=(\bfA-I)\bfu_4
\nonumber
\\
&
=(\bfA-I)[\{16\uparrow_{1,1},16\uparrow_{1,3},12\uparrow_{1,9}\}]
+(\bfA-I)[\{16\uparrow_{3,1},16\uparrow_{3,3}\}]
\nonumber
\\
&\qquad
+(\bfA-I)[\{16\uparrow_{4,1},16\uparrow_{4,3},12\uparrow_{4,9}\}]
+(\bfA-I)[\{12\uparrow_{6,9}\}]
\nonumber
\\
&\qquad
+(\bfA-I)[\{12\uparrow_{9,9}\}]
\nonumber
\\
&
=44(\bfu_1+\bfv_1)+32(\bfu_3+\bfv_3)+44(\bfu_4+\bfv_4)
\nonumber
\\
&\qquad
+12(\bfu_6+\bfv_6)+12(\bfu_9+\bfv_9),
\label{eq7.4.1}
\end{align}
\begin{displaymath}
\begin{array}{c}
\includegraphics[scale=0.8]{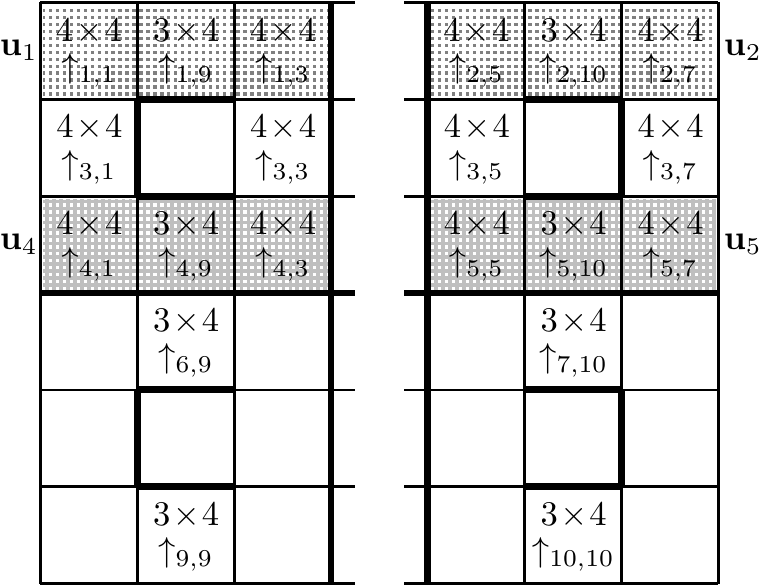}
\vspace{3pt}\vspace{3pt}\\
\mbox{Figure 7.4.3: almost vertical units of type $\uparrow$ in $\bfu_1,\bfu_2,\bfu_4,\bfu_5$}
\end{array}
\end{displaymath}

With $J_2=J_5=\{5,7,10\}$, for the horizontal streets corresponding to $\bfu_2,\bfu_5$ as highlighted in the picture on the right in Figure~7.4.3, we have, using \eqref{eq7.2.17},
\begin{align}
(\bfA-I)\bfu_2
&
=(\bfA-I)\bfu_5
\nonumber
\\
&
=44(\bfu_2+\bfv_2)+32(\bfu_3+\bfv_3)+44(\bfu_5+\bfv_5)
\nonumber
\\
&\qquad
+12(\bfu_7+\bfv_7)+12(\bfu_{10}+\bfv_{10}).
\label{eq7.4.2}
\end{align}

With $J_6=J_9=\{2,4,9\}$ and $J_7=J_{10}=\{6,8,10\}$, for the horizontal streets corresponding to $\bfu_6,\bfu_7,\bfu_9,\bfu_{10}$ highlighted in Figure~7.4.4, we have
\begin{align}
(\bfA-I)\bfu_6
&
=(\bfA-I)\bfu_9
\nonumber
\\
&
=12(\bfu_1+\bfv_1)+12(\bfu_4+\bfv_4)+44(\bfu_6+\bfv_6)
\nonumber
\\
&\qquad
+32(\bfu_8+\bfv_8)+44(\bfu_9+\bfv_9),
\label{eq7.4.3}
\\
(\bfA-I)\bfu_7
&
=(\bfA-I)\bfu_{10}
\nonumber
\\
&
=12(\bfu_2+\bfv_2)+12(\bfu_5+\bfv_5)+44(\bfu_7+\bfv_7)
\nonumber
\\
&\qquad
+32(\bfu_8+\bfv_8)+44(\bfu_{10}+\bfv_{10}).
\label{eq7.4.4}
\end{align}
\begin{displaymath}
\begin{array}{c}
\includegraphics[scale=0.8]{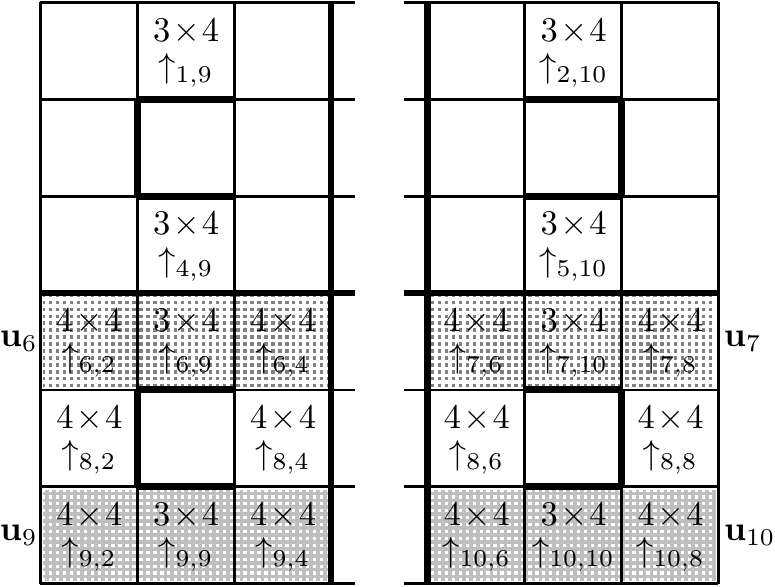}
\vspace{3pt}\vspace{3pt}\\
\mbox{Figure 7.4.4: almost vertical units of type $\uparrow$ in $\bfu_6,\bfu_7,\bfu_9,\bfu_{10}$}
\end{array}
\end{displaymath}

With $J_3=\{1,3,5,7,11,12\}$ and $J_8=\{2,4,6,8,11,12\}$, for the horizontal streets corresponding to $\bfu_3,\bfu_8$ highlighted in Figure~7.4.5, we have
\begin{align}
(\bfA-I)\bfu_3
&
=16(\bfu_1+\bfv_1)+16(\bfu_2+\bfv_2)+56(\bfu_3+\bfv_3)
\nonumber
\\
&\qquad
+16(\bfu_4+\bfv_4)+16(\bfu_5+\bfv_5)+24(\bfu_8+\bfv_8),
\label{eq7.4.5}
\\
(\bfA-I)\bfu_8
&
=24(\bfu_3+\bfv_3)+16(\bfu_6+\bfv_6)+16(\bfu_7+\bfv_7)
\nonumber
\\
&\qquad
+56(\bfu_8+\bfv_8)+16(\bfu_9+\bfv_9)+16(\bfu_{10}+\bfv_{10}).
\label{eq7.4.6}
\end{align}
\begin{displaymath}
\begin{array}{c}
\includegraphics[scale=0.8]{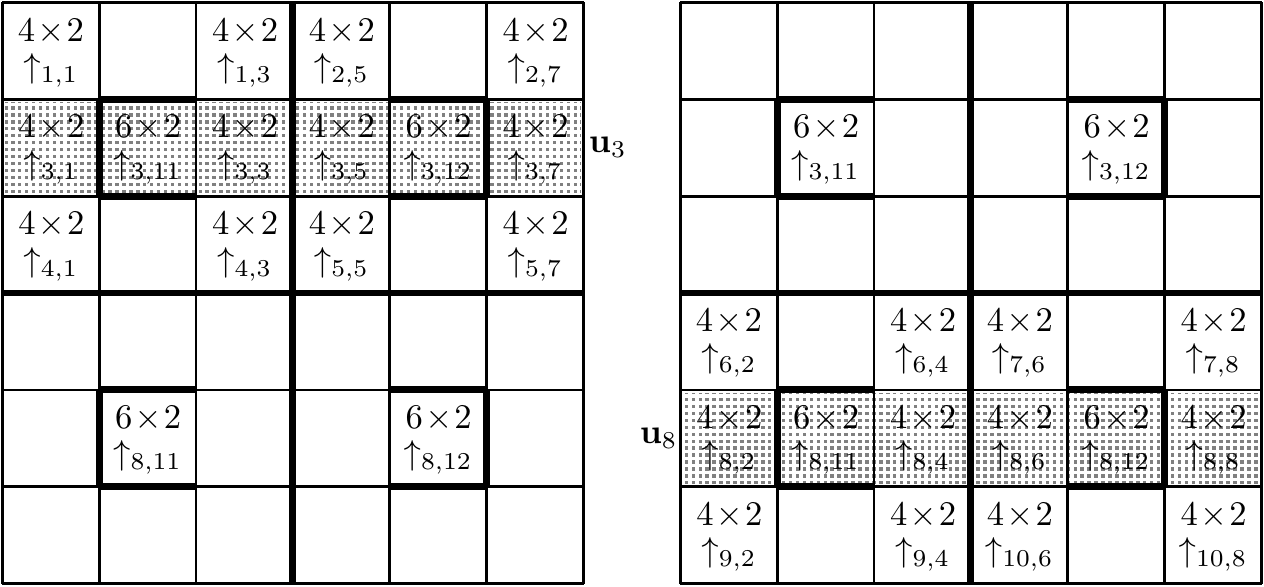}
\vspace{3pt}\vspace{3pt}\\
\mbox{Figure 7.4.5: almost vertical units of type $\uparrow$ in $\bfu_3,\bfu_8$}
\end{array}
\end{displaymath}

It follows from \eqref{eq7.4.1}--\eqref{eq7.4.6} that the street-spreading matrix is given by
\begin{equation}\label{eq7.4.7}
\bfS=\begin{pmatrix}
44&0&16&44&0&12&0&0&12&0\\
0&44&16&0&44&0&12&0&0&12\\
32&32&56&32&32&0&0&24&0&0\\
44&0&16&44&0&12&0&0&12&0\\
0&44&16&0&44&0&12&0&0&12\\
12&0&0&12&0&44&0&16&44&0\\
0&12&0&0&12&0&44&16&0&44\\
0&0&24&0&0&32&32&56&32&32\\
12&0&0&12&0&44&0&16&44&0\\
0&12&0&0&12&0&44&16&0&44
\end{pmatrix}.
\end{equation}
This has non-zero eigenvalues and corresponding eigenvectors given by
\begin{align}
\tau_1=144,
&\quad
\psi_1=(1,1,2,1,1,1,1,2,1,1)^T,
\nonumber
\\
\tau_2=112,
&\quad
\psi_2=(-1,1,0,-1,1,-1,1,0,-1,1)^T,
\nonumber
\\
\tau_3=96,
&\quad
\psi_3=(-1,-1,-2,-1,-1,1,1,2,1,1)^T,
\nonumber
\\
\tau_4=64,
&\quad
\psi_4=(1,-1,0,1,-1,-1,1,0,-1,1)^T,
\nonumber
\\
\tau_5=48,
&\quad
\psi_5=(1,1,-4,1,1,1,1,-4,1,1)^T.
\nonumber
\end{align}

Note that the street-spreading matrix $\bfS$ given by \eqref{eq7.4.7} corresponds to the choice of parameters $m=n=12$.
Switching to $m=n=12k$, where $k\ge1$ is any integer, we obtain the new street-spreading matrix $\bfS(k)=k^2\bfS$ simply by multiplying the matrix
$\bfS$ by~$k^2$.
Then of course the eigenvalues are also multiplied by~$k^2$, but the eigenvectors remain the same.
Naturally the $2$-step transition matrix $\bfA$ is modified to~$\bfA(k)$.

We next determine some of the eigenvalues of $\bfA(k)$ using \eqref{eq7.2.38}.
The largest eigenvalue of $\bfA(k)\vert_\VVV$ is
\begin{equation}\label{eq7.4.8}
\lambda_1(k)
=1+\frac{\tau_1k^2+\sqrt{\tau_1^2k^4+4\tau_1k^2}}{2}
=1+72k^2+12k\sqrt{36k^2+1},
\end{equation}
while the second largest eigenvalue of $\bfA(k)\vert_\VVV$ is
\begin{equation}\label{eq7.4.9}
\lambda_2(k)
=1+\frac{\tau_2k^2+\sqrt{\tau_2^2k^4+4\tau_2k^2}}{2}
=1+56k^2+4k\sqrt{196k^2+7},
\end{equation}
with eigenvector of the form
\begin{displaymath}
\Psi_2=(-1,1,0,-1,1,-1,1,0,-1,1,-\tau_2^*,\tau_2^*,0,-\tau_2^*,\tau_2^*,-\tau_2^*,\tau_2^*,0,-\tau_2^*,\tau_2^*)^T,
\end{displaymath}
where
\begin{displaymath}
\tau_2^*=\frac{-\tau_2k^2+\sqrt{\tau_2^2k^4+4\tau_2k^2}}{2}=4k\sqrt{196k^2+7}-56k^2.
\end{displaymath}

Let us return to Figures 7.4.1 and~7.4.2.
It is clear that the billiard moves to the up-neighbor if the $1$-direction geodesic hits any of the edges $h_1,\ldots,h_6$, and moves to the
down-neighbor if the $1$-direction geodesic hits any of the edges $h_7,\ldots,h_{12}$.
It is also clear that the billiard moves to the right-neighbor if the $1$-direction geodesic hits any of the edges $v_7,\ldots,v_{12}$, and moves to the
left-neighbor if the $1$-direction geodesic hits any of the edges $v_1,\ldots,v_6$.

We now proceed to find the edge cutting numbers of $h_1,\ldots,h_6$ and $h_7,\ldots,h_{12}$.

As in the example in Section~\ref{sec7.3}, the eigenvalue $\lambda_1$ does not contribute to their difference, as any lack of cancellation would violate the Gutkin--Veech theorem that guarantees uniformity.

The eigenvector of $\bfA(k)$ corresponding to the eigenvector $\Psi_2$ of $\bfA(k)\vert_\VVV$ is given by
\begin{align}\label{eq7.4.10}
&
-\bfu_1+\bfu_2-\bfu_4+\bfu_5-\bfu_6+\bfu_7-\bfu_9+\bfu_{10}
\nonumber
\\
&\qquad
-\tau_2^*\bfv_1+\tau_2^*\bfv_2-\tau_2^*\bfv_4+\tau_2^*\bfv_5-\tau_2^*\bfv_6+\tau_2^*\bfv_7-\tau_2^*\bfv_9+\tau_2^*\bfv_{10}.
\end{align}

We first find the number of almost vertical units counted here that cut the edges $h_1,\ldots,h_6$; see Figure~7.4.2.
The counts from $\bfu_i$, $i=1,\ldots,10$, are
\begin{equation}\label{eq7.4.11}
\begin{array}{lllll}
\bfu_1\mapsto44k^2,\ {}
&\bfu_2\mapsto44k^2,\ {}
&\bfu_3\mapsto32k^2,\ {}
&\bfu_4\mapsto44k^2,\ {}
&\bfu_5\mapsto44k^2,
\\
\bfu_6\mapsto12k^2,\ {}
&\bfu_7\mapsto12k^2,\ {}
&\bfu_8\mapsto0,\ {}
&\bfu_9\mapsto12k^2,\ {}
&\bfu_{10}\mapsto12k^2.
\end{array}
\end{equation}
On the other hand, it is not difficult to show that each of $\bfv_1,\ldots\bfv_{10}$ contributes precisely zero.
Thus for the edges $h_1,\ldots,h_6$, we have a total count of
\begin{displaymath}
-44k^2+44k^2-44k^2+44k^2-12k^2+12k^2-12k^2+12k^2=0.
\end{displaymath}

We next find the number of almost vertical units counted here that cut the edges $h_7,\ldots,h_{12}$; see Figure~7.4.2.
The counts from $\bfu_i$, $i=1,\ldots,10$, are
\begin{equation}\label{eq7.4.12}
\begin{array}{lllll}
\bfu_1\mapsto12k^2,\ {}
&\bfu_2\mapsto12k^2,\ {}
&\bfu_3\mapsto0,\ {}
&\bfu_4\mapsto12k^2,\ {}
&\bfu_5\mapsto12k^2,
\\
\bfu_6\mapsto44k^2,\ {}
&\bfu_7\mapsto44k^2,\ {}
&\bfu_8\mapsto32k^2,\ {}
&\bfu_9\mapsto44k^2,\ {}
&\bfu_{10}\mapsto44k^2.
\end{array}
\end{equation}
Again, it is not difficult to show that each of $\bfv_1,\ldots\bfv_{10}$ contributes precisely zero.
Thus for the edges $h_7,\ldots,h_{12}$, we have a total count of
\begin{displaymath}
-12k^2+12k^2-12k^2+12k^2-44k^2+44k^2-44k^2+44k^2=0.
\end{displaymath}

Thus there is perfect cancellation in the vertical direction.

We next proceed to find the edge cutting numbers of $v_1,\ldots,v_6$ and $v_7,\ldots,v_{12}$; see Figure~7.4.2.

Again, the eigenvalue $\lambda_1$ does not contribute to their difference, as any lack of cancellation would violate the Gutkin--Veech theorem that guarantees uniformity.
So we concentrate our attention on the contribution from the eigenvector $\Psi_2$ of $\bfA\vert_\VVV$.
Here, note that none of $\bfu_1,\ldots,\bfu_{10}$ makes any non-zero contribution, as only units of type $\nuparrow$ can contribute to the count.
For the contributions from $\bfv_1,\ldots,\bfv_{10}$, we use \eqref{eq7.2.15}.

We first find the number of almost vertical units counted in \eqref{eq7.4.10} that cut the edges $v_1,\ldots,v_6$; see Figure~7.4.2.
Since $\bfv_3$ and $\bfv_8$ do not feature in \eqref{eq7.4.10}, it is not necessary to any counting for them.
Clearly we have a count of $4k$ for each of $\bfv_1,\bfv_4,\bfv_6,\bfv_9$, and none for the rest, making a total of~$-16\tau_2^*k$.

We next find the number of almost vertical units counted in \eqref{eq7.4.10} that cut the edges $v_7,\ldots,v_{12}$; see Figure~7.4.2.
Clearly we have a count of $4k$ for each of $\bfv_2,\bfv_5,\bfv_7,\bfv_{10}$, and none for the rest, making a total of~$16\tau_2^*k$.

The difference, in absolute value, is therefore~$32\tau_2^*k$.
Let $c_2=c_{2,1}$ in \eqref{eq7.2.9} and \eqref{eq7.2.10}.
It follows that the second eigenvalue $\lambda_2(k)$ contributes $32c_2\tau_2^*\lambda_2^r(k)k$ to the difference between the edge cuttings numbers of $v_1,\ldots,v_6$ and $v_7,\ldots,v_{12}$.

So the \textit{horizontal} deviation from the starting point comes from the second largest eigenvalue, with order of magnitude $\lambda_2^r$ compared to the order of magnitude $\lambda_1^r$ of the main term.
Choosing $T=\lambda_1^r$, we have $\lambda_2^r\asymp T^{\kappa_0}$, where
\begin{equation}\label{eq7.4.13}
\kappa_0=\kappa_0(k)=\frac{\log\lambda_2}{\log\lambda_1}=\frac{2\log k+\log112}{2\log k+\log144}+o(1),
\end{equation}
in view of \eqref{eq7.4.8} and \eqref{eq7.4.9}, is the irregularity exponent of a $1$-direction geodesic of slope
\begin{equation}\label{eq7.4.14}
\alpha_k=[12k;12k,12k,12k,\ldots]=6k+\sqrt{36k^2+1}=\sqrt{\lambda_1(k)}
\end{equation}
on the period-surface $\Bil(2;\mathrm{C})$ in Figure~7.4.2.

Clearly $\kappa_0=\kappa_0(k)\to1$ as $k\to\infty$.
So we have just established $T^{\kappa_0}=T^{1-\eps}$ size super-fast escape rate to infinity for this infinite billiard with the explicit class of quadratic irrational slopes in \eqref{eq7.4.14} where the parameter $k\ge1$ is any integer.

It is easy to see that the irregularity exponent $\kappa_0=\kappa_0(k)$ in \eqref{eq7.4.13} is \textit{precisely} the escape rate to infinity of this infinite billiard.
The exponent of the escape rate to infinity cannot be larger than the expression \eqref{eq7.4.13} coming from the two largest eigenvalues.

Note that this super-fast escape rate to infinity comes from horizontal deviation.
As observed earlier, we have perfect cancellation in the vertical direction for the two largest eigenvalues.
So let us investigate what the third eigenvalue gives.

The third largest eigenvalue of $\bfA(k)\vert_\VVV$ is
\begin{displaymath}
\lambda_3(k)
=1+\frac{\tau_3k^2+\sqrt{\tau_3^2k^4+4\tau_3k^2}}{2}
=1+48k^2+4k\sqrt{144k^2+6},
\end{displaymath}
with eigenvector of the form
\begin{displaymath}
\Psi_3=(-1,-1,-2,-1,-1,1,1,2,1,1,-\tau_3^*,-\tau_3^*,-2\tau_3^*,-\tau_3^*,-\tau_3^*,\tau_3^*,\tau_3^*,2\tau_3^*,\tau_3^*,\tau_3^*)^T,
\end{displaymath}
where
\begin{displaymath}
\tau_3^*=\frac{-\tau_3k^2+\sqrt{\tau_3^2k^4+4\tau_3k^2}}{2}=4k\sqrt{144k^2+6}-48k^2.
\end{displaymath}
The eigenvector of $\bfA(k)$ corresponding to the eigenvector $\Psi_3$ of $\bfA(k)\vert_\VVV$ is given by
\begin{align}
&
-\bfu_1-\bfu_2-2\bfu_3-\bfu_4-\bfu_5+\bfu_6+\bfu_7+2\bfu_8+\bfu_9+\bfu_{10}
\nonumber
\\
&\qquad
-\tau_3^*\bfv_1-\tau_3^*\bfv_2-2\tau_3^*\bfv_3-\tau_3^*\bfv_4-\tau_3^*\bfv_5
\nonumber
\\
&\qquad
+\tau_3^*\bfv_6+\tau_3^*\bfv_7+2\tau_3^*\bfv_8+\tau_3^*\bfv_9+\tau_3^*\bfv_{10}.
\nonumber
\end{align}

We first find the number of almost vertical units counted here that cut the edges $h_1,\ldots,h_6$; see Figure~7.4.2.
The counts from $\bfu_i$, $i=1,\ldots,10$, are given by \eqref{eq7.4.11}.
Again, it is not difficult to show that each of $\bfv_1,\ldots\bfv_{10}$ contributes precisely zero.
Thus for the edges $h_1,\ldots,h_6$, we have a total count of
\begin{displaymath}
-44k^2-44k^2-64k^2-44k^2-44k^2+12k^2+12k^2+12k^2+12k^2=-196k^2.
\end{displaymath}

We next find the number of almost vertical units counted here that cut the edges $h_7,\ldots,h_{12}$; see Figure~7.4.2.
The counts from $\bfu_i$, $i=1,\ldots,10$, are given by \eqref{eq7.4.12}.
Again, it is not difficult to show that each of $\bfv_1,\ldots\bfv_{10}$ contributes precisely zero.
Thus for the edges $h_7,\ldots,h_{12}$, we have a total count of
\begin{displaymath}
-12k^2-12k^2-12k^2-12k^2+44k^2+44k^2+64k^2+44k^2+44k^2=196k^2.
\end{displaymath}

The difference, in absolute value, is therefore~$392\tau_3^*k^2$.
Let $c_3=c_{3,1}$ in \eqref{eq7.2.9} and \eqref{eq7.2.10}.
It follows that the third eigenvalue $\lambda_3(k)$ contributes $392c_3\tau_3^*\lambda_3^r(k)k^2$ to the difference between the edge cuttings numbers of $h_1,\ldots,h_6$ and $h_7,\ldots,h_{12}$.

So the \textit{vertical} deviation from the starting point comes from the third largest eigenvalue, with order of magnitude $\lambda_3^r$ compared to the order of magnitude $\lambda_1^r$ of the main term.
\end{example}

\begin{example}\label{ex7.4.2}
Consider the Ehrenfest wind-tree model with $a=b=1/3$, rescaled so that we have square faces of unit area.
Consider a $4$-direction billiard trajectory in the infinite region in Figure~7.4.6.
Here the building block of this infinite polysquare region is a $3\times3$ square, with walls on all sides of the middle square face as shown.

\begin{displaymath}
\begin{array}{c}
\includegraphics[scale=0.8]{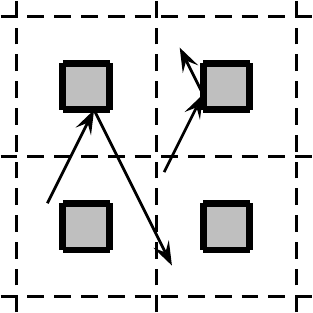}
\vspace{3pt}\\
\mbox{Figure 7.4.6: Ehrenfest wind-tree billiard model}
\end{array}
\end{displaymath}

As usual, this infinite billiard model is equivalent to a $1$-direction geodesic flow on an infinite polysquare translation surface that we denote by $\Bil(\infty;2;\mathrm{W})$.
Here the index $2$ indicates that this is double periodic, and the letter $\mathrm{W}$ refers to the wind-tree model.
To construct $\Bil(\infty;2;\mathrm{W})$, we take one of the building blocks, and unfold the $4$-direction billiard flow on it to a $1$-direction geodesic flow on a $4$-copy version of it, obtained by reflecting horizontally and vertically.
Note that each copy of the period surface has a right-neighbor, a left-neighbor, a down-neighbor and an up-neighbor in $\Bil(\infty;2;\mathrm{W})$,
and we need appropriate edge identification for gluing them together.

The period-surface $\Bil(2;\mathrm{W})$ of $\Bil(\infty;2;\mathrm{W})$ is shown in the picture on the left in Figure~7.4.7. 
The horizontal and vertical streets of $\Bil(2;\mathrm{W})$ are indicated in the picture on the right in Figure~7.4.7 where, for instance, the entries
$\leftrightarrow3$ and $\updownarrow5$ in a square face indicates that the square face is on the $3$-rd horizontal street and the $5$-th vertical street.
It is easy to see that $\Bil(2;\mathrm{W})$ has $10$ horizontal streets, of which $8$ are of length $3$ and $2$ are of length~$4$.
It also has $10$ vertical streets, of which $8$ are of length $3$ and $2$ are of length~$4$.

\begin{displaymath}
\begin{array}{c}
\includegraphics[scale=0.8]{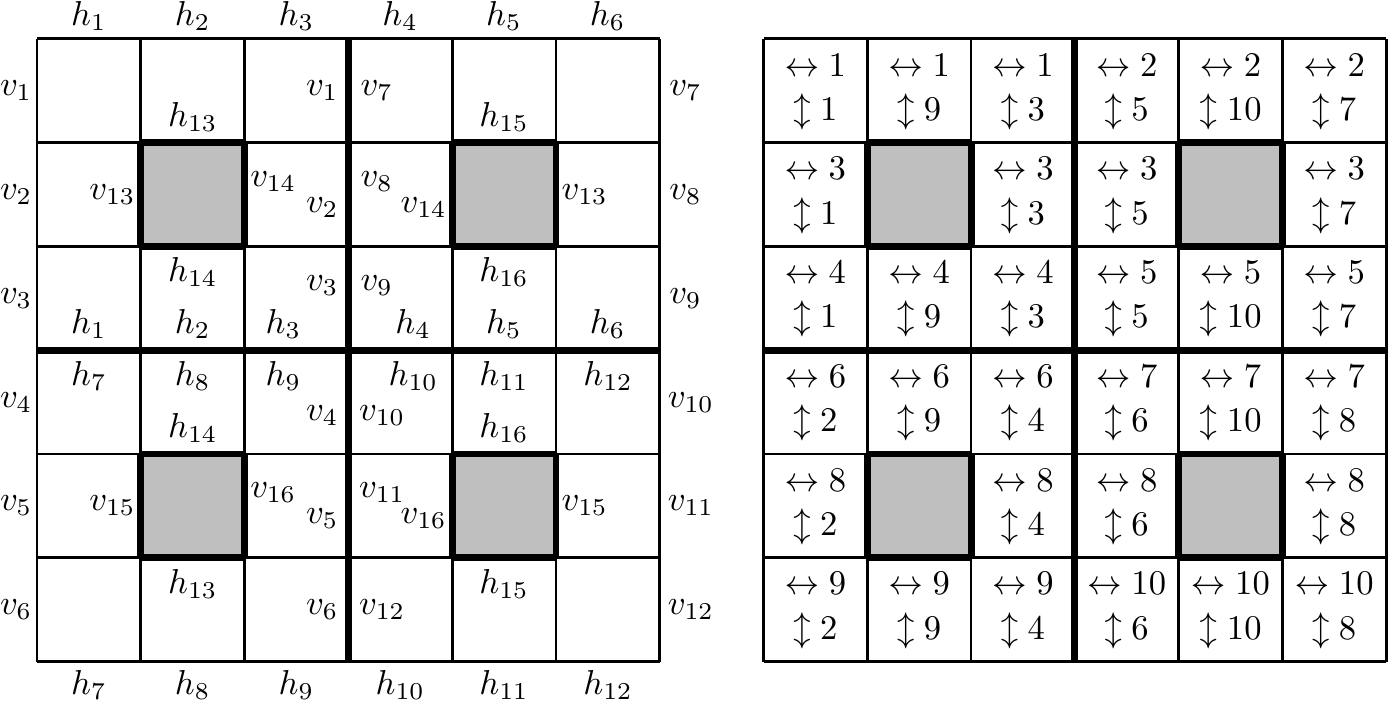}
\vspace{3pt}\\
\mbox{Figure 7.4.7: the period surface $\Bil(2;\mathrm{W})$ of $\Bil(\infty;2;\mathrm{W})$}
\end{array}
\end{displaymath}

Consider now a $1$-direction geodesic starting from some vertex of $\Bil(2;W)$ with slope $\alpha$ given by \eqref{eq7.2.1} with $m=n=12$.

With $J_1=J_4=\{1,3,9\}$ and $J_2=J_5=\{5,7,10\}$, for the horizontal streets corresponding to $\bfu_1,\bfu_2,\bfu_4,\bfu_5$ highlighted in Figure~7.4.8, we have
\begin{align}
(\bfA-I)\bfu_1
&
=(\bfA-I)\bfu_4
\nonumber
\\
&
=44(\bfu_1+\bfv_1)+32(\bfu_3+\bfv_3)+44(\bfu_4+\bfv_4)
\nonumber
\\
&\qquad
+12(\bfu_6+\bfv_6)+12(\bfu_9+\bfv_9),
\label{eq7.4.15}
\\
(\bfA-I)\bfu_2
&
=(\bfA-I)\bfu_5
\nonumber
\\
&
=44(\bfu_2+\bfv_2)+32(\bfu_3+\bfv_3)+44(\bfu_5+\bfv_5)
\nonumber
\\
&\qquad
+12(\bfu_7+\bfv_7)+12(\bfu_{10}+\bfv_{10}).
\label{eq7.4.16}
\end{align}
\begin{displaymath}
\begin{array}{c}
\includegraphics[scale=0.8]{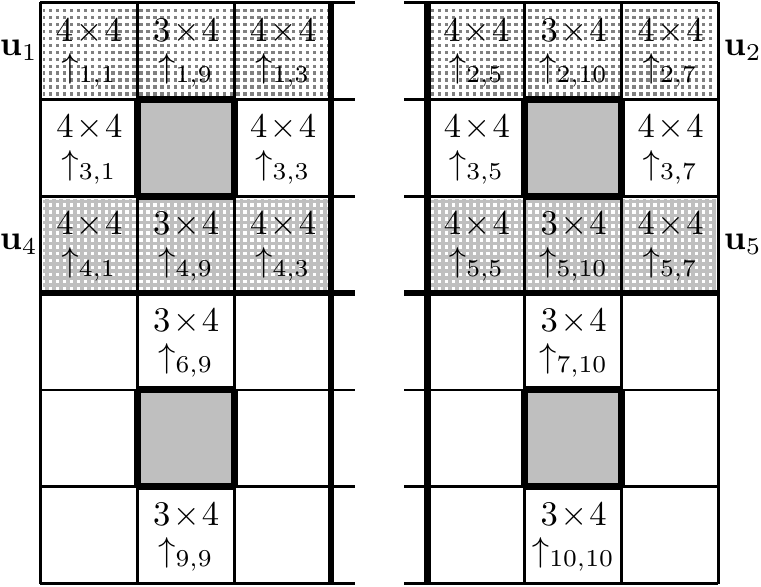}
\vspace{3pt}\\
\mbox{Figure 7.4.8: almost vertical of type $\uparrow$ in $\bfu_1,\bfu_2,\bfu_4,\bfu_5$}
\end{array}
\end{displaymath}

With $J_6=J_9=\{2,4,9\}$ and $J_7=J_{10}=\{6,8,10\}$, for the horizontal streets corresponding to $\bfu_6,\bfu_7,\bfu_9,\bfu_{10}$ highlighted in Figure~7.4.9, we have
\begin{align}
(\bfA-I)\bfu_6
&
=(\bfA-I)\bfu_9
\nonumber
\\
&
=12(\bfu_1+\bfv_1)+12(\bfu_4+\bfv_4)+44(\bfu_6+\bfv_6)
\nonumber
\\
&\qquad
+32(\bfu_8+\bfv_8)+44(\bfu_9+\bfv_9),
\label{eq7.4.17}
\\
(\bfA-I)\bfu_7
&
=(\bfA-I)\bfu_{10}
\nonumber
\\
&
=12(\bfu_2+\bfv_2)+12(\bfu_5+\bfv_5)+44(\bfu_7+\bfv_7)
\nonumber
\\
&\qquad
+32(\bfu_8+\bfv_8)+44(\bfu_{10}+\bfv_{10}).
\label{eq7.4.18}
\end{align}

\begin{displaymath}
\begin{array}{c}
\includegraphics[scale=0.8]{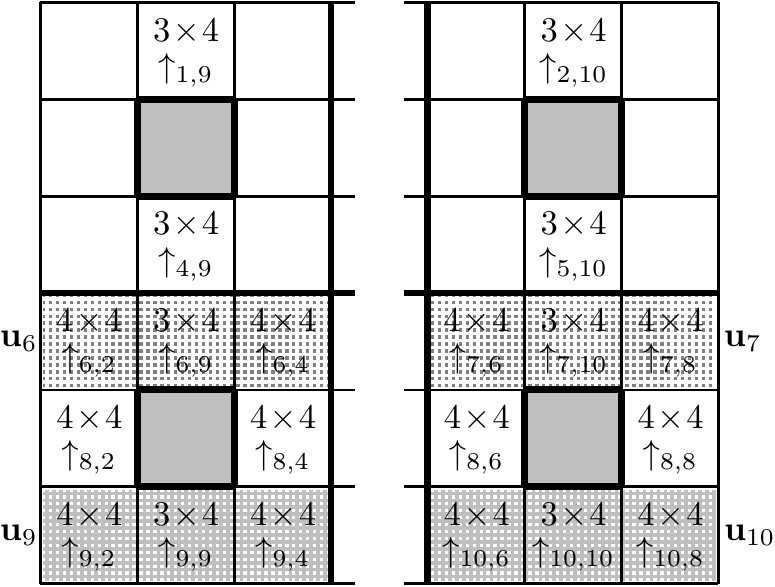}
\vspace{3pt}\\
\mbox{Figure 7.4.9: almost vertical of type $\uparrow$ in $\bfu_6,\bfu_7,\bfu_9,\bfu_{10}$}
\end{array}
\end{displaymath}

With $J_3=\{1,3,5,7\}$ and $J_8=\{2,4,6,8\}$, for the horizontal streets corresponding to $\bfu_3,\bfu_8$ highlighted in Figure~7.4.10, we have
\begin{align}
(\bfA-I)\bfu_3
&
=24(\bfu_1+\bfv_1)+24(\bfu_2+\bfv_2)+48(\bfu_3+\bfv_3)
\nonumber
\\
&\qquad
+24(\bfu_4+\bfv_4)+24(\bfu_5+\bfv_5),
\label{eq7.4.19}
\\
(\bfA-I)\bfu_8
&
=24(\bfu_6+\bfv_6)+24(\bfu_7+\bfv_7)+48(\bfu_8+\bfv_8)
\nonumber
\\
&\qquad
+24(\bfu_9+\bfv_9)+24(\bfu_{10}+\bfv_{10}).
\label{eq7.4.20}
\end{align}
\begin{displaymath}
\begin{array}{c}
\includegraphics[scale=0.8]{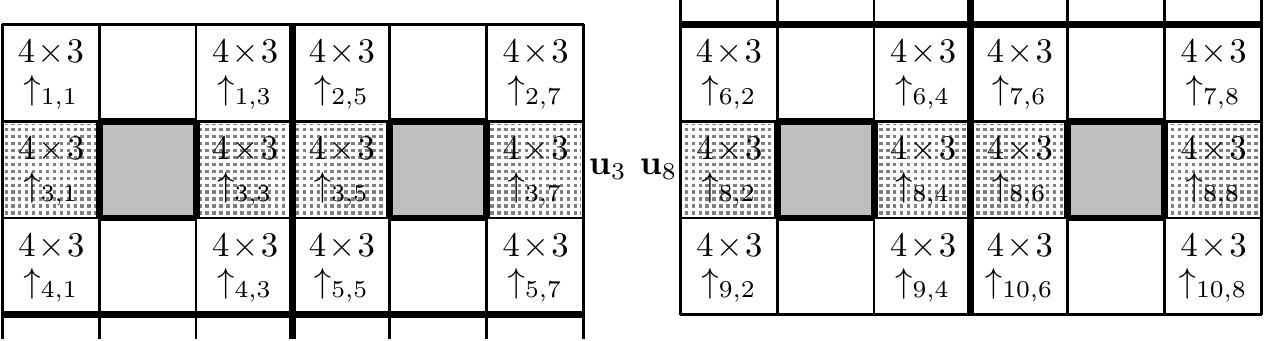}
\vspace{3pt}\\
\mbox{Figure 7.4.10: almost vertical of type $\uparrow$ in $\bfu_3$}
\end{array}
\end{displaymath}

It follows from \eqref{eq7.4.15}--\eqref{eq7.4.20} that the street-spreading matrix is given by
\begin{equation}\label{eq7.4.21}
\bfS=\begin{pmatrix}
44&0&24&44&0&12&0&0&12&0\\
0&44&24&0&44&0&12&0&0&12\\
32&32&48&32&32&0&0&0&0&0\\
44&0&24&44&0&12&0&0&12&0\\
0&44&24&0&44&0&12&0&0&12\\
12&0&0&12&0&44&0&24&44&0\\
0&12&0&0&12&0&44&24&0&44\\
0&0&0&0&0&32&32&48&32&32\\
12&0&0&12&0&44&0&24&44&0\\
0&12&0&0&12&0&44&24&0&44
\end{pmatrix}.
\end{equation}
This has non-zero eigenvalues and corresponding eigenvectors given by
\begin{align}
\tau_1=144,
&\quad
\psi_1=(3,3,4,3,3,3,3,4,3,3)^T,
\nonumber
\\
\tau_2=112,
&\quad
\psi_{2,1}=(0,-1,-1,0,-1,1,0,1,1,0)^T,
\nonumber
\\
&\quad
\psi_{2,2}=(-1,0,-1,-1,0,0,1,1,0,1)^T,
\nonumber
\\
\tau_3=64,
&\quad
\psi_3=(1,-1,0,1,-1,-1,1,0,-1,1)^T,
\nonumber
\\
\tau_4=16,
&\quad
\psi_4=(1,1,-4,1,1,1,1,-4,1,1)^T.
\nonumber
\end{align}

Note that the street-spreading matrix $\bfS$ given by \eqref{eq7.4.21} corresponds to the choice of parameters $m=n=12$.
Switching to $m=n=12k$, where $k\ge1$ is any integer, we obtain the new street-spreading matrix $\bfS(k)=k^2\bfS$ simply by multiplying the matrix $\bfS$ by~$k^2$.
Then of course the eigenvalues are also multiplied by~$k^2$, but the eigenvectors remain the same.
Naturally the $2$-step transition matrix $\bfA$ is modified to~$\bfA(k)$.

We next determine some of the eigenvalues of $\bfA(k)$ using \eqref{eq7.2.38}.
The largest eigenvalue of $\bfA(k)$ is
\begin{equation}\label{eq7.4.22}
\lambda_1(k)
=1+\frac{\tau_1k^2+\sqrt{\tau_1^2k^4+4\tau_1k^2}}{2}
=1+72k^2+12k\sqrt{36k^2+1},
\end{equation}
while the second largest eigenvalue of $\bfA(k)$ is
\begin{equation}\label{eq7.4.23}
\lambda_2(k)
=1+\frac{\tau_2k^2+\sqrt{\tau_2^2k^4+4\tau_2k^2}}{2}
=1+56k^2+4k\sqrt{196k^2+7},
\end{equation}
with eigenvectors of the form
\begin{align}
\Psi_{2,1}
&
=(0,-1,-1,0,-1,1,0,1,1,0,0,-\tau_2^*,-\tau_2^*,0,-\tau_2^*,\tau_2^*,0,\tau_2^*,\tau_2^*,0)^T,
\nonumber
\\
\Psi_{2,2}
&
=(-1,0,-1,-1,0,0,1,1,0,1,-\tau_2^*,0,-\tau_2^*,-\tau_2^*,0,0,\tau_2^*,\tau_2^*,0,\tau_2^*)^T,
\nonumber
\end{align}
where
\begin{displaymath}
\tau_2^*=\frac{-\tau_2k^2+\sqrt{\tau_2^2k^4+4\tau_2k^2}}{2}=4k\sqrt{196k^2+7}-56k^2.
\end{displaymath}

Let us return to Figures 7.4.6 and~7.4.7.
It is clear that the billiard moves to the up-neighbour if the $1$-direction geodesic hits any of the edges $h_1,\ldots,h_6$, and moves to the
down-neighbor if the $1$-direction geodesic hits any of the edges $h_7,\ldots,h_{12}$.
It is also clear that the billiard moves to the right-neighbor if the $1$-direction geodesic hits any of the edges $v_7,\ldots,v_{12}$, and moves to the
left-neighbor if the $1$-direction geodesic hits any of the edges $v_1,\ldots,v_6$.

We now proceed to find the edge cutting numbers of $h_1,\ldots,h_6$ and $h_7,\ldots,h_{12}$.
As before, it is easily shown that each of $\bfv_1,\ldots\bfv_{10}$ contributes precisely zero.

As in Examples \ref{ex7.3.1} and~\ref{ex7.3.2}, the eigenvalue $\lambda_1$ does not contribute to their difference, as any lack of cancellation would violate the
Gutkin--Veech theorem that guarantees uniformity.

The eigenvector of $\bfA(k)$ corresponding to the eigenvector $\Psi_{2,1}$ of $\bfA\vert_\VVV$ is given by
\begin{align}\label{eq7.4.24}
&
-\bfu_2-\bfu_3-\bfu_5+\bfu_6+\bfu_8+\bfu_9
\nonumber
\\
&\qquad
-\tau_2^*\bfv_2-\tau_2^*\bfv_3-\tau_2^*\bfv_5+\tau_2^*\bfv_6+\tau_2^*\bfv_8+\tau_2^*\bfv_9.
\end{align}
The eigenvector of $\bfA(k)$ corresponding to the eigenvector $\Psi_{2,2}$ of $\bfA\vert_\VVV$ is given by
\begin{align}\label{eq7.4.25}
&
-\bfu_1-\bfu_3-\bfu_4+\bfu_7+\bfu_8+\bfu_{10}
\nonumber
\\
&\qquad
-\tau_2^*\bfv_1-\tau_2^*\bfv_3-\tau_2^*\bfv_4+\tau_2^*\bfv_7+\tau_2^*\bfv_8+\tau_2^*\bfv_{10}.
\end{align}

We first find the number of almost vertical units counted here that cut the edges $h_1,\ldots,h_6$; see Figure~7.4.7.
The counts from $\bfu_i$, $i=1,\ldots,10$, are
\begin{displaymath}
\begin{array}{lllll}
\bfu_1\mapsto44k^2,\quad
&\bfu_2\mapsto44k^2,\quad
&\bfu_3\mapsto48k^2,\quad
&\bfu_4\mapsto44k^2,\quad
&\bfu_5\mapsto44k^2,
\\
\bfu_6\mapsto12k^2,\quad
&\bfu_7\mapsto12k^2,\quad
&\bfu_8\mapsto0,\quad
&\bfu_9\mapsto12k^2,\quad
&\bfu_{10}\mapsto12k^2.
\end{array}
\end{displaymath}
Thus corresponding to the eigenvector $\Psi_{2,1}$ and \eqref{eq7.4.24}, the total count is~$-112k^2$.
Corresponding to the eigenvector $\Psi_{2,2}$ and \eqref{eq7.4.25}, the total count is also~$-112k^2$.

We next find the number of almost vertical units counted here that cut the edges $h_7,\ldots,h_{12}$; see Figure~7.4.7.
The counts from $\bfu_i$, $i=1,\ldots,10$, are
\begin{displaymath}
\begin{array}{lllll}
\bfu_1\mapsto12k^2,\quad
&\bfu_2\mapsto12k^2,\quad
&\bfu_3\mapsto0,\quad
&\bfu_4\mapsto12k^2,\quad
&\bfu_5\mapsto12k^2,
\\
\bfu_6\mapsto44k^2,\quad
&\bfu_7\mapsto44k^2,\quad
&\bfu_8\mapsto48k^2,\quad
&\bfu_9\mapsto44k^2,\quad
&\bfu_{10}\mapsto44k^2.
\end{array}
\end{displaymath}
Thus corresponding to the eigenvector $\Psi_{2,1}$ and \eqref{eq7.4.24}, the total count is~$112k^2$.
Corresponding to the eigenvector $\Psi_{2,2}$ and \eqref{eq7.4.25}, the total count is also~$112k^2$.

The difference, in absolute value, is therefore~$224k^2$ in each case.
It follows that the second eigenvalue $\lambda_2(k)$ contributes $224(c_{2,1}+c_{2,2})\lambda_2^r(k)k^2$ to the difference between the edge cuttings numbers of $h_1,\ldots,h_6$ and $h_7,\ldots,h_{12}$.

Naturally, we can choose a starting vector $\bfw_0$ for the geodesic in question such that $c_{2,1}+c_{2,2}\ne0$.
This is clearly possible, since $c_{2,1}+c_{2,2}=0$ always would violate the linear independence of the eigenvectors $\Psi_{2,1}$ and~$\Psi_{2,2}$.

So the deviation from the starting point comes from the second largest eigenvalue, with order of magnitude $\lambda_2^r$ compared to the order of magnitude
$\lambda_1^r$ of the main term.
Choosing $T=\lambda_1^r$, we have $\lambda_2^r\asymp T^{\kappa_0}$, where
\begin{equation}\label{eq7.4.26}
\kappa_0=\kappa_0(k)=\frac{\log\lambda_2}{\log\lambda_1}=\frac{2\log k+\log112}{2\log k+\log144}+o(1),
\end{equation}
in view of \eqref{eq7.4.22} and \eqref{eq7.4.23}, is the irregularity exponent of a $1$-direction geodesic of slope
\begin{equation}\label{eq7.4.27}
\alpha_k=[12k;12k,12k,12k,\ldots]=6k+\sqrt{36k^2+1}=\sqrt{\lambda_1(k)}
\end{equation}
on the period-surface $\Bil(2;\mathrm{W})$ in Figure~7.4.8.

Again $\kappa_0=\kappa_0(k)\to1$ as $k\to\infty$.
So we have just established $T^{\kappa_0}=T^{1-\eps}$ size super-fast escape rate to infinity for this infinite billiard with the explicit class of quadratic irrational slopes in \eqref{eq7.4.27} where the parameter $k\ge1$ is any integer.

It is easy to see that the irregularity exponent $\kappa_0=\kappa_0(k)$ in \eqref{eq7.4.26} is {\it precisely} the escape rate to infinity of this infinite billiard.
The exponent of the escape rate to infinity cannot be larger than the expression \eqref{eq7.4.26} coming from the two largest eigenvalues.
\end{example}

In view of the examples in Sections \ref{sec7.3} and \ref{sec7.4} so far, it seems plausible to conjecture that every periodic infinite polysquare billiard with at least one \textit{infinite street} exhibits super-fast escape rate to infinity for some concrete infinite class of quadratic irrational slopes.
It would be nice to prove this conjecture in general.

We complete this section with a hybrid example.

\begin{example}\label{ex7.4.3}
Our final example is modified from the Ehrenfest wind-tree model.
Consider a $2$-direction billiard trajectory in the region in Figure~7.4.11.
Here the building block of this infinite polysquare region is a $3\times3$ square, with the middle square missing and with two vertical walls.
Note that there is no billiard in the vertical direction, so this is a hybrid problem.
The trajectory keeps on going up vertically, or going down vertically, one way but not both.

\begin{displaymath}
\begin{array}{c}
\includegraphics[scale=0.8]{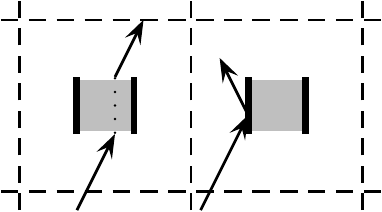}
\vspace{3pt}\\
\mbox{Figure 7.4.11: a hybrid model}
\end{array}
\end{displaymath}

As usual, this infinite hybrid model is equivalent to a $1$-direction geodesic flow on an infinite polysquare translation surface that we denote by $\Bil(\infty;2;\mathrm{H})$.
Here the index $2$ indicates that this is double periodic, and the letter $\mathrm{H}$ refers to the hybrid model.
To construct $\Bil(\infty;2;\mathrm{H})$, we take one of the building blocks, and unfold the $2$-direction billiard flow on it to a $1$-direction geodesic flow on a $2$-copy version of it, obtained by reflecting vertically.
Note that each $2$-copy version has a right-neighbor, a left-neighbor, a down-neighbor and an up-neighbor in $\Bil(\infty;2;\mathrm{H})$, and we need appropriate edge identification for gluing them together.

The period-surface $\Bil(2;\mathrm{H})$ of $\Bil(\infty;2;\mathrm{H})$ is shown in the picture on the left in Figure~7.4.12. 

\begin{displaymath}
\begin{array}{c}
\includegraphics[scale=0.8]{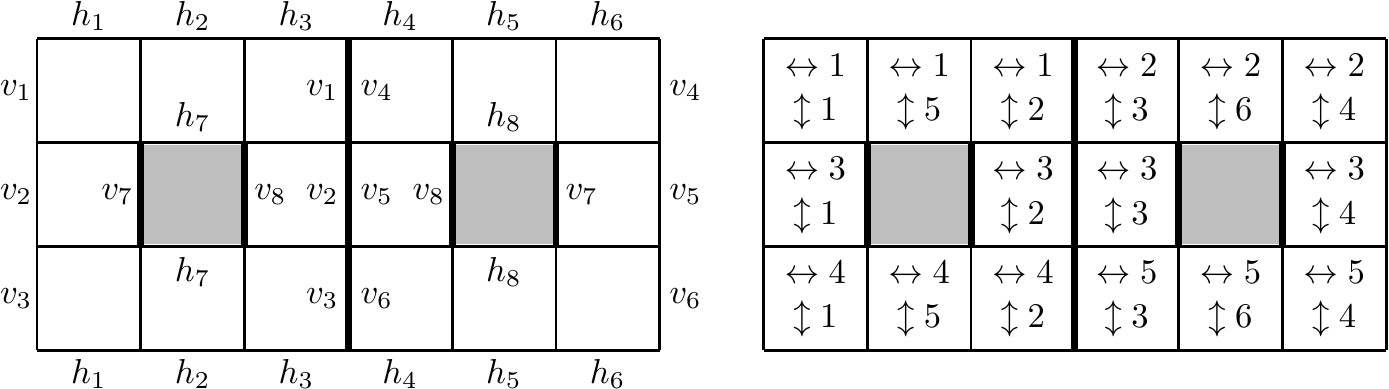}
\vspace{3pt}\\
\mbox{Figure 7.4.12: the period surface $\Bil(2;\mathrm{H})$ of $\Bil(\infty;2;\mathrm{H})$}
\end{array}
\end{displaymath}

The horizontal and vertical streets of $\Bil(2;\mathrm{H})$ are indicated in the picture on the right in Figure~7.4.12 where, for instance, the entries
$\leftrightarrow3$ and $\updownarrow4$ in a square face indicates that the square face is on the $3$-rd horizontal street and the $4$-th vertical street.
It is easy to see that $\Bil(2;\mathrm{H})$ has $5$ horizontal streets, of which $4$ are of length $3$ and $1$ is of length~$4$.
It also has $6$ vertical streets, of which $4$ are of length $3$ and $2$ are of length~$2$.

Consider now a $1$-direction geodesic starting from some vertex of $\Bil(2;H)$ with slope $\alpha$ given by \eqref{eq7.2.1} with $m=n=12$.

With $J_1=J_4=\{1,2,5\}$, for the horizontal streets corresponding to $\bfu_1,\bfu_4$ highlighted in the picture on the left in Figure~7.4.13, we have
\begin{align}
(\bfA-I)\bfu_1
&
=(\bfA-I)\bfu_4
\nonumber
\\
&
=56(\bfu_1+\bfv_1)+32(\bfu_3+\bfv_3)+56(\bfu_4+\bfv_4)
\label{eq7.4.28}
\end{align}

With $J_2=J_5=\{3,4,6\}$, for the horizontal streets corresponding to $\bfu_2,\bfu_5$ highlighted in the picture on the right in Figure~7.4.13, we have

\begin{align}
(\bfA-I)\bfu_2
&
=(\bfA-I)\bfu_5
\nonumber
\\
&
=56(\bfu_2+\bfv_2)+32(\bfu_3+\bfv_3)+56(\bfu_5+\bfv_5).
\label{eq7.4.29}
\end{align}
\begin{displaymath}
\begin{array}{c}
\includegraphics[scale=0.8]{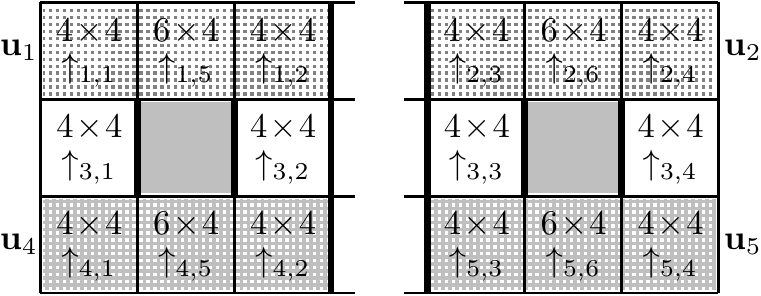}
\vspace{3pt}\\
\mbox{Figure 7.4.13: almost vertical units of type $\uparrow$ in $\bfu_1,\bfu_2,\bfu_4,\bfu_5$}
\end{array}
\end{displaymath}

With $J_3=\{1,2,3,4\}$, for the horizontal street corresponding to $\bfu_3$ highlighted in Figure~7.4.14, we have
\begin{align}\label{eq7.4.30}
(\bfA-I)\bfu_3
&
=24(\bfu_1+\bfv_1)+24(\bfu_2+\bfv_2)+48(\bfu_3+\bfv_3)
\nonumber
\\
&\qquad
+24(\bfu_4+\bfv_4)+24(\bfu_5+\bfv_5).
\end{align}
\begin{displaymath}
\begin{array}{c}
\includegraphics[scale=0.8]{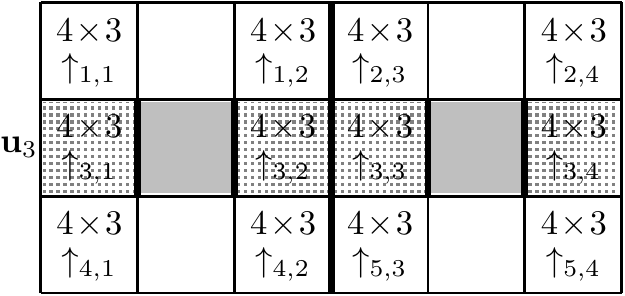}
\vspace{3pt}\\
\mbox{Figure 7.4.14: almost vertical units of type $\uparrow$ in $\bfu_3$}
\end{array}
\end{displaymath}

It follows from \eqref{eq7.4.28}--\eqref{eq7.4.30} that the street-spreading matrix is given by
\begin{equation}\label{eq7.4.31}
\bfS=\begin{pmatrix}
56&0&24&56&0\\
0&56&24&0&56\\
32&32&48&32&32\\
56&0&24&56&0\\
0&56&24&0&56
\end{pmatrix}.
\end{equation}
This has non-zero eigenvalues and corresponding eigenvectors given by
\begin{align}
\tau_1=144,
&\quad
\psi_1=(3,3,4,3,3)^T,
\nonumber
\\
\tau_2=112,
&\quad
\psi_2=(-1,1,0,-1,1)^T,
\nonumber
\\
\tau_3=16,
&\quad
\psi_3=(1,1,-4,1,1)^T.
\nonumber
\end{align}

Note that the street-spreading matrix $\bfS$ given by \eqref{eq7.4.31} corresponds to the choice of parameters $m=n=12$.
Switching to $m=n=12k$, where $k\ge1$ is any integer, we obtain the new street-spreading matrix $\bfS(k)=k^2\bfS$ simply by multiplying the matrix $\bfS$ by~$k^2$.
Then of course the eigenvalues are also multiplied by~$k^2$, but the eigenvectors remain the same.
Naturally the $2$-step transition matrix $\bfA$ is modified to~$\bfA(k)$.

We next determine some of the eigenvalues of $\bfA(k)$ using \eqref{eq7.2.38}.
The largest eigenvalue of $\bfA(k)$ is
\begin{displaymath}
\lambda_1(k)
=1+\frac{\tau_1k^2+\sqrt{\tau_1^2k^4+4\tau_1k^2}}{2}
=1+72k^2+12k\sqrt{36k^2+1},
\end{displaymath}
while the second largest eigenvalue of $\bfA(k)$ is
\begin{displaymath}
\lambda_2(k)
=1+\frac{\tau_2k^2+\sqrt{\tau_2^2k^4+4\tau_2k^2}}{2}
=1+56k^2+4k\sqrt{196k^2+7},
\end{displaymath}
with eigenvector of the form
\begin{align}
\Psi_2
&
=(-1,1,0,-1,1,-\tau_2^*,\tau_2^*,0,-\tau_2^*,\tau_2^*)^T,
\nonumber
\end{align}
where
\begin{displaymath}
\tau_2^*=\frac{-\tau_2k^2+\sqrt{\tau_2^2k^4+4\tau_2k^2}}{2}=4k\sqrt{196k^2+7}-56k^2.
\end{displaymath}

As observed earlier, there is no billiard in the vertical direction, so we now proceed to study its horizontal behaviour.
We thus proceed to find the edge cutting numbers of $v_1,v_2,v_3$ and $v_4,v_5,v_6$; see Figure~7.4.12.

It can be shown that the eigenvalue $\lambda_1$ does not contribute to their difference, although the Gutkin--Veech theorem does not apply here.
So we concentrate our attention on the contribution from the eigenvector $\Psi_2$ of $\bfA\vert_\VVV$.

The eigenvector of $\bfA(k)$ corresponding to the eigenvector $\Psi_2$ of $\bfA\vert_\VVV$ is given by
\begin{equation}\label{eq7.4.32}
-\bfu_1+\bfu_2-\bfu_4+\bfu_5
-\tau_2^*\bfv_1+\tau_2^*\bfv_2-\tau_2^*\bfv_4+\tau_2^*\bfv_5.
\end{equation}
Here, note that none of $\bfu_1,\ldots,\bfu_5$ makes any non-zero contribution, as only units of type $\nuparrow$ can contribute to the count.
For the contributions from $\bfv_1,\ldots,\bfv_5$, we use \eqref{eq7.2.15}.

We first find the number of almost vertical units counted in \eqref{eq7.4.32} that cut the edges $v_1,v_2,v_3$; see Figure~7.4.12.
Since $\bfv_3$ does not feature in \eqref{eq7.4.32}, it is not necessary to any counting for it.
Clearly we have a count of $4k$ for each of $\bfv_1,\bfv_4$, and none for $\bfv_2,\bfv_5$, making a total of~$-8\tau_2^*k$.

We next find the number of almost vertical units counted in \eqref{eq7.4.32} that cut the edges $v_4,v_5,v_6$; see Figure~7.4.12.
Clearly we have a count of $4k$ for each of $\bfv_2,\bfv_5$, and none for $\bfv_1,\bfv_4$, making a total of~$8\tau_2^*k$.

The difference, in absolute value, is therefore~$16\tau_2^*k$.
Thus the second eigenvalue $\lambda_2(k)$ contributes $16c_2\tau_2^*\lambda_2^r(k)k$ to the difference between the edge cuttings numbers of $v_1,v_2,v_3$ and $v_4,v_5,v_6$.

Again the \textit{horizontal} deviation from the starting point comes from the second largest eigenvalue, with order of magnitude $\lambda_2^r$ compared to the order of magnitude $\lambda_1^r$ of the main term.
\end{example}

%%%%%%%%%%
%
% SECTION 7.5
%
%%%%%%%%%%

\subsection{Geodesics with arbitrary starting points}\label{sec7.5}

In Sections \ref{sec7.3} and~\ref{sec7.4}, we have exhibited examples of $1$-direction geodesics in infinite polysquare translation surfaces that start from a vertex and which exhibit super-fast escape rate to infinity.
In Example~\ref{ex7.3.2}, we have also shown that there are $1$-direction geodesics in square-maze translation surfaces that start from a vertex and which exhibit super-slow escape rate to infinity.

We now describe a method by which we can relax the restriction that the starting point of the geodesic has to be a vertex of a polysquare translation surface, so that only the angle $\alpha$ of the geodesic matters, and can still establish results on the escape rate to infinity as before.

Note, first of all, that a $1$-direction geodesic \textit{modulo one} is equivalent to a torus line in the unit square.
To help us visualize the situation even more clearly, we can replace the unit square by $\Rr^2$ and the torus line by a line on~$\Rr^2$.

Next, note that we are considering slopes $\alpha$ that are badly approximable numbers.
Such numbers satisfy the following two properties.

\begin{propa}
Suppose that $\alpha$ is badly approximable.
Then there exists a constant $C^*=C^*(\alpha)>0$ such that for every real number $\ell\ge1$, every segment of length $\ell$ of a straight line of slope
$\alpha$ has distance at most $C^*/\ell$ from the nearest integer lattice point in~$\Zz^2$.
\end{propa}

\begin{displaymath}
\begin{array}{c}
\includegraphics[scale=0.8]{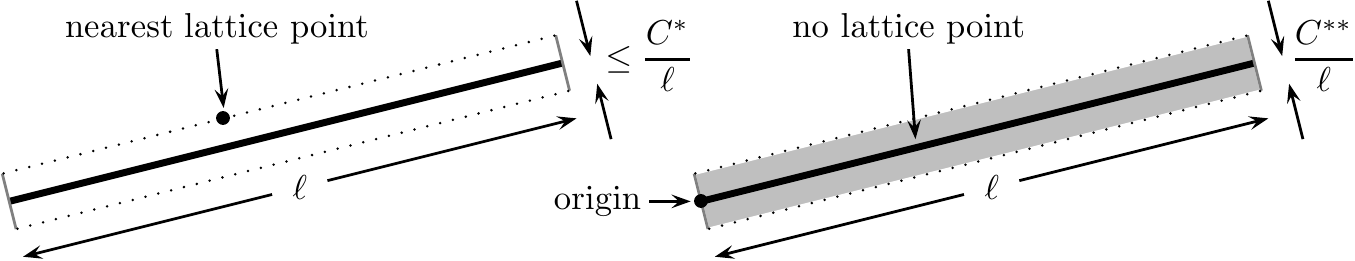}
\vspace{3pt}\\
\mbox{Figure 7.5.1: illustrating Property A and Property B}
\end{array}
\end{displaymath}

\begin{propb}
Suppose that $\alpha$ is badly approximable.
Then there exists a constant $C^{**}=C^{**}(\alpha)>0$ such that for any real number $\ell\ge1$, any tilted rectangle with one side starting from the origin, with slope $\alpha$ and length~$\ell$, and with the perpendicular side of length $C^{**}/\ell$, does not contain any integer lattice point in $\Zz^2$ except the origin.
\end{propb}

The idea is to consider geodesics that are parallel to our given geodesic and which share some of its characteristics.

Consider the initial segment $L(S;t)$, $0\le t\le T$, with $T\ge2$, of a $1$-direction geodesic $L$ of slope~$\alpha$, starting at a point~$S$, not necessarily a vertex of a polysquare translation surface~$\PPP$, as shown in the top half of Figure~7.5.2.

\begin{displaymath}
\begin{array}{c}
\includegraphics[scale=0.8]{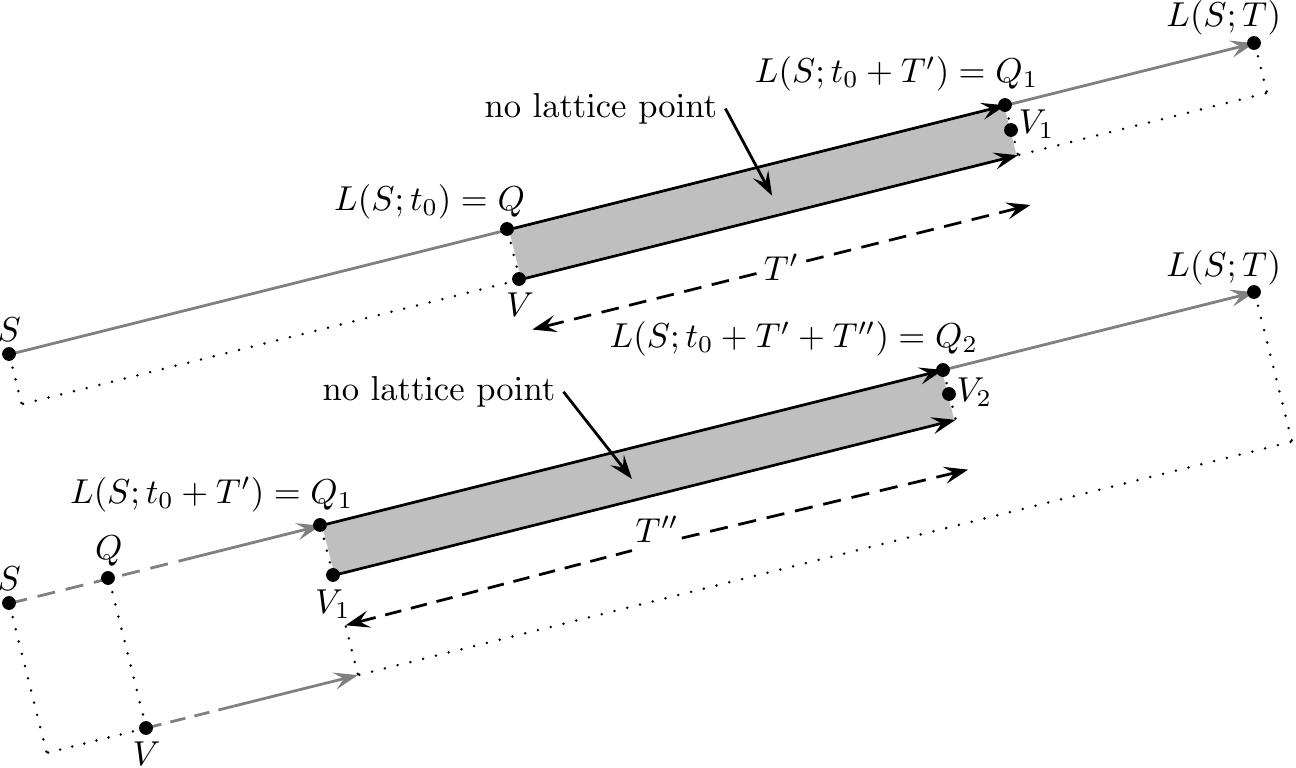}
\vspace{3pt}\\
\mbox{Figure 7.5.2: using parallel geodesics}
\end{array}
\end{displaymath}

Let $Q=Q(T)=L(S;t_0)$ denote the point on this segment which is closest to a vertex of~$\PPP$, and let $V$ be this closest vertex.
Using Property~A, we see that the distance between $Q$ and $V$ is at most~$C^*/T$.

Starting from the vertex~$V$, we draw a new geodesic $L'$ of slope~$\alpha$, in the same forward direction as~$L$.
Consider the longest vertex-free rectangle between $L$ and~$L'$.
Using Property~B, we see that the side of this rectangle parallel to $L$ has length
\begin{equation}\label{eq7.5.1}
T'\ge\frac{C^{**}}{C^*}T.
\end{equation}

Suppose that $T\le t_0+T'$.
Then the segment $L(S;t)$, $t_0\le t\le T$, of $L$ remains close to the segment $L'(V;t)$, $0\le t\le T-t_0$, of the new geodesic~$L'$, since the rectangle between them does not contain any vertex of $\PPP$ apart from $V$ and another vertex at the far end if $T=t_0+T'$.

Suppose that $T>t_0+T'$.
Then the segment $L(S;t)$, $t_0\le t\le t_0+T'$, of $L$ remains close to the segment $L'(V;t)$, $0\le t\le T'$, of the new geodesic~$L'$, since the rectangle between them does not contain any vertex of $\PPP$ apart from at the two ends.
Since this rectangle is the longest vertex-free rectangle that we can draw, clearly there is another vertex $V_1$ of $\PPP$ at the far end.
We now repeat the argument.

Starting from the vertex~$V_1$, we draw a new geodesic $L''$ of slope~$\alpha$, in the same forward direction as~$L$.
Consider the longest vertex-free rectangle between $L$ and~$L''$, as shown in the bottom half of Figure~7.5.2.
Using Property~B, we see that the side of this rectangle parallel to $L$ has length
\begin{equation}\label{eq7.5.2}
T''\ge\frac{C^{**}}{C^*}T.
\end{equation}

Suppose that $t_0+T'<T\le t_0+T'+T''$.
Then the segment $L(S;t)$, $t_0+T'\le t\le T$, of $L$ remains close to the segment $L''(V_1;t)$, $0\le t\le T-t_0-T'$, of the new geodesic~$L''$, since the rectangle between them does not contain any vertex of $\PPP$ apart from $V_1$ and another vertex at the far end if $T=t_0+T'+T''$.

Suppose that $T>t_0+T'+T''$.
Then the segment $L(S;t)$, $t_0+T'\le t\le t_0+T'+T''$, of $L$ remains close to the segment $L''(V_1;t)$, $0\le t\le T''$, of the new geodesic~$L''$, since the rectangle between them does not contain any vertex of $\PPP$ apart from at the two ends.
Since this rectangle is the longest vertex-free rectangle that we can draw, clearly there is another vertex $V_2$ of $\PPP$ at the far end.
We now repeat the argument again.

Clearly the argument must stop after a finite number of steps, in view of estimates such as \eqref{eq7.5.1} and \eqref{eq7.5.2} and their analogs.
In particular it must stop after at most $C^*/C^{**}$ steps.

Note that we have moved forward from the point $Q$ to the point $L(S;T)$.
Clearly a similar argument applies when we move backward from the point $Q$ to the starting point $S=L(S;0)$.

What we have shown is that the finite geodesic $L(S;t)$, $0\le t\le T$, can be broken up into a finite number of parts, each of which remains close to a parallel finite geodesic of the same length that starts from a vertex of~$\PPP$.

We can now use this observation to study escape rates to infinity.

Suppose that we have established that for a polysquare translation surface~$\PPP$, the escape rate to infinity for a $1$-direction geodesic $L$ of slope $\alpha$ that starts at a vertex of $\PPP$ is $O(f(T))$ as a function of time~$T$, where $f$ is an increasing function satisfying $f(t)\to\infty$ as $t\to\infty$.
This means that the diameter of the finite geodesic $L(t)$, $0\le t\le T$, is $O(f(T))$.
Consider now a $1$-direction geodesic $L(S;t)$, $0\le t\le T$, of slope $\alpha$ and starting point $S$ that is not necessarily a vertex of~$\PPP$.
We now break $L(S;t)$, $0\le t\le T$, into a finite number of parts as described above, and note that each part remains close to a parallel finite geodesic that starts at a vertex of $\PPP$ and has length at most~$T$.
The diameter of each of these parallel finite geodesics is $O(f(T))$.
It follows that the diameter of $L(S;t)$, $0\le t\le T$, is also $O(f(T))$.

For super-slow logarithmic escape rate to infinity, we have $f(T)=\log T$.
For super-fast escape rate to infinity as in the examples in Sections \ref{sec7.3} and~\ref{sec7.4}, we have $f(T)=T^{\kappa_0}$.

To exhibit fluctuations of the required order of magnitude in the escape rate to infinity, we have a similar approach but with a slightly different first step.

Let $n\ge2$ be an arbitrarily large but fixed integer.
Consider a $1$-direction geodesic $L(S;t)$, $t\ge0$, on a polysquare translation surface $\PPP$ that starts from an arbitrary point with slope~$\alpha$.
Suppose that $t_0$ is the first value of $t\ge0$ such that $L(S;t_0)$ has perpendicular distance at most $1/n$ from a vertex of~$\PPP$, as shown in Figure~7.5.3.

\begin{displaymath}
\begin{array}{c}
\includegraphics[scale=0.8]{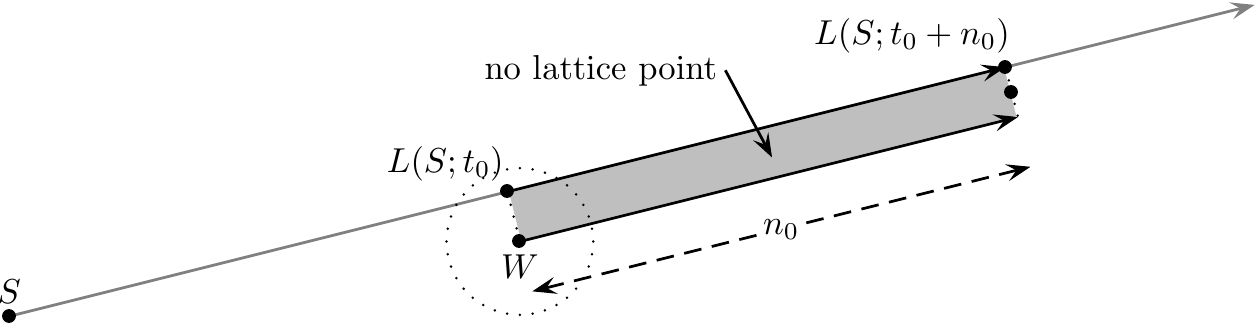}
\vspace{3pt}\\
\mbox{Figure 7.5.3: using parallel geodesics}
\end{array}
\end{displaymath}

In view of Property~A, we know that $0\le t_0\le C^*n$.

Suppose that the vertex of $\PPP$ in question is~$W$.
Starting from this vertex~$W$, we draw a new geodesic $L'$ of slope~$\alpha$, in the same forward direction as~$L$.
Consider the longest vertex-free rectangle between $L$ and~$L'$.
Using Property~B, we see that the side of this rectangle parallel to $L$ has length
\begin{equation}\label{eq7.5.3}
n_0\ge C^{**}n.
\end{equation}
Furthermore, the segment $L(S;t)$, $t_0\le t\le t_0+n_0$, of $L$ remains close to the segment $L'(W;t)$, $0\le t\le n_0$, of the new geodesic~$L'$, since the rectangle between them does not contain any vertex of $\PPP$ apart from at the two ends.

Suppose now that fluctuations of size $C_1T^{\kappa_0}$ are exhibited for geodesics of slope $\alpha$ and length $T$ that start from a vertex
of~$\PPP$.
Then $L'(W;t)$, $0\le t\le n_0$, and hence also $L(S;t)$, $t_0\le t\le t_0+n_0$, exhibits fluctuations of size
\begin{displaymath}
C_1n_0^{\kappa_0}\ge C_2n^{\kappa_0},
\end{displaymath}
in view of \eqref{eq7.5.3}.
If the distance between $S$ and $L(S;t_0)$ exceeds
\begin{equation}\label{eq7.5.4}
\frac{1}{2}C_2n^{\kappa_0},
\end{equation}
then $L(S;t)$, $0\le t\le t_0$, has diameter at least equal to \eqref{eq7.5.4}.
Otherwise $L(S;t)$, $0\le t\le t_0+n_0$, has diameter at least equal to
\begin{displaymath}
C_2n^{\kappa_0}-\frac{1}{2}C_2n^{\kappa_0}=\frac{1}{2}C_2n^{\kappa_0}.
\end{displaymath}
%
%

%%%%%%%%%%
%
% SECTION 8
%
%%%%%%%%%%

\section{Beyond polysquare surfaces}\label{sec8}

%%%%%%%%%%
%
% SECTION 8.1
%
%%%%%%%%%%

\subsection{Street-rational polyrectangle translation surfaces}\label{sec8.1}

We have developed two different versions of the shortline method for polysquare translation surfaces.
Using the eigenvalue-based version of the method, we can prove time-quantitative uniformity in terms of the irregularity exponent, as in \cite{BDY1,BDY2} and Section~\ref{sec7} of the present paper.
Using the eigenvalue-free version of the method, we can prove time-quantitative density, including superdensity in some cases, as in~\cite{BCY}.

The purpose of this section is to show that both versions of the shortline method can be extended to the class of \textit{street-rational polyrectangle translation surfaces}, a class which includes all polysquare translation surfaces but goes far beyond.
The first example of this larger class comes from right triangle billiard with angle~$\pi/8$.
This is perhaps the simplest non-integrable right triangle billiard. 
The standard trick of unfolding implies that this billiard can be described in terms of a $1$-direction geodesic flow on the regular octagon surface, as shown in Figure~8.1.1.

\begin{displaymath}
\begin{array}{c}
\includegraphics[scale=0.8]{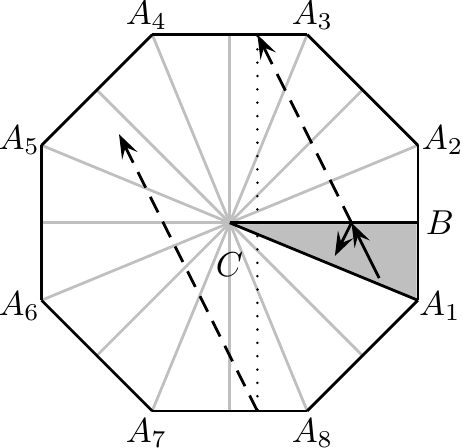}
\vspace{3pt}\\
\mbox{Figure 8.1.1: right triangle billiard with angle $\pi/8$, and geodesic flow}
\\
\mbox{on the regular octagon surface via unfolding}
\end{array}
\end{displaymath}

The definition of the \textit{regular octagon surface} is rather straightforward.
We identify the opposite parallel boundary edges by translation.
Thus the four identified pairs are
\begin{displaymath}
(A_1A_2,A_6A_5),
\quad
(A_2A_3,A_7A_6),
\quad
(A_3A_4,A_8A_7),
\quad
(A_4A_5,A_1A_8).
\end{displaymath}
Hence a $1$-direction geodesic flow on the regular octagon surface, a compact orientable surface, is a $16$-fold covering of the right triangle billiard with angle~$\pi/8$, in much the same way as a torus line flow on a $2\times2$ square is a $4$-fold covering of the square billiard.
Needless to say, the regular octagon surface looks completely different from a polysquare surface.

In Figure~8.1.1, the point $C$ represents the center of the octagon.
The right triangle $A_1BC$ has angle $\pi/8$ at~$C$.
Reflecting the $A_1BC$ billiard across the side $CB$ is the first step in the unfolding process, and leads to the triangle~$A_1A_2C$.
There are three more steps.
Reflecting $A_1A_2C$ across the side $CA_2$ leads to a polygon $A_1A_2A_3C$.
Reflecting $A_1A_2A_3C$ across the side $CA_3$ leads to a polygon $A_1A_2A_3A_4A_5C$.
Finally, reflecting $A_1A_2A_3A_4A_5C$ across the side $CA_5$ leads to the whole octagon.

Non-integrability is clear from the unfolding, since the vertices of the octagon are split-singularities of the geodesic flow on the surface.

Of course the same elegant construction works for any right triangle billiard with angle $\pi/k$ where $k\ge4$ is \textit{even}.
Unfolding will then convert the billiard orbit to a $1$-direction geodesic flow on the \textit{regular $k$-gon surface}, defined by identifying parallel boundary edges of the $k$-gon by translation.
These are non-integrable systems for every even integer $k\ge8$, so the regular octagon surface is the simplest such system.

\begin{remark}
Consider the regular $k$-gon surface with even $k\ge4$.
If $k$ is divisible by~$4$, then the boundary identification gives $1$ vertex, $k/2$ edges, and $1$ region, so Euler's formula
\begin{displaymath}
2-2g=\chi=V-E+R=1-(k/2)+1
\end{displaymath}
gives $g=k/4$.
If $k$ is not divisible by~$4$, then the boundary identification gives $2$ vertices, $k/2$ edges, and $1$ region, so Euler's formula
\begin{displaymath}
2-2g=\chi=V-E+R=2-(k/2)+1
\end{displaymath}
gives $g=(k-2)/4$.
Thus the regular $k$-gon surface with even $k$ has genus $1$ when $k=4,6$.
This is consistent with the well known fact that we can tile the plane with squares or regular hexagons, but not with any other regular polygons with an even number of sides.
\end{remark}

Let us return to the regular octagon surface.
While it looks completely different from a polysquare surface, there is a hidden similarity.
The regular octagon surface is in fact equivalent to a \textit{polyrectangle translation surface}; see Figure~8.1.2.

\begin{displaymath}
\begin{array}{c}
\includegraphics[scale=0.8]{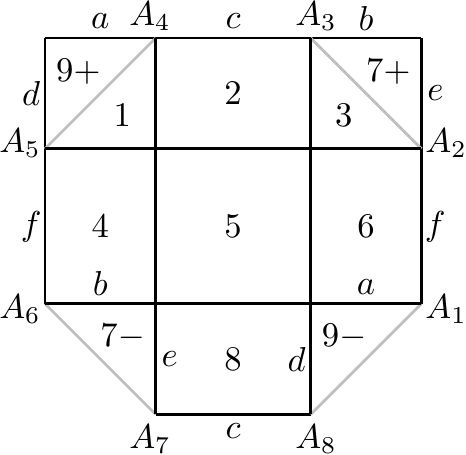}
\vspace{3pt}\\
\mbox{Figure 8.1.2: the regular octagon surface seen as}
\\
\mbox{a polyrectangle translation surface}
\end{array}
\end{displaymath}

The edge $A_2A_3$ is identified with the edge~$A_7A_6$.
This allows us to replace the triangle labelled $7-$ by the triangle labelled $7+$, with the two horizontal edges $b$ identified and the two vertical edges $e$ identified.
Likewise, the edge $A_4A_5$ is identified with the edge~$A_1A_8$.
This allows us to replace the triangle labelled $9-$ by the triangle labelled $9+$, with the two horizontal edges $a$ identified and the two vertical edges $d$ identified.
Thus the regular octagon surface becomes a polyrectangle translation surface consisting of $7$ rectangles, labelled
\begin{displaymath}
(1,9+),2,(3,7+),4,5,6,8.
\end{displaymath}
With the edge identification, this surface has $2$ horizontal streets
\begin{displaymath}
(1,9+),2,(3,7+),8
\quad\mbox{and}\quad
4,5,6,
\end{displaymath}
as well as $2$ vertical streets
\begin{displaymath}
(1,9+),4,(3,7+),6
\quad\mbox{and}\quad
2,5,8.
\end{displaymath}

The two identified edges $e$ also allows us to replace the rectangle $8$ at the bottom in Figure~8.1.2, indicated by $8-$ in Figure~8.1.3, by a rectangle on the top right indicated by $8+$.

\begin{displaymath}
\begin{array}{c}
\includegraphics[scale=0.8]{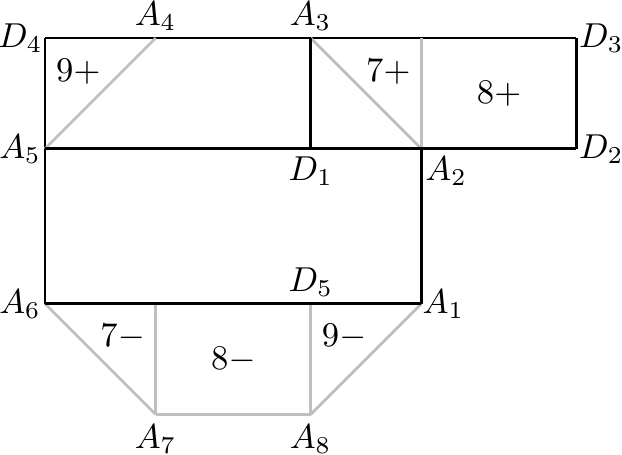}
\vspace{3pt}\\
\mbox{Figure 8.1.3: street decomposition into similar rectangles}
\end{array}
\end{displaymath}

Figures 8.1.2 and~8.1.3 justify the claim that the regular octagon surface is in fact a polyrectangle translation surface.
Furthermore, we have a second crucial property of the regular octagon surface, that the three rectangles
\begin{displaymath}
D_4A_5D_1A_3,
\quad
A_3D_1D_2D_3,
\quad
A_5A_6A_1A_2
\end{displaymath}
are similar.
To see this, assume, without loss of generality, that $\length(A_3D_1)=1$, so that each edge of the regular octagon has length~$\sqrt{2}$.
Then
\begin{equation}\label{eq8.1.1}
\frac{\length(A_5D_1)}{\length(A_3D_1)}=\frac{\length(D_1D_2)}{\length(D_3D_2)}=1+\sqrt{2},
\end{equation}
and
\begin{equation}\label{eq8.1.2}
\frac{\length(A_6A_1)}{\length(A_2A_1)}=\frac{2+\sqrt{2}}{\sqrt{2}}=1+\sqrt{2}.
\end{equation}
We also have
\begin{equation}\label{eq8.1.3}
\frac{\length(A_5D_2)}{\length(D_3D_2)}=2(1+\sqrt{2}).
\end{equation}
Note that \eqref{eq8.1.3} shows that the cotangent of the diagonal $A_5D_3$ of the horizontal street $D_4A_5D_2D_3$ is equal to $2(1+\sqrt{2})$, while \eqref{eq8.1.2} shows that the cotangent of the diagonal $A_6A_2$ of the horizontal street $A_5A_6A_1A_2$ is equal to $1+\sqrt{2}$.
We can draw similar conclusions for the vertical streets.

\begin{remark}
The equality of the ratios in \eqref{eq8.1.1} and \eqref{eq8.1.2} can also be established by using the well known geometric property of the circle very often known as the \textit{chord-angle relation} and which states that if we fix any chord with endpoints $P,Q$ on a circle, then the angle $PRQ$ remains the same for any point $R$ on the same circular arc~$PQ$.
Since $A_1A_2$ and $A_2A_3$ are two chords of the same length, the angles $A_1A_6A_2$ and $A_2A_5A_3$ are equal.
\end{remark}

The ratio of the cotangents of the diagonals of the horizontal streets, and similarly, the ratio of the cotangents of the diagonals of the vertical streets, is \textit{rational}.
Thus we refer to the regular octagon surface as a \textit{street-rational polyrectangle translation surface}.

\begin{remark}
Consider a horizontal street on a finite polyrectangle translation surface~$\PPP$.
The width of this street, \textit{i.e.} the perpendicular distance between its top and bottom horizontal edges, may not be equal to~$1$.
Let us expand or contract this street to obtain a similar copy where the width is now equal to~$1$.
Then we call the length of this expanded or contracted horizontal street the \textit{normalized length} of the horizontal street.
Some authors also use the terms \textit{modulus} and \textit{cylinder} in place of the terms \textit{normalized length} and \textit{street}.
We can also think of this as the cotangent of the (angle that the) diagonals (make with the direction) of the horizontal street.
For a finite street-rational polyrectangle translation surface~$\PPP$, there clearly exists a smallest real number $h^*$ which is an integer multiple of the normalized length of every horizontal street.
We call $h^*$ the \textit{normalized horizontal street-LCM} of~$\PPP$.
Analogous to this, we denote by $v^*$ the \textit{normalized vertical street-LCM} of~$\PPP$.

To have a suitable version of the surplus shortline method on a street-rational polyrectangle translation surface~$\PPP$, we must consider slopes of the form
\begin{equation}\label{eq8.1.4}
\alpha=v^*a_0+\frac{1}{h^*a_1+\frac{1}{v^*a_2+\frac{1}{h^*a_3+\cdots}}},
\end{equation}
where $a_0,a_1,a_2,a_3,\ldots$ are positive integers.

Note that $h^*$ plays the analogous role of an integer which is the least integer multiple of the lengths of the horizontal streets in a polysquare translation surface, and $v^*$ plays the analogous role of an integer which is the least integer multiple of the lengths of the vertical streets in a polysquare translation surface.
We can view the expression \eqref{eq8.1.4} as an extension of the concept of continued fraction expansion.
For the success of the surplus shortline method in a street-rational polyrectangle translation surface, we require these \textit{digits} to be integer multiples of $h^*$ and $v^*$ as relevant.

Note that it is not necessary that $h^*$ and $v^*$ are rational multiples of each other, so there may not be a quantity that corresponds to the
street-LCM of a finite polysquare translation surface.
\end{remark}

Indeed, we have the two key ingredients needed for the success of the shortline method.
Working with a polyrectangle translation surface automatically gives rise to the horizontal and vertical directions, vital for the shortline method which is an alternating process between two directions, while \textit{street-rationality} guarantees that the concept of shortline is well defined.

Thus the eigenvalue-free version of the shortline method developed in \cite[Sections 6.2 and~6.4]{BCY} also works for street-rational polyrectangle translation surfaces, and establishes superdensity of the right triangle billiard with angle $\pi/8$ for some special slopes.

The normalized lengths of the streets, or the cotangents of the diagonals of the streets, of the regular octagon surface represented as a street-rational polyrectangle translation surface are $2(1+\sqrt{2})$ and $1+\sqrt{2}$, so that $h^*=v^*=2(1+\sqrt{2})$.
Thus we consider the special slopes of the form
\begin{equation}\label{eq8.1.5}
\alpha=2(1+\sqrt{2})a_0+\frac{1}{2(1+\sqrt{2})a_1+\frac{1}{2(1+\sqrt{2})a_2+\cdots}},
\end{equation}
and their reciprocals $\alpha^{-1}$, where $a_0,a_1,a_2,a_3,\ldots$ is an infinite sequence of positive integers bounded from above.

The choice \eqref{eq8.1.5} represents the analog of those badly approximable slopes for which the eigenvalue-free version of the shortline method works in the case of polysquare translation surfaces.

Figure~8.1.4 illustrates some almost horizontal detour crossings and their almost vertical shortcuts.

\begin{displaymath}
\begin{array}{c}
\includegraphics[scale=0.8]{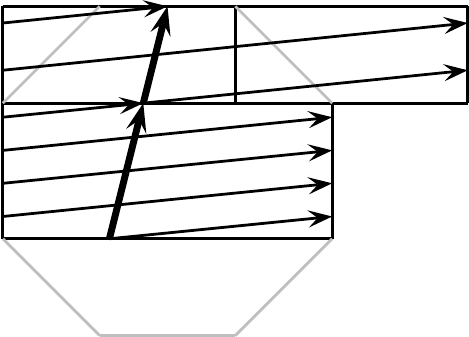}
\vspace{3pt}\\
\mbox{Figure 8.1.4: almost horizontal detour crossings}
\\
\mbox{and their almost vertical shortcuts}
\end{array}
\end{displaymath}

By a straightforward adaptation of the proof of \cite[Theorems 6.1.1 and~6.4.1]{BCY} to the street-rational polyrectangle translation surface in Figures 8.1.2--8.1.4 that represents the regular octagon surface, we obtain the following result concerning the right triangle billiard with angle~$\pi/8$.

\begin{thm}\label{thm8.1.1}
\emph{(i)}
Consider the right triangle with angle~$\pi/8$.
Let $\alpha>1$ be a real number of the form \eqref{eq8.1.5}. 
Then any half-infinite billiard orbit in the right triangle with angle~$\pi/8$ with initial slope $\alpha$ exhibits superdensity.
\end{thm}

Figures 8.1.2--8.1.4 are based on the parallel decomposition of the regular octagon into $3$ parts as shown in the picture on the left in Figure~8.1.5.
The picture on the right shows a different decomposition into $4$ parts.

\begin{displaymath}
\begin{array}{c}
\includegraphics[scale=0.8]{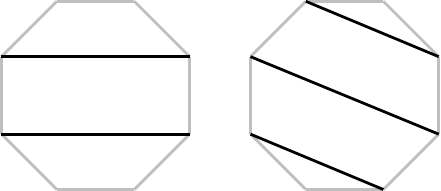}
\vspace{3pt}\\
\mbox{Figure 8.1.5: parallel decompositions of the regular octagon}
\end{array}
\end{displaymath}

Figure~8.1.6 illustrates the representation of the regular octagon surface as a street-rational polyrectangle translation surface, based on this second decomposition.
The $4$ rectangles are similar.
Since $A_3A_4$ and $A_4A_5$ are two chords of the same length, the angles $A_3A_2A_4$ and $A_4A_1A_5$ are equal.

\begin{displaymath}
\begin{array}{c}
\includegraphics[scale=0.8]{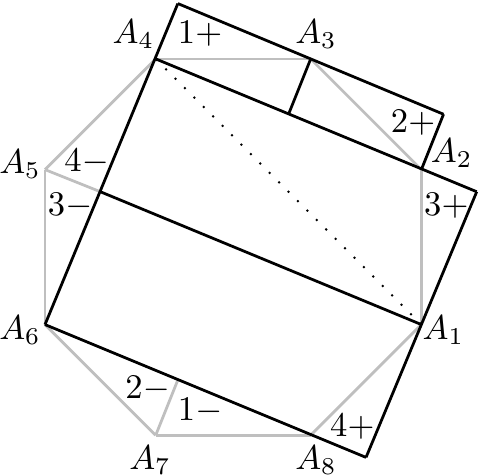}
\vspace{3pt}\\
\mbox{Figure 8.1.6: another street decomposition into similar rectangles}
\end{array}
\end{displaymath}

It is not difficult to see that this polyrectangle translation surface, suitably rotated, has $2$ horizontal streets and $2$ vertical streets.
For instance, in one of these two directions, one street comprises the two smaller rectangles, and the other street comprises the two larger rectangles.
Furthermore, it can be checked that the normalized lengths of the streets are equal to $2(1+\sqrt{2})$, so we can consider slopes of the form \eqref{eq8.1.5} for the suitably rotated copy of this polyrectangle translation surface.
This leads to \textit{new} explicit slopes such that any half-infinite orbit with such an initial slope in the right triangle billiard with angle $\pi/8$ exhibits superdensity.

It can be shown that the same method works for the right triangle billiard with angle~$\pi/k$, where $k\ge8$ is an even integer.
Again we use the elementary fact that the vertices of a regular $k$-gon are all on a circle, and street-rationality comes from the chord-angle relation on this circle.
We leave the details to the reader.

\begin{theorem811c}
\emph{(ii)}
Let $k\ge8$ be an even integer, and consider the right-triangle with angle~$\pi/k$.
There exist infinitely many slopes, depending on~$k$, such that any half-infinite billiard orbit in the right triangle with angle $\pi/k$ with such an initial slope exhibits superdensity.

\emph{(iii)}
In general, consider an arbitrary finite street-rational polyrectangle translation surface~$\PPP$.
There exist infinitely many slopes, depending on~$\PPP$, such that any half-infinite $1$-direction geodesic on $\PPP$ having such a slope exhibits superdensity.

\emph{(iv)}
As in Theorem~\ref{thm7.1.1}, for many of these slopes we can explicitly compute the irregularity exponent.
Combining the irregularity exponent with the method of zigzagging introduced in \cite[Section~3.3]{BDY1}, we can also describe, for a geodesic flow on $\PPP$ with such a slope, the time-quantitative behavior of the edge cutting and face crossing numbers, as well as equidistribution relative to all convex sets.
\end{theorem811c}

In the course on the next few sections, we shall emphasize on parts (iii) and (iv) by giving a few examples; see Theorems \ref{thm8.3.1} and Theorems
\ref{thm8.4.1}--\ref{thm8.4.3}.

\begin{remark}
This paper is not about ergodic theory.
Instead, we focus on the time-quantitative evolution of individual orbits instead of geodesic flow as a whole.
Nevertheless, it is very interesting to point out a substantial difference between geodesic flow on a polysquare translation surface and geodesic flow on a regular $k$-gon surface with even $k\ge8$.
From the viewpoint of ergodic theory these two flows are very different.
The latter is \textit{weakly mixing} in almost every direction, while the former is \textit{not} weakly mixing in any direction; see~\cite{AD}.

Our main point is that, despite this big difference between the two flows, the shortline method works \textit{equally well} for individual orbits in either case.

However, whereas for a polysquare translation surface, a geodesic \textit{modulo one} is equivalent to a torus line in the unit square, there is no corresponding analog for the street-rational polyrectangle translation surface corresponding to the regular octagon surface.
\end{remark}

In the next section, we illustrate the proof of part (iv) in a special case that reflects the whole difficulty of the general case.

%%%%%%%%%%
%
% SECTION 8.2
%
%%%%%%%%%%

\subsection{Computing the irregularity exponent for the regular octagon surface}\label{sec8.2}

We now apply the eigenvalue-based version of the shortline method on the regular octagon surface, and compute the irregularity exponent of $1$-direction geodesics with certain slopes.
In particular, we consider special slopes of the form
\begin{equation}\label{eq8.2.1}
\alpha=2(1+\sqrt{2})a_0+\frac{1}{2(1+\sqrt{2})a_1+\frac{1}{2(1+\sqrt{2})a_2+\cdots}},
\end{equation}
where $a_0,a_1,a_2,\ldots$ form a sequence of positive integers which is eventually periodic.
This last requirement distinguishes \eqref{eq8.2.1} from \eqref{eq8.1.5}.

Since the common ratio of the vertical and horizontal sides of the rectangles in Figure~8.1.4 is $1+\sqrt{2}$, the shortline of a geodesic on the regular octagon surface with slope $\alpha$ given by \eqref{eq8.2.1} has slope~$\alpha_1^{-1}$, where
\begin{displaymath}
\alpha_1=2(1+\sqrt{2})a_1+\frac{1}{2(1+\sqrt{2})a_2+\frac{1}{2(1+\sqrt{2})a_3+\cdots}},
\end{displaymath}
and the shortline of this shortline has slope
\begin{displaymath}
\alpha_2=2(1+\sqrt{2})a_2+\frac{1}{2(1+\sqrt{2})a_3+\frac{1}{2(1+\sqrt{2})a_4+\cdots}},
\end{displaymath}
and so on.
Note that for $\alpha,\alpha_1,\alpha_2,\ldots,$ the usual shift of digits $a_0,a_1,a_2,\ldots$ applies.

\begin{thm}\label{thm8.2.1}
Consider $1$-direction geodesic flow on the regular octagon surface with slope $\alpha$ given by \eqref{eq8.2.1}.
If the sequence $a_0,a_1,a_2,\ldots$ has period $k_1,k_2,\ldots,k_r$ eventually, then the irregularity exponent of the geodesic with slope $\alpha$ is equal to
\begin{displaymath}
\frac{\log\vert\lambda\vert}{\log\vert\Lambda\vert},
\end{displaymath}
where $\lambda$ is the eigenvalue with the larger absolute value of the product matrix
\begin{displaymath}
\prod_{i=1}^r
\begin{pmatrix}
-2(\sqrt{2}-1)k_i&1\\
1&0
\end{pmatrix},
\end{displaymath}
and $\Lambda$ is the eigenvalue with the larger absolute value of the product matrix
\begin{displaymath}
\prod_{i=1}^r
\begin{pmatrix}
2(1+\sqrt{2})k_i&1\\
1&0
\end{pmatrix}.
\end{displaymath}
Alternatively, we have the product formula
\begin{displaymath}
\Lambda=\prod_{i=1}^r\beta_i,
\end{displaymath}
where for every $i=1,\ldots,r$,
\begin{displaymath}
\beta_i=2(1+\sqrt{2})k_i+\frac{1}{2(1+\sqrt{2})k_{i+1}+\frac{1}{2(1+\sqrt{2})k_{i+2}+\cdots}},
\end{displaymath}
corresponding to the periodic sequence
\begin{displaymath}
k_i,k_{i+1},k_{i+2},\ldots,k_r,k_1,k_2,\ldots,k_r,\ldots.
\end{displaymath}
\end{thm}

\begin{remark}
If the sequence $a_0,a_1,a_2,\ldots$ has period $k_1,k_2$ eventually, of length~$2$, then it is possible to use the street-spreading matrix in the same spirit as in Theorem~\ref{thm7.2.2}, noting that the method there works even for street-rational polyrectangle translation surfaces.
But here the period can be arbitrarily long.
\end{remark}

\begin{proof}[Proof of Theorem~\ref{thm8.2.1}]
We apply an adaptation of the original eigenvalue-based version of the surplus shortline method developed in \cite[Section~3]{BDY1} and \cite[Section~4]{BDY2}.

First of all, we note that the regular octagon region, viewed as a street-rational polyrectangle translation surface as in Figures 8.1.2--8.1.4, has $7$ horizontal and $7$ vertical edges, some with identified pairs, as shown in Figure~8.2.1.

\begin{displaymath}
\begin{array}{c}
\includegraphics[scale=0.8]{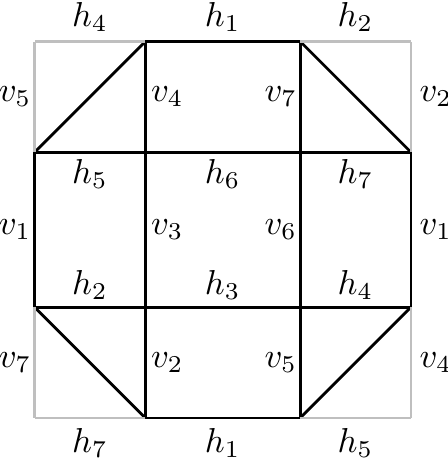}
\vspace{3pt}\\
\mbox{Figure 8.2.1: horizontal and vertical edges of the regular octagon surface}
\\
\mbox{viewed as a street-rational polyrectangle translation surface}
\end{array}
\end{displaymath}

We distinguish the $14$ types of almost vertical units
\begin{displaymath}
h_1h_3,h_1h_4,
h_2h_5,h_2h_6,
h_3h_6,h_3h_7,
h_4h_5,
h_4h_7,
h_5h_1,h_5h_4,
h_6h_1,h_6h_2,
h_7h_2,h_7h_3,
\end{displaymath}
as shown in Figure~8.2.2.

\begin{displaymath}
\begin{array}{c}
\includegraphics[scale=0.8]{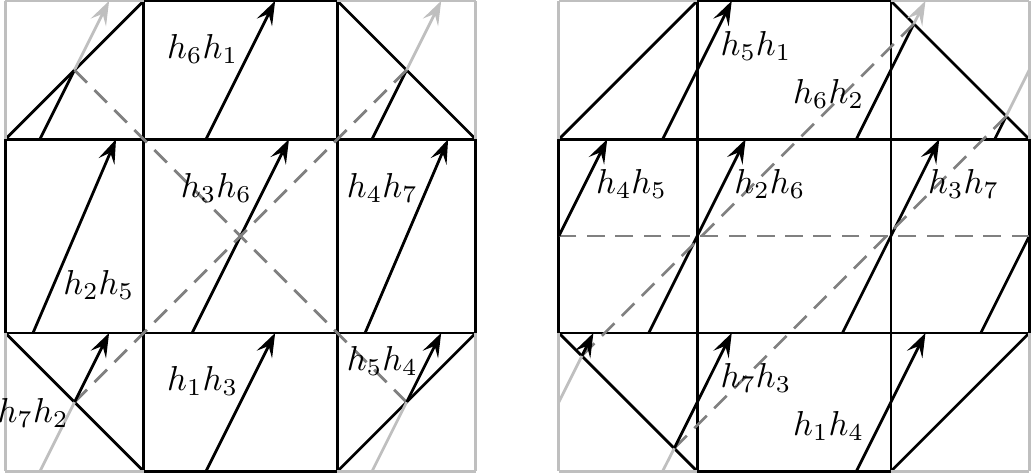}
\vspace{3pt}\\
\mbox{Figure 8.2.2: almost vertical units of the regular octagon surface}
\end{array}
\end{displaymath}

We distinguish the $14$ types of almost horizontal units
\begin{displaymath}
v_1v_3,v_1v_4,
v_2v_5,v_2v_6,
v_3v_6,v_3v_7,
v_4v_5,
v_4v_7,
v_5v_1,v_5v_4,
v_6v_1,v_6v_2,
v_7v_2,v_7v_3,
\end{displaymath}
as shown in Figure~8.2.3.

\begin{displaymath}
\begin{array}{c}
\includegraphics[scale=0.8]{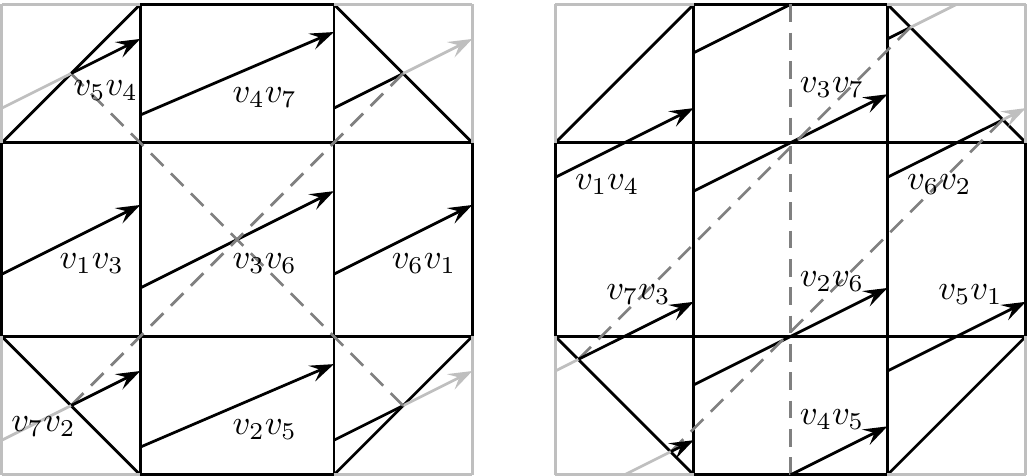}
\vspace{3pt}\\
\mbox{Figure 8.2.3: almost horizontal units of the regular octagon surface}
\end{array}
\end{displaymath}

Consider first the almost vertical unit~$h_1h_3$.
It is not difficult to see from Figures 8.2.2 and~8.2.3 that $h_1h_3$ is the shortcut of an almost horizontal detour crossing of a horizontal street, made up of a fractional almost horizontal unit $v_4v_5$, full almost horizontal units $v_5v_4,v_4v_7,v_7v_2$, then $(k-1)$ copies of full almost horizontal units $v_2v_5,v_5v_4,v_4v_7,v_7v_2$, and finally a fractional almost horizontal unit~$v_2v_6$.
This can be summarized by
\begin{displaymath}
h_1h_3\to
v_4v_5,v_5v_4,v_4v_7,v_7v_2,(v_2v_5,v_5v_4,v_4v_7,v_7v_2)^{k-1},v_2v_6,
\end{displaymath}
where $k\ge1$ is the branching parameter.
For bookkeeping purposes, we next apply the delete end rule, keep the initial fractional almost horizontal unit $v_4v_5$ as a full unit, discard the final fractional almost horizontal unit~$v_2v_6$, and call these units the ancestor units of~$h_1h_3$.

Repeating the same exercise on the other $13$ almost vertical units, we summarize this ancestor process by
\begin{align}
h_1h_3&\rightharpoonup
v_4v_5,
k\times v_5v_4,
k\times v_4v_7,
k\times v_7v_2,
(k-1)\times v_2v_5,
\label{eq8.2.2}
\\
h_1h_4&\rightharpoonup
v_4v_5,
k\times v_5v_4,
k\times v_4v_7,
k\times v_7v_2,
k\times v_2v_5,
\label{eq8.2.3}
\\
h_2h_5&\rightharpoonup
v_7v_3,
2k\times v_3v_6,
2k\times v_6v_1,
(2k-1)\times v_1v_3,
\label{eq8.2.4}
\\
h_2h_6&\rightharpoonup
v_7v_3,
2k\times v_3v_6,
2k\times v_6v_1,
2k\times v_1v_3,
\label{eq8.2.5}
\\
h_3h_6&\rightharpoonup
v_2v_6,
2k\times v_6v_1,
2k\times v_1v_3,
(2k-1)\times v_3v_6,
\label{eq8.2.6}
\\
h_3h_7&\rightharpoonup
v_2v_6,
2k\times v_6v_1,
2k\times v_1v_3,
2k\times v_3v_6,
\label{eq8.2.7}
\\
h_4h_5&\rightharpoonup
v_5v_1,
2k\times v_1v_3,
2k\times v_3v_6,
2k\times v_6v_1,
\label{eq8.2.8}
\\
h_4h_7&\rightharpoonup
v_5v_1,
2k\times v_1v_3,
2k\times v_3v_6,
(2k-1)\times v_6v_1,
\label{eq8.2.9}
\\
h_5h_1&\rightharpoonup
v_1v_4,
k\times v_4v_7,
k\times v_7v_2,
k\times v_2v_5,
k\times v_5v_4,
\label{eq8.2.10}
\\
h_5h_4&\rightharpoonup
v_1v_4,
k\times v_4v_7,
k\times v_7v_2,
k\times v_2v_5,
(k-1)\times v_5v_4,
\label{eq8.2.11}
\\
h_6h_1&\rightharpoonup
v_3v_7,
k\times v_7v_2,
k\times v_2v_5,
k\times v_5v_4,
(k-1)\times v_4v_7,
\label{eq8.2.12}
\\
h_6h_2&\rightharpoonup
v_3v_7,
k\times v_7v_2,
k\times v_2v_5,
k\times v_5v_4,
k\times v_4v_7,
\label{eq8.2.13}
\\
h_7h_2&\rightharpoonup
v_6v_2,
k\times v_2v_5,
k\times v_5v_4,
k\times v_4v_7,
(k-1)\times v_7v_2,
\label{eq8.2.14}
\\
h_7h_3&\rightharpoonup
v_6v_2,
k\times v_2v_5,
k\times v_5v_4,
k\times v_4v_7,
k\times v_7v_2.
\label{eq8.2.15}
\end{align}
These lead to a $14\times14$ transition matrix
\begin{displaymath}
M(k)=\begin{pmatrix}
M_{1,1}(k)&M_{1,2}(k)\\
M_{2,1}(k)&M_{2,2}(k)
\end{pmatrix},
\end{displaymath}
where any particular row captures the information in the ancestor relation of the almost vertical unit in question by displaying the multiplicities of each of its ancestor almost horizontal units.

We have
\begin{displaymath}
M_{1,1}(k)=\bordermatrix{
&v_1v_3&v_1v_4&v_2v_5&v_2v_6&v_3v_6&v_3v_7&v_4v_5\cr
h_1h_3&0&0&k-1&0&0&0&1\cr 
h_1h_4&0&0&k&0&0&0&1\cr 
h_2h_5&2k-1&0&0&0&2k&0&0\cr
h_2h_6&2k&0&0&0&2k&0&0\cr
h_3h_6&2k&0&0&1&2k-1&0&0\cr
h_3h_7&2k&0&0&1&2k&0&0\cr
h_4h_5&2k&0&0&0&2k&0&0},
\end{displaymath}
\begin{displaymath}
M_{1,2}(k)=\bordermatrix{
&v_4v_7&v_5v_1&v_5v_4&v_6v_1&v_6v_2&v_7v_2&v_7v_3\cr
h_1h_3&k&0&k&0&0&k&0\cr 
h_1h_4&k&0&k&0&0&k&0\cr 
h_2h_5&0&0&0&2k&0&0&1\cr
h_2h_6&0&0&0&2k&0&0&1\cr
h_3h_6&0&0&0&2k&0&0&0\cr
h_3h_7&0&0&0&2k&0&0&0\cr
h_4h_5&0&1&0&2k&0&0&0},
\end{displaymath}
\begin{displaymath}
M_{2,1}(k)=\bordermatrix{
&v_1v_3&v_1v_4&v_2v_5&v_2v_6&v_3v_6&v_3v_7&v_4v_5\cr
h_4h_7&2k&0&0&0&2k&0&0\cr 
h_5h_1&0&1&k&0&0&0&0\cr 
h_5h_4&0&1&k&0&0&0&0\cr 
h_6h_1&0&0&k&0&0&1&0\cr 
h_6h_2&0&0&k&0&0&1&0\cr 
h_7h_2&0&0&k&0&0&0&0\cr 
h_7h_3&0&0&k&0&0&0&0},
\end{displaymath}
and
\begin{displaymath}
M_{2,2}(k)=\bordermatrix{
&v_4v_7&v_5v_1&v_5v_4&v_6v_1&v_6v_2&v_7v_2&v_7v_3\cr
h_4h_7&0&1&0&2k-1&0&0&0\cr 
h_5h_1&k&0&k&0&0&k&0\cr 
h_5h_4&k&0&k-1&0&0&k&0\cr 
h_6h_1&k-1&0&k&0&0&k&0\cr 
h_6h_2&k&0&k&0&0&k&0\cr 
h_7h_2&k&0&k&0&1&k-1&0\cr 
h_7h_3&k&0&k&0&1&k&0}.
\end{displaymath}

Similarly, we can study the ancestor relation of each of the almost horizontal units, again using the delete end rule, and obtain the analogs of
\eqref{eq8.2.2}--\eqref{eq8.2.15}.
These will lead to another $14\times14$ transition matrix.
Since we have listed the almost vertical units and almost horizontal units in lexicographical order, these two $14\times14$ transition matrices are the same.

Of the $14$ eigenvalues of~$M(k)$, there are $10$ irrelevant ones of the form $\pm1$, $\pm\ii$ and $(-1\pm\sqrt{3}\ii)/2$.
The more interesting ones are
\begin{displaymath}
(1+\sqrt{2})k\pm\left((1+\sqrt{2})^2k^2+1\right)^{1/2},
\quad
(1-\sqrt{2})k\pm\left((1-\sqrt{2})^2k^2+1\right)^{1/2}.
\end{displaymath}
Clearly the eigenvalue with the largest absolute value is
\begin{equation}\label{eq8.2.16}
\Lambda=(1+\sqrt{2})k+\left((1+\sqrt{2})^2k^2+1\right)^{1/2},
\end{equation}
and the eigenvalue with the second largest absolute value is
\begin{equation}\label{eq8.2.17}
\lambda=-(\sqrt{2}-1)k-\left((\sqrt{2}-1)^2k^2+1\right)^{1/2}.
\end{equation}
The other two eigenvalues among these four are also irrelevant.

We shall show that the transition matrix $M(k)$ has a conjugate with the form
\begin{equation}\label{eq8.2.18}
P^{-1}M(k)P=\begin{pmatrix}
T&?\\
0&A(k)
\end{pmatrix},
\end{equation}
where $T$ is a $10\times10$ triangular matrix with main diagonal entries
\begin{equation}\label{eq8.2.19}
1,-1,-1,-1,-\ii,\ii,-1,1,\frac{-1-\sqrt{3}\ii}{2},\frac{-1+\sqrt{3}\ii}{2},
\end{equation}
in this order, the entries of $T$ below the main diagonal are all zero, and
\begin{equation}\label{eq8.2.20}
A(k)=\begin{pmatrix}
2(1+\sqrt{2})k&1&?&?\\
1&0&?&?\\
0&0&-2(\sqrt{2}-1)k&1\\
0&0&1&0
\end{pmatrix}.
\end{equation}
The description of the matrix $M(k)$ by \eqref{eq8.2.18}--\eqref{eq8.2.20} is extremely convenient.
It reduces the necessary eigenvalue computation of arbitrary products 
\begin{displaymath}
\prod_{i=1}^rM(k_i)
\end{displaymath}
of $14\times14$ matrices with different values of the branching parameter $k_i$ to the much simpler eigenvalue computation of products   
\begin{displaymath}
\prod_{i=1}^r
\begin{pmatrix}
2(1+\sqrt{2})k_i&1\\
1&0
\end{pmatrix}
\quad\mbox{and}\quad
\prod_{i=1}^r
\begin{pmatrix}
-2(\sqrt{2}-1)k_i&1\\
1&0
\end{pmatrix}
\end{displaymath}
of $2\times2$ matrices, as the remaining eigenvalues $\pm1$, $\pm\ii$ and $(-1\pm\sqrt{3}\ii)/2$ are irrelevant.

We now outline the routine deduction of \eqref{eq8.2.18}--\eqref{eq8.2.20}.

We first make use of the fact that $M(k)$ has $6$ eigenvectors that are independent of the branching parameter~$k$.
Together with the eigenvalues, they are
\begin{align}
\lambda_1=1,
&\quad
\bfv_1=(-1,0,1,0,0,0,0,-1,0,0,1,0,0,0)^T,
\nonumber
\\
\lambda_2=-1,
&\quad
\bfv_2=(-1,0,-1,0,1,0,0,0,0,1,0,0,0,0)^T,
\nonumber
\\
\lambda_3=-1,
&\quad
\bfv_3=(-1,0,-1,0,0,0,0,1,0,0,1,0,0,0)^T,
\nonumber
\\
\lambda_4=-1,
&\quad
\bfv_4=(-1,0,-1,0,1,0,0,0,0,0,0,0,1,0)^T,
\nonumber
\\
\lambda_5=-\ii,
&\quad
\bfv_5=(1,0,\ii,2\ii,-1+\ii,-2,0,-1,0,0,-\ii,-2\ii,1-\ii,2)^T,
\nonumber
\\
\lambda_6=\ii,
&\quad
\bfv_6=(1,0,-\ii,-2\ii,-1-\ii,-2,0,-1,0,0,\ii,2\ii,1+\ii,2)^T.
\nonumber
\end{align}
This observation allows us to apply a \textit{partial diagonalization} trick first discussed in \cite[Lemma~4.1.1]{BDY2}.
Let $Q$ be a $14\times14$ invertible matrix such that the first $6$ columns are $\bfv_1,\ldots,\bfv_6$.
Then for every $1\le i\le6$, the $i$-th column of the conjugate $Q^{-1}M(k)Q$ has the special form that its $i$-th element is~$\lambda_i$, and the remaining elements are all zero.
A concrete choice of $Q$ gives
\begin{displaymath}
Q^{-1}M(k)Q=\begin{pmatrix}
D_6&?\\
0&M_8(k)
\end{pmatrix},
\end{displaymath}
where $D_6$ is a $6\times6$ diagonal matrix with diagonal entries $\lambda_1,\ldots\lambda_6$, and
\begin{displaymath}
M_8(k)=\begin{pmatrix}
-\frac{1}{4}\!-\!\frac{\sqrt{3}\ii}{4}
&\frac{1}{4}\!-\!\frac{\sqrt{3}\ii}{4}
&-\frac{\sqrt{3}\ii}{12}
&\frac{1}{4}\!-\!\frac{\sqrt{3}k\ii}{6}
&\frac{1}{4}\!-\!\frac{\sqrt{3}k\ii}{6}
&0
&0
&0
\vspace{4pt}\\
\frac{1}{4}\!+\!\frac{\sqrt{3}\ii}{4}
&-\frac{1}{4}\!+\!\frac{\sqrt{3}\ii}{4}
&\frac{\sqrt{3}\ii}{12}
&\frac{1}{4}\!+\!\frac{\sqrt{3}k\ii}{6}
&\frac{1}{4}\!+\!\frac{\sqrt{3}k\ii}{6}
&0
&0
&0
\vspace{4pt}\\
-2k\!-\!\frac{3+3\sqrt{3}\ii}{2}
&-2k\!-\!\frac{3-3\sqrt{3}\ii}{2}
&\frac{1}{2}
&2k\!-\!\frac{3}{2}
&k\!-\!\frac{5}{2}
&0
&k
&0
\vspace{4pt}\\
-6k\!-\!\frac{7+3\sqrt{3}\ii}{2}
&-6k\!-\!\frac{7-3\sqrt{3}\ii}{2}
&\frac{1}{2}
&6k\!+\!\frac{3}{2}
&3k\!+\!\frac{3}{2}
&0
&3k
&0
\vspace{4pt}
\\
2k\!-\!\frac{3+3\sqrt{3}\ii}{2}
&2k\!-\!\frac{3-3\sqrt{3}\ii}{2}
&-\frac{1}{2}
&-2k\!-\!\frac{3}{2}
&-k\!-\!\frac{3}{2}
&0
&-k
&0
\vspace{4pt}\\
-2k
&-2k
&0
&3k
&2k
&0
&k
&1
\vspace{2pt}\\
2k\!-\!\frac{7+3\sqrt{3}\ii}{2}
&2k\!-\!\frac{7-3\sqrt{3}\ii}{2}
&\frac{1}{2}
&-\frac{5}{2}
&k\!-\!\frac{5}{2}
&1
&-k\!-\!1
&1
\vspace{4pt}\\
-2k
&-2k
&0
&3k
&2k
&1
&k
&0
\end{pmatrix}.
\end{displaymath}

We next make use of the fact that $M_8(k)$ has $4$ eigenvectors that are independent of the branching parameter~$k$.
Together with the eigenvalues, they are
\begin{align}
\tau_1=-1,
&\quad
\bfw_1=(0,0,0,0,0,-1,0,1)^T,
\nonumber
\\
\tau_2=1,
&\quad
\bfw_2=\left(\frac{\sqrt{3}\ii}{18},-\frac{\sqrt{3}\ii}{18},-1,-1,1,1,1,1\right)^T,
\nonumber
\\
\tau_3=-\frac{1+\sqrt{3}\ii}{2},
&\quad
\bfw_3=\left(\frac{1}{6},-\frac{1}{6},1,-1,1,0,1,0\right)^T,
\nonumber
\\
\tau_4=-\frac{1-\sqrt{3}\ii}{2},
&\quad
\bfw_4=\left(-\frac{1}{6},\frac{1}{6},1,-1,1,0,1,0\right)^T.
\nonumber
\end{align}

We apply the same partial diagonalization trick.
Let $R$ be an $8\times8$ invertible matrix such that the first $4$ columns are $\bfw_1,\ldots,\bfw_4$.
Then for every $1\le i\le4$, the $i$-th column of the conjugate $R^{-1}M_8(k)R$ has the special form that its $i$-th element is $\tau_i$, and the remaining elements are all zero.
A concrete choice of $R$ gives
\begin{displaymath}
R^{-1}M_8(k)R=\begin{pmatrix}
D_4&?\\
0&M_4(k)
\end{pmatrix},
\end{displaymath}
where $D_4$ is a $4\times4$ diagonal matrix with diagonal entries $\tau_1,\ldots,\tau_4$, and
\begin{displaymath}
M_4(k)=\begin{pmatrix}
6k\!-\!2
&2\!-\!4k
&2k
&0
\vspace{2pt}\\
6k\!-\!\frac{3}{2}
&2\!-\!4k
&2k
&0
\vspace{4pt}\\
8k\!-\!1
&1\!-\!4k
&2k\!-\!1
&1
\vspace{2pt}\\
14k\!-\!\frac{3}{2}
&1\!-\!8k
&4k
&1
\end{pmatrix}.
\end{displaymath}

Finally, one more routine conjugation turns $M_4(k)$ into $A(k)$ in \eqref{eq8.2.20}.
This completes the deduction of \eqref{eq8.2.18}--\eqref{eq8.2.20}.

One can derive Theorem~\ref{thm8.2.1} from \eqref{eq8.2.18}--\eqref{eq8.2.20} in the usual way; we leave it to the reader.
\end{proof}

To complete this section, we shall compute the eigenvalues of the $2$-step transition matrix $\bfA$ of the regular octagon surface.

With Figure~8.1.2 in mind, we view the regular octagon surface as a polyrectangle translation surface with $7$ rectangle faces, as shown in Figure~8.2.4.
The horizontal streets are $1,2,3,8$ and $4,5,6$, while the vertical streets are $1,4,3,6$ and $2,5,8$.
Note here that two distinct rectangle faces can fall into the same horizontal and the same vertical street simultaneously, for instance, faces $2,8$ or faces $4,6$, so we shall adapt our notation from that used in earlier street-spreading matrix determination.
We also show the almost vertical units of type~$\uparrow$ in the polyrectangle translation surface.
We have not shown the almost vertical units of type $\nuparrow$ here.

\begin{displaymath}
\begin{array}{c}
\includegraphics[scale=0.8]{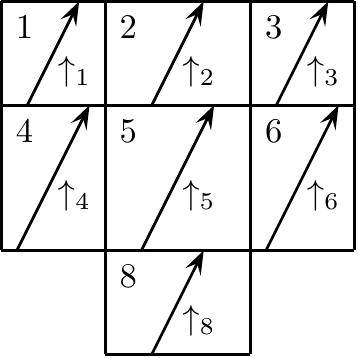}
\vspace{3pt}\\
\mbox{Figure 8.2.4: the regular octagon surfaces and almost vertical units of type $\uparrow$}
\end{array}
\end{displaymath}

Let
\begin{displaymath}
J_1=\{1,2,3,8\}
\quad\mbox{and}\quad
J_2=\{4,5,6\}
\end{displaymath}
denote the horizontal streets, and let 
\begin{displaymath}
I_1=I_3=I_4=I_6=\{1,3,4,6\}
\quad\mbox{and}\quad
I_2=I_5=I_8=\{2,5,8\}
\end{displaymath}
denote the vertical streets.

For simplicity, we consider only the special case with branching parameter $k=1$, so that we are considering a slope of the form
\begin{displaymath}
\alpha=2(1+\sqrt{2})+\frac{1}{2(1+\sqrt{2})+\frac{1}{2(1+\sqrt{2})+\cdots}}.
\end{displaymath}

Corresponding to \eqref{eq7.2.14}, we define the column matrices
\begin{equation}\label{eq8.2.21}
\bfu_1=[\{\uparrow_s:j\in J^*_1,s\in I^*_j\}]
\quad\mbox{and}\quad
\bfu_2=[\{\uparrow_s:j\in J^*_2,s\in I^*_j\}].
\end{equation}
Here $J^*_1$, $J^*_2$ and $I^*_j$ denote that the edges are counted with multiplicity.

The horizontal street $J_1$ has length $2(1+\sqrt{2})$ and width~$1$, so has normalized length $2(1+\sqrt{2})$.
Since the normalized length is precisely the value of the digit $2(1+\sqrt{2})$ of~$\alpha$, this means that an almost horizontal detour crossing of slope $\alpha^{-1}$ travels along this street essentially once, giving rise to a single count and $J^*_1=\{1,2,3,8\}$.
On the other hand, the horizontal street $J_2$ has smaller length $2+\sqrt{2}$ but greater width~$\sqrt{2}$, and so has normalized length $1+\sqrt{2}$.
Thus an almost horizontal detour crossing of slope $\alpha^{-1}$ travels along the horizontal street $J_2$ essentially twice, giving rise to a double count and $J^*_2=\{4,5,6,4,5,6\}$.

A similar argument now gives
\begin{displaymath}
I^*_1=I^*_3=I^*_4=I^*_6=\{1,3,4,6\}
\quad\mbox{and}\quad
I^*_2=I^*_5=I^*_8=\{2,5,8,2,5,8\}.
\end{displaymath}

We also define the column matrices $\bfv_1$ and $\bfv_2$ analogous to \eqref{eq7.2.15}, but their details are not important.
Also, analogous to \eqref{eq7.2.17}, we have
\begin{equation}\label{eq8.2.22}
(\bfA-I)[\{\uparrow_s\}]=\left\{\begin{array}{ll}
\bfu_1+\bfv_1,&\mbox{if $s\in J_1$},\\
\bfu_2+\bfv_2,&\mbox{if $s\in J_2$}.
\end{array}\right.
\end{equation}

We now combine \eqref{eq8.2.21} and \eqref{eq8.2.22}.
For the horizontal street corresponding to~$\bfu_1$, as highlighted in the picture on the left in Figure~8.2.5, we have
\begin{align}\label{eq8.2.23}
(\bfA-I)\bfu_1
&
=(\bfA-I)[\{2\uparrow_1,4\uparrow_2,2\uparrow_3,4\uparrow_8\}]
\nonumber
\\
&\qquad
+(\bfA-I)[\{2\uparrow_4,4\uparrow_5,2\uparrow_6\}]
\nonumber
\\
&
=12(\bfu_1+\bfv_1)+8(\bfu_2+\bfv_2).
\end{align}
For the horizontal street corresponding to~$\bfu_2$, as highlighted in the picture on the right in Figure~8.2.5, we have
\begin{align}\label{eq8.2.24}
(\bfA-I)\bfu_2
&
=(\bfA-I)[\{4\uparrow_1,4\uparrow_2,4\uparrow_3,4\uparrow_8\}]
\nonumber
\\
&\qquad
+(\bfA-I)[\{4\uparrow_4,4\uparrow_5,4\uparrow_6\}]
\nonumber
\\
&
=16(\bfu_1+\bfv_1)+12(\bfu_2+\bfv_2).
\end{align}
\begin{displaymath}
\begin{array}{c}
\includegraphics[scale=0.8]{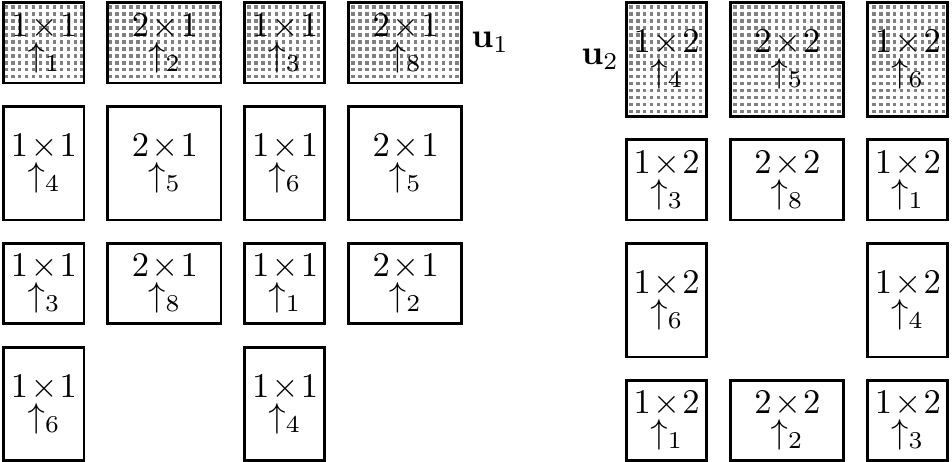}
\vspace{3pt}\\
\mbox{Figure 8.2.5: almost vertical units of type $\uparrow$ in $\bfu_1$ and $\bfu_2$}
\end{array}
\end{displaymath}

It follows from \eqref{eq8.2.23} and \eqref{eq8.2.24} that the street-spreading matrix is given by
\begin{displaymath}
\bfS=\begin{pmatrix}
12&16\\
8&12
\end{pmatrix},
\end{displaymath}
with eigenvalues $\tau_1=12+8\sqrt{2}$ and $\tau_2=12-8\sqrt{2}$.
Using \eqref{eq7.2.38}, the corresponding eigenvalues of $\bfA$ are
\begin{displaymath}
\lambda(12+8\sqrt{2};\pm)=7+4\sqrt{2}\pm2\!\left(20+14\sqrt{2}\right)^{1/2}
\end{displaymath}
and
\begin{displaymath}
\lambda(12-8\sqrt{2};\pm)=7-4\sqrt{2}\pm2\!\left(20-14\sqrt{2}\right)^{1/2}.
\end{displaymath}
The two largest eigenvalues are therefore
\begin{displaymath}
\lambda_1=7+4\sqrt{2}+2\!\left(20+14\sqrt{2}\right)^{1/2}
\quad\mbox{and}\quad
\lambda_2=7-4\sqrt{2}+2\!\left(20-14\sqrt{2}\right)^{1/2}.
\end{displaymath}
Recall \eqref{eq8.2.16} and \eqref{eq8.2.17} that, for the branching parameter $k=1$, the two largest eigenvalues for the $1$-step transition matrix are
\begin{displaymath}
\Lambda=(1+\sqrt{2})+\left((1+\sqrt{2})^2+1\right)^{1/2}
\quad\mbox{and}\quad
\lambda=-(\sqrt{2}-1)-\left((\sqrt{2}-1)^2+1\right)^{1/2}.
\end{displaymath}
Note that $\lambda_1=\Lambda^2$ and $\lambda_2=\lambda^2$.

%%%%%%%%%%
%
% SECTION 8.3
%
%%%%%%%%%%

\subsection{Regular octagon billiard}\label{sec8.3}

We now switch to billiard flow on the regular octagon.
This also leads to a street-rational polyrectangle translation surface with $1$-direction geodesic flow, but the details are more complicated than in Section~\ref{sec8.1}.

Clearly the method earlier of iterated reflection on a side, like in Figure~8.1.1, does not apply here, since we start with the regular octagon.
In this case, we need a somewhat different, but ultimately equivalent, approach based on the construction of the \textit{reflected net}.
The first step is illustrated by Figure~8.3.1 which shows a \textit{reflected double-octagon net} of the regular octagon billiard region, where the identified boundary edges are marked with the same letter.
Using Euler's formula, it is easy to see that the Euler characteristic $\chi$ of the compact surface in Figure~8.3.1 is $\chi=8-8+2=2$, so the genus $g$ is $g=1-(\chi/2)=0$.
Hence this double-octagon surface is homeomorphic to the sphere.

\begin{displaymath}
\begin{array}{c}
\includegraphics[scale=0.8]{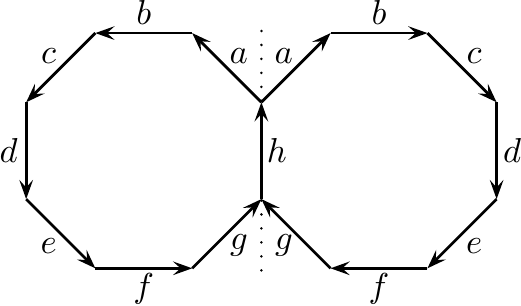}
\vspace{3pt}\\
\mbox{Figure 8.3.1: double-octagon net of the regular octagon billiard surface}
\end{array}
\end{displaymath}

Note the reflection symmetry of the boundary labelling, where the labelling on the right octagon is obtained from the labelling on the left octagon by a reflection across the vertical dotted line that we may call a \textit{mirror}.

Regular octagon billiard is a complicated flow.
While the double-octagon net in Figure~8.3.1 is quite simple, the corresponding surface still exhibits a rather complicated geodesic flow.
We illustrate how we construct this geodesic flow in Figure~8.3.2.
We start with a short billiard orbit in the left octagon, labelled~$\mathbf{1}$, that hits the boundary edge~$a$.
Following the usual billiard rule, $\mathbf{1}$ bounces back as shown by the dashed arrow.
We reflect this dashed arrow across the mirror in the middle, and obtain~$\mathbf{2}$. 
If $\mathbf{1}$ is considered an initial segment of a geodesic on the compact surface, then $\mathbf{2}$ is the continuation of the same geodesic on this compact surface.
In this way, the process can be viewed as a partial unfolding of the $8$-direction billiard flow in the left octagon into a $4$-direction flow in the
double-octagon, noting that a billiard orbit which starts horizontal or vertical or with slope $\pm1$ is not particularly interesting.
Note that the two octagons are reflections of each other.

\begin{displaymath}
\begin{array}{c}
\includegraphics[scale=0.8]{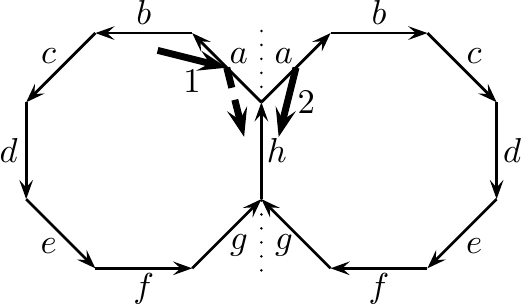}
\vspace{3pt}\\
\mbox{Figure 8.3.2: partial unfolding of the billiard in the left regular octagon}
\end{array}
\end{displaymath}

Next, we join up octagons in such a way that neighboring octagons are reflections of each other.
We end up with a ring of $8$ octagons, as shown in Figure~8.3.3.
This is sometimes known as the translation surface for regular octagon billiard.

\begin{displaymath}
\begin{array}{c}
\includegraphics[scale=0.8]{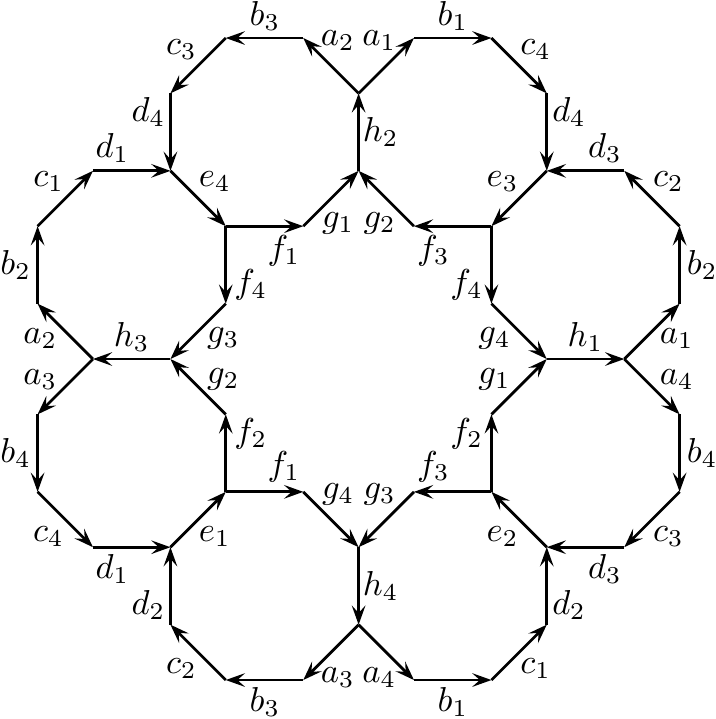}
\vspace{3pt}\\
\mbox{Figure 8.3.3: translation surface for regular octagon billiard}
\end{array}
\end{displaymath}

\begin{remark}
In general, for any even $k\ge4$, we can glue $k$ copies of a regular $k$-gon together in a perfect ring formation analogous to that in Figure~8.3.3 such that the midpoints of the common edges between neighboring $k$-gons all lie on a circle.
For odd $k\ge3$, we can do likewise with $2k$ copies of a regular $k$-gon.

The edge labellings in Figure~8.3.3 look very cumbersome at first sight.
However, we can follow a simple convention.
Start with one copy of the octagon, and label the directed edges of this octagon by $a,b,c,d,e,f,g,h$, initially without subscripts.
The adjacent octagon has reflected labelling, and the next one in the ring has labelling that is a reflection of the labelling of the second one, and so on.
When we complete this process, we have labellings $a,b,c,d,e,f,g,h$, still without subscripts, on each octagon.
The second step of the process is to look at all the edges labelled~$a$, and they occur as $4$ directed parallel pairs.
We now identify such directed parallel pairs by labelling them $a_1,a_2,a_3,a_4$.
In Figure~8.3.3, we have labelled them in increasing order of the angles they make with the positive horizontal direction.
We then repeat this step with the other edges.
Note that the edges $e$ and $h$ do not appear to come in pairs, but they actually do, but the identified edges overlap.
\end{remark}

We clearly need more, as we need a corresponding net that exhibits a $1$-direction geodesic flow and which shows the streets in a more transparent fashion.
The impatient reader may jump ahead to the net in Figures 8.3.6 and~8.3.7 from which we can easily obtain the corresponding street-rational polyrectangle translation surface, adapt the shortline method in the proof of \cite[Theorems 6.1.1 and~6.4.1]{BCY}, and exhibit explicit slopes for which we can establish superdensity.

\begin{remark}
The remarkable result of Veech~\cite{V3} on the \textit{uniform-periodic dichotomy} in flat systems, particularly that every infinite billiard orbit in a regular polygon is either periodic or exhibits uniformity.
His method is ergodic in nature, however, and does not give time-quantitative results.
It is the purpose of our work here to make some quantitative statements.
\end{remark}

The plan is quite straightforward, and the details are not too cumbersome.
We shall proceed in steps and illustrate our ideas with pictures.
Indeed, regular octagon billiard represents the whole difficulty.
Once we fully understand the special case of regular octagon billiard, it is easy to visualize the general case of regular $k$-gon billiard, where $k\ge8$ is any even integer. 

We shall use the parallel decomposition of the regular polygon shown in the picture on the right in Figure~8.1.5.
With the help of the edge labellings, we can easily work out some streets, as shown in Figure~8.3.4.
We consider tilted streets in the direction shown.

\begin{displaymath}
\begin{array}{c}
\includegraphics[scale=0.8]{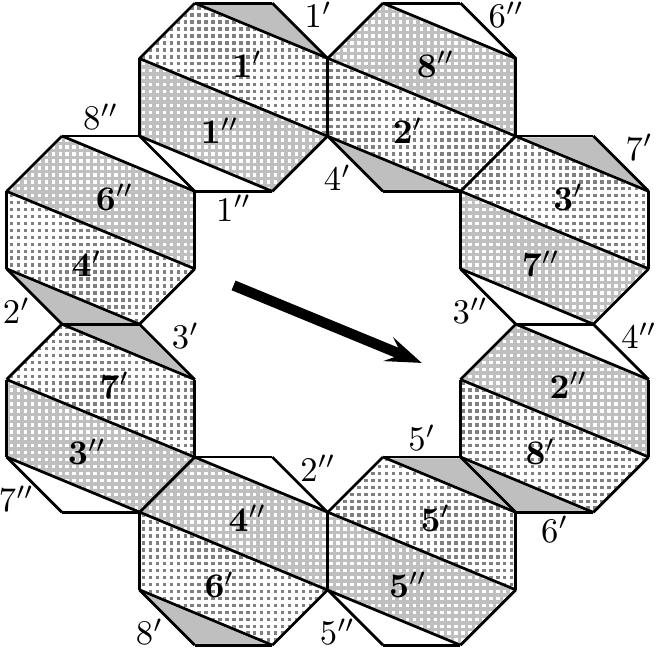}
\vspace{3pt}\\
\mbox{Figure 8.3.4: first set of tilted streets on the translation surface}
\\
\mbox{of regular octagon billiard}
\end{array}
\end{displaymath}

The first big street in this direction consists of $8$ trapezoids labelled
\begin{displaymath}
\mathbf{1'},\mathbf{2'},\mathbf{3'},\mathbf{4'},\mathbf{5'},\mathbf{6'},\mathbf{7'},\mathbf{8'}.
\end{displaymath}
Note the edge labellings in Figure~8.3.3.
The right edge of the trapezoid $\mathbf{3'}$ is~$b_2$, the same as the left edge of the trapezoid $\mathbf{4'}$.
The right edge of the trapezoid $\mathbf{4'}$ is~$g_3$, the same as the left edge of the trapezoid $\mathbf{5'}$.
The right edge of the trapezoid $\mathbf{5'}$ is~$d_2$, the same as the left edge of the trapezoid $\mathbf{6'}$.
The right edge of the trapezoid $\mathbf{6'}$ is~$a_3$, the same as the left edge of the trapezoid $\mathbf{7'}$.
The right edge of the trapezoid $\mathbf{7'}$ is~$f_2$, the same as the left edge of the trapezoid $\mathbf{8'}$.
The right edge of the trapezoid $\mathbf{8'}$ is~$c_3$, the same as the left edge of the trapezoid $\mathbf{1'}$.
This completes the street.

The second big street in this direction consists of $8$ trapezoids labelled
\begin{displaymath}
\mathbf{1''},\mathbf{2''},\mathbf{3''},\mathbf{4''},\mathbf{5''},\mathbf{6''},\mathbf{7''},\mathbf{8''}.
\end{displaymath}

There are also two small streets in this direction, consisting of $8$ triangles
\begin{displaymath}
1',2',3',4',5',6',7',8'
\quad\mbox{and}\quad
1'',2'',3'',4'',5'',6'',7'',8''
\end{displaymath}
each.

To use the shortline method, we need to consider streets in a second direction.
Figure~8.3.5 is an analog of Figure~8.3.4 in this new direction.
Again, we see that there are two big streets, each consisting of $8$ trapezoids, and two small streets, each consisting of $8$ triangles.

\begin{displaymath}
\begin{array}{c}
\includegraphics[scale=0.8]{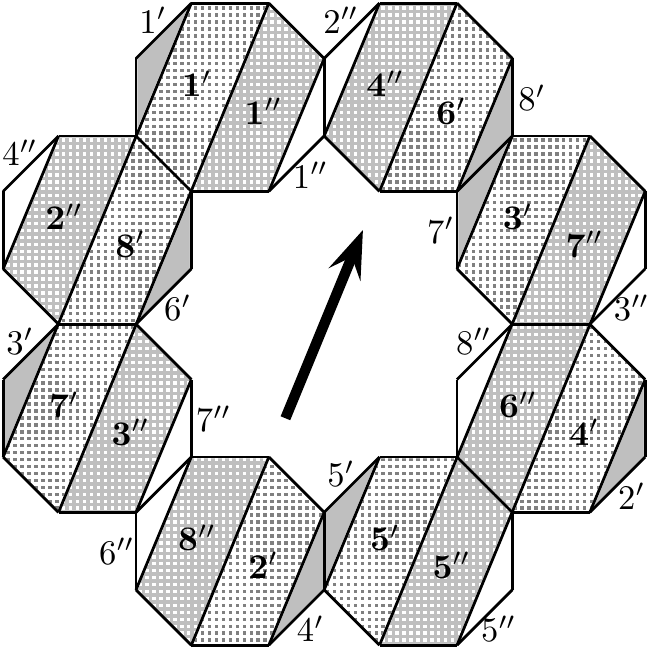}
\vspace{3pt}\\
\mbox{Figure 8.3.5: second set of tilted streets on the translation surface}
\\
\mbox{of regular octagon billiard}
\end{array}
\end{displaymath}

Having determined the streets in two perpendicular directions, we now attempt to visualize the translation surface of regular octagon billiard as a polyrectangle translation surface~$\PPP$.

To do so, we must be able to visualize the intersection of any two perpendicular streets in the translation surface of regular octagon billiard as one of the rectangle faces in~$\PPP$.
Figure~8.3.6 is a suitable modification of Figure~8.3.4.

\begin{displaymath}
\begin{array}{c}
\includegraphics[scale=0.8]{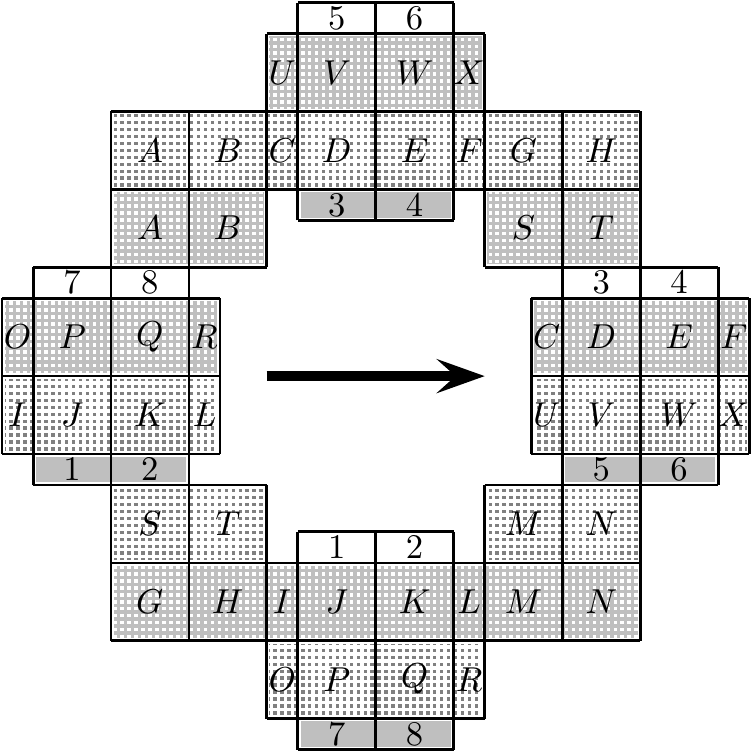}
\vspace{3pt}\\
\mbox{Figure 8.3.6: a modification of the streets in Figure 8.3.4}
\end{array}
\end{displaymath}

Note that in Figure~8.3.4, each big tilted street consists of $8$ trapezoids, each the union of $2$ squares and $2$ $(\pi/8)$-right-triangles, while each small tilted street consists of $8$ triangles, each the union of $2$ $(\pi/8)$-right-triangles.
This suggests that we attempt to visualize these $(\pi/8)$-right-triangles as pairs that combine to form $(\pi/8)$-rectangles.
Here a $(\pi/8)$-right-triangle means a right triangle with an angle equal to~$\pi/8$, whereas a $(\pi/8)$-rectangle means a rectangle that is the union of $2$ $(\pi/8)$-right-triangles.

If we compare Figures 8.3.4 and~8.3.6, we see that $A$ and $B$ are the squares in the trapezoid~$\mathbf{1'}$, while the $(\pi/8)$-rectangle $C$ is made up of the $(\pi/8)$-right-triangle on the right hand end of the trapezoid $\mathbf{1'}$ and the $(\pi/8)$-right-triangle on the left hand end of the trapezoid~$\mathbf{2'}$.
On the other hand, the $(\pi/8)$-rectangle $X$ is made up of the $(\pi/8)$-right-triangle on the right hand end of the trapezoid $\mathbf{8'}$ and the
$(\pi/8)$-right-triangle on the left hand end of the trapezoid~$\mathbf{1'}$.

Likewise, Figure~8.3.7 is a suitable modification of Figure~8.3.5.

\begin{displaymath}
\begin{array}{c}
\includegraphics[scale=0.8]{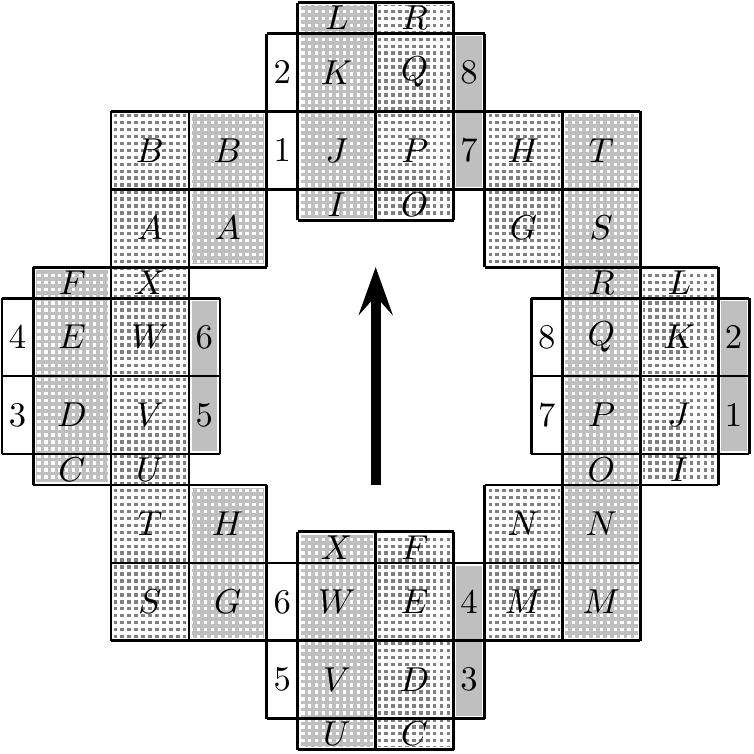}
\vspace{3pt}\\
\mbox{Figure 8.3.7: a modification of the streets in Figure 8.3.5}
\end{array}
\end{displaymath}

Figures 8.3.6 and~8.3.7 together clearly demonstrate that we can now visualize the translation surface of regular octagon billiard as a street-rational polyrectangle translation surface.

Figures 8.3.4--8.3.7 correspond to the decomposition of the regular octagon into $4$ parts as shown in the picture on the right in Figure~8.1.5.
If we repeat our argument in this section for the decomposition of the regular octagon into $3$ parts as shown in the picture on the left in Figure~8.1.5, we obtain a different $8$-octagon net which nevertheless has the same $1$-direction geodesic flow.

Adapting the shortline method proof of \cite[Theorems 6.1.1 and~6.4.1]{BCY} to any such concrete street-rational polyrectangle translation surface that arise, we then obtain the special case $k=8$ of Theorem~\ref{thm8.3.1} below.

It is easy to see that the same method works for any regular $k$-gon, where $k\ge6$ is even.
Indeed, the first step is to have an analog of Figure~8.3.1.

We simply need a regular $k$-gon with oriented boundary edges, and reflect it on a side.
This gives a $2$-copy version of the regular $k$-gon with boundary identification, which is a net of the corresponding billiard region.
To adapt the shortline method of \cite[Theorems 6.1.1 and~6.4.1]{BCY}, we need a flat surface with $1$-direction geodesic flow.
To construct such a surface we repeat the arguments illustrated by Figures 8.3.1--8.3.3.
At the end we have the translation surface of regular $k$-gon billiard with boundary identification.

We have street-rationality, courtesy of the elementary geometric fact that the vertices of a regular polygon lie on a circle, so that we can use the chord-angle relation on this circle.
The last simple step is to rearrange this flat polygonal surface to a street-rational polyrectangle translation surface by translating some corresponding parts, like converting Figures 8.3.4 and~8.3.5 to Figures 8.3.6 and~8.3.7.

To get a polyrectangle, we need two perpendicular decompositions into parallel strips like in Figure~8.1.5 for the regular octagon.
And of course it can be done for every regular polygon of $k$ sides, where $k\ge6$ is even; see Figure~8.3.8 for the special case $k=6$.
In the picture on the left, the horizontal line divides the hexagon into $2$ parts.
In the picture on the right, the $2$ vertical lines divide it into $3$ parts.

\begin{displaymath}
\begin{array}{c}
\includegraphics[scale=0.8]{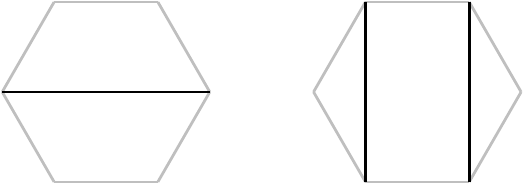}
\vspace{3pt}\\
\mbox{Figure 8.3.8: two perpendicular decompositions of the regular hexagon}
\end{array}
\end{displaymath}

Thus we obtain the following general result. 

\begin{thm}\label{thm8.3.1}
Let $k\ge6$ be an even integer, and consider billiard in a regular polygon of $k$ sides.
There exist infinitely many explicit slopes, depending on~$k$, such that any half-infinite billiard orbit having such an initial slope exhibits superdensity in the polygon.

For infinitely many of these initial slopes that give rise to superdensity, we can explicitly compute the corresponding irregularity exponent.

Combining the irregularity exponent with the method of zigzagging introduced in \cite[Section 3.3]{BDY1}, we can also describe, for billiard orbits having these initial slopes, the time-quantitative behavior of the edge cutting and face crossing numbers, as well as equidistribution relative to all convex sets.
\end{thm}

%%%%%%%%%%
%
% SECTION 8.4
%
%%%%%%%%%%

\subsection{More superdensity results}\label{sec8.4}

Next we switch to the regular polygons of $k$ sides, or $k$-gons for short, where $k\ge5$ is odd.

The first example of this class comes from right triangle billiard with angle~$\pi/5$.
The standard trick of unfolding implies that the billiard can be described in terms of a $1$-direction geodesic flow on the \textit{regular double-pentagon surface}, as shown in Figure~8.4.1.
Since we identify parallel edges, this surface has $1$-direction geodesic flow.

\begin{displaymath}
\begin{array}{c}
\includegraphics[scale=0.8]{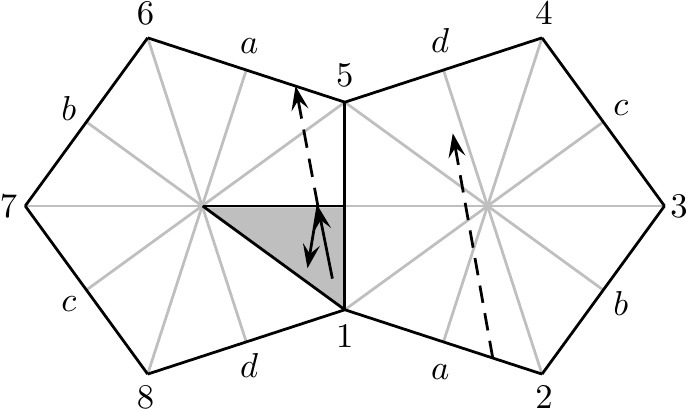}
\vspace{3pt}\\
\mbox{Figure 8.4.1: right triangle billiard with angle $\pi/5$, and geodesic flow}
\\
\mbox{on the regular double-pentagon surface via unfolding}
\end{array}
\end{displaymath}

Here the horizontal and vertical are no longer natural directions.
However, viewed in the appropriate way, $1$-direction geodesic flow on the regular double-pentagon surface can be shown to be equivalent to $1$-direction geodesic flow on a street-rational \textit{polyparallelogram} translation surface.
By moving the two triangles on the bottom left to the top right, we end up with two rhombi and two other parallelograms, as shown in the picture on the left in Figure~8.4.2.

\begin{displaymath}
\begin{array}{c}
\includegraphics[scale=0.8]{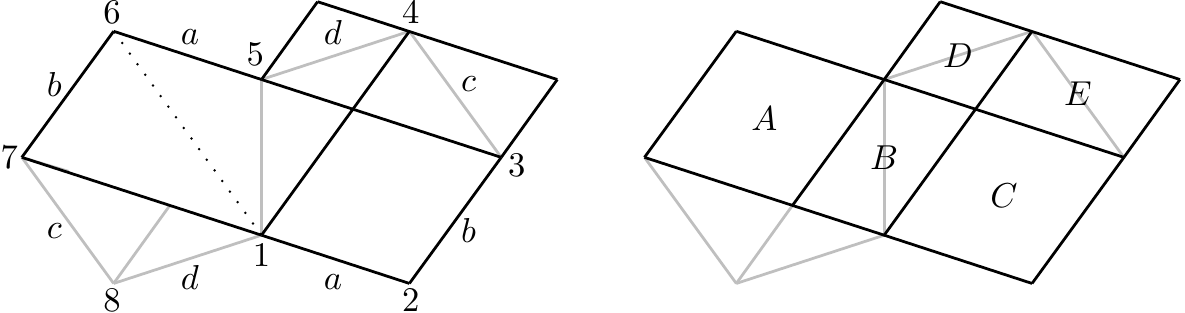}
\vspace{3pt}\\
\mbox{Figure 8.4.2: the double-pentagon surface as}
\\
\mbox{a polyparallelogram translation surface}
\end{array}
\end{displaymath}

Here street-rationality comes from the geometric fact that the vertices of a regular polygon all lie on the same circle, and so we can use the chord-angle relation on this circle.
Using this, we see that the angles $716$ and $534$ are the same.
Hence these two other parallelograms are similar.
Note from the picture on the right that we have two northwest-to-southeast streets, namely $A,B,C$ and~$D,E$.
We also have two southwest-to-northeast streets, namely $A,C,E$ and~$B,D$.
We shall return to this surface in the next section.

Similar constructions show that $1$-direction geodesic flow on every regular double-$k$-gon surface with odd $k\ge5$ is equivalent to $1$-direction geodesic flow on a street-rational polyparallelogram translation surface.

Regular pentagon billiard, on the other hand, is a complicated flow.
While the double-pentagon net in Figure~8.4.3 is quite simple, the corresponding surface still exhibits a rather complicated geodesic flow.
Note that the two pentagons are reflections of each other.

\begin{displaymath}
\begin{array}{c}
\includegraphics[scale=0.8]{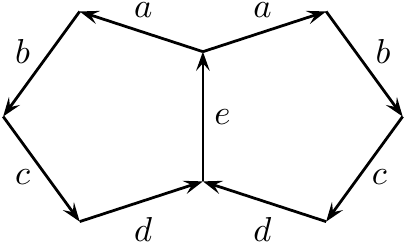}
\vspace{3pt}\\
\mbox{Figure 8.4.3: double-pentagon net of the regular pentagon billiard surface}
\end{array}
\end{displaymath}

As in the case for regular octagon billiard in Section~\ref{sec8.3}, we join up pentagons in such a way that neighboring pentagons are reflections of each other.
We end up with a ring of $10$ pentagons, as shown in Figure~8.4.4.
This is sometimes known as the translation surface of regular pentagon billiard.
Note that the edge labellings in Figure~8.4.4 are obtained by following the convention discussed in the Remark after Figure~8.3.3.

\begin{displaymath}
\begin{array}{c}
\includegraphics[scale=0.8]{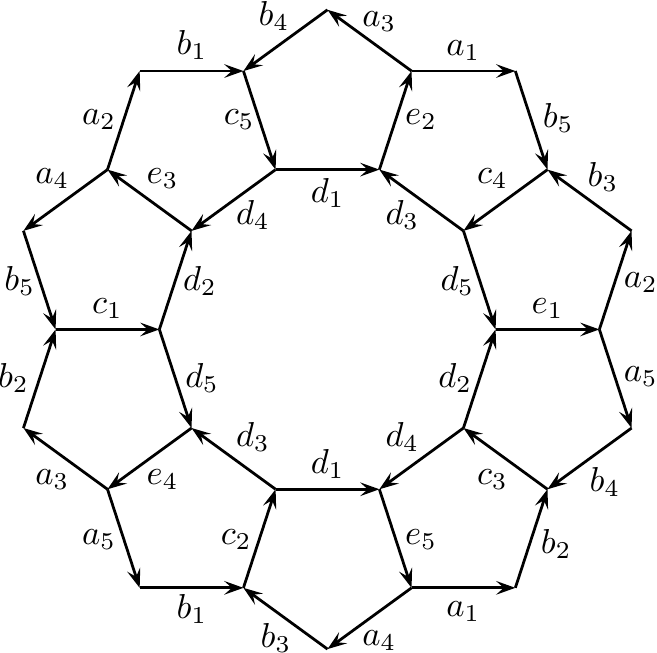}
\vspace{3pt}\\
\mbox{Figure 8.4.4: translation surface for regular pentagon billiard}
\end{array}
\end{displaymath}

We clearly need more, as we need a corresponding net that exhibits a $1$-direction geodesic flow and which shows the streets in a more transparent fashion.
The impatient reader may jump ahead to the nets in Figure~8.4.6 from which we can easily obtain the corresponding street-rational polyparallelogram translation surface, adapt the shortline method in the proof of \cite[Theorems 6.1.1 and~6.4.1]{BCY}, and exhibit explicit slopes for which we can establish superdensity.

The plan is quite straightforward, and again the details are not too cumbersome.
We shall proceed in steps and illustrate our ideas with pictures.
Indeed, regular pentagon billiard represents the whole difficulty.
Once we fully understand the special case of regular pentagon billiard, it is easy to visualize the general case of regular $k$-gon billiard, where $k\ge5$ is any odd integer. 

We shall use the decomposition of a regular pentagon into a trapezoid and a triangle.
With the help of edge labellings, we can easily work out some streets, as shown in Figure~8.4.5.
We first consider tilted streets in the direction shown in the picture on the left.

\begin{displaymath}
\begin{array}{c}
\includegraphics[scale=0.8]{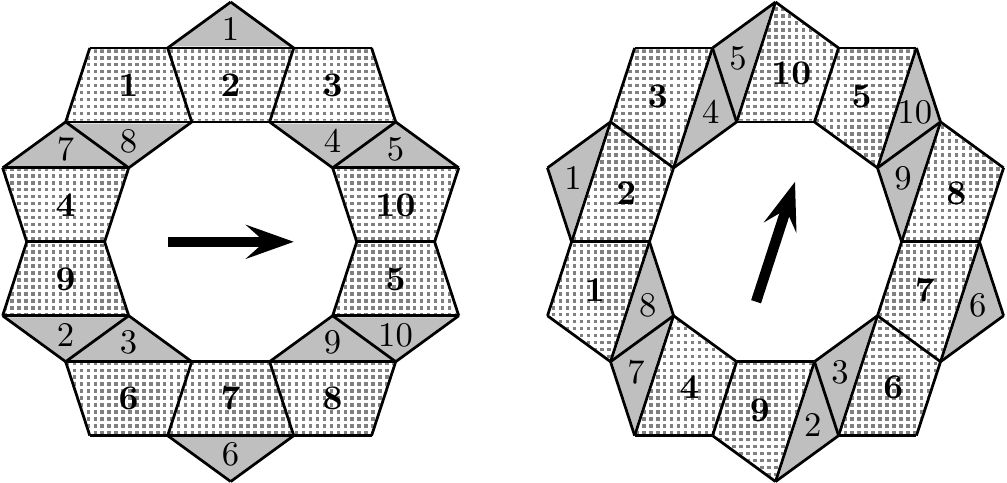}
\vspace{3pt}\\
\mbox{Figure 8.4.5: tilted streets on the translation surface}
\\
\mbox{of regular pentagon billiard}
\end{array}
\end{displaymath}

The big street in this direction consists of $10$ trapezoids labelled
\begin{displaymath}
\mathbf{1},\mathbf{2},\mathbf{3},\mathbf{4},\mathbf{5},\mathbf{6},\mathbf{7},\mathbf{8},\mathbf{9},\mathbf{10}.
\end{displaymath}

There is also a small street in this direction, consisting of $10$ triangles
\begin{displaymath}
1,2,3,4,5,6,7,8,9,10.
\end{displaymath}

To use the shortline method, we need to consider streets in a second direction, shown in the picture on the right in Figure~8.4.5.
Again, we see that there are a big street, consisting of $10$ trapezoids, and a small street, consisting of $10$ triangles.

Having determined the streets in two different directions, we now attempt to visualize the translation surface of regular pentagon billiard as a polyparallelogram translation surface~$\PPP$.

To do so, we must be able to visualize the intersection of any two streets in different directions in the translation surface of regular pentagon billiard as one of the paralleogram faces in~$\PPP$.
Figure~8.4.6 is a suitable modification of Figure~8.4.5.

\begin{displaymath}
\begin{array}{c}
\includegraphics[scale=0.8]{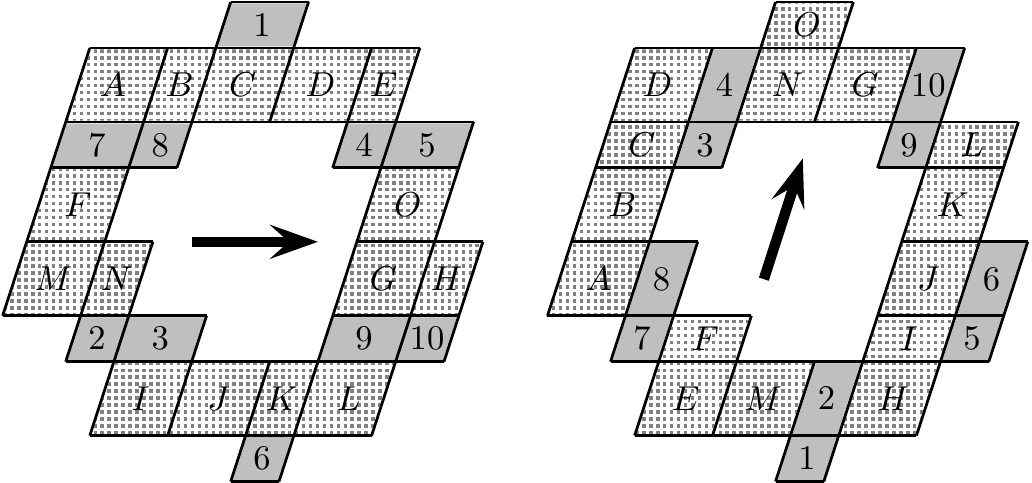}
\vspace{3pt}\\
\mbox{Figure 8.4.6: a modification of the streets in Figure 8.4.5}
\\
\mbox{of regular pentagon billiard}
\end{array}
\end{displaymath}

Note that in Figure~8.4.5, each big tilted street consists of $10$ trapezoids, each the union of a rhombus and a triangle.
If we compare the pictures on the left in Figures 8.4.5 and~8.4.6, we see that $A$ and $C$ are the rhombi in the trapezoids $\mathbf{1}$
and~$\mathbf{2}$, while the parallelogram $B$ is made up of the triangle on the right hand end of the trapezoid $\mathbf{1}$ and the triangle on the left hand end of the trapezoid~$\mathbf{2}$.
On the other hand, the parallelogram $E$ is made up of the triangle on the right hand end of the trapezoid $\mathbf{3}$ and the triangle on the left hand end of the trapezoid~$\mathbf{4}$.

Each small tilted street consists of $10$ triangles, and each can be split into unequal parts by a segment parallel to the other direction and intersecting one of its vertices.
If we compare the pictures on the left in Figures 8.4.5 and~8.4.6, we see that the parallelogram $1$ is the union of the bigger half of the triangle $1$ and the bigger half of the triangle~$2$, while the rhombus $2$ is the union of the smaller half of the triangle $2$ and the triangle~$3$.

Thus the resulting Figure~8.4.6 shows that we have a street-rational polyparallelogram translation surface.

We have already studied the superdensity of certain geodesic flow on the regular tetrahedron surface and on the cube surface.
The next member of the famous list of the five platonic solids is the regular dodecahedron which has regular pentagon faces.
Thus it is not surprising that a similar decomposition works for geodesic flow on the dodecahedron surface.
The standard net of the dodecahedron surface consists of $12$ pentagons.
To construct the corresponding polygonal surface with $1$-direction geodesic flow we have to glue together $10$ copies of the $12$-pentagon standard net.
This means $120$ copies of the pentagon, far too complicated to be explicitly included here.
For illustration, Figure~8.4.7 shows one large street on the dodecahedron surface.
We leave the complicated details of the explicit construction of the whole $120$-copy version of the pentagon to the interested reader.

\begin{displaymath}
\begin{array}{c}
\includegraphics[scale=0.8]{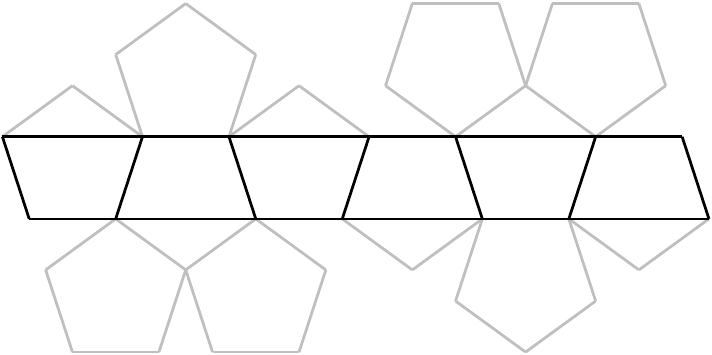}
\vspace{3pt}\\
\mbox{Figure 8.4.7: one large street on the dodecahedron surface}
\end{array}
\end{displaymath}

Note that the parity condition that $k\ge5$ is odd requires us to have decomposition into street-rational parallelograms rather than into street-rational polyrectangles for the case when $k\ge6$ is even.
This change is irrelevant, and does not prevent us from adapting the shortline method proof of \cite[Theorems 6.1.1 and~6.4.1]{BCY}.
Thus we obtain the following analog of Theorems \ref{thm8.1.1} and~\ref{thm8.3.1} for odd $k\ge5$.

\begin{thm}\label{thm8.4.1}
Let $k\ge5$ be an odd integer.

\emph{(i)}
Consider the right-triangle with angle $\pi/k$.
There exist infinitely many slopes, depending on~$k$, such that any half-infinite billiard orbit with such an initial slope exhibits superdensity in the
right-triangle with angle~$\pi/k$.

\emph{(ii)}
Consider geodesic flow on the regular double-polygon of $k$ sides.
There exist infinitely many explicit slopes, depending on~$k$, such that any half-infinite geodesic having such a slope exhibits superdensity in the double-polygon.

\emph{(iii)}
Consider billiard in a regular polygon of $k$ sides.
There exist infinitely many explicit slopes, depending on~$k$, such that any half-infinite billiard orbit having such an initial slope exhibits superdensity in the polygon.

\emph{(iv)}
Consider geodesic flow on the dodecahedron surface.
There exist infinitely many explicit slopes such that any half-infinite geodesic having such a slope exhibits superdensity on the surface.

\emph{(v)}
For infinitely many of these slopes in parts \emph{(i)}, \emph{(ii)}, \emph{(iii)} and \emph{(iv)} that give rise to superdensity, we can explicitly compute the corresponding irregularity exponent.
Combining the irregularity exponent with the method of zigzagging introduced in \cite[Section~3.3]{BDY1}, we can also describe, for trajectories having these initial slopes, the time-quantitative behavior of the edge cutting and face crossing numbers, as well as equidistribution relative to all convex sets.
\end{thm}

At this point, it is perhaps appropriate to mention the study of the remarkable uniform-periodic dichotomy in flat systems.
This study essentially originates from the work of Veech~\cite{V3}, with the result that every infinite billiard orbit in a regular polygon is either periodic or exhibits uniformity.

As we have shown earlier, surfaces arising from regular polygons of $k$ sides can be visualized as street-rational polyrectangle translation surfaces for even $k\ge6$ and as street-rational polyparallelogram translation surfaces for odd $k\ge5$.

Another flat system that exhibits the uniform-periodic dichotomy is billiard on a street-rational polyrectangle translation surface called the \textit{golden L-surface}.
As can be seen in the picture on the left in Figure~8.4.8, this surface exhibits $45$-degree reflection symmetry, and it becomes a flat surface if we carry out the boundary identification indicated in the picture on the right in Figure~8.4.8.

\begin{displaymath}
\begin{array}{c}
\includegraphics[scale=0.8]{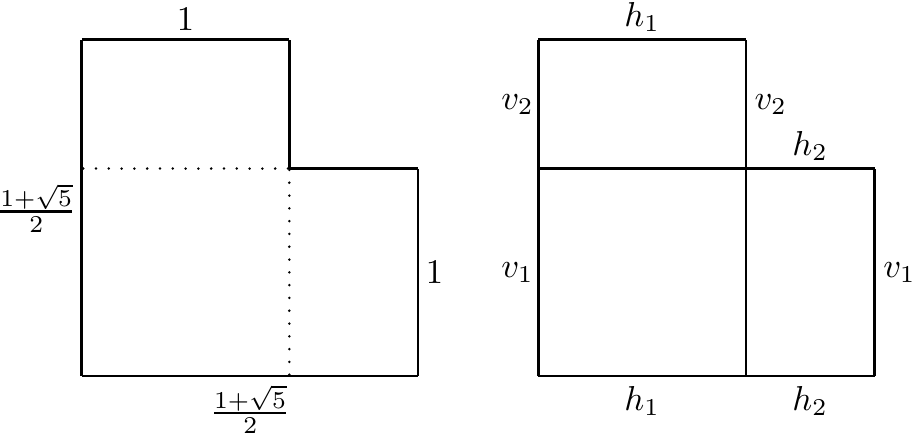}
\vspace{3pt}\\
\mbox{Figure 8.4.8: the golden L-surface}
\end{array}
\end{displaymath}

The street-rationality of the golden L-surface comes from a simple arithmetic property of the golden ratio, that
\begin{displaymath}
\frac{1+\sqrt{5}}{2}=\frac{1}{\frac{1+\sqrt{5}}{2}-1}.
\end{displaymath}

Of course the two special cases of square billiard and equilateral triangle billiard are well known classical results dating back to the 1920s. 
These classical results represent the \textit{easy} case of integrable systems.

Billiard in a regular polygon of $k$ sides, with $k\ge5$, on the other hand, is a \textit{non-integrable system}, where the vertices represent split-singularities of the orbits.

A flat system exhibiting the uniform-periodic dichotomy is called \textit{optimal}.

For a long time regular polygon billiard remains the only known infinite family of primitive optimal systems.
Here \textit{primitive} means that an elementary building block like an \textit{atom} cannot be obtained from a simpler system by some covering space construction.
We illustrate this concept on a familiar example.

Consider geodesic flow on the class of polysquare translation surfaces.
By the Gutkin--Veech theorem, each member of this class is an optimal system, and together they exhibit every high genus.
But there is only one primitive member, namely, geodesic flow on the unit torus $[0,1)^2$.
In other words, the fractional part of geodesic flow on a polysquare translation surface, or modulo one, is the torus line flow on the flat torus.
The inverse map of modulo one gives back the polysquare translation surface as a (branching) covering surface of the flat torus.

Similarly, we can glue together an arbitrary number of, for instance, congruent regular octagons in a natural horizontal or vertical way.
The class of such \textit{polyoctagon surfaces} exhibits arbitrarily high genera.
But there is only one primitive member, namely, the regular octagon surface itself. 

Let us return to regular polygon billiard, the first infinite family of primitive optimal systems, discovered in the 1980s.
A completely different infinite family of primitive optimal systems, discovered in the early 2000s independently by Calta~\cite{C} and McMullen~\cite{Mc}, using different methods, consists of L-shaped tables for which the billiard is optimal.
We shall refer to this as the Calta--McMullen family.
Translation surfaces arising from members of this family have genus~$2$, and the simplest member of this family is the golden L-shape.
The L-shape has the very special property that it is the same to consider billiard flow in the L-shape region or geodesic flow on the L-surface, as they are equivalent systems.

The general member of the Calta--McMullen family is an L-shaped billiard table, as shown in the picture on the left in Figure~8.4.9, or an L-surface, with boundary identification given as in the picture on the right in Figure~8.4.9.

\begin{displaymath}
\begin{array}{c}
\includegraphics[scale=0.8]{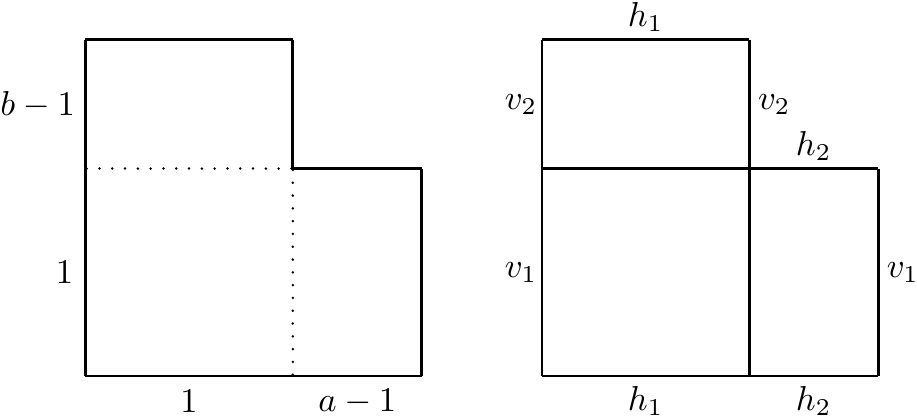}
\vspace{3pt}\\
\mbox{Figure 8.4.9: L-shaped billiard table}
\end{array}
\end{displaymath}

The bottom horizontal side has length~$a$, the left-most vertical side has length~$b$, the right-most vertical side has length~$1$, and the top-most horizontal side has length~$1$.
The numbers
\begin{displaymath}
a=r_1\sqrt{d}+r_2
\quad\mbox{and}\quad
b=r_1\sqrt{d}+1-r_2
\end{displaymath}
are such that $r_1,r_2$ are rational numbers and $d\ge2$ is a square-free integer.
Here the choice of length $1$ for two of the edges excludes equivalent systems, \textit{i.e.}, it represents normalization.

It is easy to check that the L-surface is a street-rational polyrectangle translation surface.
For the two horizontal streets, the width-height ratios of the two rectangle are
\begin{displaymath}
\frac{a}{1}
\quad\mbox{and}\quad
\frac{1}{b-1},
\end{displaymath}
and their ratio is $a(b-1)=(r_1\sqrt{d}+r_2)(r_1\sqrt{d}-r_2)=r_1^2d-r_2^2$, which is rational.
For the two vertical streets, the width-height ratios of the two rectangle are
\begin{displaymath}
\frac{a-1}{1}
\quad\mbox{and}\quad
\frac{1}{b},
\end{displaymath}
and their ratio is $(a-1)b=(r_1\sqrt{d}+r_2-1)(r_1\sqrt{d}+1-r_2)=r_1^2d-(r_2-1)^2$, which is also rational.

The shortline method gives the following result.

\begin{thm}\label{thm8.4.2}
Consider any surface or region in the Calta--McMullen family.

\emph{(i)}
There exist infinitely many explicit slopes, depending on the surface, such that any half-infinite geodesic on the surface having such a slope exhibits superdensity.

\emph{(ii)}
There exist infinitely many explicit slopes, depending on the region, such that any billiard orbit in the region having such an initial slope exhibits superdensity.

\emph{(iii)} For infinitely many of these slopes in parts \emph{(i)} and \emph{(ii)} that give rise to superdensity, we can explicitly compute the corresponding irregularity exponents.
Combining the irregularity exponent with the method of zigzagging introduced in \cite[Section~3.3]{BDY1}, we can also describe, for trajectories having these initial slopes, the time-quantitative behavior of the edge cutting and face crossing numbers, as well as equidistribution relative to all convex sets.
\end{thm}

We shall justify part (iii) for the golden L-surface in Section~\ref{sec8.5}.

Optimality and superdensity of some orbits are two different aspects concerning a dynamical system, both exhibiting \textit{perfect} behavior.
Optimality for these systems is established by Calta and McMullen via ergodicity, making use of Birkhoff's ergodic theorem.
Since Birkhoff's ergodic theorem does not give any error term, these results do not say anything quantitative about the speed of convergence to uniformity, or even about time-quantitative density.
Our superdensity results, on the other hand, establish a best possible form
of time-quantitative density, at least for \textit{some} slopes.
And this is the best that we can hope for, since superdensity fails for almost every slope anyway.

Note next that every member of the Calta--McMullen family is what we may call a \textit{$1$-step L-staircase}, and the corresponding surface has the same genus~$2$.
It is a $2$-parameter family, since
\begin{displaymath}
a=\sqrt{q_1}+r_2
\quad\mbox{and}\quad
b=\sqrt{q_1}+1-r_2
\end{displaymath}
depend only on two rational parameters $q_1=r_1^2d$ and~$r_2$.

We next describe a far-reaching extension of the Calta--McMullen family, leading to surfaces with arbitrarily large genus.
For any integer $k\ge1$, we construct an infinite family of street-rational polysquares that are \textit{$k$-step L-staircases}, and each corresponding surface has genus~$k+1$.

The picture on the left in Figure~8.4.10 illustrates a typical street-rational $2$-step L-staircase.

\begin{displaymath}
\begin{array}{c}
\includegraphics[scale=0.8]{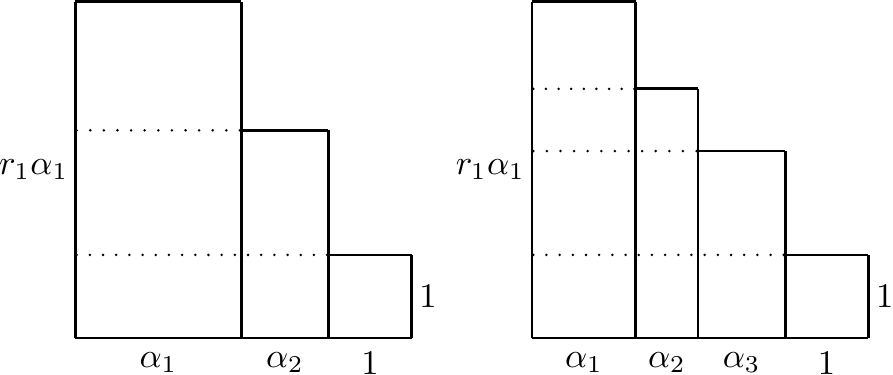}
\vspace{3pt}\\
\mbox{Figure 8.4.10: street-rational $2$-step and $3$-step L-staircases}
\end{array}
\end{displaymath}

The right-most vertical rectangle is a unit square, representing normalization to exclude equivalent systems.
The left-most vertical rectangle has size $\alpha_1\times r_1\alpha_1$, while the second vertical rectangle has size $\alpha_2\times r_2\alpha_2$, where $r_1\alpha_1>r_2\alpha_2>1$ and $r_1,r_2$ are positive rational numbers.

The top horizontal rectangle now has size $\alpha_1\times(r_1\alpha_1-r_2\alpha_2)$, the middle horizontal rectangle has size
$(\alpha_1+\alpha_2)\times (r_2\alpha_2-1)$, while the bottom horizontal rectangle has size $(\alpha_1+\alpha_2+1)\times1$.
Thus the $2$-step L-staircase surface is a street-rational polyrectangle translation surface precisely if we have the relations
\begin{equation}\label{eq8.4.1}
\frac{\alpha_1}{r_1\alpha_1-r_2\alpha_2}
=r_3\frac{\alpha_1+\alpha_2}{r_2\alpha_2-1}
=r_4(\alpha_1+\alpha_2+1)
\end{equation}
for the width-height ratios of the $3$ horizontal rectangles, where $r_3,r_4$ are some positive rational numbers.

The analogous requirement for the $3$ vertical rectangles is guaranteed by definition.

We wish to express the two positive real variables $\alpha_1,\alpha_2$ in \eqref{eq8.4.1} in terms of the given positive rational parameters
$r_1,r_2,r_3,r_4$.
Routine calculation shows that both $\alpha_1,\alpha_2$ are algebraic numbers.
Indeed, using the second equality in \eqref{eq8.4.1}, we obtain
\begin{equation}\label{eq8.4.2}
\alpha_1=\frac{r_2r_4\alpha^2_2+(r_2r_4-r_3-r_4)\alpha_2-r_4}{r_3+r_4-r_2r_4\alpha_2}.
\end{equation}
Substituting this into the first equality in \eqref{eq8.4.1} eliminates the variable~$\alpha_1$, and gives rise to a polynomial equation of degree $4$ in the
variable~$\alpha_2$ with rational coefficients.
Thus $\alpha_2$ is algebraic.
In view of \eqref{eq8.4.2}, $\alpha_1$ is also algebraic.
Crucially, for the typical choice of the rational parameters $r_1,r_2,r_3,r_4$, this polynomial equation of degree $4$ has a positive root which is not rational.

This completes the construction of the class of $2$-step street-rational L-staircases.
It is a $4$-parameter family, depending on the rational parameters $r_1,r_2,r_3,r_4$, and each surface in this family has genus~$3$.

The picture on the right in Figure~8.4.10 illustrates a typical street-rational $3$-step L-staircase.

The right-most vertical rectangle is a unit square, representing normalization to exclude equivalent systems.
The left-most vertical rectangle has size $\alpha_1\times r_1\alpha_1$, the second vertical rectangle has size $\alpha_2\times r_2\alpha_2$, while the third vertical rectangle has size $\alpha_3\times r_3\alpha_3$, where $r_1\alpha_1>r_2\alpha_2>r_3\alpha_3>1$ and $r_1,r_2,r_3$ are positive rational numbers.

The top horizontal rectangle now has size $\alpha_1\times(r_1\alpha_1-r_2\alpha_2)$, the second horizontal rectangle has size
$(\alpha_1+\alpha_2)\times(r_2\alpha_2-r_3\alpha_3)$, the third horizontal rectangle has size $(\alpha_1+\alpha_2+\alpha_3)\times(r_3\alpha_3-1)$, while the bottom horizontal rectangle has size $(\alpha_1+\alpha_2+\alpha_3+1)\times1$.
Thus the $3$-step L-staircase surface is a street-rational polyrectangle translation surface precisely if we have the relations
\begin{equation}\label{eq8.4.3}
\frac{\alpha_1}{r_1\alpha_1-r_2\alpha_2}
=r_4\frac{\alpha_1+\alpha_2}{r_2\alpha_2-r_3\alpha_3}
=r_5\frac{\alpha_1+\alpha_2+\alpha_3}{r_3\alpha_3-1}
=r_6(\alpha_1+\alpha_2+\alpha_3+1)
\end{equation}
for the width-height ratios of the $4$ horizontal rectangles, where $r_4,r_5,r_6$ are some positive rational numbers.

The analogous requirement for the $4$ vertical rectangles is guaranteed by definition.

We wish to express the three positive real variables $\alpha_1,\alpha_2,\alpha_3$ in \eqref{eq8.4.3} in terms of the given positive rational parameters $r_1,r_2,r_3,r_4,r_5,r_6$.
Routine calculation shows that $\alpha_1,\alpha_2,\alpha_3$ are all algebraic numbers.
We leave the details to the reader.

This completes the construction of the class of $3$-step street-rational L-staircases.
It is a $6$-parameter family, depending on the rational parameters $r_1,r_2,r_3,r_4,r_5,r_6$, and each surface in this family has genus~$4$.

The construction of street-rational $k$-step L-staircases, where $k\ge1$, goes in a similar way.
This gives rise to a $2k$-parameter family, and each surface in the family has genus~$k+1$.

The construction gives a street-rational polyrectangle translation surface with a $1$-direction geodesic flow, and the corresponding billiard region is also street-rational.
Thus the shortline method works, and gives the following generalization of Theorem~\ref{thm8.4.2}.

\begin{thm}\label{thm8.4.3}
Let $k\ge1$ be an integer, and consider a $k$-step street-rational L-staircase surface or region.

\emph{(i)}
There exist infinitely many explicit slopes, depending on the surface, such that any half-infinite geodesic on the surface having such a slope exhibits superdensity.

\emph{(ii)}
There exist infinitely many explicit slopes, depending on the region, such that any billiard orbit in the region having such an initial slope exhibits superdensity.

\emph{(iii)}
For infinitely many of these slopes in parts \emph{(i)} and \emph{(ii)} that give rise to superdensity, we can explicitly compute the corresponding irregularity exponents.
Combining the irregularity exponent with the method of zigzagging introduced in \cite[Section~3.3]{BDY1}, we can also describe, for trajectories having these initial slopes, the time-quantitative behavior of the edge cutting and face crossing numbers, as well as equidistribution relative to all convex sets.
\end{thm}

\begin{remark}
We are not aware whether these street-rational $k$-step L-staircases, where $k\ge2$, in general are optimal.
What we can prove is that they all have explicit superdense orbits.
More precisely, our construction gives primitive street-rational L-staircase surfaces in \textit{every} positive genus and which exhibit explicit superdense orbits. 
On the other hand, we are not aware of any infinite family of surfaces in a fixed high genus.
\end{remark}

To obtain new street-rational polyrectangle translation surfaces, there is a technique which is algebraic in nature.
We shall describe this technique by first looking at a simple example.

Consider the polysquare translation surface which is a $3\times3$ square with a missing square in the middle, as shown in Figure~8.4.11 with boundary identification given by perpendicular translation.

\begin{displaymath}
\begin{array}{c}
\includegraphics[scale=0.8]{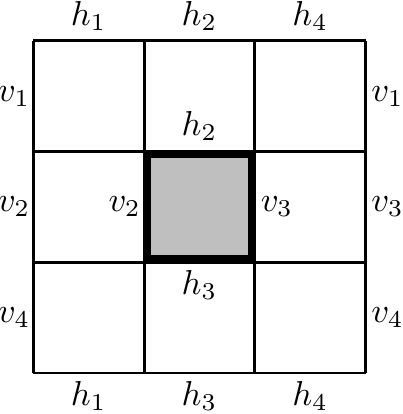}
\vspace{3pt}\\
\mbox{Figure 8.4.11: a polysquare translation surface}
\end{array}
\end{displaymath}

We can obtain a polyrectangle translation surface by varying the lengths of the edges of the square faces to change them to rectangle faces.
To ensure that we obtain affine-different polyrectangles, we fix a horizontal side and a vertical side to have length~$1$.
For simplicity, we have ensured that the bottom left square face does not get altered.
As shown in Figure~8.4.12, we allow the other faces to have different lengths.

\begin{displaymath}
\begin{array}{c}
\includegraphics[scale=0.8]{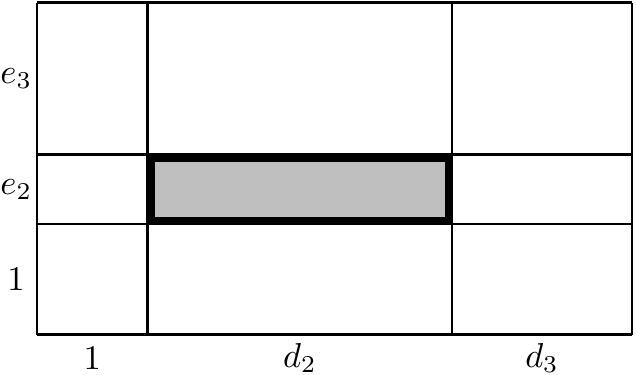}
\vspace{3pt}\\
\mbox{Figure 8.4.12: a polyrectangle translation surface corresponding}
\\
\mbox{to the polysquare translation surface in Figure 8.4.11}
\end{array}
\end{displaymath}

The original polysquare translation surface has $4$ horizontal streets and $4$ vertical streets, corresponding to the edge pairings $v_1,v_2,v_3,v_4$ and $h_1,h_2,h_3,h_4$.
The $4$ corresponding horizontal streets in the polyrectangle translation surface have normalized lengths, \textit{i.e.} street lengths normalized by division by street widths, equal to
\begin{displaymath}
\frac{1+d_2+d_3}{e_3},
\quad
\frac{1}{e_2},
\quad
\frac{d_3}{e_2},
\quad
\frac{1+d_2+d_3}{1}.
\end{displaymath}
For horizontal street-rationality, we require positive rational numbers $r_1,r_2,r_3$ such that
\begin{equation}\label{eq8.4.4}
e_3=r_1,
\quad
d_3=r_2,
\quad
e_2(1+d_2+d_3)=r_3.
\end{equation}
On the other hand, the $4$ corresponding vertical streets in the polyrectangle translation surface have normalized lengths equal to
\begin{displaymath}
\frac{1+e_2+e_3}{1},
\quad
\frac{e_3}{d_2},
\quad
\frac{1}{d_2},
\quad
\frac{1+e_2+e_3}{d_3}.
\end{displaymath}
For vertical street-rationality, again we require $e_3$ and $d_3$ to be rational, already taken care of in \eqref{eq8.4.4}, as well as another positive rational numbers $r_4$ such that
\begin{equation}\label{eq8.4.5}
d_2(1+e_2+e_3)=r_4.
\end{equation}
Combining \eqref{eq8.4.4} and \eqref{eq8.4.5}, we obtain
\begin{displaymath}
e_2(d_2+1+r_2)=r_3
\quad\mbox{and}\quad
d_2(e_2+1+r_1)=r_4,
\end{displaymath}
leading to a linear equation in $e_2$ and $d_2$ with rational coefficients.
On combining this linear equation with \eqref{eq8.4.5} and eliminating the variable~$e_2$, say, we obtain a quadratic equation with rational coefficients in the variable~$d_2$.
One can check that for infinitely many choices of the rational numbers $r_1,r_2,r_3,r_4$, this quadratic equation has a positive irrational root $d_2$ with corresponding positive~$e_2$, leading to a street-rational polyrectangle translation surface that is not a polysquare surface.

A similar argument seems to work if we start with a polysquare translation surface.
For illustration, we consider the slightly more complicated polysquare translation surface shown in Figure~8.4.13.

\begin{displaymath}
\begin{array}{c}
\includegraphics[scale=0.8]{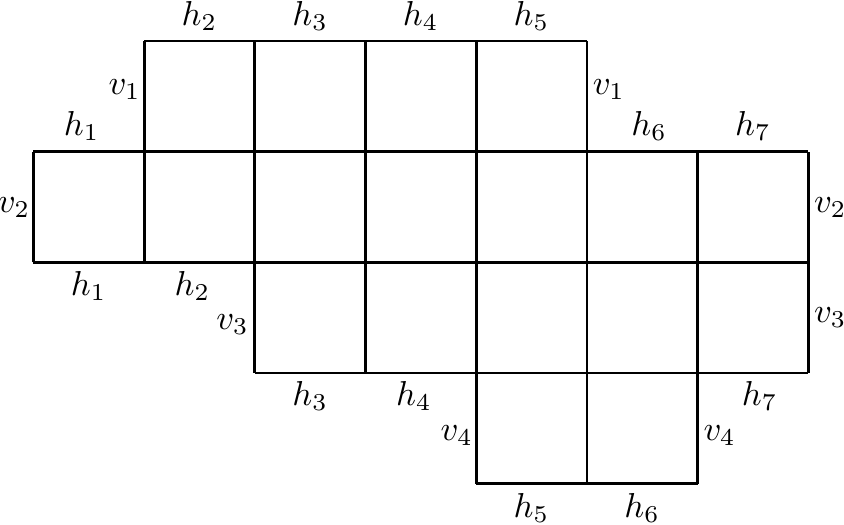}
\vspace{3pt}\\
\mbox{Figure 8.4.13: another polysquare translation surface}
\end{array}
\end{displaymath}

We can obtain a polyrectangle translation surface by varying the lengths of the edges of the square faces to change them to rectangle faces.
To ensure that we obtain affine-different polyrectangles, we fix a horizontal side and a vertical side to have length~$1$.
For simplicity, we have ensured that the left most square face does not get altered.
As shown in Figure~8.4.14, we allow the other faces to have different lengths.

\begin{displaymath}
\begin{array}{c}
\includegraphics[scale=0.8]{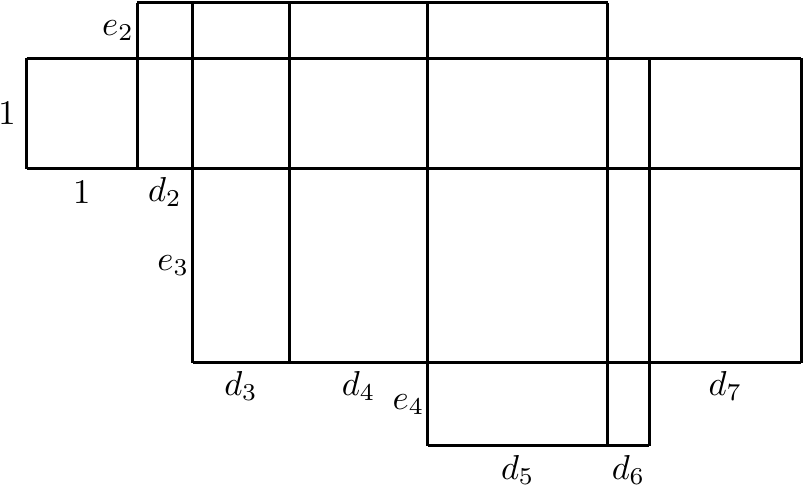}
\vspace{3pt}\\
\mbox{Figure 8.4.14: a polyrectangle translation surface corresponding}
\\
\mbox{to the polysquare translation surface in Figure 8.4.13}
\end{array}
\end{displaymath}

The original polysquare translation surface has $4$ horizontal streets and $7$ vertical streets.
The $4$ corresponding horizontal streets in the polyrectangle translation surface have normalized lengths equal to
\begin{displaymath}
\frac{d_2+\ldots+d_5}{e_2},
\quad
\frac{1+d_2+\ldots+d_7}{1},
\quad
\frac{d_3+\ldots+d_7}{e_3},
\quad
\frac{d_5+d_6}{e_4},
\end{displaymath}
and horizontal street-rationality leads to $3$ equations that involve $3$ positive rational numbers $r_1,r_2,r_3$, say.
On the other hand, the $7$ corresponding vertical streets in the polyrectangle translation surface have normalized lengths equal to
\begin{displaymath}
\frac{1}{1},
\ \frac{1+e_2}{d_2},
\ \frac{1+e_2+e_3}{d_3},
\ \frac{1+e_2+e_3}{d_4},
\ \frac{1+e_2+e_3+e_4}{d_5},
\ \frac{1+e_3+e_4}{d_6},
\ \frac{1+e_3}{d_7},
\end{displaymath}
and vertical street-rationality leads to $6$ equations that involve $6$ positive rational numbers $r_4,\ldots,r_9$, say.
In other words, we have $9$ equations in the $9$ variables
\begin{displaymath}
d_2,d_3,d_4,d_5,d_6,d_7,e_2,e_3,e_4
\end{displaymath}
that involve $9$ positive rational parameters $r_1,\ldots,r_9$.
And it is reasonable to expect that there are infinitely many choices of $r_1,\ldots,r_9$ that lead to positive solutions for all the variables and irrational solutions for some.

Indeed, we may generalize this to any arbitrary finite polysquare translation surface with $m$ horizontal streets and $n$ vertical streets.
Street-rationality will give rise to $m+n-2$ equations in $m+n-2$ variables of the type $d_i$ and~$e_j$, and these equations involve $m+n-2$ positive rational parameters.

In sharp contrast, as far as we know, the current literature shows only $3$ types of polysquare translation surfaces that will lead to infinitely many primitive optimal translation surfaces.
The L-surface on the left in Figure~8.4.15 leads to the Calta--McMullen family of L-surfaces with genus~$2$.
The second type, shown in the middle of Figure~8.4.15, has genus~$3$, while the third type, shown on the right in Figure~8.4.15, has genus~$4$.

\begin{displaymath}
\begin{array}{c}
\includegraphics[scale=0.8]{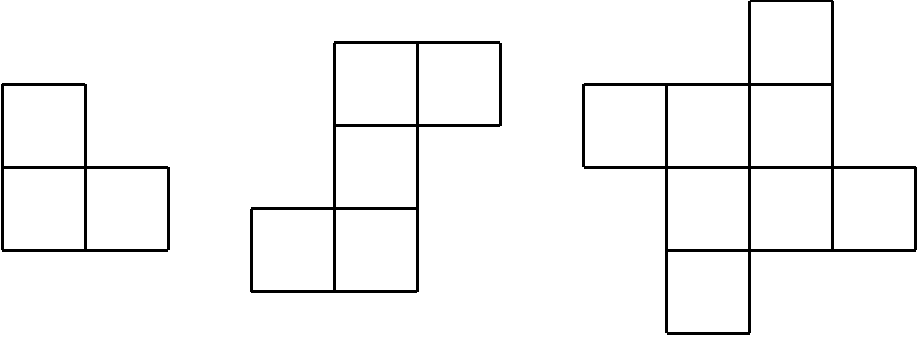}
\vspace{3pt}\\
\mbox{Figure 8.4.15: polysquare translation surfaces leading to infinitely many}
\\
\mbox{primitive optimal translation surfaces}
\end{array}
\end{displaymath}

For the details of the corresponding optimal families, we refer the reader to the paper of McMullen~\cite{Mc4}.

In particular, Calta~\cite{C} and McMullen~\cite{Mc2,Mc3,Mc4} have given a complete list of optimal systems of genus~$2$.
Those missing from this list include the family shown in Figure~8.4.16, with perpendicular boundary identification, and genus~$2$.

\begin{displaymath}
\begin{array}{c}
\includegraphics[scale=0.8]{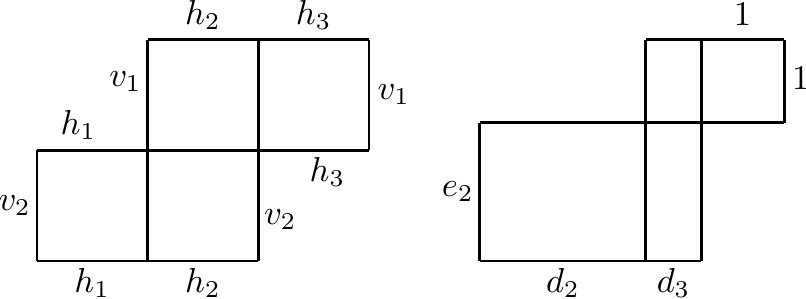}
\vspace{3pt}\\
\mbox{Figure 8.4.16: an infinite family of non-optimal street-rational}
\\
\mbox{polyrectangle translation surfaces}
\end{array}
\end{displaymath}

The polysquare translation surface has $2$ horizontal streets and $3$ vertical streets.
Let us consider the corresponding polyrectangle translation surfaces shown in the picture on the right in Figure~8.4.16.
The $3$ vertical streets have normalized lengths equal to
\begin{displaymath}
\frac{e_2}{d_2},
\quad
\frac{1+e_2}{d_3},
\quad
\frac{1}{1},
\end{displaymath}
and for vertical street-rationality, we need rational numbers $r_1,r_2>0$ such that
\begin{displaymath}
r_1\frac{e_2}{d_2}=r_2\frac{1+e_2}{d_3}=1,
\end{displaymath}
so we have the two conditions
\begin{equation}\label{eq8.4.6}
d_2=r_1e_2
\quad\mbox{and}\quad
d_3=r_2(1+e_2).
\end{equation}
On the other hand, the $2$ horizontal streets have normalized lengths equal to
\begin{displaymath}
\frac{1+d_3}{1},
\quad
\frac{d_2+d_3}{e_2},
\end{displaymath}
and for horizontal street-rationality, we need a rational number $r_3>0$ such that
\begin{displaymath}
r_3\frac{1+d_3}{1}=\frac{d_2+d_3}{e_2},
\end{displaymath}
so we have the extra condition
\begin{equation}\label{eq8.4.7}
d_2+d_3=r_3e_2(1+d_3).
\end{equation}

The equations \eqref{eq8.4.6} and \eqref{eq8.4.7} have quadratic irrational solutions in $d_2,d_3,e_2$ that can be written in a simple explicit form, in sharp contrast to L-staircases where the corresponding requirements \eqref{eq8.4.1}--\eqref{eq8.4.3} lead to unpleasant higher degree equations, so we work out the details here.

Substituting \eqref{eq8.4.6} into \eqref{eq8.4.7}, we obtain
\begin{equation}\label{eq8.4.8}
r_1e_2+r_2(1+e_2)=r_3e_2(1+r_2(1+e_2)).
\end{equation}
This quadratic equation in the variable $e_2$ can be rewritten in the form
\begin{displaymath}
r_2r_3e_2^2+(r_2r_3+r_3-r_1-r_2)e_2-r_2=0.
\end{displaymath}
Since $r_2,r_3>0$, such a typical quadratic equation has complex conjugate roots $e_2=R_2\pm R_1\sqrt{D}$, where $R_1,R_2$ are rational numbers and $D\ge2$ is a positive squarefree integer.
Furthermore,
\begin{equation}\label{eq8.4.9}
R_1^2D-R_2^2=-(R_2+R_1\sqrt{D})(R_2-R_1\sqrt{D})=\frac{1}{r_3}>0.
\end{equation}
For simplicity, let us assume that
\begin{equation}\label{eq8.4.10}
e_2=R_1\sqrt{D}+R_2>0.
\end{equation}
Using \eqref{eq8.4.10}, the left hand side of \eqref{eq8.4.8} becomes
\begin{displaymath}
r_1(R_1\sqrt{D}+R_2)+r_2(R_1\sqrt{D}+R_2+1),
\end{displaymath}
and the right hand side of \eqref{eq8.4.8} becomes
\begin{displaymath}
r_3(R_1\sqrt{D}+R_2)(1+r_2(R_1\sqrt{D}+R_2+1)).
\end{displaymath}
The coefficient for $\sqrt{D}$ in \eqref{eq8.4.8} is equal to
\begin{equation}\label{eq8.4.11}
(r_1+r_2)R_1=r_3(1+r_2)R_1+2r_2r_3R_1R_2,
\end{equation}
while the rational term in \eqref{eq8.4.8} is equal to
\begin{equation}\label{eq8.4.12}
(r_1+r_2)R_2+r_2=r_3(1+r_2)R_2+r_2r_3R_2^2+r_2r_3R_1^2D.
\end{equation}
Combining \eqref{eq8.4.11} and \eqref{eq8.4.12} and eliminating $r_3$, we deduce that
\begin{displaymath}
\frac{(r_1+r_2)R_1}{(1+r_2)R_1+2r_2R_1R_2}
=\frac{(r_1+r_2)R_2+r_2}{(1+r_2)R_2+r_2R_2^2+r_2R_1^2D},
\end{displaymath}
which can be simplified to
\begin{displaymath}
\frac{r_1+r_2}{1+r_2+2r_2R_2}
=\frac{(r_1+r_2)R_2+r_2}{(1+r_2)R_2+r_2R_2^2+r_2R_1^2D}.
\end{displaymath}
For any given $R_1,R_2,D$ and~$r_2$, this is a linear equation in~$r_1$, with solution
\begin{equation}\label{eq8.4.13}
r_1=\frac{1+r_2(2R_2+1)}{R_1^2D-R_2^2}-r_2.
\end{equation}
We require $r_1>0$, so we have to assume the \textit{positivity condition}
\begin{displaymath}
\frac{1+r_2(2R_2+1)}{R_1^2D-R_2^2}-r_2>0.
\end{displaymath}
In view of \eqref{eq8.4.9}, this is equivalent to $1+r_2(2R_2+1)-(R_1^2D-R_2^2)r_2>0$, \textit{i.e.},
\begin{equation}\label{eq8.4.14}
(R_1^2D-(R_2+1)^2)r_2<1.
\end{equation}
Note that $R_1^2D-(R_2+1)^2$ is never zero, so there are only two possibilities:

(i)
If $R_1^2D-(R_2+1)^2>0$, then $r_2$ is restricted to the range
\begin{displaymath}
0<r_2<\frac{1}{R_1^2D-(R_2+1)^2}.
\end{displaymath}

(ii)
If $R_1^2D-(R_2+1)^2<0$, then any $r_2>0$ is possible.

Combining \eqref{eq8.4.6} and \eqref{eq8.4.10}, we now have
\begin{equation}\label{eq8.4.15}
e_2=R_1\sqrt{D}+R_2,
\quad
d_2=r_1(R_1\sqrt{D}+R_2),
\quad
d_3=r_2(R_1\sqrt{D}+R_2+1).
\end{equation}
Note that \eqref{eq8.4.13} and \eqref{eq8.4.14} become particularly simple if $R_1^2D-R_2^2=1$ and $R_2>0$.
Then case (ii) applies.
If we write $r=r_2$, then using \eqref{eq8.4.13}, we can reduce \eqref{eq8.4.15} to
\begin{displaymath}
e_2=R_1\sqrt{D}+R_2,
\quad
d_2=(2rR_2+1)(R_1\sqrt{D}+R_2),
\quad
d_3=r(R_1\sqrt{D}+R_2+1).
\end{displaymath}
It is a lucky coincidence that we know all the integral solutions of the Pell equation $R_1^2D-R_2^2=1$ and $R_2>0$.
The simplest such choice is $D=2$ and $R_1=R_2=1$, leading to the solution
\begin{displaymath}
e_2=\sqrt{2}+1,
\quad
d_2=(2r+1)(\sqrt{2}+1),
\quad
d_3=r(\sqrt{2}+2),
\end{displaymath}
where $r>0$ is any positive rational number.

It is clear from the work of Calta and McMullen that neither geodesic flow on any of these street-rational polyrectangle translation surfaces nor billiard flow on any such polyrectangle billiard table is optimal.
Our shortline method, nevertheless, does give time-quantitative results for infinitely many directions.

Next, we leave this algebraic approach and outline a completely different geometric way to construct primitive non-polysquare street-rational polyrectangle translation surfaces.
We shall call this \textit{double-rational gluing}.
For a simple illustration of the ideas, we consider Figure~8.4.17, where we glue together two regular octagons with edge lengths that are rational multiples of each other.

\begin{displaymath}
\begin{array}{c}
\includegraphics[scale=0.8]{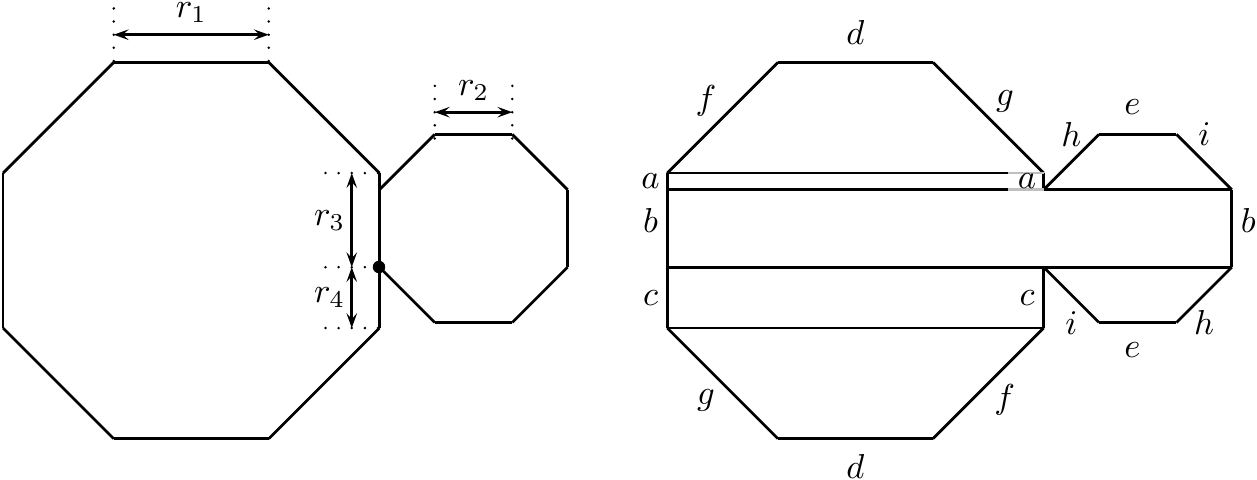}
\vspace{3pt}\\
\mbox{Figure 8.4.17: gluing together two regular octagons}
\end{array}
\end{displaymath}

Since the edge lengths of the two regular octagons are rational multiples of each other, we may assume, without loss of generality, that both edge lengths are rational numbers, equal to $r_1$ and $r_2$ as shown in the picture on the left in Figure~8.4.17.
We glue the two octagons together in such a way that the vertex, indicated in the picture by the dot, of the octagon on the right lies on a point on the edge of the octagon on the left which has rational distances $r_3$ and $r_4$ from the two nearest vertices.
With edge identification shown in the picture on the right in Figure~8.4.17, the union of the two regular octagons becomes a surface.

It is clear that the edge identification give rise to $4$ vertical streets and $5$ horizontal streets, as shown in Figure~8.4.18.
To check for street-rationality, we need to look at the normalized lengths, \textit{i.e.} street lengths normalized by division by street widths, of these streets.

\begin{displaymath}
\begin{array}{c}
\includegraphics[scale=0.8]{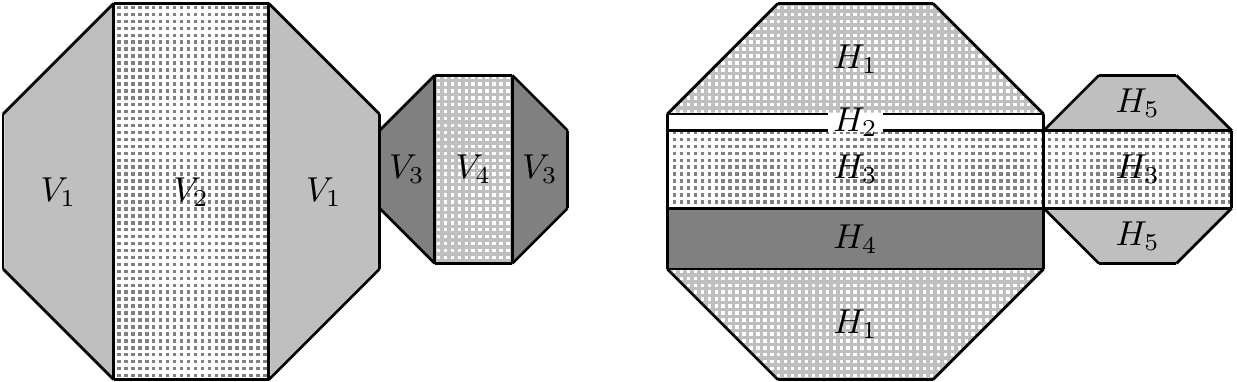}
\vspace{3pt}\\
\mbox{Figure 8.4.18: vertical and horizontal streets}
\end{array}
\end{displaymath}

The normalized lengths of the vertical streets $V_1,V_2,V_3,V_4$ are respectively
\begin{displaymath}
\frac{2+\sqrt{2}}{1/\sqrt{2}},
\quad
\frac{1+\sqrt{2}}{1},
\quad
\frac{2+\sqrt{2}}{1/\sqrt{2}},
\quad
\frac{1+\sqrt{2}}{1},
\end{displaymath}
and these are all rational multiples of $1+\sqrt{2}$.
On the other hand, the normalized lengths of the horizontal streets $H_1,H_2,H_3,H_4,H_5$ are respectively
\begin{displaymath}
\frac{2+\sqrt{2}}{1/\sqrt{2}},
\quad
\frac{r_1(1+\sqrt{2})}{r_3-r_2},
\quad
\frac{(r_1+r_2)(1+\sqrt{2})}{r_2},
\quad
\frac{r_1(1+\sqrt{2})}{r_4},
\quad
\frac{2+\sqrt{2}}{1/\sqrt{2}},
\end{displaymath}
and these are also all rational multiples of $1+\sqrt{2}$.
Thus the surface comprising the two octagons is a street-rational polyrectangle translation surface which is not a polysquare surface.

Note that a bigger regular octagon is not a covering surface of a smaller regular octagon.
Thus this street-rational polyrectangle translation surface is primitive, and cannot be obtained from a simpler surface via covering construction.

We need not stop at two regular octagons.
In general, we can glue together, across horizontal or vertical edges, arbitrarily many regular octagons with edge lengths that are rational multiples of each other in a similar way.
With a typical choice of rational edge length parameters, the resulting street-rational polyrectangle translation surface is primitive.

Of course, we can replace the regular octagon by any regular polygon of $k$ sides, where $k\ge8$ is divisible by~$4$, and obtain an analogous class of street-rational polyrectangle translation surfaces via double-rational gluing.
The divisibility by $4$ guarantees that each copy of the regular polygon has horizontal and vertical sides which play a special role in the gluing process.

We can also replace the regular octagon with the golden L-shape or, more precisely, the golden cross made up of $4$ reflected copies of the golden L-shape, as shown in Figure~8.4.19.
Then double-rational gluing gives rise to street-rational polyrectangle translation surfaces where the normalized length of any street is a rational multiple of the golden ratio $(1+\sqrt{5})/2$.

\begin{displaymath}
\begin{array}{c}
\includegraphics[scale=0.8]{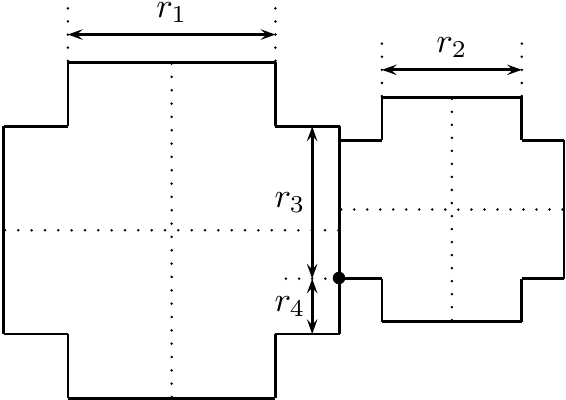}
\vspace{3pt}\\
\mbox{Figure 8.4.19: gluing together two golden crosses}
\end{array}
\end{displaymath}

In Figure~8.4.20, we have a street-rational polyrectangle translation surface made up of nine golden crosses, with six copies of the same size, two copies with double edge length, and one copy of triple edge length.
Indeed, if we consider this figure as a billiard table, the the corresponding translation surface for the billiard is also a street-rational polyrectangle translation surface.

\begin{displaymath}
\begin{array}{c}
\includegraphics[scale=0.8]{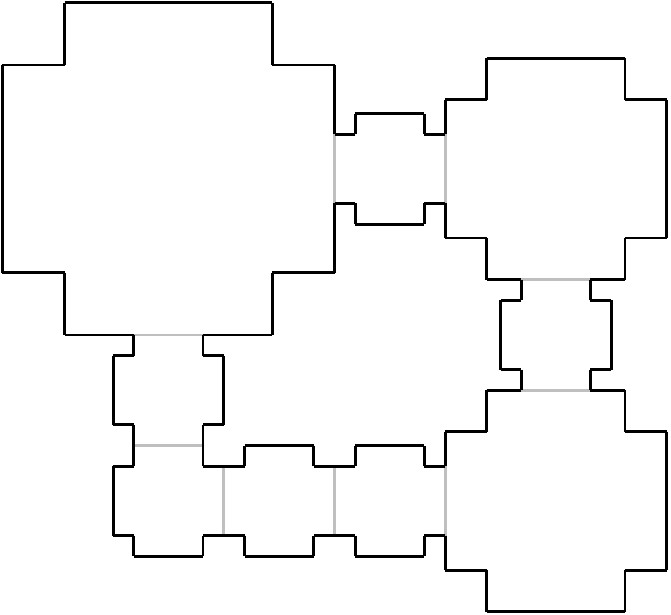}
\vspace{3pt}\\
\mbox{Figure 8.4.20: a street-rational polyrectangle surface}
\\
\mbox{consisting of nine golden crosses}
\end{array}
\end{displaymath}

A similar consideration arises if we start with a cross made up of $4$ reflected copies of any member of the Calta--McMullen family, or a surface made up of $4$ reflected copies of an street-rational $k$-step L-staircase.

Note that there is no analogous results for optimal systems.
We are not aware of any construction using double-rational gluing that builds optimal systems from optimal components.

We conclude this section by considering surfaces of some rectangular boxes.
This is a generalization of Example~7.2.4 concerning the surface of the unit cube.

One such example is what we may call the \textit{golden brick}, or a brick with golden ratio, as shown in Figure~8.4.21.

\begin{displaymath}
\begin{array}{c}
\includegraphics[scale=0.8]{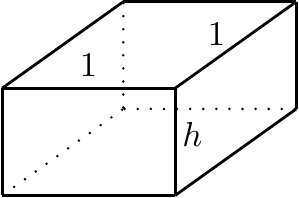}
\vspace{3pt}\\
\mbox{Figure 8.4.21: a brick with the golden ratio}
\end{array}
\end{displaymath}

This has two opposite faces that are unit squares, and the four parallel edges joining them have length $h=(\sqrt{5}-1)/2$. 
It is easy to see that this gives rise to a street-rational polyrectangle surface.
Indeed, it has $3$ streets, each of length~$4$.
One of these streets has size $h\times4$, while the other two have common size $1\times(2+2h)$.
Street-rationality is now a consequence of the observation that
\begin{displaymath}
\frac{1}{2+2h}=2\cdot\frac{h}{4},
\end{displaymath}
due to $h^2+h=1$.
We shall return to this example in the next section where we shall determine some irregularity exponents.

Clearly, this example can be generalized.
Starting with two opposite unit square faces, we now let $z$ denote the length of the four parallel edges joining them.
One of these streets has size $z\times4$, while the other two have common size $1\times(2+2z)$.
Thus we have street-rationality if we can find coprime integers $a,b\ge1$ such that
\begin{displaymath}
\frac{1}{2+2z}=\frac{a}{b}\cdot\frac{z}{4},
\end{displaymath}
precisely when $az^2+az-2b=0$.
Thus
\begin{displaymath}
z=\frac{\sqrt{a^2+8ab}-a}{2a}.
\end{displaymath}
Note that the choice $a=2$ and $b=1$ gives $z=h$.

Thus Theorem~\ref{thm8.4.3} can be extended to all of these special box-surfaces. 

We can further generalize by starting with two opposite rectangle faces.
We leave the details to the reader.
Unfortunately the analogous problem for an \textit{arbitrary} box-surface remains wide open.

%%%%%%%%%%
%
% SECTION 8.5
%
%%%%%%%%%%0

\subsection{Computing more irregularity exponents}\label{sec8.5}

In this section, we first return to the double-pentagon surface as shown in Figure~8.4.2.
Our purpose is to find the eigenvalues of its $2$-step transition matrix~$\bfA$.

This surface can be viewed as a polyparallelogram translation surface, and we shall use an analogy with a corresponding polyrectangle translation surface as shown in Figure~8.5.1.

\begin{displaymath}
\begin{array}{c}
\includegraphics[scale=0.8]{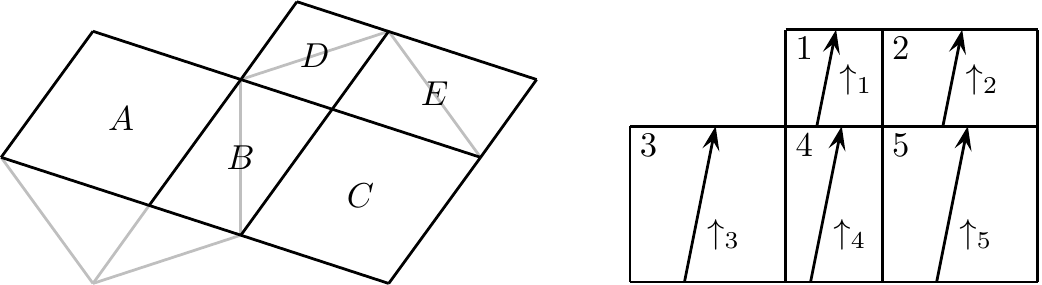}
\vspace{3pt}\\
\mbox{Figure 8.5.1: analogy of the double-pentagon surface}
\\
\mbox{with a polyrectangle translation surface}
\end{array}
\end{displaymath}

The rhombi $A,C,D$ in the picture on the left become the squares $3,5,1$ respectively, while the parallelograms $B,E$ become the rectangles $4,2$ respectively.
Here the horizontal streets are $1,2$ and $3,4,5$, while the vertical streets are $1,4$ and $2,5,3$.
It is not difficult to check that both horizontal streets and both vertical streets have normalized lengths $(3+\sqrt{5})/2$.
Note here that the two distinct rectangle faces $3,5$ can fall into the same horizontal street and the same vertical street, so we shall adapt our notation from Section~\ref{sec7.2} for the determination of the street-spreading matrix.

We also show in the picture on the right in Figure~8.5.1 the almost vertical units of type $\uparrow$ in the polyrectangle translation surface.
We have not shown that almost vertical units of type $\nuparrow$ here.

Let
\begin{displaymath}
J_1=\{1,2\}
\quad\mbox{and}\quad
J_2=\{3,4,5\}
\end{displaymath}
denote the horizontal streets, and let
\begin{displaymath}
I_1=I_4=\{1,4\}
\quad\mbox{and}\quad
I_2=I_3=I_5=\{2,3,5\}
\end{displaymath}
denote the vertical streets.

For the polyrectangle translation surface, we thus consider slopes of the form
\begin{displaymath}
\alpha_k=\frac{3+\sqrt{5}}{2}k+\frac{1}{\frac{3+\sqrt{5}}{2}k+\frac{1}{\frac{3+\sqrt{5}}{2}k+\cdots}}.
\end{displaymath}
For simplicity, however, we consider only the special case with branching parameter $k=1$.

Corresponding to \eqref{eq7.2.14}, we define the column matrices
\begin{equation}\label{eq8.5.1}
\bfu_1=[\{\uparrow_s:j\in J_1^*,s\in I_j^*\}]
\quad\mbox{and}\quad
\bfu_2=[\{\uparrow_s:j\in J_2^*,s\in I_j^*\}].
\end{equation}
Here $J_1^*$, $J_2^*$ and $I_j^*$ denote that the edges are counted with multiplicity.
We also define the column matrices $\bfv_1$ and $\bfv_2$ analogous to \eqref{eq7.2.15}, but their details are not important.
Also, analogous to \eqref{eq7.2.17}, we have
\begin{equation}\label{eq8.5.2}
(\bfA-I)[\{\uparrow_s\}]=\left\{\begin{array}{ll}
\bfu_1+\bfv_1,&\mbox{if $s\in J_1$},\\
\bfu_2+\bfv_2,&\mbox{if $s\in J_2$}.
\end{array}\right.
\end{equation}

We now combine \eqref{eq8.5.1} and \eqref{eq8.5.2}.
For the horizontal street corresponding to~$\bfu_1$, as highlighted in the picture on the left in Figure~8.5.2, we have
\begin{align}\label{eq8.5.3}
(\bfA-I)\bfu_1
&
=(\bfA-I)[\{\uparrow_1,\uparrow_2\}]
+(\bfA-I)[\{\uparrow_3,\uparrow_4,\uparrow_5\}]
\nonumber
\\
&
=2(\bfu_1+\bfv_1)+3(\bfu_2+\bfv_2).
\end{align}
For the horizontal street corresponding to~$\bfu_2$, as highlighted in the picture on the right in Figure~8.5.2, we have
\begin{align}\label{eq8.5.4}
(\bfA-I)\bfu_2
&
=(\bfA-I)[\{\uparrow_1,2\uparrow_2\}]
+(\bfA-I)[\{2\uparrow_3,\uparrow_4,2\uparrow_5\}]
\nonumber
\\
&
=3(\bfu_1+\bfv_1)+5(\bfu_2+\bfv_2).
\end{align}
\begin{displaymath}
\begin{array}{c}
\includegraphics[scale=0.8]{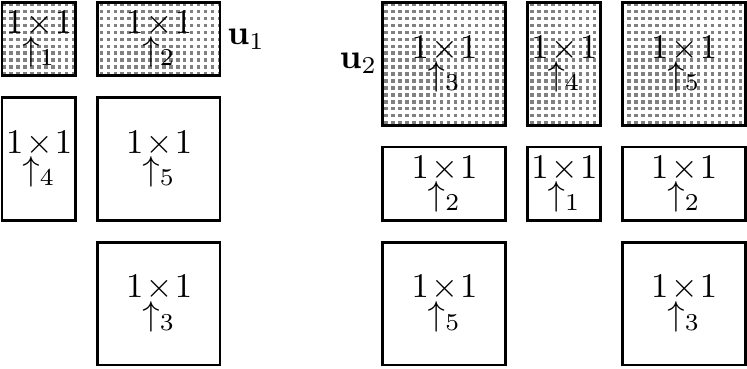}
\vspace{3pt}\\
\mbox{Figure 8.5.2: almost vertical units of type $\uparrow$ in $\bfu_1$ and $\bfu_2$}
\end{array}
\end{displaymath}

It follows from \eqref{eq8.5.3} and \eqref{eq8.5.4} that the street-spreading matrix is given by
\begin{displaymath}
\bfS=\begin{pmatrix}
2&3\\
3&5
\end{pmatrix},
\end{displaymath}
with eigenvalues
\begin{displaymath}
\tau_1=\frac{7+3\sqrt{5}}{2}
\quad\mbox{and}\quad
\tau_2=\frac{7-3\sqrt{5}}{2}.
\end{displaymath}
Using \eqref{eq7.2.38}, the corresponding eigenvalues of $\bfA$ are
\begin{displaymath}
\lambda\left(\frac{7+3\sqrt{5}}{2};\pm\right)=\frac{11+3\sqrt{5}}{4}\pm\frac{1}{4}\left(150+66\sqrt{5}\right)^{1/2}
\end{displaymath}
and
\begin{displaymath}
\lambda\left(\frac{7-3\sqrt{5}}{2};\pm\right)=\frac{11-3\sqrt{5}}{4}\pm\frac{1}{4}\left(150-66\sqrt{5}\right)^{1/2}.
\end{displaymath}
The two largest eigenvalues are therefore
\begin{displaymath}
\lambda_1=\frac{11+3\sqrt{5}}{4}+\frac{1}{4}\left(150+66\sqrt{5}\right)^{1/2}
\end{displaymath}
and
\begin{displaymath}
\lambda_2=\frac{11-3\sqrt{5}}{4}+\frac{1}{4}\left(150-66\sqrt{5}\right)^{1/2}.
\end{displaymath}

Next, as promised earlier, we return to Theorem~\ref{thm8.4.2}(iii) concerning the golden L-surface.

As in Section~\ref{sec8.2}, we shall apply the eigenvalue-based shortline method to compute the irregularity exponent for geodesic flow on this surface. 

What makes the golden L-surface particularly simple is that its two horizontal streets are similar rectangles, and its two vertical streets are also similar rectangles; see Figure~8.5.3.

\begin{displaymath}
\begin{array}{c}
\includegraphics[scale=0.8]{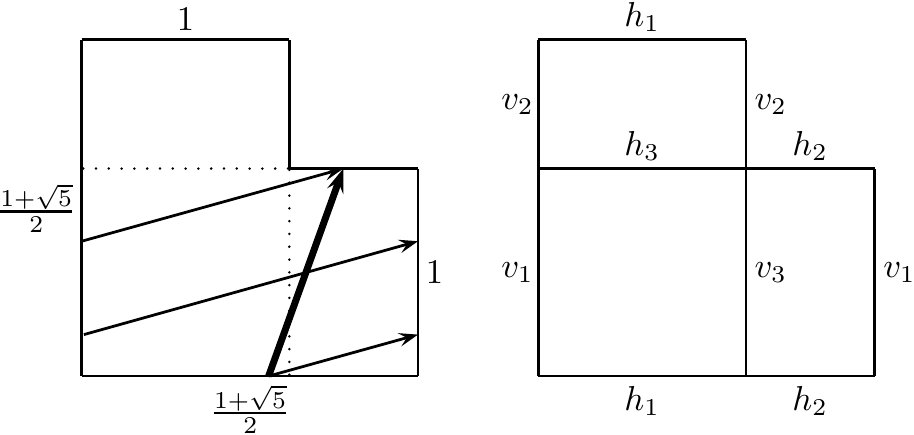}
\vspace{3pt}\\
\mbox{Figure 8.5.3: the golden L-surface, detour crossing and shortcut}
\end{array}
\end{displaymath}

To apply the surplus shortline method, we consider slopes of the special form
\begin{equation}\label{eq8.5.5}
\alpha=\frac{1+\sqrt{5}}{2}a_0+\frac{1}{\frac{1+\sqrt{5}}{2}a_1+\frac{1}{\frac{1+\sqrt{5}}{2}a_2+\cdots}},
\end{equation}
where $a_0,a_1,a_2,\ldots$ are positive integers.
Furthermore, to determine the irregularity exponent explicitly, we need eventual periodicity of the sequence $a_0,a_1,a_2,\ldots.$

The golden L-surface has $3$ faces.

The almost vertical units $h_1h_2,h_1h_3,h_2h_2,h_2h_3,h_3h_1,h_3h_1^*$ can be defined in the same way as shown in Figure~7.1.6.

The almost horizontal units $v_1v_2,v_1v_3,v_2v_2,v_2v_3,v_3v_1,v_3v_1^*$ can be defined in a similar way.

The picture on the left in Figure~8.5.3 illustrates an almost horizontal detour crossing of a horizontal street and its almost vertical shortcut~$h_1h_2$ in the special case when the branching parameter is equal to~$2$.
Using the delete end rule, for a general branching parameter $k\ge1$, the ancestor process can be summarized by
\begin{align}
h_1h_2
&\rightharpoonup v_2v_3,k\times v_3v_1,k\times v_1v_3,
\label{eq8.5.6}
\\
h_1h_3
&\rightharpoonup v_2v_3,k\times v_3v_1,(k-1)\times v_1v_3,
\label{eq8.5.7}
\\
h_2h_2
&\rightharpoonup v_3v_1^*,k\times v_1v_3,(k-1)\times v_3v_1,
\label{eq8.5.8}
\\
h_2h_3
&\rightharpoonup v_3v_1^*,k\times v_1v_3,k\times v_3v_1,
\label{eq8.5.9}
\\
h_3h_1
&\rightharpoonup v_1v_2,(k-1)\times v_2v_2,
\label{eq8.5.10}
\\
h_3h_1^*
&\rightharpoonup v_1v_2,k\times v_2v_2.
\label{eq8.5.11}
\end{align}
These lead to a $6\times6$ transition matrix
\begin{displaymath}
M(k)=\bordermatrix{
&v_1v_2&v_1v_3&v_2v_2&v_2v_3&v_3v_1&v_3v_1^*\cr
h_1h_2&0&k&0&1&k&0\cr 
h_1h_3&0&k-1&0&1&k&0\cr 
h_2h_2&0&k&0&0&k-1&1\cr
h_2h_3&0&k&0&0&k&1\cr
h_3h_1&1&0&k-1&0&0&0\cr
h_3h_1^*&1&0&k&0&0&0}.
\end{displaymath}

Similarly, we can study the ancestor relation of each of the almost horizontal units, again using the delete end rule, and obtain the analogs of
\eqref{eq8.5.6}--\eqref{eq8.5.11}.
These will lead to another $6\time6$ transition matrix.
Since we have listed the almost vertical units and almost horizontal units in lexicographical order, these two $6\times6$ transition matrices are the same.

The eigenvalues of $M(k)$ are two pairs of algebraic conjugates
\begin{align}
&\frac{1+\sqrt{5}}{4}k\pm\left(\left(\frac{1+\sqrt{5}}{4}k\right)^2+1\right)^{1/2},
\nonumber
\\
&\frac{1-\sqrt{5}}{4}k\pm\left(\left(\frac{1-\sqrt{5}}{4}k\right)^2+1\right)^{1/2},
\nonumber
\end{align}
and $(-1\pm\sqrt{3}\ii)/2$ having absolute value~$1$.
Of these, the eigenvalue with the largest absolute value is
\begin{displaymath}
\Lambda 
=\frac{1+\sqrt{5}}{4}k+\left(\left(\frac{1+\sqrt{5}}{4}k\right)^2+1\right)^{1/2}
=\frac{1+\sqrt{5}}{2}k+\frac{1}{\frac{1+\sqrt{5}}{2}k+\frac{1}{\frac{1+\sqrt{5}}{2}k+\cdots}},
\end{displaymath}
\textit{i.e.}, $\Lambda$ is equal to the slope $\alpha$ in \eqref{eq8.5.5} in the special case $a_i=k$ for all $i\ge0$.
The eigenvalue with the second largest absolute value is
\begin{displaymath}
\lambda=\frac{1-\sqrt{5}}{4}k-\left(\left(\frac{1-\sqrt{5}}{4}k\right)^2+1\right)^{1/2}.
\end{displaymath}
The remaining $4$ eigenvalues are irrelevant.

We shall show that the transition matrix $M(k)$ has a conjugate with the form
\begin{equation}\label{eq8.5.12}
P^{-1}M(k)P=\begin{pmatrix}
T&?\\
0&A(k)
\end{pmatrix},
\end{equation}
where
\begin{equation}\label{eq8.5.13}
T=\begin{pmatrix}
\frac{-1-\sqrt{3}\ii}{2}&?
\vspace{2pt}\\
0&\frac{-1+\sqrt{3}\ii}{2}
\end{pmatrix},
\end{equation}
and
\begin{equation}\label{eq8.5.14}
A(k)=\begin{pmatrix}
\frac{(1+\sqrt{5})k}{2}&1&?&?
\vspace{4pt}\\
1&0&?&?\\
0&0&\frac{(1-\sqrt{5})k}{2}&1
\vspace{4pt}\\
0&0&1&0
\end{pmatrix}.
\end{equation}
The description of the matrix $M(k)$ by \eqref{eq8.5.12}--\eqref{eq8.5.14} is extremely convenient.
It reduces the necessary eigenvalue computation of arbitrary products 
\begin{displaymath}
\prod_{i=1}^rM(k_i)
\end{displaymath}
of $6\times6$ matrices with different values of the branching parameter $k_i$ to the much simpler eigenvalue computation of products   
\begin{displaymath}
\prod_{i=1}^r
\begin{pmatrix}
\frac{(1+\sqrt{5})k_i}{2}&1
\vspace{4pt}\\
1&0
\end{pmatrix}
\quad\mbox{and}\quad
\prod_{i=1}^r
\begin{pmatrix}
\frac{(1-\sqrt{5})k_i}{2}&1
\vspace{4pt}\\
1&0
\end{pmatrix}
\end{displaymath}
of $2\times2$ matrices, as the remaining eigenvalues $(-1\pm\sqrt{3}\ii)/2$ are irrelevant.

We now outline the routine deduction of \eqref{eq8.5.12}--\eqref{eq8.5.14}.

We first make use of the fact that $M(k)$ has $2$ eigenvectors that are independent of the branching parameter~$k$.
Together with the eigenvalues, they are
\begin{align}
\lambda_1=\frac{-1-\sqrt{3}\ii}{2},
&\quad\bfv_1=\left(\frac{-1-\sqrt{3}\ii}{2},-1,0,\frac{-1+\sqrt{3}\ii}{2},1,1\right)^T,
\nonumber
\\
\lambda_2=\frac{-1+\sqrt{3}\ii}{2},
&\quad\bfv_2=\left(\frac{-1+\sqrt{3}\ii}{2},-1,0,\frac{-1-\sqrt{3}\ii}{2},1,1\right)^T.
\nonumber
\end{align}
We use a partial diagonalization trick first discussed in \cite[Lemma~4.1.1]{BDY2}.
Let $Q$ be a $6\times6$ invertible matrix such that the first $2$ columns are $\bfv_1,\bfv_2$, and where the remaining columns are
\begin{align}
\bfv_3&=(1,0,-1,-1,-1,0)^T,
\nonumber
\\
\bfv_4&=(0,1,1,-1,1,0)^T,
\nonumber
\\
\bfv_5&=(0,0,0,0,0,-1)^T,
\nonumber
\\
\bfv_6&=(-1,1,1,1,0,0)^T.
\nonumber
\end{align}
Then
\begin{equation}\label{eq8.5.15}
Q^{-1}M(k)Q=\begin{pmatrix}
T&?\\
0&M_4(k)
\end{pmatrix},
\end{equation}
where
\begin{equation}\label{eq8.5.16}
M_4(k)=\begin{pmatrix}
3k-1&2-5k&1&1-3k\\
-k&k-2&1&k-2\\
4k&3-6k&0&2-4k\\
3k&3-4k&-1&3-3k
\end{pmatrix}.
\end{equation}

Let $R$ be a $4\times4$ auxiliary matrix such that its first two columns are
\begin{align}
\bfw_1
&=\left(\frac{\sqrt{5}+1}{2},\frac{\sqrt{5}-3}{2},2,\frac{5-\sqrt{5}}{2}\right)^T,
\nonumber
\\
\bfw_2
&=\left(1,0,\frac{\sqrt{5}+1}{2},1\right)^T.
\nonumber
\end{align}
Routine calculation shows that if $R$ is invertible, then independently of the choice of its third and fourth columns, the conjugate $R^{-1}M_4(k)R$ has the simpler triangular form
\begin{equation}\label{eq8.5.17}
R^{-1}M_4(k)R=\begin{pmatrix}
A_1(k)&?\\
0&A_2(k)
\end{pmatrix},
\end{equation}
where
\begin{equation}\label{eq8.5.18}
A_1(k)=\begin{pmatrix}
\frac{(1+\sqrt{5})k}{2}&1
\vspace{4pt}\\
1&0
\end{pmatrix},
\end{equation}
and $A_2(k)$ is a $2\times2$ matrix that depends on the third and fourth columns of~$R$.
Again with some routine calculation we can find suitable third and fourth columns of $R$ which give
\begin{equation}\label{eq8.5.19}
A_2(k)=\begin{pmatrix}
\frac{(1-\sqrt{5})k}{2}&1
\vspace{4pt}\\
1&0
\end{pmatrix}.
\end{equation}

Clearly \eqref{eq8.5.12}--\eqref{eq8.5.14} follow on combining \eqref{eq8.5.15}--\eqref{eq8.5.19}.

Consider now a geodesic on the golden L-surface with slope $\alpha$ of the form \eqref{eq8.5.5}, where the sequence $a_0,a_1,a_2,\ldots$ of integers is eventually periodic.
Then the shortline of this geodesic has slope $\alpha_1^{-1}$, where
\begin{displaymath}
\alpha_1=\frac{1+\sqrt{5}}{2}a_1+\frac{1}{\frac{1+\sqrt{5}}{2}a_2+\frac{1}{\frac{1+\sqrt{5}}{2}a_3+\cdots}},
\end{displaymath}
and the shortline of this shortline has slope
\begin{displaymath}
\alpha_2=\frac{1+\sqrt{5}}{2}a_2+\frac{1}{\frac{1+\sqrt{5}}{2}a_3+\frac{1}{\frac{1+\sqrt{5}}{2}a_4+\cdots}},
\end{displaymath}
and so on.

\begin{thm}\label{thm8.5.1}
Consider geodesic flow on the golden L-surface with slope $\alpha$ given by \eqref{eq8.5.5}.
If the sequence $a_0,a_1,a_2,\ldots$ has period $k_1,k_2,\ldots,k_r$ eventually, then the irregularity exponent of the geodesic with slope $\alpha$ is equal to
\begin{displaymath}
\frac{\log\vert\lambda\vert}{\log\vert\Lambda\vert},
\end{displaymath}
where $\lambda$ is the eigenvalue with the larger absolute value of the product matrix
\begin{equation}\label{eq8.5.20}
\prod_{i=1}^r
\begin{pmatrix}
\frac{(1-\sqrt{5})k_i}{2}&1
\vspace{4pt}\\
1&0
\end{pmatrix},
\end{equation}
and $\Lambda$ is the eigenvalue with the larger absolute value of the product matrix
\begin{equation}\label{eq8.5.21}
\prod_{i=1}^r
\begin{pmatrix}
\frac{(1+\sqrt{5})k_i}{2}&1
\vspace{4pt}\\
1&0
\end{pmatrix}.
\end{equation}
Alternatively, we have the product formula
\begin{displaymath}
\Lambda=\prod_{i=1}^r\beta_i,
\end{displaymath}
where for every $i=1,\ldots,r$,
\begin{displaymath}
\beta_i=\frac{1+\sqrt{5}}{2}k_i+\frac{1}{\frac{1+\sqrt{5}}{2}k_{i+1}+\frac{1}{\frac{1+\sqrt{5}}{2}k_{i+2}+\cdots}},
\end{displaymath}
corresponding to the periodic sequence
\begin{displaymath}
k_i,k_{i+1},k_{i+2},\ldots,k_r,k_1,k_2,\ldots,k_r,\ldots.
\end{displaymath}
\end{thm}

For the golden L-surface, let us now calculate the eigenvalues of $2$-step transition matrix $\bfA$ by first finding the street-spreading matrix for the slope $\alpha$ given by \eqref{eq8.5.5} with $a_i=1$ for every $i\ge0$.
We shall follow the notation in Section~\ref{sec7.2}, and the details are shown in Figure~8.5.4.

\begin{displaymath}
\begin{array}{c}
\includegraphics[scale=0.8]{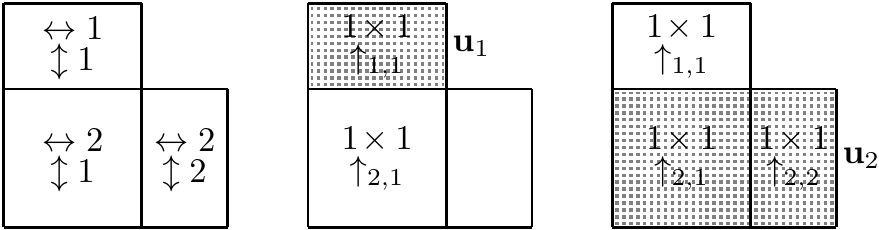}
\vspace{3pt}\\
\mbox{Figure 8.5.4: almost vertical units of type $\uparrow$ in $\bfu_1,\bfu_2$}
\end{array}
\end{displaymath}

Consider the horizontal street corresponding to~$\bfu_1$, as highlighted in the picture in the middle in Figure~8.5.4.
Using \eqref{eq7.2.17}, we have
\begin{align}\label{eq8.5.22}
(\bfA-I)\bfu_1
&
=(\bfA-I)[\{\uparrow_{1,1}\}]+(\bfA-I)[\{\uparrow_{2,1}\}]
\nonumber
\\
&
=(\bfu_1+\bfv_1)+(\bfu_2+\bfv_2).
\end{align}
For the horizontal street corresponding to~$\bfu_2$, as highlighted in the picture on the right in Figure~8.5.4, a similar argument gives
\begin{align}\label{eq8.5.23}
(\bfA-I)\bfu_2
&
=(\bfA-I)[\{\uparrow_{1,1}\}]+(\bfA-I)[\{\uparrow_{2,1},\uparrow_{2,2}\}]
\nonumber
\\
&
=(\bfu_1+\bfv_1)+2(\bfu_2+\bfv_2).
\end{align}

It follows from \eqref{eq8.5.22} and \eqref{eq8.5.23} that the street-spreading matrix is given by
\begin{displaymath}
\bfS=\begin{pmatrix}
1&1\\
1&2
\end{pmatrix},
\end{displaymath}
with eigenvalues
\begin{displaymath}
\tau_1=\frac{3+\sqrt{5}}{2}
\quad\mbox{and}\quad
\tau_2=\frac{3-\sqrt{5}}{2}.
\end{displaymath}
Using \eqref{eq7.2.38}, the corresponding eigenvalues of $\bfA$ are
\begin{displaymath}
\lambda\left(\frac{3+\sqrt{5}}{2};\pm\right)
=\frac{7+\sqrt{5}}{4}\pm\frac{1}{2}\left(\left(\frac{3+\sqrt{5}}{2}\right)^2+4\left(\frac{3+\sqrt{5}}{2}\right)\right)^{1/2}
\end{displaymath}
and
\begin{displaymath}
\lambda\left(\frac{3-\sqrt{5}}{2};\pm\right)
=\frac{7-\sqrt{5}}{4}\pm\frac{1}{2}\left(\left(\frac{3-\sqrt{5}}{2}\right)^2+4\left(\frac{3-\sqrt{5}}{2}\right)\right)^{1/2}.
\end{displaymath}
The two largest eigenvalues are therefore
\begin{displaymath}
\lambda_1=\frac{7+\sqrt{5}}{4}+\frac{1}{2}\left(\left(\frac{3+\sqrt{5}}{2}\right)^2+4\left(\frac{3+\sqrt{5}}{2}\right)\right)^{1/2}
\end{displaymath}
and
\begin{displaymath}
\lambda_2=\frac{7-\sqrt{5}}{4}+\frac{1}{2}\left(\left(\frac{3-\sqrt{5}}{2}\right)^2+4\left(\frac{3-\sqrt{5}}{2}\right)\right)^{1/2}.
\end{displaymath}

For the corresponding $1$-step transition matrix, we shall determine the eigenvalues $\Lambda$ and $\lambda$ by using \eqref{eq8.5.20} and \eqref{eq8.5.21} with $r=1$ and $k_1=1$.
The eigenvalue with the larger absolute value of the matrices
\begin{displaymath}
\begin{pmatrix}
\frac{1+\sqrt{5}}{2}&1
\vspace{4pt}\\
1&0
\end{pmatrix}
\quad\mbox{and}\quad
\begin{pmatrix}
\frac{1-\sqrt{5}}{2}&1
\vspace{4pt}\\
1&0
\end{pmatrix}
\end{displaymath}
are respectively
\begin{displaymath}
\Lambda=\frac{1+\sqrt{5}}{4}+\frac{1}{2}\left(\left(\frac{1+\sqrt{5}}{2}\right)^2+4\right)^{1/2}
\end{displaymath}
and
\begin{displaymath}
\lambda=\frac{1-\sqrt{5}}{4}+\frac{1}{2}\left(\left(\frac{1-\sqrt{5}}{2}\right)^2+4\right)^{1/2}.
\end{displaymath}
Note that $\lambda_1=\Lambda^2$ and $\lambda_2=\lambda^2$.

We complete this section by calculating the eigenvalues of the $2$-step transition matrix $\bfA$ for the golden brick, first introduced at the end of Section~\ref{sec8.4}, for a geodesic with a suitable slope.

For the $4$-copy translation surface of the golden brick, note that the streets of width $1$ has length $2+2h$, and so normalized length $2+2h$, while the streets of width $h$ has length~$4$, and so normalized length
\begin{displaymath}
\frac{4}{h}=2(2+2h),
\end{displaymath}
so that $h^*=v^*=2(2+2h)$.
To visualize the $4$-copy version of the golden brick, we refer the reader to Figures 7.2.12--7.2.14 which illustrate the $4$-copy version of the surface of the unit cube.
We obtain the $4$-copy version of the golden brick if we shorten the edges $a_1,a_2,a_3,a_4$, and those in between, in Figure~7.2.12 from length $1$ to length~$h$.
Then the horizontal streets $3,6$ and vertical streets $2,6$ in Figure~7.2.13 still have length $4$ but now have width $h$ instead of~$1$, and we can obtain an analog of Figure~7.2.14 with thinner rows $3,6$ and columns $2,6$.

For simplicity, we consider a geodesic with slope
\begin{displaymath}
\alpha=2(2+2h)+\frac{1}{2(2+2h)+\frac{1}{2(2+2h)+\frac{1}{2(2+2h)+\cdots}}}.
\end{displaymath}

As before, we define the matrices $\bfu_i,\bfv_i$, $i=1,\ldots,6$, according to \eqref{eq7.2.14} and \eqref{eq7.2.15}.
Then Figures 8.5.5--8.5.7 below are the analogs of Figures 7.2.15 and 7.2.16 for the surface of the cube.

\begin{displaymath}
\begin{array}{c}
\includegraphics[scale=0.8]{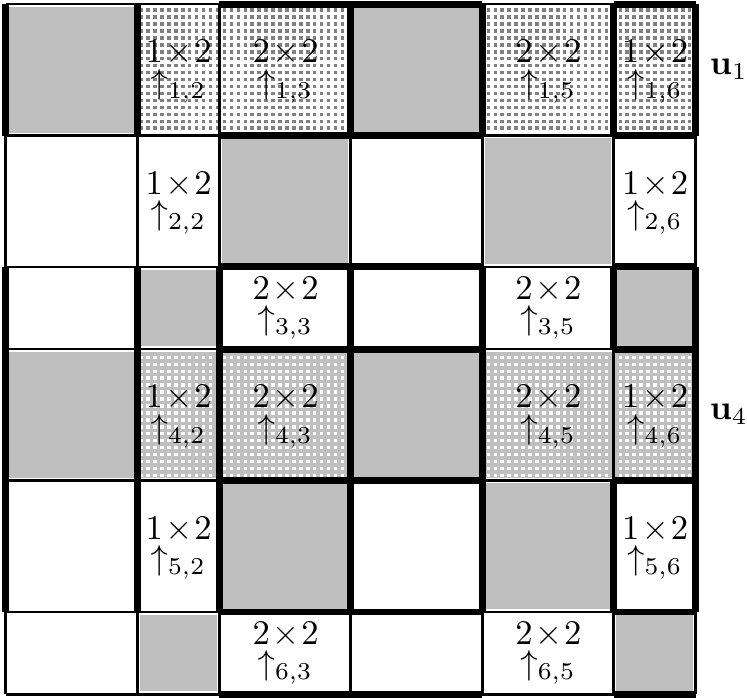}
\vspace{3pt}\\
\mbox{Figure 8.5.5: almost vertical units of type $\uparrow$ in $\bfu_1,\bfu_4$}
\end{array}
\end{displaymath}

For the horizontal streets corresponding to $\bfu_1,\bfu_4$, as highlighted in Figure~8.5.5, we have, applying \eqref{eq7.2.17},
\begin{align}\label{eq8.5.24}
(\bfA-I)\bfu_1
&
=(\bfA-I)\bfu_4
\nonumber
\\
&
=(\bfA-I)[\{2\uparrow_{1,2},4\uparrow_{1,3},4\uparrow_{1,5},2\uparrow_{1,6}\}]+(\bfA-I)[\{2\uparrow_{2,2},2\uparrow_{2,6}\}]
\nonumber
\\
&\qquad
+(\bfA-I)[\{4\uparrow_{3,3},4\uparrow_{3,5}\}]
+(\bfA-I)[\{2\uparrow_{4,2},4\uparrow_{4,3},4\uparrow_{4,5},2\uparrow_{4,6}\}]
\nonumber
\\
&\qquad
+(\bfA-I)[\{2\uparrow_{5,2},2\uparrow_{5,6}\}]
+(\bfA-I)[\{4\uparrow_{6,3},4\uparrow_{6,5}\}]
\nonumber
\\
&
=12(\bfu_1+\bfv_1)+4(\bfu_2+\bfv_2)+8(\bfu_3+\bfv_3)
\nonumber
\\
&\qquad
+12(\bfu_4+\bfv_4)+4(\bfu_5+\bfv_5)+8(\bfu_6+\bfv_6).
\end{align}
\begin{displaymath}
\begin{array}{c}
\includegraphics[scale=0.8]{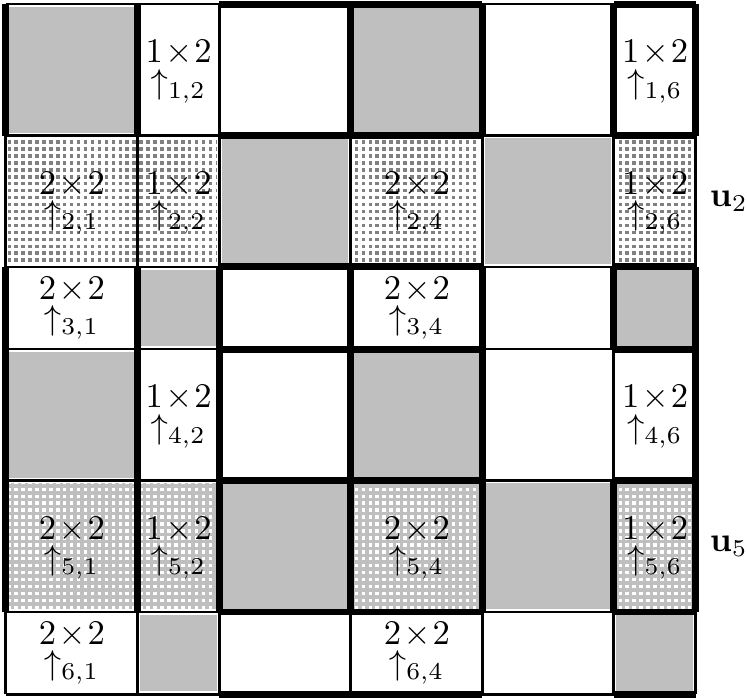}
\vspace{3pt}\\
\mbox{Figure 8.5.6: almost vertical units of type $\uparrow$ in $\bfu_2,\bfu_5$}
\end{array}
\end{displaymath}

For the horizontal streets corresponding to~$\bfu_2,\bfu_5$, as highlighted in Figure~8.5.6, a similar argument gives
\begin{align}\label{eq8.5.25}
(\bfA-I)\bfu_2
&
=(\bfA-I)\bfu_5
\nonumber
\\
&
=4(\bfu_1+\bfv_1)+12(\bfu_2+\bfv_2)+8(\bfu_3+\bfv_3)
\nonumber
\\
&\qquad
+4(\bfu_4+\bfv_4)+12(\bfu_5+\bfv_5)+8(\bfu_6+\bfv_6).
\end{align}
\begin{displaymath}
\begin{array}{c}
\includegraphics[scale=0.8]{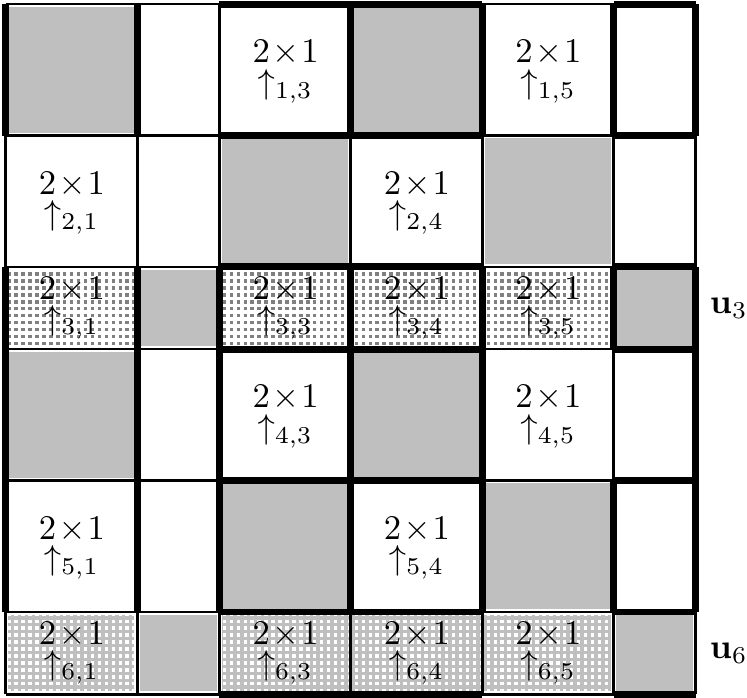}
\vspace{3pt}\\
\mbox{Figure 8.5.7: almost vertical units of type $\uparrow$ in $\bfu_3,\bfu_6$}
\end{array}
\end{displaymath}

For the horizontal streets corresponding to~$\bfu_3,\bfu_6$, as highlighted in Figure~8.5.7, a similar argument gives
\begin{align}\label{eq8.5.26}
(\bfA-I)\bfu_3
&
=(\bfA-I)\bfu_6
\nonumber
\\
&
=4(\bfu_1+\bfv_1)+4(\bfu_2+\bfv_2)+8(\bfu_3+\bfv_3)
\nonumber
\\
&\qquad
+4(\bfu_4+\bfv_4)+4(\bfu_5+\bfv_5)+8(\bfu_6+\bfv_6).
\end{align}

It follows from \eqref{eq8.5.24}--\eqref{eq8.5.26} that the street-spreading matrix is given by
\begin{displaymath}
\bfS=\begin{pmatrix}
12&4&4&12&4&4\\
4&12&4&4&12&4\\
8&8&8&8&8&8\\
12&4&4&12&4&4\\
4&12&4&4&12&4\\
8&8&8&8&8&8
\end{pmatrix},
\end{displaymath}
with non-zero eigenvalues $\tau_1=24+8\sqrt{5}$, $\tau_2=16$ and $\tau_3=24-8\sqrt{5}$.
Using \eqref{eq7.2.38}, the corresponding eigenvalues of $\bfA$ are
\begin{align}
\lambda(24+8\sqrt{5};\pm)
&=13+4\sqrt{5}\pm\left(248+104\sqrt{5}\right)^{1/2},
\nonumber
\\
\lambda(16;\pm)
&=9\pm4\sqrt{5},
\nonumber
\\
\lambda(24-8\sqrt{5};\pm)&
=13-4\sqrt{5}\pm\left(248-104\sqrt{5}\right)^{1/2}.
\nonumber
\end{align}

%%%%%%%%%%
%
% SECTION 8.6
%
%%%%%%%%%%

\subsection{Surfaces tiled with congruent equilateral triangles}\label{sec8.6}

We wish to complete our study of geodesic flow on the surface of each of the $5$ platonic solids.

We have already discussed in \cite{BDY1} the superdensity of geodesic flow on the regular tetrahedron surface which is integrable.
Earlier in this paper, we have also discussed the superdensity of geodesic flow on the cube surface and the regular dodacahedron surface which are non-integrable.

The last two examples of the $5$ platonic solids are the regular octahedron and the regular icosahedron.
Both of these surfaces belong to the large class of flat surfaces where the faces are congruent equilateral triangles.
We refer to such a surface as a \textit{polytriangle surface}.
Every polytriangle surface has the crucial property that the union of any two equilateral triangle faces sharing an edge forms a $60$-degree rhombus, which turns out to be a perfect substitute for squares.
Thus polytriangle surfaces are perfect analogs of polysquare translation surfaces, and both versions of the shortline method work for this class.

We give here a detailed discussion of geodesic flow on the regular octahedron surface, as shown in Figure~8.6.1.
Geodesic flow on the icosahedron surface goes in a similar way, and we omit it.

\begin{displaymath}
\begin{array}{c}
\includegraphics[scale=0.8]{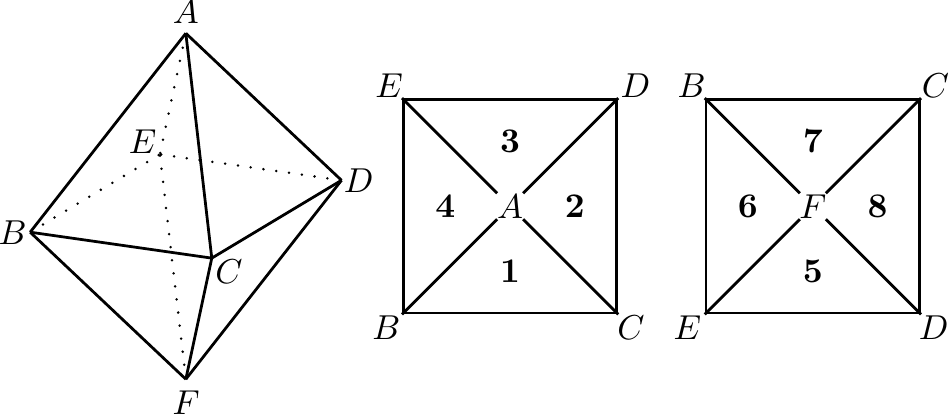}
\vspace{3pt}\\
\mbox{Figure 8.6.1: the regular octahedron, top view and bottom view}
\end{array}
\end{displaymath}

For our convenience, we label the faces as in Figure~8.6.1.
Thus we have face and vertex pairings given by
\begin{displaymath}
\begin{array}{llll}
(\mathbf{1},ABC),
&\quad(\mathbf{2},ACD),
&\quad(\mathbf{3},ADE),
&\quad(\mathbf{4},AEB),
\vspace{4pt}\\
(\mathbf{5},FED),
&\quad(\mathbf{6},FBE),
&\quad(\mathbf{7},FCB),
&\quad(\mathbf{8},FDC),
\end{array}
\end{displaymath}
where the three vertices of each triangle face are given in anticlockwise order.
Note that we have four pairs of vertex-disjoint and parallel faces, given by
\begin{displaymath}
(\mathbf{1},\mathbf{5}),
\quad(\mathbf{2},\mathbf{6}),
\quad(\mathbf{3},\mathbf{7}),
\quad(\mathbf{4},\mathbf{8}).
\end{displaymath}

Between any vertex-disjoint pair of faces we have a \textit{street} made up of $6$ triangle faces.
In Figure~8.6.2, we show $4$ different nets of the regular octahedron, together with a surplus detour crossing of a geodesic in a given direction.

\begin{displaymath}
\begin{array}{c}
\includegraphics[scale=0.8]{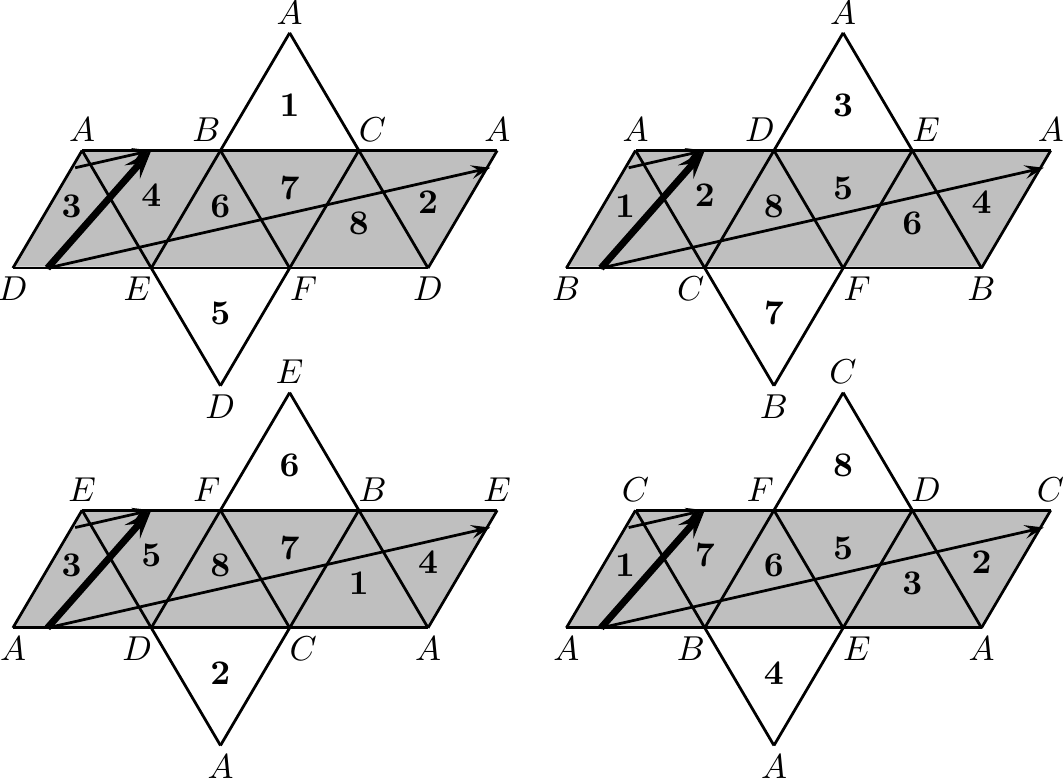}
\vspace{3pt}\\
\mbox{Figure 8.6.2: nets of the regular octahedron surface,}
\\
\mbox{detour crossings and shortcuts}
\end{array}
\end{displaymath}

The picture on the top left of Figure~8.6.2 is a net of the regular octahedron surface.
The long parallelogram in the middle is a street between faces $\mathbf{1}$ and~$\mathbf{5}$.

More precisely, the cycle of triangle faces $\mathbf{3},\mathbf{4},\mathbf{6},\mathbf{7},\mathbf{8},\mathbf{2}$ is the street, with the left and right edges of the parallelogram identified.
This street also has a natural decomposition into $60$-degree rhombi formed from the union of the pairs $(\mathbf{3},\mathbf{4})$, $(\mathbf{6},\mathbf{7})$ and
$(\mathbf{8},\mathbf{2})$.

The detour crossings in all $4$ different nets can be thought of as all going in the same direction.

The copies on the top left and bottom right have parallel edges~$AB$, while the copies on the top right and bottom left have parallel edges~$AD$.
On the other hand, the two copies on the left have parallel edges~$CA$, while the two copies on the right have parallel edges~$EA$.

Note that each shaded triangle in Figure~8.6.2 has $2$ more copies of itself rotated by $120$ degrees and $240$ degrees.
We therefore conclude that a geodesic crosses every triangle face in $3$ different directions, at $120$ degrees to each other.

Note also that every triangle face is complemented by $3$ other triangle faces to form $60$-degree rhombi.
For instance, for the triangle face $\mathbf{3}$, we have the pair $(\mathbf{3},\mathbf{4})$ in the top left, the pair $(\mathbf{3},\mathbf{5})$ in the bottom left, and the pair
$(\mathbf{3},\mathbf{2})$ in the bottom right.

Since a $60$-degree rhombus is just a tilted analog of a square, and can be viewed as one, it is straightforward to adapt the shortline method.
To obtain some special slopes, we consider the analogy illustrated in Figure~8.6.3.

\begin{displaymath}
\begin{array}{c}
\includegraphics[scale=0.8]{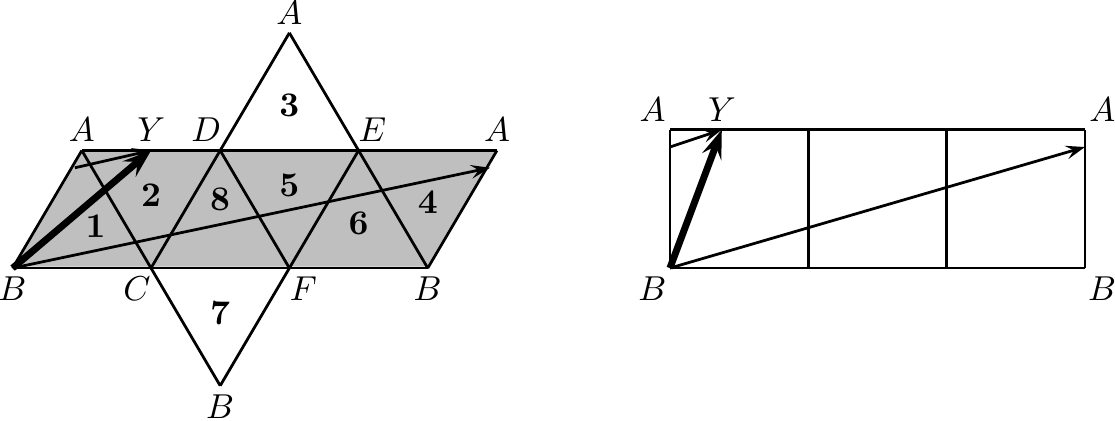}
\vspace{3pt}\\
\mbox{Figure 8.6.3: a particular geodesic and mutual shortlines}
\end{array}
\end{displaymath}

For the polysquare model in the picture on the right, it is clear that we consider an almost horizontal geodesic of slope $\alpha_k^{-1}$, where
\begin{equation}\label{eq8.6.1}
\alpha_k=[3k;3k,3k,3k,\ldots]=\frac{\sqrt{9k^2+4}+3k}{2}.
\end{equation}
Then the shortline has slope~$\alpha_k$.
Returning to the original polytriangle surface, a suitable angle corresponding to $\alpha_k$ can be determined.
In particular, if the length of the edge $DE$ is equal to~$1$, then the length of the segment $AY$ should be equal to~$\alpha_k^{-1}$.
Elementary calculation shows that the detour crossing in the picture on the left in Figure~8.6.3 must have slope
\begin{equation}\label{eq8.6.2}
\beta_k^{-1}=\frac{\sqrt{3}}{2(3+\frac{1}{2}+\alpha_k^{-1})}=\frac{\sqrt{3}}{7+2\alpha_k^{-1}},
\end{equation}
which is $\sqrt{3}$ times a quadratic irrational.

Indeed, the case of the regular octahedron surface illustrates very well the whole class of polytriangle surfaces.

We next turn our attention to $60$-degree rhombus billiard.
We use the same idea as for the regular octagon billiard and regular pentagon billiard.
We join up rhombi in such a way that neighboring rhombi are reflections of each other, and end up with a double ring of $6$ rhombi, as shown in Figure~8.6.4.
It is a double ring, as the third rhombus does not join up with the first, due to the edges $a_1$ and $b_2$ having different directions.
This is sometimes known as the \textit{translation surface for $60$-degree rhombus billiard}.
Note that the edge labellings in Figure~8.6.4 are obtained by following the convention discussed in the Remark after Figure~8.3.3.
Identifying the boundary edges $a_1,b_2,c_1,c_2,c_3,d_1,d_2,d_3$, we obtain a polytriangle surface, which is intuitively a \textit{double hexagon}.

\begin{displaymath}
\begin{array}{c}
\includegraphics[scale=0.8]{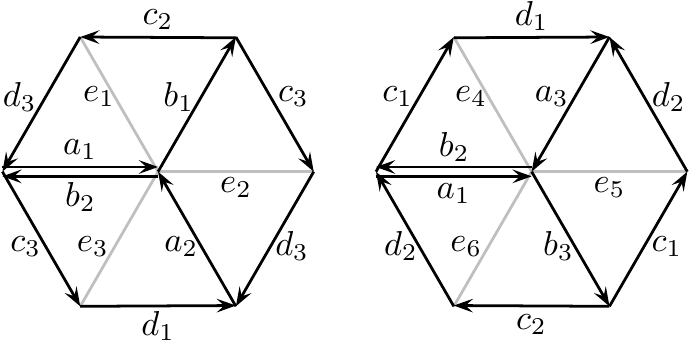}
\vspace{3pt}\\
\mbox{Figure 8.6.4: the translation surface of $60$-degree rhombus billiard}
\\
\mbox{with edge labellings}
\end{array}
\end{displaymath}

Furthermore, we label the $6$ rhombi as in Figure~8.6.5.

\begin{displaymath}
\begin{array}{c}
\includegraphics[scale=0.8]{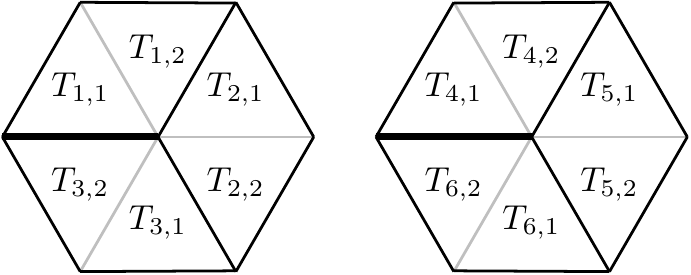}
\vspace{3pt}\\
\mbox{Figure 8.6.5: the faces of the translation surface of $60$-degree rhombus billiard}
\end{array}
\end{displaymath}

It is not difficult to see that $60$-degree rhombus billiard is equivalent to $1$-direction geodesic flow on the polytriangle surface defined by the double hexagon shown in Figures 8.6.4 and~8.6.5.
We refer to this polytriangle surface as the RB-surface, or the rhombus billiard surface.

We can list the streets of the RB-surface explicitly.

The long parallelogram in the middle of Figure~8.6.6 is a street of the RB-surface, comprising the cycle
\begin{equation}\label{eq8.6.3}
T_{1,1},T_{1,2},T_{2,1},T_{3,2},T_{3,1},T_{2,2}
\end{equation}
of $6$ equilateral triangles.
The key fact is that
this street has a natural decomposition into $60$-degree rhombi
\begin{displaymath}
T_{1,1}\cup T_{1,2},
\quad
T_{2,1}\cup T_{3,2},
\quad
T_{3,1}\cup T_{2,2}.
\end{displaymath}
\begin{displaymath}
\begin{array}{c}
\includegraphics[scale=0.8]{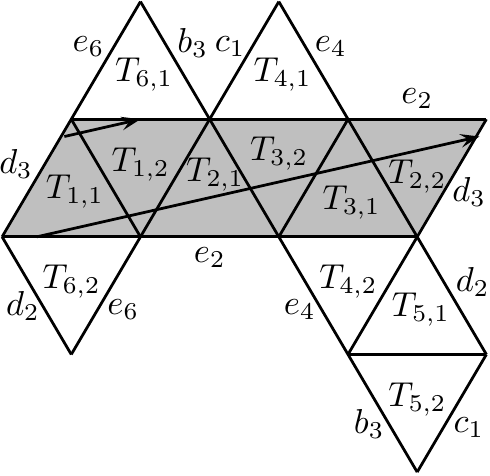}
\vspace{3pt}\\
\mbox{Figure 8.6.6: a net and a street with detour crossing of the RB-surface}
\end{array}
\end{displaymath}

Beside the street \eqref{eq8.6.3}, there are $5$ more streets of the RB-surface, made up of the cycles
\begin{align}
&
T_{1,2},T_{2,1},T_{2,2},T_{1,1},T_{6,2},T_{6,1},
\label{eq8.6.4}
\\
&
T_{2,1},T_{2,2},T_{3,1},T_{4,2},T_{4,1},T_{3,2},
\label{eq8.6.5}
\\
&
T_{3,1},T_{3,2},T_{4,1},T_{5,2},T_{5,1},T_{4,2},
\label{eq8.6.6}
\\
&
T_{4,1},T_{4,2},T_{5,1},T_{6,2},T_{6,1},T_{5,2},
\label{eq8.6.7}
\\
&
T_{5,1},T_{5,2},T_{6,1},T_{1,2},T_{1,1},T_{6,2}.
\label{eq8.6.8}
\end{align}
However, we only need some of these streets.

Figure~8.6.7 gives a description of the RB-surface as a polyrhombus translation surface.

\begin{displaymath}
\begin{array}{c}
\includegraphics[scale=0.8]{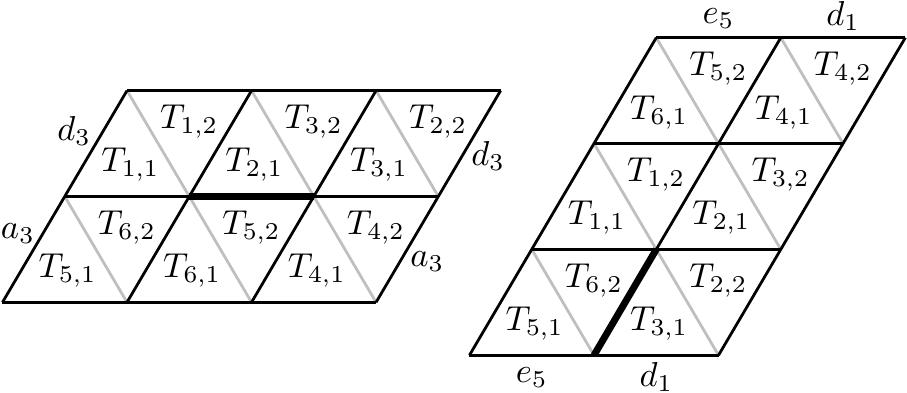}
\vspace{3pt}\\
\mbox{Figure 8.6.7: the RB-surface as a polyrhombus translation surface}
\end{array}
\end{displaymath}

In the picture on the left, we see $2$ horizontal streets.
The top horizontal street is \eqref{eq8.6.3}, while the bottom horizontal street is \eqref{eq8.6.7}.
In the picture on the right, we see $2$ tilted streets that are the analogs of vertical streets in polysquare translation surfaces.
The left tilted street is \eqref{eq8.6.8}, while the right tilted street is \eqref{eq8.6.5}.

Noting the similarity between Figure~8.6.3 and Figure~8.6.6, we conclude that for the shortline method to work, the detour crossing in Figure~8.6.6 must have slope
$\beta_k^{-1}$, given by \eqref{eq8.6.1} and \eqref{eq8.6.2}.

We have the following analog of Theorem~\ref{thm8.1.1}(iii)--(iv).

\begin{thm}\label{thm8.6.1}
\emph{(i)}
For any arbitrary polytriangle surface, there exist infinitely many explicit slopes, depending on the surface, such that any half-infinite geodesic with such a slope exhibits superdensity.

\emph{(ii)}
For any arbitrary table of polytriangle shape, there exist infinitely many explicit slopes, depending on the shape of the table, such that any half-infinite billiard orbit having such an initial slope exhibits superdensity.

\emph{(iii)}
For infinitely many of these slopes in parts \emph{(i)} and \emph{(ii)} that give rise to superdensity, we can explicitly compute the corresponding irregularity exponents which describe the time-quantitative aspects of equidistribution.
\end{thm}

For the RB-surface, we now attempt to find the eigenvalues of its $2$-step transition matrix~$\bfA$.
This surface can be viewed as a polyrhombus translation surface, and we shall use an analogy with a corresponding polysquare translation surface as shown in Figure~8.6.8.

\begin{displaymath}
\begin{array}{c}
\includegraphics[scale=0.8]{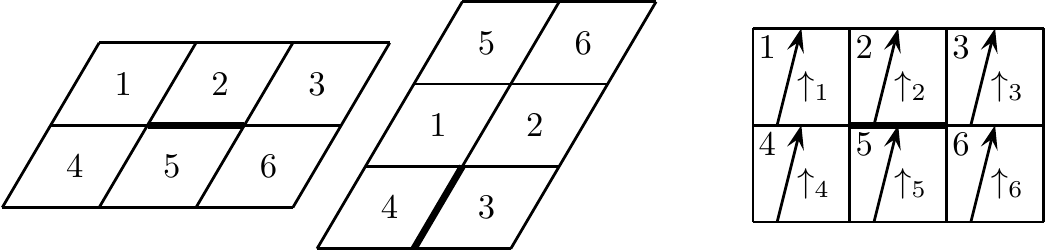}
\vspace{3pt}\\
\mbox{Figure 8.6.8: analogy of the RB-surface with a polysquare translation surface}
\end{array}
\end{displaymath}

For simplicity of notation, the rhombi are renumbered $1,2,3,4,5,6$, as shown in the picture on the left in Figure~8.6.8, and we consider the corresponding polysquare translation surface in the picture on the right.
Here the horizontal streets are $1,2,3$ and $4,5,6$, while the vertical streets are $5,1,4$ and $6,2,3$.
Both horizontal streets and both vertical streets have lengths~$3$.
Note here that distinct square faces can fall into the same horizontal street and the same vertical street, so we shall adapt our notation from Section~\ref{sec7.2} for the determination of the street-spreading matrix.

We also show in the picture on the right in Figure~8.6.8 the almost vertical units of type $\uparrow$ in the polysquare translation surface.
We have not shown the almost vertical units of type $\nuparrow$ here.

Let
\begin{displaymath}
J_1=\{1,2,3\}
\quad\mbox{and}\quad
J_2=\{4,5,6\}
\end{displaymath}
denote the horizontal streets, and let
\begin{displaymath}
I_1=I_4=I_5=\{1,4,5\}
\quad\mbox{and}\quad
I_2=I_3=I_6=\{2,3,6\}
\end{displaymath}
denote the vertical streets.

For the polysquare translation surface, we thus consider slopes of the form
\begin{displaymath}
\alpha_k=3k+\frac{1}{3k+\frac{1}{3k+\cdots}}.
\end{displaymath}
For simplicity, however, we consider only the special case with branching parameter $k=1$.

Corresponding to \eqref{eq7.2.14}, we define the column matrices
\begin{equation}\label{eq8.6.9}
\bfu_1=[\{\uparrow_s:j\in J_1^*,s\in I_j^*\}]
\quad\mbox{and}\quad
\bfu_2=[\{\uparrow_s:j\in J_2^*,s\in I_j^*\}].
\end{equation}
Here $J_1^*$, $J_2^*$ and $I_j^*$ denotes that the edges are counted with multiplicity.
We also define the column matrices $\bfv_1$ and $\bfv_2$ analogous to \eqref{eq7.2.15}, but their details are not important.
Also, analogous to \eqref{eq7.2.17}, we have
\begin{equation}\label{eq8.6.10}
(\bfA-I)[\{\uparrow_s\}]=\left\{\begin{array}{ll}
\bfu_1+\bfv_1,&\mbox{if $s\in J_1$},\\
\bfu_2+\bfv_2,&\mbox{if $s\in J_2$}.
\end{array}\right.
\end{equation}

We now combine \eqref{eq8.6.9} and \eqref{eq8.6.10}.
For the horizontal street corresponding to~$\bfu_1$, as highlighted in the picture on the left in Figure~8.6.9, we have
\begin{align}\label{eq8.6.11}
(\bfA-I)\bfu_1
&
=(\bfA-I)[\{\uparrow_1,2\uparrow_2,2\uparrow_3\}]
+(\bfA-I)[\{\uparrow_4,\uparrow_5,2\uparrow_6\}]
\nonumber
\\
&
=5(\bfu_1+\bfv_1)+4(\bfu_2+\bfv_2).
\end{align}
For the horizontal street corresponding to~$\bfu_2$, as highlighted in the picture on the right in Figure~8.6.9, we have
\begin{align}\label{eq8.6.12}
(\bfA-I)\bfu_2
&
=(\bfA-I)[\{2\uparrow_1,\uparrow_2,\uparrow_3\}]
+(\bfA-I)[\{2\uparrow_4,2\uparrow_5,\uparrow_6\}]
\nonumber
\\
&
=4(\bfu_1+\bfv_1)+5(\bfu_2+\bfv_2).
\end{align}
\begin{displaymath}
\begin{array}{c}
\includegraphics[scale=0.8]{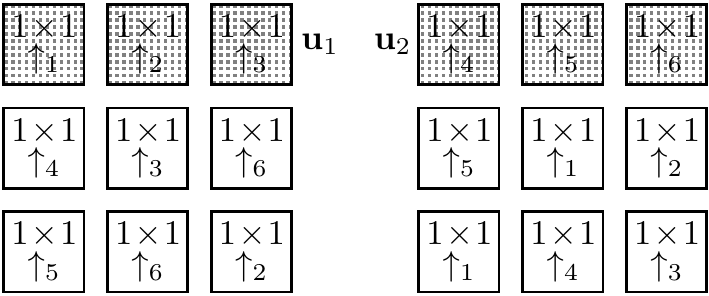}
\vspace{3pt}\\
\mbox{Figure 8.6.9: almost vertical units of type $\uparrow$ in $\bfu_1$ and $\bfu_2$}
\end{array}
\end{displaymath}

It follows from \eqref{eq8.6.11} and \eqref{eq8.6.12} that the street-spreading matrix is given by
\begin{displaymath}
\bfS=\begin{pmatrix}
5&4\\
4&5
\end{pmatrix}.
\end{displaymath}
The eigenvalues of $\bfS$ are $\tau_1=9$ and $\tau_2=1$.
Using \eqref{eq7.2.38}, the corresponding eigenvalues of $\bfA$ are
\begin{displaymath}
\lambda(9;\pm)=\frac{11}{2}\pm\frac{1}{2}\sqrt{117}
\quad\mbox{and}\quad
\lambda(1;\pm)=\frac{3}{2}\pm\frac{1}{2}\sqrt{5}.
\end{displaymath}
The two largest eigenvalues are therefore
\begin{displaymath}
\lambda_1=\frac{11}{2}+\frac{1}{2}\sqrt{117}
\quad\mbox{and}\quad
\lambda_2=\frac{3}{2}+\frac{1}{2}\sqrt{5}.
\end{displaymath}

We have studied geodesics and billiards on many flat translation surfaces, for which both versions of the shortline method work, and identified many explicit slopes that yield superdensity and for which we can compute the irregularity exponents.
Such surfaces include polysquare translation surfaces, street-rational polyrectangle translation surfaces and street-rational polyparallelogram translation surfaces.
They also include regular polygon surfaces, L-staircase surfaces, and polytriangle surfaces. 
We may refer to them under the general name of \textit{street-rational surfaces}.  

%%%%%%%%%%
%
% SECTION 8.7
%
%%%%%%%%%%

\subsection{The shortline method works for all Veech surfaces, and beyond}\label{sec8.7}

The title of this section summarizes Sections \ref{sec8.1}--\ref{sec8.6}, that the surplus shortline method works for all Veech surfaces, providing time-quantitative results like superdensity and uniformity with explicit values of the irregularity exponents.

Note that we have not defined here the concept of \textit{Veech surfaces}.

As motivation, consider an arbitrary polysquare translation surface~$\PPP$, together with $1$-direction geodesic flow on it.
We now consider the following two statements:

(1)
The surface $\PPP$ has a street-rational decomposition in any arbitrary direction with rational slope.

(2)
The surface $\PPP$ is optimal, \textit{i.e.}, $1$-direction geodesic flow on $\PPP$ exhibits uniform-periodic dichotomy.

Perhaps the best way to describe Veech surfaces is to say that they are translation surfaces that exhibit properties like (1) and~(2).
Unfortunately, the precise definition is a rather technical algebraic one, and we do not really need to know that here.
The interested reader is referred to the introductory article~\cite{HS}.

Instead, our wish is to have a list of known examples of Veech surfaces, and some of the key properties of such surfaces.
In particular, the following two properties established by Veech~\cite{V1} are analogous to (1) and (2) respectively.

\begin{theorema}
Suppose that $S$ is a Veech surface, and that $\bfv$ is an arbitrary direction such that there exists a non-empty finite geodesic segment on $S$ that goes between two not necessarily distinct singularities.
Then $S$ has a street-rational decomposition in this direction.
\end{theorema}

\begin{theoremb}
Every Veech surface $S$ is optimal, \textit{i.e.}, $1$-direction geodesic flow on $S$ exhibits uniform-periodic dichotomy.
\end{theoremb}

\begin{remark}
Note that Veech and other authors use the term \textit{cylinder} instead of the term \textit{street}.
\end{remark}

Classifying the Veech surfaces is a very difficult problem.
Early examples are the translation surfaces of regular polygon billiards discovered by Veech, together with a few related examples.
Progress has been slow, despite the effort of many.
The subject is on the borderline of many fields, including ergodic theory, topology and algebraic geometry.
The goal is to understand \textit{optimal dynamics}, a time-qualitative property of the long-term behavior of orbits, and the existing methods do not seem to say anything about the time-quantitative aspects.
Our goal here is to complement the time-qualitative work by making some time-quantitative statements on the long-term behavior of orbits.

Indeed, our surplus shortline method works for all Veech surfaces, and the reason is very simple.
For the success of our method, we need a translation surface with street-rational decomposition in only \textit{two} different directions, and Theorem~A gives far more than that.

The class of Veech surfaces include the following:

(i)
polysquare surfaces including the flat torus, and polytriangle surfaces;

(ii)
translation surfaces of regular polygon billiards and some related systems such as infinitely many special triangle billiards, including the regular decagon surface and the Ward system;

(iii)
Calta--McMullen L-staircases;

(iv)
the class of cathedral surfaces;

(v)
the Bouw--M\"{o}ller family;

(vi)
isolated triangle billiards with angles $(\pi/4,\pi/3,5\pi/12)$, $(2\pi/9,\pi/3,4\pi/9)$ and $(\pi/5,\pi/3,7\pi/15)$;

(vii)
and new Veech surfaces via the covering construction of some old ones.

This list is not exhaustive.
There are a few more countable families of quadrilateral billiards, and so on, discovered very recently.
Nevertheless, it does appear that the most elegant examples among these remain the regular polygon billiards discovered by Veech~\cite{V1} in 1989.

In view of (vii) above, it suffices to know all the primitive Veech surfaces.

Here we do not say anything about the Ward system or the Bouw--M\"{o}ller family.
However, for illustration, we elaborate a little on the class of \textit{cathedral surfaces}, first introduced by McMullen, Mukamel and Wright~\cite{MMW} in 2017.

Figure~8.7.1 shows two identical copies of the cathedral polygon $\cp(a,b)$.
The polygon is symmetric across the horizontal axis, and each boundary edge either is horizontal or vertical, or has slope~$\pm1$.
To convert the polygon into a translation surface, the identification of boundary horizontal and vertical edges comes from perpendicular translation, while the identification of edges of slope $\pm1$ are indicated by the decomposition into horizontal streets $H_1,\ldots,H_5$ and the decomposition into vertical streets $V_1,\ldots,V_5$.
Thus the cathedral surface is well defined, and has $5$ horizontal streets and $5$ vertical streets.

\begin{displaymath}
\begin{array}{c}
\includegraphics[scale=0.8]{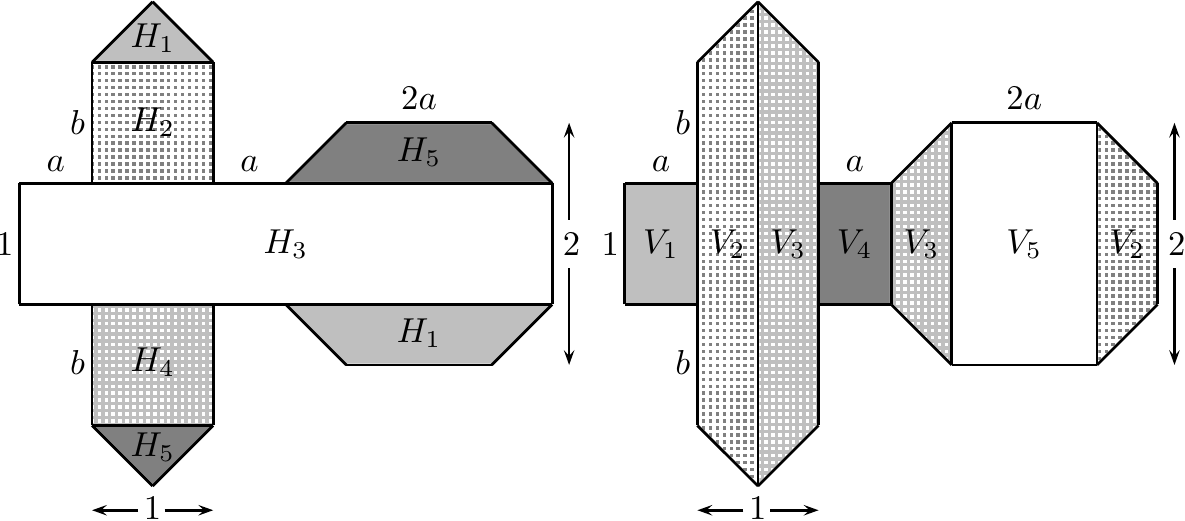}
\vspace{3pt}\\
\mbox{Figure 8.7.1: horizontal and vertical streets of $\cp(a,b)$}
\end{array}
\end{displaymath}

Let
\begin{equation}\label{eq8.7.1}
a=r_1\sqrt{d}+r_2>0
\quad\mbox{and}\quad
b=3r_1\sqrt{d}-3r_2-\frac{3}{2}>0,
\end{equation}
where $d\ge2$ is a square-free integer, and $r_1,r_2$ are rational numbers.
Remarkably, as long as \eqref{eq8.7.1} is satisfied, geodesic flow on the cathedral surface $\cp(a,b)$ is optimal, as shown by McMullen, Mukamel and Wright.

We now verify that under the condition \eqref{eq8.7.1}, the cathedral surface $\cp(a,b)$ is a street-rational polyrectangle translation surface.

Consider first the horizontal streets $H_1,\ldots,H_5$ as indicated in the picture on the left in Figure~8.7.1.
Here the ratio of the (horizontal) lengths and (vertical) widths of these streets are respectively
\begin{displaymath}
\frac{1+2a}{1/2},\quad
\frac{1}{b},\quad
\frac{2+4a}{1},\quad
\frac{1}{b},\quad
\frac{1+2a}{1/2},
\end{displaymath}
so these streets have only two different shapes, with ratio
\begin{displaymath}
2b(1+2a)=3(2r_1\sqrt{d}-2r_2-1)(2r_1\sqrt{d}+2r_2+1)=3(4r_1^2d-(2r_2+1)^2),
\end{displaymath}
which is clearly rational.

Consider next the vertical streets $V_1,\ldots,V_5$ as indicated in the picture on the right in Figure~8.7.1.
Here the ratio of the (vertical) lengths and (horizontal) widths of these streets are respectively
\begin{displaymath}
\frac{1}{a},\quad
\frac{3+2b}{1/2},\quad
\frac{3+2b}{1/2},\quad
\frac{1}{a},\quad
\frac{2}{2a},
\end{displaymath}
so these streets have only two different shapes, with ratio
\begin{displaymath}
2a(3+2b)=12(r_1\sqrt{d}+r_2)(r_1\sqrt{d}-r_2)=12(r_1^2d-r_2^2),
\end{displaymath}
which is also clearly rational.

Thus the surplus shortline method works for any cathedral surface $\cp(a,b)$ that satisfies \eqref{eq8.7.1}.
We therefore conclude that geodesic flow on such surfaces exhibits \textit{time-qualitative} optimality, and also exhibits \textit{time-quantitative} behavior in the sense that there are infinitely many explicit slopes such that geodesics with such slopes are superdense and we can also compute the irregularity exponents.

It is easy to see that the class of cathedral surfaces $\cp(a,b)$ depends only on two rational parameters $q_1=r_1^2d$ and~$r_2$, so this is a
$2$-parameter class.

While we remain very far from the complete classification of all primitive Veech surfaces, we now know all the primitive Veech surfaces that have genus~$2$, through the breakthrough of McMullen \cite{Mc2,Mc3}.

\begin{theoremc}
The affine-different primitive Veech surfaces in genus $2$ are precisely the following:

\emph{(i)}
the infinite class of Calta--McMullen L-staircase surfaces; and

\emph{(ii)}
the regular decagon surface with parallel edge identification.
\end{theoremc}

The following partial converse to Theorem~B is also due to McMullen~\cite{Mc2}.

\begin{theoremd}
Suppose that a translation surface $S$ in genus $2$ is not a Veech surface.
Then there exist geodesics on $S$ which are neither dense nor periodic.
\end{theoremd}

A different sort of partial converse to Theorem~B is due to Cheung and Masur~\cite{CM}, building on the work of McMullen.

\begin{theoreme}
Suppose that a translation surface $S$ in genus $2$ is not a Veech surface.
Then there exist geodesics on $S$ which are dense but not uniformly distributed.
\end{theoreme}

Recall that the surplus shortline method works for all Veech surfaces, as we need only street-rational decomposition in two different directions, while Veech surfaces exhibit street-rational decomposition in infinitely many different directions.
This makes it plausible to expect that there are many translation surfaces that are not Veech surfaces, or not optimal, but nevertheless have sufficient street-rationality for the surplus shortline method to work, giving rise to time-quantitative results.
Thus the surplus shortline method goes beyond the class of optimal systems.

We next give an example of a $3$-parameter family of non-optimal systems that have sufficient street-rationality for the surplus shortline method to work.
Figure~8.7.2 shows a family of decagons which are affine-different, due to the fixed height $2$ and fixed width~$2$.
These decagons are also symmetric across the horizontal axis and the vertical axis as shown.
Identifying the parallel edges of every such decagon, we obtain a translation surface $S_{10}(a,b,c)$ in genus~$2$.

\begin{displaymath}
\begin{array}{c}
\includegraphics[scale=0.8]{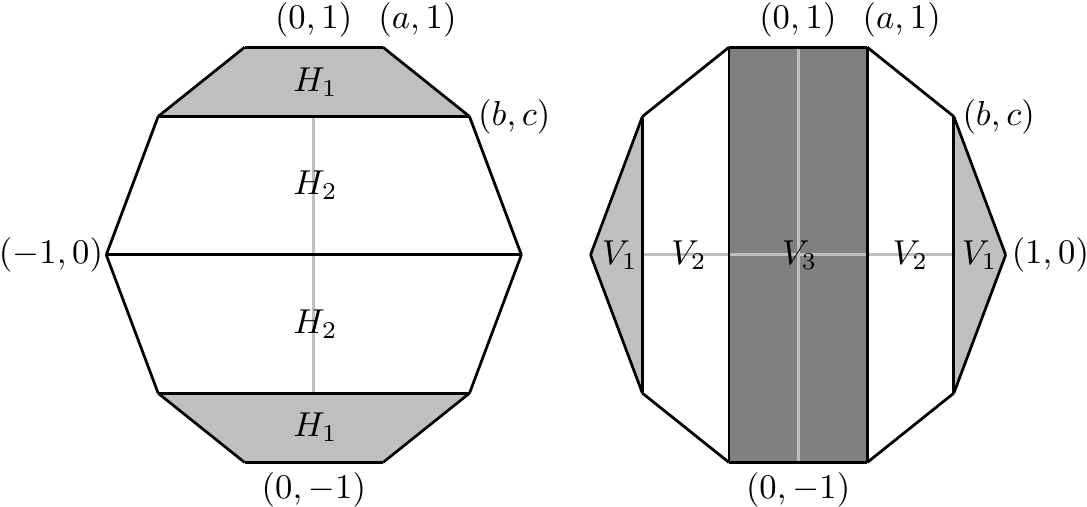}
\vspace{3pt}\\
\mbox{Figure 8.7.2: $S_{10}(a,b,c)$ and its horizontal and vertical streets}
\end{array}
\end{displaymath}

Note first that $S_{10}(a,b,c)$ has $2$ horizontal streets $H_1,H_2$, as shown in the picture on the left in Figure~8.7.2.
Here the ratio of the (horizontal) lengths and (vertical) widths of these streets are respectively
\begin{displaymath}
\frac{2a+2b}{1-c},\quad
\frac{2+2b}{c},
\end{displaymath}
so we have street rationality if there exists a positive rational number $r_1$ such that
\begin{equation}\label{eq8.7.2}
\frac{a+b}{1-c}=r_1\frac{1+b}{c}.
\end{equation}
Note next that $S_{10}(a,b,c)$ has $3$ vertical streets $V_1,V_2,V_3$, as shown in the picture on the right in Figure~8.7.2.
Here the ratio of the (vertical) lengths and (horizontal) widths of these streets are respectively
\begin{displaymath}
\frac{2c}{1-b},\quad
\frac{2+2c}{b-a},\quad
\frac{2}{2a},
\end{displaymath}
so we have street rationality if there exist positive rational numbers $r_2,r_3$ such that
\begin{equation}\label{eq8.7.3}
a=r_2\frac{1-b}{c}
\quad\mbox{and}\quad
a=r_3\frac{b-a}{1+c}.
\end{equation}

Given positive rational numbers $r_1,r_2,r_3$, we have $3$ equations \eqref{eq8.7.2} and \eqref{eq8.7.3} in the $3$ variables $a,b,c$.
It follows that there exist infinitely many affine-different non-regular decagons $S_{10}(a,b,c)$ that have street-rational decompositions in the horizontal and vertical directions.
Every one of these non-regular decagons is primitive and has genus~$2$, differs from the Veech surfaces listed in Theorem~C, and is therefore not optimal.
It then follows from Theorem~D that each has geodesics which are neither dense nor periodic.
It also follows from Theorem~E that each has geodesics which are dense but not uniformly distributed.
However, it also follows from the surplus shortline method that each has geodesics that exhibit superdensity and quantitative equidistribution with explicit irregularity exponents.
These vastly different types of orbits that can arise demonstrate the intriguing fact that we have a wide spectrum of possibilities.

Unfortunately, we currently do not have analogs of Theorems C and~D when the genus exceeds~$2$, and this prevents us from proving plausible conjectures about more infinite families of surfaces that are not optimal but have sufficient street-rationality for the surplus shortline method to work.
Nevertheless, we are of the opinion that the class of surfaces that have sufficient street-rationality for the surplus shortline method to work is \textit{much}, \textit{much} larger than the class of surfaces with optimal dynamics.

%%%%%%%%%%
%
% SECTION 8.8
%
%%%%%%%%%%

\subsection{Density in generalized mazes}\label{sec8.8}

In this section, we switch to infinite flat surfaces, and begin by generalizing the concept of square-maze translation surfaces first introduced in \cite[Section~6.5]{BCY}.

It is natural to start the discussion here with arguably the simplest example of generalized square-maze translation surfaces, the so-called \textit{infinite halving staircase surface} $S(\infty;1/2)$, arising from the \textit{infinite halving staircase region} $R(\infty;1/2)$ shown in Figure~8.8.1.

\begin{displaymath}
\begin{array}{c}
\includegraphics[scale=0.8]{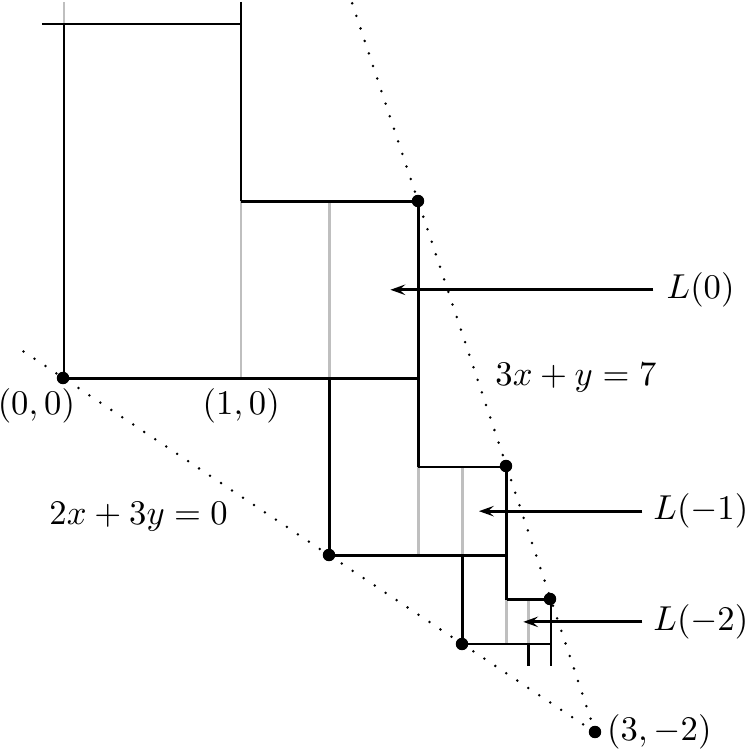}
\vspace{3pt}\\
\mbox{Figure 8.8.1: three halving L-shapes $L(0),L(-1),L(-2)$ of $R(\infty;1/2)$}
\end{array}
\end{displaymath}

The building blocks of $R(\infty;1/2)$ are L-shapes $L(n)$, $n\in\Zz$, that consist of three squares of the same size $2^n\times2^n$.
These L-shapes of different sizes are glued together as shown.
Thus
\begin{displaymath}
R(\infty;1/2)=\bigcup_{-\infty<n<\infty}L(n).
\end{displaymath}
For reference, we take the bottom left vertex of $L(0)$ to be the origin $(0,0)$.

Note that as we move from $L(i)$ to $L(i+1)$, we move up and there is doubling, and that as we move from $L(i)$ to $L(i-1)$, we move down and there is halving.

Moving towards the bottom right of the staircase region $R(\infty;1/2)$, we converge to the limit point $(3,-2)$.
On the other hand, moving towards the top left of the staircase region $R(\infty;1/2)$, we approach $(-\infty,\infty)$, and the region is bounded between the two lines $3x+y=7$ and $2x+3y=0$.

To obtain the desired closed surface $S(\infty;1/2)$ from the infinite region $R(\infty;1/2)$, we have to identify pairs of boundary edges of $R(\infty;1/2)$.
The edge identification comes from the simplest perpendicular translation, both horizontal and vertical, as shown in Figure~8.8.2.

We have indicated all the horizontal and vertical boundary edges of~$L(i)$.

The vertical boundary edge identification is simple, and comprises a pair of identified edges $v'_i$ on the top square of $L(i)$ and a pair of identified edges $v''_i$ on the two bottom squares of~$L(i)$.

The horizontal edge identification is a little less straightforward.
First of all, there is a pair of identified edges $h'_i$ on the left half of the right square of~$L(i)$.
Then there are two further horizontal boundary edges.
The edge $h''_i$ on the top of the right half of the right square of $L(i)$ is identified with a horizontal edge of the L-shape $L(i-1)$ below, while the edge $h''_{i+1}$ on the bottom of the bottom left square of $L(i)$ is identified with a horizontal edge of the L-shape $L(i+1)$ above.

\begin{displaymath}
\begin{array}{c}
\includegraphics[scale=0.8]{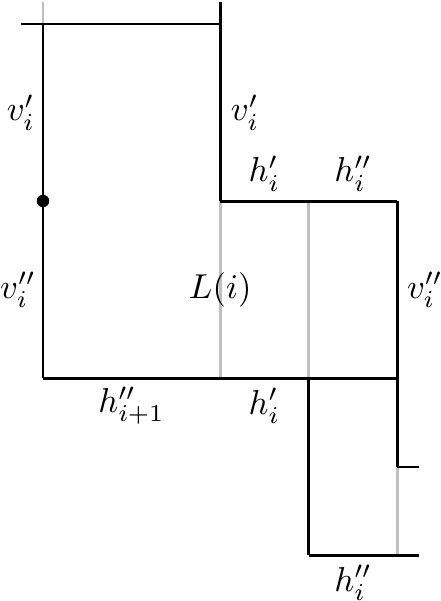}
\vspace{3pt}\\
\mbox{Figure 8.8.2: identifying pairs of boundary edges of $R(\infty;1/2)$}
\end{array}
\end{displaymath}

Carrying out this boundary edge identification pattern throughout on $R(\infty;1/2)$, we obtain a closed surface.
If it is equipped with a flat metric, \textit{i.e.}, every square has zero curvature, then it becomes the desired flat infinite closed surface $S(\infty;1/2)$.

Recall that in a finite or infinite polysquare translation surface, we have building blocks that are congruent unit size squares.
This property is clearly violated in $S(\infty;1/2)$, since the building blocks can be arbitrarily large and arbitrarily small.
This is why we refer to $S(\infty;1/2)$ as a \textit{generalized maze surface} -- more about this later.

The infinite surface $S(\infty;1/2)$ can be interpreted as a limit of a sequence of compact surfaces $S(n;1/2)$, $n=0,1,2,3,\ldots.$
For an arbitrary integer $n\ge0$, the surface $S(n;1/2)$ is defined by closing the open surface
\begin{displaymath}
\bigcup_{-n\le i\le n}L(i),
\end{displaymath}
which is a union of $2n+1$ L-shapes. 
Here \textit{closing} means setting up an appropriate pattern for identifying pairs of parallel boundary edges.

We illustrate the pattern in the special cases $n=0$ and $n=1$ in Figure~8.8.3.
The general case goes in a similar way.

\begin{displaymath}
\begin{array}{c}
\includegraphics[scale=0.8]{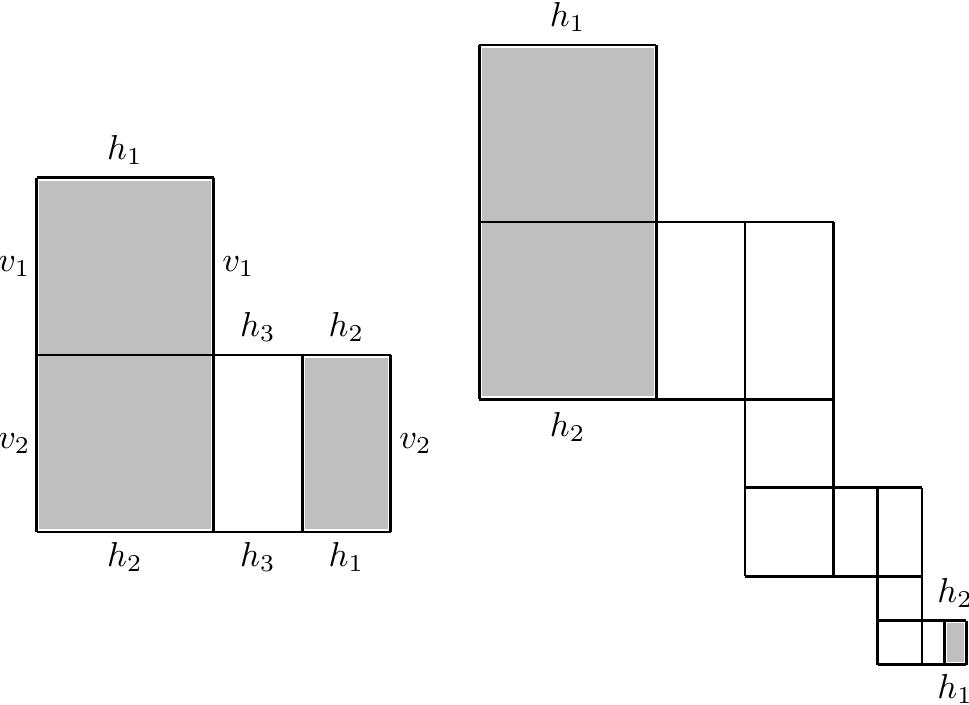}
\vspace{3pt}\\
\mbox{Figure 8.8.3: $S(0;1/2)$ and $S(1;1/2)$ (not to scale)}
\end{array}
\end{displaymath}

In the picture of $S(0;1/2)$ on the left, the edges $v_1$, $v_2$ and $h_3$ correspond respectively to $v'_0$, $v''_0$ and $h'_0$ in Figure~8.8.2.
The only novelty comes from the two copies of $h_1$ and the two copies of~$h_2$, as these identified edges do \textit{not} have the same length.
The edge $h_1$ on the left has twice the length of the edge $h_1$ on the right, and the edge $h_2$ on the left has twice the length of the edge $h_2$ on the right.
Thus the identification comes from a translation combined with a dilation or contraction with factor~$2$.

In the picture of $S(1;1/2)$ on the right, we have only indicated the edge pairings $h_1$ and~$h_2$.
The remaining edge pairings follow the pattern set in Figure~8.8.2.
Again, the pairs of edges $h_1$ and $h_2$ do not have the same length.
The edge $h_1$ on the left has $8$ times the length of the edge $h_1$ on the right, and the edge $h_2$ on the left has $8$ times the length of the edge $h_2$ on the right.
Thus the identification comes from a translation combined with a dilation or contraction with factor~$8$.

It is easy to see that for an arbitrary integer $n\ge0$, the identification of the special edges $h_1$ and $h_2$ comes from a translation combined with a dilation or contraction with factor~$2^{2n+1}$.

In view of this dilation or contraction, the bounded surfaces $S(n;1/2)$, $n\ge0$, fall outside of the class of polysquare surfaces.
We say that $S(n;1/2)$, $n\ge0$, belong to the class of \textit{d-c-polysquare surfaces}, where \textit{d-c} refers to the \textit{dilation} and \textit{contraction} of some boundary edges of the polysquare.

It is easy to see that $S(\infty;1/2)$ has infinite genus which, intuitively, represents \textit{infinite complexity of the geodesic flow}.
Indeed, all square-maze translation surfaces have infinite genus.
We have already mentioned this simple fact several times without proof.
For the sake of completeness, we include here a sketch of the proof in the special case of $S(\infty;1/2)$.
It is a routine application of Euler's formula, which gives
\begin{displaymath}
2-2g_i=V_i-E_i+R_i,
\end{displaymath}
where $g_i$ denotes the genus of $S(i;1/2)$, whereas $V_i,E_i,R_i$ denote respectively the numbers of vertices, edges and regions of $S(i;1/2)$ after boundary edge identification.
Using Figure~8.8.3, we see that
\begin{align}
2-2g_0
&
=V_0-E_0+R_0=2-8+4=-2,
\nonumber
\\
2-2g_1
&
=V_1-E_1+R_1=2-24+12=-10.
\nonumber
\end{align}
Hence the genus of $S(0;1/2)$ is~$2$, and the genus of $S(1;1/2)$ is~$6$.
Similarly, one can show that for an arbitrary integer $n\ge0$, the genus of $S(n;1/2)$ is~$4n+2$.
Since $S(\infty;1/2)$ is a limit of $S(n;1/2)$ as $n\to\infty$, we conclude that $S(\infty;1/2)$ has infinite genus.

Since $S(\infty;1/2)$ is a flat surface, its geodesics consist of parallel straight line segments, and we have a $1$-direction flow.
We may call  $S(0;1/2)$, shown in the picture on the left in Figure~8.8.3, the period-surface of $S(\infty;1/2)$.
Clearly it is a d-c-polysquare translation surface.

Figure~8.8.4 shows the horizontal and vertical streets of the surface $S(\infty;1/2)$.

\begin{displaymath}
\begin{array}{c}
\includegraphics[scale=0.8]{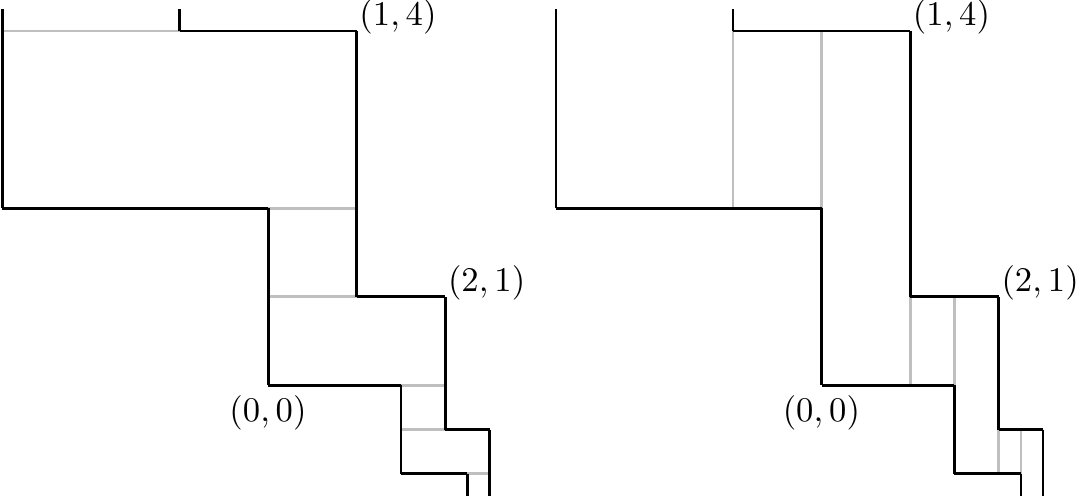}
\vspace{3pt}\\
\mbox{Figure 8.8.4: the horizontal and vertical streets of $S(\infty;1/2)$}
\end{array}
\end{displaymath}

It is clear that the normalized lengths of horizontal streets are $1$ and~$2$, while the normalized lengths of vertical streets are $2$ and~$4$.
Thus the normalized horizontal street-LCM is equal to~$2$, while the normalized vertical street-LCM is equal to~$4$.

With a view to using the surplus shortline method on $S(\infty;1/2)$, we now focus on two particular geodesics $V_*(t)$ and $H_*(t)$ on
$S(\infty;1/2)$, shown in Figure~8.8.5. 

\begin{displaymath}
\begin{array}{c}
\includegraphics[scale=0.8]{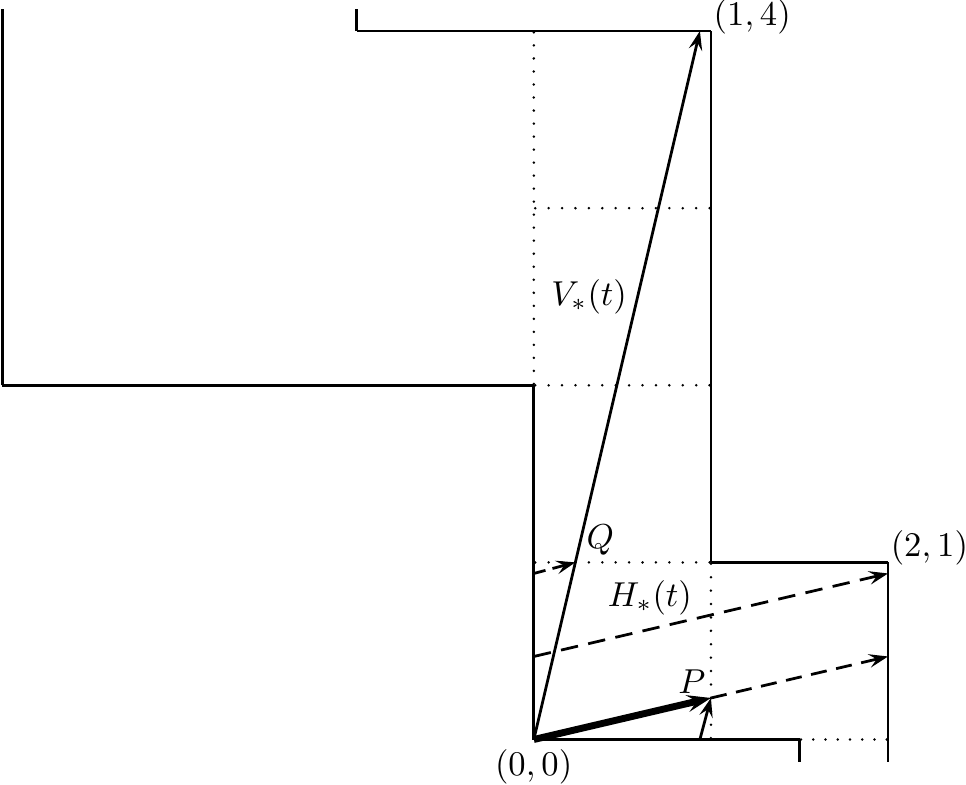}
\vspace{3pt}\\
\mbox{Figure 8.8.5: two particular geodesics of $S(\infty;1/2)$}
\end{array}
\end{displaymath}

The almost vertical geodesic $V_*(t)$, $t\ge0$, starts from the origin and has special slope
\begin{equation}\label{eq8.8.1}
\alpha=2+\sqrt{5}=4+\frac{1}{4+\frac{1}{4+\cdots}},
\end{equation}
where the continued fraction digits are all equal to~$4$.
The almost horizontal geodesic $H_*(t)$, $t\ge0$, starts from the origin and has the reciprocal slope $\alpha^{-1}=\sqrt{5}-2$.

Note that $V_*(0)=H_*(0)$ is the origin and, as usual, the parameter $t$ denotes time; in other words, for both $V_*(t)$ and $H_*(t)$, we use the
arc-length parametrization, which represents unit-speed motion of a particle.

The key fact is that $V_*(t)$ and $H_*(t)$ are shortlines of each other.

Indeed, the initial segment of $V_*(t)$, $t\ge0$, between the origin and the point $P$ is a detour crossing of a vertical street of $4$ building block squares, and the initial segment of $H_*(t)$, $t\ge0$, between the origin and the point $P$ is the shortcut.

Similarly, the initial segment of $H_*(t)$, $t\ge0$, between the origin and the point $Q$ is a detour crossing of a horizontal street of $2$ building block squares, and the initial segment of $V_*(t)$, $t\ge0$, between the origin and the point $Q$ is the shortcut.

The \textit{mutual shortcut} property of the special geodesics $V_*(t)$, $t>0$, and $H_*(t)$, $t>0$, means that for every vertical street of $S(\infty;1/2)$, the detour crossing points of $V_*(t)$ coincides with the shortcut crossing points of~$H_*(t)$.
Similarly, for every horizontal street of $S(\infty ;1/2)$, the detour crossing points of $H_*(t)$ coincides with the shortcut crossing points of~$V_*(t)$. 

As a consequence of an intrinsic symmetry within $S(\infty;1/2)$, we can easily list all the different types of shortcuts or units.

Although the surface $S(\infty;1/2)$ is infinite, this symmetry makes it sufficient to distinguish only $8$ different types of almost vertical shortcuts or units of slope~$\alpha$, where the prototypes are $a_i$, $1\le i\le8$, in the picture on the left in Figure~8.8.6.
Similarly, it is sufficient to distinguish only 8 different types of almost horizontal shortcuts of slope~$\alpha^{-1}$, where the prototypes are $b_j$, $1\le j\le8$, in the picture on the right in Figure~8.8.6.

\begin{displaymath}
\begin{array}{c}
\includegraphics[scale=0.8]{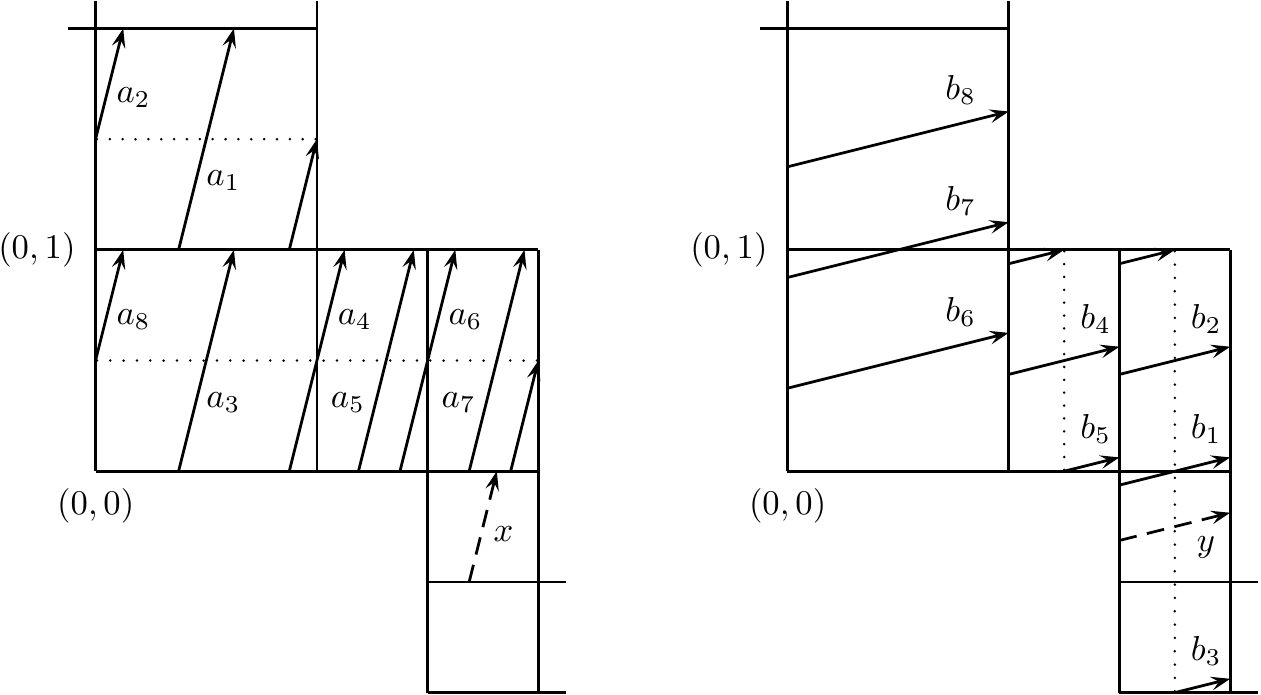}
\vspace{3pt}\\
\mbox{Figure 8.8.6: almost vertical and almost horizontal units of $S(\infty;1/2)$}
\end{array}
\end{displaymath}

We recall that the infinite halving staircase region $R(\infty;1/2)$ is built from similar L-shapes $L(i)$, $i\in\Zz$, where $L(i+1)$ and $L(i-1)$ are respectively obtained from $L(i)$ by doubling and halving.
This gives rise to a transformation~$\bfU$, which represents \textit{up and doubling}, as well as an inverse transformation~$\bfD$, which represents \textit{down and halving}.

Consider a prototype almost vertical unit $a_1$ and another almost vertical unit $x$ in the picture on the left in Figure~8.8.6, where $x$ is the same type of unit as~$a_1$.
The transformation $\bfU$ maps $x$ to~$a_1$, which we denote formally by $a_1=\bfU x$.
The transformation $\bfD$ maps $a_1$ to~$x$, which we denote formally by $x=\bfD a_1$.

Consider a prototype almost horizontal unit $b_8$ and another almost horizontal unit $y$ in the picture on the right in Figure~8.8.6, where $y$ is the same type of unit as~$b_8$.
The transformation $\bfU$ maps $y$ to~$b_8$, which we denote formally by $b_8=\bfU y$.
The transformation $\bfD$ maps $b_8$ to~$y$, which we denote formally by $y=\bfD b_8$.

There are analogs for the remaining prototypes $a_i$, $2\le i\le8$, and $b_j$, $1\le j\le7$.
Moreover, we can take any power of the transformation~$\bfU$, and the symmetry group of $S(\infty;1/2)$ contains the infinite cyclic group
$\{\bfU^k:k\in\Zz\}$.

It is straightforward to determine the ancestor units of the almost vertical units $a_i$, $1\le i\le8$.
We have
\begin{align}
a_1
&
\hookrightarrow\{b_7\},3b_8,\{\bfU b_1\},
\label{eq8.8.2}
\\
a_2
&
\hookrightarrow\{b_7\},4b_8,\{\bfU b_1\},
\label{eq8.8.3}
\\
a_3
&
\hookrightarrow\{\bfU b_3\},2b_2,2b_4,b_6,\{b_7\},
\label{eq8.8.4}
\\
a_4
&
\hookrightarrow\{\bfU b_3\},2b_2,2b_4,2b_6,\{b_5\},
\label{eq8.8.5}
\\
a_5
&
\hookrightarrow\{b_5\},2b_2,b_4,2b_6,\{b_5\},
\label{eq8.8.6}
\\
a_6
&
\hookrightarrow\{b_5\},2b_2,2b_4,2b_6,\{b_3\},
\label{eq8.8.7}
\\
a_7
&
\hookrightarrow\{b_1\},b_2,2b_4,2b_6,\{b_3\},
\label{eq8.8.8}
\\
a_8
&
\hookrightarrow\{b_1\},2b_2,2b_4,2b_6,\{b_7\},
\label{eq8.8.9}
\end{align}
These include a fractional unit at the beginning and a fractional unit at the end, with whole units in the middle.

Here and later, fractional units are indicated by $\{\ldots\}$, and the whole units in the middle are listed lexicographically.

We also determine the ancestor units of the almost horizontal units $b_j$, $1\le j\le8$.
We have
\begin{align}
b_1
&
\hookrightarrow\{\bfD a_2\},\bfD a_1,\bfD a_3,a_7,\{a_8\},
\label{eq8.8.10}
\\
b_2
&
\hookrightarrow\{a_6\},\bfD a_1,\bfD a_3,\{a_8\},
\label{eq8.8.11}
\\
b_3
&
\hookrightarrow\{a_6\},\bfD a_1,\bfD a_3,a_7,\{\bfD a_4\},
\label{eq8.8.12}
\\
b_4
&
\hookrightarrow\{a_4\},a_5,\{a_6\},
\label{eq8.8.13}
\\
b_5
&
\hookrightarrow\{a_4\},2a_5,\{a_6\},
\label{eq8.8.14}
\\
b_6
&
\hookrightarrow\{a_8\},a_1,\bfU a_7,\{a_4\},
\label{eq8.8.15}
\\
b_7
&
\hookrightarrow\{a_8\},a_1,a_3,\bfU a_7,\{a_2\},
\label{eq8.8.16}
\\
b_8
&
\hookrightarrow\{a_2\},a_3,\bfU a_7,\{a_2\}.
\label{eq8.8.17}
\end{align}

The symmetry group of $S(\infty;1/2)$ contains the infinite cyclic group $\{\bfU^k:k\in\Zz\}$.
It follows that by applying arbitrary integer powers of the transformation $\bfU$ to \eqref{eq8.8.2}--\eqref{eq8.8.17}, we obtain their analogs over the whole infinite surface.

The infinite halving staircase surface $S(\infty;1/2)$ is \textit{not} a square-maze translation surface.
However, as Figure~8.8.6, the relations \eqref{eq8.8.2}--\eqref{eq8.8.17} and their analogs illustrate, the key concept of shortline is well defined for some explicit slopes including the particular slope given by \eqref{eq8.8.1}.
We can therefore easily adapt the proof of \cite[Theorem~6.5.1]{BCY} and obtain density of  some explicit geodesics on this infinite surface.

We may therefore say, intuitively speaking, that $S(\infty;1/2)$ is an example of an infinite polysquare translation surface \textit{with bounded-ratio streets formed from square faces with side lengths equal to integer powers of~$2$}.
We now make this intuition precise with the following definition, which includes the infinite halving staircase surface as a special case. 

\begin{2maze}
An infinite closed flat translation surface $S$ is called a $2$-power square-maze translation surface if it satisfies the following four requirements:

(1) The building blocks of $S$ are axis-parallel squares with possibly different side lengths~$2^k$, where $k\in\Zz$.

(2) Any two building block squares that have a common edge have a common vertex, and either they have the same side length, or the side length of one is half the side length of the other.

(3) $S$ does not contain an infinite horizontal or vertical line.
Furthermore, there is an integer $r$ such that any horizontal or vertical line segment fully inside $S$ can be covered by at most $r$ congruent squares fully inside~$S$.

(4) We apply the simplest boundary identification via perpendicular translation, giving rise to $1$-direction geodesic flow in~$S$.
\end{2maze}

The requirement (3) expresses the analogous maze property that the surface $S$ has \textit{bounded-ratio streets}, where the ratio of the long side and the short side of any street is less than an absolute constant. 

If the minimum value of $r$ in (3) is~$r_0$, then we call this a \textit{$2$-power $r_0$-square-maze translation surface}. 
The infinite halving staircase surface $S(\infty;1/2)$ is a an example of a $2$-power $4$-square-maze translation surface.

In a similar way, we can define the class of \textit{$k$-power $r$-square-maze translation surfaces} for every fixed integers $k\ge2$ and $r\ge2$.

It is easy to see that the special symmetry of the infinite halving staircase surface $S(\infty;1/2)$ provided by $\bfU$ and $\bfD$ is not really important.
What is important is that the concept of shortline is well defined for any $k$-power $r$-square-maze translation surface at least for some explicit slopes.
Moreover, the requirement (3) on bounded-ratio streets guarantees that, despite the different sizes of the squares, the magnification process in the shortline method still works.
Thus we can adapt the proof of \cite[Theorem~6.5.1]{BCY} and obtain density of some explicit geodesics on these generalized maze translation surfaces.

Next we return to the ordinary square-maze translation surface of congruent unit square building blocks.
We can obtain a different kind of generalization by replacing the unit square in the definition of the square-maze translation surface with, say, any other fixed regular
$n$-gon, where $n\ge8$ is divisible by~$4$. 
From a supply of infinitely many congruent copies of such a fixed regular $n$-gon, we can form a horizontal-vertical $\Zz^2$-like grid.
Making infinitely many appropriate \textit{holes} in the grid, we can easily enforce the maze property that the lengths of the horizontal and vertical streets are uniformly bounded. 
The maze property then guarantees that we can adapt the proof of \cite[Theorem~6.5.1]{BCY} and obtain density of some explicit geodesics on these
\textit{$n$-gon-maze translation surfaces}.

%%%%%%%%%%
%
% SECTION 8.9
%
%%%%%%%%%%

\subsection{When the flow does not preserve the area}\label{sec8.9}

We now return to the infinite halving staircase surface $S(\infty;1/2)$, with the period-surface $S(0;1/2)$ shown in the picture on the left in Figure~8.8.3, and recall that $S(0;1/2)$ is not a polysquare translation surface, but belongs to the class of d-c-polysquare translation surfaces, where there is boundary edge identification involving dilation or contraction.

We study \textit{$1$-direction line-flow} on the period-surface $S(0;1/2)$, which comes from the \textit{projection} of $1$-direction geodesic flow on the infinite halving staircase surface $S(\infty;1/2)$.

In our investigation on flat systems in \cite{BDY1,BDY2,BCY} and up to Section~\ref{sec8.6} in this paper, we have been studying $2$-dimensional flow that preserves the $2$-dimensional Lebesgue measure.
Since area is homogeneous, it is natural to study whether or not an infinite orbit exhibits uniform distribution on the surface. 
What therefore makes the period-surface $S(0;1/2)$ particularly interesting is that it is our first example of a flat system with a natural $1$-direction
line-flow that is \textit{not} area-preserving, a consequence of the dilation or contraction at the boundary.
Thus we cannot expect uniform distribution of the orbits.
Our purpose in this section is therefore to examine what we can still say about the distribution of an orbit.

Let us return to the picture on the left in Figure~8.8.6.
Since every almost vertical unit $a_i$, $1\le i\le8$, has the same length, it is fairly easy to compute asymptotically the relative time a $1$-direction line-flow with slope
$\alpha$ spends in the top square of $S(0;1/2)$.

Indeed, the visiting time mainly depends on the ratio of the coordinates of an eigenvector corresponding to the largest eigenvalue $\Lambda$ of the $2$-step transition matrix $\bfA$ of the shortcut-ancestor process, assuming of course that the shortline method works for the slope~$\alpha$.

More precisely, the relative time a $1$-direction line-flow with slope $\alpha$ spends in the top square of $S(0;1/2)$ is asymptotically the ratio
\begin{equation}\label{eq8.9.1}
\frac{v(1)+v(2)}{v(1)+\ldots+v(8)},
\end{equation}
where $(v(1),\ldots,v(8))^T$ is an eigenvector corresponding to the largest eigenvalue $\Lambda$ of $\bfA$ where, for $j=1,\ldots,8$, $v(j)$ denotes the coefficient of the almost vertical unit~$a_j$.
Here we make the usual assumption that the $1$-direction line-flow moves with constant speed. 

We can generalize the slope $\alpha$ given by \eqref{eq8.8.1} to any slope of the form
\begin{equation}\label{eq8.9.2}
\alpha
=[n;m,n,m,\ldots]
=n+\frac{1}{m+\frac{1}{n+\frac{1}{m+\cdots}}}
=\frac{n}{2}\left(1+\sqrt{1+\frac{4}{mn}}\right),
\end{equation}
where both $n,m\ge4$ are divisible by~$4$.
This is motivated by the observation that for $S(\infty;1/2)$ and the period-surface $S(0;1/2)$, the normalized lengths of horizontal streets are $1$ and~$2$, while the normalized lengths of vertical streets are $2$ and~$4$.
Thus the surplus shortline method works for these slopes.

Indeed, it is clear from the picture on the left in Figure~8.8.3 that $S(0;1/2)$ has $2$ horizontal streets, where the top street consists of $1$ square, and the bottom street consists of $2$ squares.
It also has $2$ vertical streets.
One of these is the white rectangle with top and bottom edges~$h_3$.
The more complicated one is shaded, with the left part consisting of $2$ squares with side length~$1$, and the right part consisting of $2$ squares with side length~$1/2$.

The good news is that, despite the fact that $S(0;1/2)$ is not a polysquare surface, we can still apply Theorem~\ref{thm7.2.2}.
The pessimistic reader may verify that this result can indeed be extended to d-c-polysquare surfaces.

Using this, we now attempt to find the eigenvalues of the $2$-step transition matrix~$\bfA$.
For simplicity of notation, we number the various parts of $S(0;1/2)$ as in Figure~8.9.1.
Here the horizontal streets are $1$ and $2,3,4$, while the vertical streets are $1,2,4$ and~$3$.
We also show the almost vertical units of type~$\uparrow$, so that $\uparrow_1,\uparrow_2,\uparrow_3,\uparrow_4$ are respectively $a_1,a_3,a_5,a_7$ in Figure~8.8.6.
We have not shown the almost vertical units of type~$\nuparrow$, but $\nuparrow_1,\nuparrow_2,\nuparrow_3,\nuparrow_4$ are respectively $a_2,a_4,a_6,a_8$.
Let
\begin{displaymath}
J_1=\{1\}
\quad\mbox{and}\quad
J_2=\{2,3,4\}
\end{displaymath}
denote the horizontal streets, and let
\begin{displaymath}
I_1=I_2=I_4=\{1,2,4\}
\quad\mbox{and}\quad
I_3=\{3\}
\end{displaymath}
denote the vertical streets.

\begin{displaymath}
\begin{array}{c}
\includegraphics[scale=0.8]{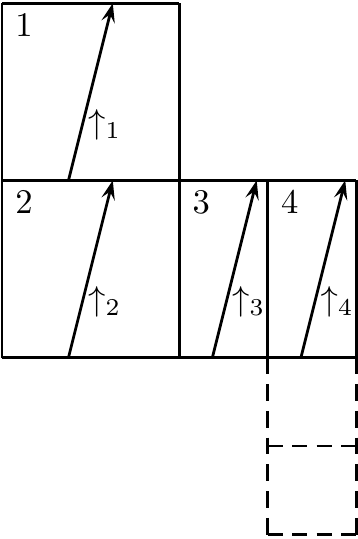}
\vspace{3pt}\\
\mbox{Figure 8.9.1: $S(0;1/2)$ and almost vertical units of type $\uparrow$}
\end{array}
\end{displaymath}

We consider slopes of the form \eqref{eq8.9.2}.

Corresponding to \eqref{eq7.2.14}, we define the column matrices
\begin{equation}\label{eq8.9.3}
\bfu_1=[\{\uparrow_s:j\in J_1^*,s\in I_j^*\}]
\quad\mbox{and}\quad
\bfu_2=[\{\uparrow_s:j\in J_2^*,s\in I_j^*\}].
\end{equation}
Here $J_1^*$, $J_2^*$ and $I_j^*$ denote that the edges are counted with multiplicity.
Corresponding to \eqref{eq7.2.15}, we also define the column matrices
\begin{align}
\bfv_1
&=[\{\nuparrow_j:j\in J_1^*\}]-[\{\uparrow_j:j\in J_1^*\}],
\label{eq8.9.4}
\\
\bfv_2
&=[\{\nuparrow_j:j\in J_2^*\}]-[\{\uparrow_j:j\in J_2^*\}].
\label{eq8.9.5}
\end{align}
Also, analogous to \eqref{eq7.2.17}, we have
\begin{equation}\label{eq8.9.6}
(\bfA-I)[\{\uparrow_s\}]=\left\{\begin{array}{ll}
\bfu_1+\bfv_1,&\mbox{if $s\in J_1$},\\
\bfu_2+\bfv_2,&\mbox{if $s\in J_2$}.
\end{array}\right.
\end{equation}

We now combine \eqref{eq8.9.3} and \eqref{eq8.9.6}.
For the horizontal street corresponding to~$\bfu_1$, as highlighted in the picture on the left in Figure~8.9.2, we have
\begin{align}\label{eq8.9.7}
(\bfA-I)\bfu_1
&
=(\bfA-I)\left[\left\{\frac{mn}{4}\uparrow_1\right\}\right]
+(\bfA-I)\left[\left\{\frac{mn}{4}\uparrow_2,\frac{mn}{4}\uparrow_4\right\}\right]
\nonumber
\\
&
=\frac{mn}{4}(\bfu_1+\bfv_1)+\frac{mn}{2}(\bfu_2+\bfv_2).
\end{align}
For the horizontal street corresponding to~$\bfu_2$, as highlighted in the picture on the right in Figure~8.9.2, we have
\begin{align}\label{eq8.9.8}
(\bfA-I)\bfu_2
&
=(\bfA-I)\left[\left\{\frac{mn}{4}\uparrow_1\right\}\right]
+(\bfA-I)\left[\left\{\frac{mn}{4}\uparrow_2,\frac{mn}{4}\uparrow_3,\frac{mn}{4}\uparrow_4\right\}\right]
\nonumber
\\
&
=\frac{mn}{4}(\bfu_1+\bfv_1)+\frac{3mn}{4}(\bfu_2+\bfv_2).
\end{align}
\begin{displaymath}
\begin{array}{c}
\includegraphics[scale=0.8]{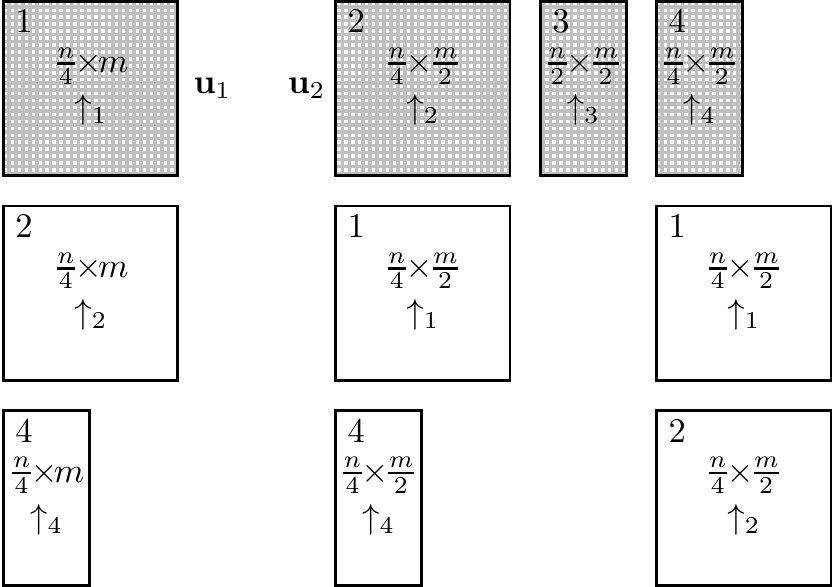}
\vspace{3pt}\\
\mbox{Figure 8.9.2: almost vertical units of type $\uparrow$ in $\bfu_1$ and $\bfu_2$}
\end{array}
\end{displaymath}

It follows from \eqref{eq8.9.7} and \eqref{eq8.9.8} that the street-spreading matrix is given by
\begin{displaymath}
\bfS=\frac{mn}{4}\begin{pmatrix}
1&1\\
2&3
\end{pmatrix},
\end{displaymath}
with eigenvalues
\begin{equation}\label{eq8.9.9}
\tau_1=\frac{(2+\sqrt{3})mn}{4}
\quad\mbox{and}\quad
\tau_2=\frac{(2-\sqrt{3})mn}{4},
\end{equation}
and eigenvector
\begin{displaymath}
\psi_1=\left(\frac{\sqrt{3}-1}{2},1\right)^T
\end{displaymath}
corresponding to~$\tau_1$.
By \eqref{eq7.2.38} and \eqref{eq7.2.39}, the largest eigenvalue of $\bfA\vert_\VVV$ is
\begin{displaymath}
\Lambda=1+\frac{\tau_1+\sqrt{\tau_1^2+4\tau_1}}{2},
\end{displaymath}
with eigenvector
\begin{displaymath}
\Psi=(y,1,xy,x)^T,
\end{displaymath}
where
\begin{equation}\label{eq8.9.10}
y=\frac{\sqrt{3}-1}{2}
\quad\mbox{and}\quad
x=\frac{-\tau_1+\sqrt{\tau_1^2+4\tau_1}}{2}.
\end{equation}
We know that $\Lambda$ is the largest eigenvalue of~$\bfA$.

Suppose that a corresponding eigenvector is given by
\begin{equation}\label{eq8.9.11}
\bfV=(v(1),\ldots,v(8))^T,
\end{equation}
where $v(i)$, $i=1,\ldots,8$, denote respectively the coefficients of $a_i$, $i=1,\ldots,8$, also represented in alternative form respectively by
\begin{displaymath}
\uparrow_1,\nuparrow_1,\uparrow_2,\nuparrow_2,\uparrow_3,\nuparrow_3,\uparrow_4,\nuparrow_4.
\end{displaymath}
Then
\begin{equation}\label{eq8.9.12}
\bfV=y\bfu_1+\bfu_2+xy\bfv_1+x\bfv_2.
\end{equation}
From \eqref{eq8.9.3}--\eqref{eq8.9.5}, we have
\begin{align}
\bfu_1
&=\left[\left\{\frac{mn}{4}\uparrow_1,\frac{mn}{4}\uparrow_2,\frac{mn}{4}\uparrow_4\right\}\right],
\nonumber
\\
\bfu_2
&=\left[\left\{\frac{mn}{4}\uparrow_1,\frac{mn}{4}\uparrow_2,\frac{mn}{4}\uparrow_3,\frac{mn}{4}\uparrow_4\right\}\right],
\nonumber
\\
\bfv_1
&=[\{m\,\nuparrow_1\}]-[\{m\uparrow_1\}],
\nonumber
\\
\bfv_2
&=\left[\left\{\frac{m}{2}\,\nuparrow_2,\frac{m}{2}\,\nuparrow_3,\frac{m}{2}\,\nuparrow_4\right\}\right]
-\left[\left\{\frac{m}{2}\uparrow_2,\frac{m}{2}\uparrow_3,\frac{m}{2}\uparrow_4\right\}\right],
\nonumber
\end{align}
so that
\begin{align}
\bfu_1
&=\left(\frac{mn}{4},0,\frac{mn}{4},0,0,0,\frac{mn}{4},0\right)^T,
\label{eq8.9.13}
\\
\bfu_2
&=\left(\frac{mn}{4},0,\frac{mn}{4},0,\frac{mn}{4},0,\frac{mn}{4},0\right)^T,
\label{eq8.9.14}
\\
\bfv_1
&=(-m,m,0,0,0,0,0,0)^T,
\label{eq8.9.15}
\\
\bfv_2
&=\left(0,0,-\frac{m}{2},\frac{m}{2},-\frac{m}{2},\frac{m}{2},-\frac{m}{2},\frac{m}{2}\right)^T.
\label{eq8.9.16}
\end{align}
Combining \eqref{eq8.9.11}--\eqref{eq8.9.16}, we conclude that
\begin{align}
&
v(1)=\frac{(y+1)mn}{4}-xym,
\quad
v(2)=xym,
\label{eq8.9.17}
\\
&
v(3)=v(7)=\frac{(y+1)mn}{4}-\frac{xm}{2},
\label{eq8.9.18}
\\
&
v(5)=\frac{mn}{4}-\frac{xm}{2},
\quad
v(4)=v(6)=v(8)=\frac{xm}{2}.
\label{eq8.9.19}
\end{align}
It follows from \eqref{eq8.9.17}--\eqref{eq8.9.19} that
\begin{equation}\label{eq8.9.20}
v(1)+v(2)=\frac{(y+1)mn}{4}
\quad\mbox{and}\quad
v(1)+\ldots+v(8)=\frac{(3y+4)mn}{4},
\end{equation}
so that the ratio \eqref{eq8.9.1} is equal to
\begin{equation}\label{eq8.9.21}
\frac{v(1)+v(2)}{v(1)+\ldots+v(8)}=\frac{y+1}{3y+4}=2-\sqrt{3},
\end{equation}
in view of \eqref{eq8.9.10}.

Recall that the ratio \eqref{eq8.9.21} is asymptotically the relative time a $1$-direction line-flow in $S(0;1/2)$ with slope $\alpha$ given by \eqref{eq8.9.2} spends in the top square.
Its value of $2-\sqrt{3}\approx0.268$ is independent of the choice of the integer parameters $m$ and~$n$, and crucially it is less than~$1/3$.
Thus any such $1$-direction line-flow cannot be uniform in $S(0;1/2)$, and the top square is \textit{under-visited}.

Consider next the left half-square of the right square of $S(0;1/2)$, denoted by the label $3$ in Figure~8.9.1.
We wish to study the relative time a $1$-direction line-flow in $S(0;1/2)$ with slope $\alpha$ given by \eqref{eq8.9.2} spends in this half-square.

Note first of all from the edge identification given in the picture on the left in Figure~8.8.3 that every time a $1$-direction line-flow enters this
half-square, it spends a constant time $t_1$ in it before leaving.
On the other hand, it is clear from the picture on the left in Figure~8.8.6 that the only way that such $1$-direction line-flow can enter this half-square is along an almost vertical unit of type~$a_4$.
Thus the total time the line-flow spends in this half-square is equal to~$t_1v(4)$.

To evaluate~$t_1$, we note that from the edge identification given in the picture on the left in Figure~8.8.3 that every time a $1$-direction line-flow enters the top square, it spends a constant time $t_2$ in it before leaving.
Furthermore, $t_2$ is equal to the length of an almost vertical unit.
Simple geometric consideration now shows that $t_1/t_2=\alpha/2$.

Finally, note that the total time of the line-flow is equal to $t_2(v(1)+\ldots+v(8))$, which is the product of the total number of almost vertical units and the length of such a unit.
Thus the relative time a $1$-direction line-flow in $S(0;1/2)$ with slope $\alpha$ given by \eqref{eq8.9.2} spends in the half-square under consideration is equal to
\begin{equation}\label{eq8.9.22}
\frac{t_1v(4)}{t_2(v(1)+\ldots+v(8))}
=\frac{\alpha x}{(3y+4)n}
=(3\sqrt{3}-5)\frac{1+\sqrt{1+\frac{4}{mn}}}{1+\sqrt{1+\frac{16}{(2+\sqrt{3})mn}}},
\end{equation}
in view of \eqref{eq8.9.2}, \eqref{eq8.9.9}, \eqref{eq8.9.10}, \eqref{eq8.9.19} and \eqref{eq8.9.20}.

Taking the limit $m,n\to\infty$ in \eqref{eq8.9.22}, the relative time a $1$-direction line-flow in $S(0;1/2)$ with slope $\alpha$ given by \eqref{eq8.9.2} spends in the half-square under consideration converges to $3\sqrt{3}-5\approx0.196$.
In fact, easy computation shows that the value of \eqref{eq8.9.22} is greater than $0.195$ for every choice $m,n\ge4$ of the parameters.
Crucially, it is greater than $1/6$.
Thus this half-square is \textit{over-visited}.

We summarize our observations in the form of a theorem.

\begin{thm}\label{thm8.9.1}
Let $m,n\ge4$ be arbitrary integers such that both are divisible by~$4$, and let $\alpha=\alpha(m,n)$ be the slope given by \eqref{eq8.9.2}. 
Let $L_\alpha(t)$, $t\ge0$, be any half-infinite line-flow with slope $\alpha$ on the d-c-polysquare surface $S(0;1/2)$. 

\emph{(i)}
The relative visiting time of $L_\alpha$ to the top square of $S(0;1/2)$ is asymptotically equal to $2-\sqrt{3}\approx0.268$, which is
independent of the choice of the integer parameters $m$ and~$n$.
In particular, the top square is under-visited.
Furthermore, in an arbitrary time interval $0\le t\le T$ with $T\ge1$, the actual visiting time of the line-flow to the top square is equal to
\begin{displaymath}
(2-\sqrt{3})T+O(T^{\kappa_0}),
\quad\mbox{where}\quad
\kappa_0=\frac{\log\vert\Lambda_2\vert}{\log\vert\Lambda_1\vert}.
\end{displaymath}
Here $\Lambda_1$ and $\Lambda_2$ are respectively the eigenvalues of the $2$-step transition matrix of the line-flow with the largest and second largest absolute value.

\emph{(ii)}
The relative visiting time of $L_\alpha$ to the left half of the right square of $S(0;1/2)$ is asymptotically equal to \eqref{eq8.9.22} which depends on the choice of the integer parameters $m$ and~$n$.
In particular, this part of $S(0;1/2)$ is over-visited.
The error term of the actual visiting time remains the same $O(T^{\kappa_0})$.
\end{thm}

%%%%%%%%%%
%
% REFERENCES
%
%%%%%%%%%%

\end{document}